\documentclass[preprint]{imsart}
\usepackage{amsmath}
\usepackage{amsfonts}
\usepackage{amssymb}
\usepackage{graphicx}
\usepackage{graphicx}
\usepackage{subfigure}
\usepackage{algorithm}
\usepackage{tikz}
\usepackage{pgfplots}
\usepackage{url}
\usepackage{afterpage}
\providecommand{\U}[1]{\protect\rule{.1in}{.1in}}
\usepackage[round,authoryear]{natbib}
\newtheorem{theorem}{Theorem}

\newtheorem{definition}[theorem]{Definition}

\newtheorem{lemma}[theorem]{Lemma}

\newtheorem{proposition}[theorem]{Proposition}
\newtheorem{remark}[theorem]{Remark}

\newenvironment{proof}[1][Proof]{\noindent\textbf{#1.} }{\ \rule{0.5em}{0.5em}}

\makeatletter
\newcommand\footnoteref[1]{\protected@xdef\@thefnmark{\ref{#1}}\@footnotemark}
\makeatother

\let\cite\citep
\graphicspath{{./}{./figures/}{../figures/}{../figures/simus/}{../figures/real/}}
\usetikzlibrary{arrows}
\usetikzlibrary{positioning}
\usetikzlibrary{fit}
\usetikzlibrary{calc}
\tikzstyle{A}=[circle,
thick,
minimum size=.5cm,
draw=blue!80,
fill=blue!20
]
\tikzstyle{B}=[circle,
thick,
minimum size=.5cm,
draw=orange!80,
fill=orange!25]

\DeclareMathOperator{\CRM}{CRM}

\DeclareMathOperator*{\Gam}{Gamma}

\DeclareMathOperator{\Poi}{Poisson}
\DeclareMathOperator{\Ber}{Ber}

\DeclareMathOperator{\Unif}{Unif}

\DeclareMathOperator{\PK}{PK}

\arxiv{1401.1137}
\begin{document}
	
\begin{frontmatter}

\title{Sparse graphs using exchangeable random measures}
\runtitle{Bayesian nonparametric random graphs}

\begin{aug}
\author{\fnms{Fran\c{c}ois} \snm{Caron}\thanksref{t1}\ead[label=e1]{caron@stats.ox.ac.uk}}
\and
\author{\fnms{Emily B.} \snm{Fox}\thanksref{t2}\ead[label=e2]{ebfox@stat.washington.edu}}

\thankstext{t1}{FC acknowledges the support of the European Commission under the Marie Curie Intra-European
Fellowship Programme.}
\thankstext{t2}{EBF was supported in part by DARPA Grant FA9550-12-1-0406 negotiated by AFOSR and AFOSR Grant FA9550-12-1-0453.}
\runauthor{F. Caron and E. Fox}

\affiliation{University of Oxford and University of Washington}

\address{Department of Statistics\\
University of Oxford\\
1 South Parks Road\\
Oxford, OX1 3TG United Kingdom\\
\printead{e1}}

\address{Department of Statistics\\
University of Washington\\
Box 354322\\
Seattle, WA 98195-4322\\
\printead{e2}}

\end{aug}

\begin{abstract}
Statistical network modeling has focused on representing the graph as a discrete structure, namely the adjacency matrix, and considering the exchangeability of this array.  In such cases, the Aldous-Hoover representation theorem~\cite{Aldous1981,Hoover1979} applies and informs us that the graph is necessarily either dense or empty.  In this paper, we instead consider representing the graph as a measure on $\mathbb{R}_+^2$. For the associated definition of exchangeability in this continuous space, we rely on the Kallenberg representation theorem~\cite{Kallenberg2005}. We show that for certain choices of such exchangeable random measures underlying our graph construction, our network process is sparse with power-law degree distribution.  In particular, we build on the framework of completely random measures (CRMs) and use the theory associated with such processes to derive important network properties, such as an urn representation for our analysis and network simulation.
Our theoretical results are explored empirically and compared to common network models. We then present a Hamiltonian Monte Carlo algorithm for efficient exploration of the posterior distribution and demonstrate that we are able to recover graphs ranging from dense to sparse---and perform associated tests---based on our flexible CRM-based formulation. We explore network properties in a range of real datasets, including Facebook social circles, a political blogosphere, protein networks, citation networks, and world wide web networks, including networks with hundreds of thousands of nodes and millions of edges.
\end{abstract}

\begin{keyword}[class=MSC]
\kwd[Primary ]{62F15}
\kwd{05C80}
\kwd[; secondary ]{60G09}
\kwd{60G51}
\kwd{60G55}
\end{keyword}

\begin{keyword}
\kwd{random graphs}
\kwd{L\'evy measure}
\kwd{point process}
\kwd{exchangeability}
\kwd{generalized gamma process}
\end{keyword}

\end{frontmatter}
\newpage
\tableofcontents
\newpage
\section{Introduction}

\label{sec:intro}

The rapid increase in the availability and importance of network data has been
a driving force behind the significant recent attention on random graph
models. This effort builds on a long history, with a popular early model being
the Erd\" os R\'{e}nyi random graph~\citep{Erdos1959}. However, the Erd\" os
R\'{e}nyi formulation has since been dismissed as overly simplistic since it
fails to capture important real-world network properties.
A plethora of other network models have been proposed in recent
years, with some overviews of such models provided in 
\cite{Newman2003,Newman2009,Bollobas2001,Durrett2007,Goldenberg2010, Fienberg2012}.

In many scenarios, it is appealing conceptually to assume that the order in which nodes are observed is of no importance~\cite{BickelChen2009,Hoff2009}. In statistical network models, this equates with the notion of exchangeability.  Classically, the graph has been represented by a
discrete structure, or \emph{adjacency matrix}, $Z$ where $Z_{ij}$ is a binary
variable with $Z_{ij}=1$ indicating an edge from node $i$ to node $j$. In the
case of undirected graphs, we furthermore restrict $Z_{ij}=Z_{ji}$. For generic matrices $Z$ in some space $\mathbf{Z}$, an (infinite)
\emph{exchangeable random array}~\cite{Diaconis2008,Lauritzen2008} is one such that
\begin{equation}
(Z_{ij})\overset{d}{=}(Z_{\pi(i)\sigma(j)})\text{ for }(i,j)\in\mathbb{N}^{2}%
\label{eq:sequenceexchangeability}
\end{equation}
for any permutation $\pi,\sigma$ of $\mathbb{N}$, with $\pi=\sigma$ in the
jointly exchangeable case.

The celebrated Aldous-Hoover theorem~\cite{Aldous1981,Hoover1979} states that infinite exchangeability
implies a mixture model representation for the matrix involving transformations of uniform random variables (see Theorem~\ref{thm:AldousHoover}).
For undirected graphs, this transformation is specified by the \emph{graphon}.

The Aldous-Hoover constructive definition has motivated the development of Bayesian statistical models for arrays \cite{Lloyd2012} and many popular network models can be recast in this
framework~\cite{Hoff2002,Nowicki2001,Airoldi2008,KimLescovec2012,Miller2009}. Estimators
of models in this class and their associated properties have been studied  extensively in recent years~\cite{BickelChen2009,BickelChenLevina2011,Rohe2011,ZhaoLevinaZhu2012,Airoldi2013,ChoiWolfe2013}.

However, one unpleasing consequence of the Aldous-Hoover theorem is that graphs represented by an exchangeable
random array are either trivially empty or dense\footnote{Note that we refer to graphs with $\Theta(n^2)$ edges as dense graphs and to graphs  with $o(n^2)$ edges as \emph{sparse graphs}, following the terminology of \citet{BollobasRiordan2009}.}, i.e. the number of edges grows
quadratically with the number of nodes $n$ (see Theorem \ref{theorem:excharraydense}). 
To quote the survey of \citet{Orbanz2015}
\textit{\textquotedblleft the theory also clarifies the limitations of exchangeable
models. It shows, for example, that most Bayesian models of network data are
inherently misspecified."}
 The conclusion is that we cannot have
both exchangeability of the nodes (in the sense of \eqref{eq:sequenceexchangeability}), a cornerstone of Bayesian modeling, and sparse graphs, which is what we observe in the real world~\cite{Newman2009}, especially for large networks.
Several models have been developed which give up exchangeability in order to obtain sparse graphs~\cite{Barabasi1999}.
Alternatively, there is a body of literature that examines rescaling graph properties with network size $n$, leading to sparse graph sequences where each graph is finitely exchangeable \citep{Bollobas2007,BollobasRiordan2009,WolfeOlhede2013,BorgsChayesCohnZhao2014}.  However, any method building on a rescaling-based approach provides a graph distribution, $\pi_n$, that lacks projectivity: marginalizing node $n$ does not yield $\pi_{n-1}$, the distribution on graphs of size $n-1$.

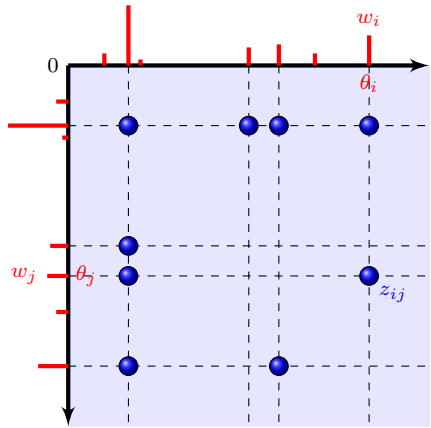
\begin{figure}[t]
\begin{tikzpicture}[node distance=1.4cm, auto,>=latex',scale=.8]
\draw[->,ultra thick] node[left] {0} (0,0) -- (6,0) node[right] {};
\draw[->,ultra thick] (0,0) -- (0,-6) node[right] {};
\draw[fill=blue,opacity=0.1] (0,0) -- (0,-6) -- (6,-6) -- (6,0)  ;
\draw[dashed] (1,0) -- (1,-6) node[right] {};
\draw[dashed] (3,0) -- (3,-6) node[right] {};
\draw[dashed] (3.5,0) -- (3.5,-6) node[right] {};
\draw[dashed] (5,0) -- (5,-6) node[right] {};
\draw[dashed] (0,-1) -- (6,-1) node[right] {};
\draw[dashed] (0,-3) -- (6,-3) node[right] {};
\draw[dashed] (0,-3.5) -- (6,-3.5) node[right] {};
\draw[dashed] (0,-5) -- (6,-5) node[right] {};
\node[draw,circle,inner sep=2.5pt,fill=blue,shading=ball] at (1,-1) {};
\node[draw,circle,inner sep=2.5pt,fill=blue,shading=ball] at (1,-5) {};
\node[draw,circle,inner sep=2.5pt,fill=blue,shading=ball] at (3,-1) {};
\node[draw,circle,inner sep=2.5pt,fill=blue,shading=ball] at (3.5,-1) {};
\node[draw,circle,inner sep=2.5pt,fill=blue,shading=ball] at (5,-1) {};
\node[draw,circle,inner sep=2.5pt,fill=blue,shading=ball] at (1,-3) {};
\node[draw,circle,inner sep=2.5pt,fill=blue,shading=ball] at (1,-3.5) {};
\node[draw,circle,inner sep=2.5pt,fill=blue,shading=ball] at (5,-3.5) {};
\node[draw,circle,inner sep=2.5pt,fill=blue,shading=ball] at (3.5,-5) {};
\node[] at (5.4,-3.8) {$\color{blue}z_{ij}$};
\node[red] at (5,-.3) {$\color{red}\theta_{i}$};
\node[red] at (.3,-3.5) {$\color{red}\theta_{j}$};
\draw[-,ultra thick,red] (1,0) -- (1,1) node[right] {};
\draw[-,ultra thick,red] (3,0) -- (3,.3) node[right] {};
\draw[-,ultra thick,red] (3.5,0) -- (3.5,.35) node[right] {};
\draw[-,ultra thick,red] (5,0) -- (5,.5) node[above] {$\color{red}w_i$};
\draw[-,ultra thick,red] (.6,0) -- (.6,.2) node[right] {};
\draw[-,ultra thick,red] (1.2,0) -- (1.2,.1) node[right] {};
\draw[-,ultra thick,red] (4.1,0) -- (4.1,.2) node[above] {};
\draw[-,ultra thick,red] (0, -1) -- (-1,-1) node[right] {};
\draw[-,ultra thick,red] (0,-3) -- (-.3,-3) node[right] {};
\draw[-,ultra thick,red] (0,-3.5) -- (-.35,-3.5) node[left] {$\color{red}w_j$};
\draw[-,ultra thick,red] (0,-5) -- (-.5,-5) node[right] {};
\draw[-,ultra thick,red] (0,-.6) -- (-.2,-.6) node[right] {};
\draw[-,ultra thick,red] (0,-1.2) -- (-.1,-1.2) node[right] {};
\draw[-,ultra thick,red] (0,-4.1) -- (-.2,-4.1) node[right] {};
\end{tikzpicture}
\caption{Point process representation of a random graph. Each node $i$ is
embedded in $\mathbb{R}_{+}$ at some location $\theta_{i}$ and is associated with a
sociability parameter $w_{i}$. An edge between nodes $\theta_{i}$ and
$\theta_{j}$ is represented by a point at locations $(\theta_{i},\theta_{j})$
and $(\theta_{j},\theta_{i})$ in $\mathbb{R}_{+}^{2}$.}%
\label{fig:pointprocess}%
\end{figure}
To leverage some of the benefits of \emph{generative} exchangeable modeling while producing sparse graphs with power-law behavior, we set aside the discrete array structure of the adjacency matrix and instead consider a different notion of exchangeability of a continuous-space representation of networks based on a \emph{point process} on $\mathbb{R}_{+}^{2}$ (see
Figure~\ref{fig:pointprocess})
\begin{equation}
Z=\sum_{i,j}z_{ij}\delta_{(\theta_{i},\theta_{j})},
\label{eq:pointprocessrep}
\end{equation}
where $z_{ij}=1$ if there is a link between nodes $\theta_{i}$ and $\theta
_{j}$ in $\mathbb{R}_{+}$, and is 0 otherwise. Our notion of exchangeability in this framework is as follows.  Paralleling \eqref{eq:sequenceexchangeability}, the point
process $Z$ on $\mathbb{R}_{+}^{2}$ is exchangeable if and only if, for any $h>0$ and for any permutations $\pi,\sigma$ of $\mathbb{N}$, %
\begin{equation}
(Z(A_{i}\times A_{j}))\overset{d}{=}(Z(A_{\pi(i)}\times A_{\sigma(j)}))\text{ for
}(i,j)\in\mathbb{N}^{2},%
\label{eq:pointprocessexchangeability}
\end{equation}
where here we consider \emph{intervals} $A_{i}=[h(i-1),hi]$ with $i\in\mathbb{N}$.  Considering arbitrarily small intervals $A_i$, such that two nodes $\theta_j$ and $\theta_k$ are unlikely to fall into the same interval, leads to a similar intuition and statistical implication of exchangeability as in the Aldous-Hoover framework. Note, however, that if we order nodes in \eqref{eq:pointprocessrep} by the first time an edge appears for that node, and look at the associated adjacency matrix, then this array is not exchangeable in the sense of \eqref{eq:sequenceexchangeability}. Importantly, though, our notion of exchangeability allows us to define a practical and efficient inference algorithm (described in Section~\ref{sec:MCMC}) due to the invariance property in the continuous space specified in \eqref{eq:pointprocessexchangeability}.

In place of the Aldous-Hoover theorem, we now appeal to the continuous-space counterpart~\citep[Chapter 9]{Kallenberg2005} which provides a
representation theorem for exchangeable point processes
on $\mathbb{R}_{+}^{2}$:
a point process is exchangeable if and only if it can be
represented as a transformation of unit-rate Poisson processes and uniform
random variables (see Theorem~\ref{thm:kallenberg}); this is in direct analogy to the graphon transformation of uniform random variables in the Aldous-Hoover representation. More precisely, within the Kallenberg framework, we consider that two nodes $i\neq j$ connect with probability
\begin{equation}
\Pr(z_{ij}=1|w_i,w_j)=1-e^{-2w_i w_j}
\end{equation}
where the positive \emph{sociability} parameters $(w_i)_{i=1,2,\ldots}$ are the points of a Poisson point process, or equivalently the jumps of a \emph{completely random measure} (CRM) \cite{Kingman1967,Kingman1993,Lijoi2010}. We show that by carefully choosing the L\'evy measure characterizing this CRM, we are able to construct graphs ranging from \emph{sparse} to \emph{dense}. In particular, any L\'{e}vy measure yielding an infinite activity CRM leads to sparse graphs; alternatively, finite activity CRMs, whose associated point processes are in the compound Poisson process family, yield dense graphs. When building on a specific class of infinite activity \emph{regularly varying} CRMs, we can obtain graphs where the number of edges increases at a rate below $n^{a}$ for some constant $1<a<2$ that depends on the L\'evy measure. The associated degree distribution has a \emph{power-law} form.

By building on the framework of CRMs, we are able to harness the considerable theory and practicality of such processes to (1) derive important properties of our proposed model and (2) develop an efficient statistical estimation procedure.  The CRM construction enables us to relate the sparsity properties of the graph to the properties of the L\'evy measure.  We also utilize the CRM-based formulation to develop a scalable Hamiltonian Monte Carlo sampler that can automatically handle a range of graphs from dense to sparse based on inferring a graph sparsity parameter. We show in Section~\ref{sec:experiments} that our methods scale to graphs with hundreds of thousands of nodes and millions of edges. Thus, our generative specification enjoys both an analytic representation in the Kallenberg framework and a formulation in terms of CRMs.  The former allows us to nicely connect with existing random graph models whereas the latter provides (1) connections to the Bayesian nonparametric modeling and inference literature and (2) interpretability and theoretical analysis of the formulation.

In summary, our
proposed framework captures a number of desirable properties:

\begin{itemize}
\item \textbf{Sparsity.} We can obtain graphs where the number of edges
increases sub-quadratically with the number of nodes.

\item \textbf{Power Law.} Our formulation yields a power-law form, which is useful in modeling many real-world graphs~\cite{Newman2009}.

\item \textbf{Exchangeability} in the sense of \eqref{eq:pointprocessexchangeability}.

\item \textbf{Simplicity.} Three hyperparameters tune the expected number of nodes, power-law
properties, etc.

\item \textbf{Interpretability.} The node-specific sociability parameters, $w_i$, lead to straightforward interpretability of the model.

\item \textbf{Scalable inference.} Our CRM-based Hamiltonian Monte Carlo sampler efficiently scales to large, real-world graphs, allowing for rapid analysis of graph properties such as sparsity, power-law, etc.

\end{itemize}

A bipartite random graph formulation with power-law behavior building on CRMs was first proposed by~\citet{Caron2012}.  In this paper, we consider a more general CRM-based framework for bipartite graphs, directed multigraphs, and undirected graphs.  More importantly, we prove that the resulting formulation yields sparse graphs under certain conditions---a notion not explored in \cite{Caron2012}---and cast exchangeability within the Kallenberg representation theorem.  Both of these represent important and non-trivial extensions of this work.  A number of other theoretical results are explored in Section~\ref{sec:simprop} as well.  Finally, we note that the sampler of \citet{Caron2012} simply does not apply to our undirected graphs.  Instead, we present new and efficient posterior computations with demonstrated scalability on a range of large, real-world networks.

Our paper is organized as follows.  In Section~\ref{sec:background}, we provide background on exchangeability for sequences, arrays, and random measures on $\mathbb{R}_+^2$.  The latter provides an important theoretical foundation for the graph structures we propose.  We also present background on CRMs, which form the key building block of our graph construction.  The generic formulation for directed multigraphs, undirected graphs, and bipartite graphs is presented in Section~\ref{sec:model}.  Properties, such as exchangeability and sparsity, and methods for simulation are presented in Section~\ref{sec:simprop}.  Specific cases of our formulation leading to dense and sparse graphs are considered in Section~\ref{sec:specialcases}, including an empirical analysis of network properties of our proposed formulation relative to common network models.  Our Markov chain Monte Carlo (MCMC) based posterior computations are in Section~\ref{sec:MCMC}.  Finally, Section~\ref{sec:experiments} provides a simulated study and an extensive analysis of a variety of large, real-world graphs.

\section{Background}
\label{sec:background}

\subsection{Exchangeability and de Finetti-type representation theorems}

Our focus is on exchangeable random structures that can represent networks. To
build to such constructs, we first present a brief review of exchangeability
for random sequences, continuous-time processes, and discrete network arrays.
Thorough and accessible overviews of exchangeability of random structures are
presented in the surveys of~\citet{Aldous1985} and~\citet{Orbanz2015}. Here, we simply abstract
away the notions relevant to placing our network formulation in context,
as summarized in Table~\ref{table:exchangeability}.

\begin{table}[h]
\caption{Overview of representation theorems}%
\begin{tabular}
[c]{l|l|l}
& Discrete structure & Continuous time/space\\\hline
Exchangeability & de Finetti (1931) & B\"uhlmann (1960)\\
Joint/separate exchangeability & Aldous-Hoover (1979-1981) & Kallenberg (1990)
\end{tabular}
\label{table:exchangeability}
\end{table}

The classical representation theorem arising from a notion of exchangeability
for discrete \emph{sequences} of random variables is due to~\citet{DeFinetti1931}. The theorem states that a sequence $Z_{1}%
,Z_{2},\dots$ with $Z_i \in \mathbf{Z}$ is exchangeable if and only if there exists a random probability
measure $\Theta$ on $\mathbf{Z}$ with law $\nu$ such that the $Z_{i}$ are
conditionally i.i.d. given $\Theta$. That is, all exchangeable infinite
sequences can be represented as a mixture with directing measure $\Theta$ and
mixing measure $\nu$. If examining continuous-time \emph{processes} instead of
sequences, the representation associated with exchangeable \emph{increments}
is given by~\citet{Buhlmann1960} (see also~\citet{Freedman1996}) in terms of mixing L\'{e}vy
processes.

The focus of our work, however, is on graph structures.  Recall the definition of exchangeability of arrays in \eqref{eq:sequenceexchangeability}.
A representation theorem for exchangeability of the classical discrete adjacency \emph{matrix}, $Z$, follows in Theorem~\ref{thm:AldousHoover} by considering a special case of the Aldous-Hoover theorem to \emph{2-arrays}.  We additionally focus here on \emph{joint exchangeability}---that is, symmetric permutations of rows and columns---which is applicable to matrices $Z$ where both rows and columns index the same set of nodes. \emph{Separate exchangeability} allows for different row and column permutations,
making it applicable to scenarios where one has distinct node identities on
rows and columns, such as in the bipartite graphs we consider in
Section~\ref{sec:bipartite}. Extensions of Theorem~\ref{thm:AldousHoover} to higher dimensional arrays are likewise
straightforward~\cite{Orbanz2015}.

\begin{theorem}
{\normalfont\textbf{(Aldous-Hoover representation of jointly exchangeable
matrices~\cite{Aldous1981,Hoover1979})}}. A random 2-array $(Z_{ij})_{i,j
\in\mathbb{N}}$ is jointly exchangeable if and only if there exists a random
measurable function $f: [0,1]^{3} \rightarrow\mathbf{Z}$ such that
\begin{align}
(Z_{ij}) \overset{d}{=} (f(U_{i},U_{j},U_{ij})),
\end{align}
where $(U_{i})_{i\in\mathbb{N}}$ and $(U_{ij})_{i,j>i \in\mathbb{N}}$ with
$U_{ij}=U_{ji}$ are a sequence and matrix, respectively, of i.i.d.
$\mbox{Uniform}[0,1]$ random variables.
\label{thm:AldousHoover}
\end{theorem}

For undirected graphs where $Z$ is a binary, symmetric adjacency matrix, the Aldous-Hoover representation
can be expressed as the existence of a \emph{graphon} $\omega: [0,1]^2 \rightarrow [0,1]$,
symmetric in its arguments, where
\begin{equation}
f(U_{i},U_{j},U_{ij}) = \left\{\begin{array}{ll} 1 & U_{ij} < \omega(U_i,U_j)\\
												0 & \mbox{otherwise}.
											\end{array}\right.
\label{eq:graphon}%
\end{equation}

Exchangeability is a fundamentally important concept in modeling. For example,
an assumption of joint exchangeability in network models implies that the
probability of a given graph depends on certain structural features, such as
number of edges, triangles, and five-stars, but not on where these features
occur in the network. Likewise, for separate exchangeability, the probability
of the matrix is invariant to reordering of the rows and columns, e.g., users
and items in a recommender system application. However, based on the Aldous-Hoover representation theorem, one can
derive the important consequence that \textit{if a random graph is
exchangeable, it is either dense or empty.} Note, crucially, that this result
assumes the graph is modeled via a \emph{discrete} adjacency matrix structure and exchangeability is considered in this framework.

Throughout this paper, we instead consider representing a graph as a point process
$Z=\sum_{i,j}z_{ij}\delta_{(\theta_{i},\theta_{j})}$ with nodes $\theta_{i}$ embedded in $\mathbb{R}_{+}$, as in \eqref{eq:pointprocessrep}, and then examine notions of exchangeability in this context. \citet{Kallenberg1990}
derived de-Finetti-style representation theorems for separately and jointly
exchangeable random measures on $\mathbb{R}_{+}^{2}$, which we present for the jointly exchangeable case in Theorem~\ref{thm:kallenberg}. Recall the definition of
joint exchangeability of a random measure on $\mathbb{R}_+$ in \eqref{eq:pointprocessexchangeability}. In the following,
$\lambda$ denotes the Lebesgue measure on $\mathbb{R}_{+}$, $\lambda_{D}$ the Lebesgue measure on
the diagonal $D=\{(s,t)\in\mathbb{R}_{+}^{2}|s=t\}$, and
$\widetilde{\mathbb{N}}_{2}=\{\{i,j\}|(i,j)\in\mathbb{N}^{2}\}$. We also
define a \emph{U-array} to be an array of independent uniform random variables.

\begin{theorem}
{\normalfont\textbf{(Representation theorem for jointly exchangeable random
measures on $\mathbb{R}_{+}^{2}$ \cite[Theorem 9.24]%
{Kallenberg1990,Kallenberg2005}).}} \newline\ A\ random measure $\xi$ on
$\mathbb{R}_{+}^{2}$ is jointly exchangeable if and only if almost surely
\begin{align}
\begin{aligned} \xi & =\sum_{i,j}f(\alpha_0,\vartheta_{i},\vartheta_{j},\zeta_{\{i,j\}})\delta_{\theta_{i},\theta_{j}}+\beta_0\lambda_{D}+\gamma_0(\lambda\times\lambda)\\ & +\sum_{j,k}\left( g(\alpha_0,\vartheta_{j},\chi_{jk})\delta_{\theta _{j},\sigma_{jk}}+g^{\prime}(\alpha_0,\vartheta_{j},\chi_{jk})\delta _{\sigma_{jk},\theta_{j}}\right) \\ & +\sum_{j}\left( h(\alpha_0,\vartheta_{j})(\delta_{\theta_{j}}\times \lambda)+h^{\prime}(\alpha_0,\vartheta_{j})(\lambda\times\delta_{\theta_{j}})\right) \\ & +\sum_{k}\left( l(\alpha_0,\eta_{k})\delta_{\rho_{k},\rho_{k}^{\prime}}+l^{\prime}(\alpha_0,\eta_{k})\delta_{\rho_{k}^{\prime},\rho_{k}}\right) \end{aligned} \label{theorem:kallenbergjoint}%
\end{align}
for some measurable functions $f:\mathbb{R}_{+}^{4}\rightarrow\mathbb{R}_{+}$, $g,g':\mathbb{R}_{+}^{3}\rightarrow\mathbb{R}_{+}$ and $h,h^{\prime},l,l^{\prime}$: $\mathbb{R}%
_{+}^{2}\rightarrow \mathbb{R}_{+}$. Here, $(\zeta_{\{i,j\}})$ with $\{i,j\}\in\widetilde{\mathbb{N}%
}_{2}$ is a U-array. $\{(\theta_{j},\vartheta_{j})\}$ and $\{(\sigma_{ij}%
,\chi_{ij})\}$ on $\mathbb{R}_{+}^{2}$ and $\{(\rho_{j},\rho_{j}^{\prime}%
,\eta_{j})\}$ on $\mathbb{R}_{+}^{3}$ are independent, unit-rate Poisson
processes. Furthermore, $\alpha_0,\beta_0,\gamma_0\geq0$ are an independent set of
random variables. \label{thm:kallenberg}
\end{theorem}

We place our proposed network model of Section~\ref{sec:model} within this Kallenberg representation in Section~\ref{sec:kallengbergmapping}, yielding direct analogs to the classical graphon representation of graphs based on exchangeability of the adjacency matrix.

\subsection{Completely Random Measures}

Our models for graphs build on the completely random measure
(CRM)~\cite{Kingman1967} framework. CRMs have been used extensively in the
Bayesian nonparametric literature for proposing flexible classes of priors
over functional spaces, \citep[cf.][]{Regazzini2003,Lijoi2010}. We recall in this section basic properties of CRMs; the reader can refer to the monograph of \citet{Kingman1993} for an exhaustive coverage.

A CRM $W$ on\ $\mathbb{R}_{+}$ is a random measure such that for any countable
number of disjoint measurable sets $A_{1}, A_{2}, \ldots$ of $\mathbb{R}_{+}$, the
random variables $W(A_{1}), W(A_{2}), \ldots$ are independent and%
\begin{equation}
W(\cup_{j}A_{j})=\sum_{j}W(A_{j}).
\end{equation}
If one additionally assumes that the distribution of $W([t,s])$ only depends
on $t-s$, (i.e. we have i.i.d. increments of fixed size) then the CRM takes
the following form%
\begin{equation}
W=\sum_{i=1}^{\infty}w_{i}\delta_{\theta_{i}},
\end{equation}
where $(w_{i},\theta_i)_{i\in\mathbb{N}}$ are the points of a Poisson point process on $\mathbb R_+^2$ with mean (or L\'evy) measure $\nu(dw,d\theta)=\rho(dw)\lambda(d\theta)$; moreover, the
Laplace transform of $W(A)$ for any measurable set $A$ admits the following
representation:
\begin{equation}
\mathbb{E}[\exp(-tW(A))]=\exp\left(  -\int_{\mathbb{R}_{+}\times A}\left[
1-\exp(-tw)\right]  \rho(dw)\lambda(d\theta)\right)  ,
\end{equation}
for any $t>0$ and $\rho$ a measure on $\mathbb{R}_{+}$ such that%
\begin{equation}
\int_{0}^{\infty}(1-e^{-w})\rho(dw)<\infty.
\end{equation}
The measure $\rho$ is referred to as the jump part of the L\'{e}vy measure. For a CRM $W$ with i.i.d. increments, which are intimately
connected to subordinators~\cite[Chapter 8]{Kingman1993}, $\rho$ characterizes
these increments. We denote this process as $W\sim \CRM(\rho,\lambda)$. Note
that $W([0,T])<\infty$ for any $T<\infty$, while $W(\mathbb{R}_{+})=\infty$ if
$\rho$ is not degenerate at 0.

The jump part $\rho$ of the L\'{e}vy measure is of particular interest for our
construction for graphs. If $\rho$ satisfies the condition
\begin{equation}
\int_{0}^{\infty}\rho(dw)=\infty,\label{eq:infiniteactivity}%
\end{equation}
then there will be an infinite number of jumps in any interval $[0,T]$, and we refer to the CRM as \emph{infinite activity}.
Otherwise, the number of jumps will be finite almost surely. In our models of Section~\ref{sec:model},
these jumps will map directly to the nodes in the graph.

Finally, throughout we let $\psi(t)$ be the Laplace exponent, defined as
\begin{equation}
\psi(t)=\int_0^\infty (1-e^{-wt})\rho(w)dw
\label{eq:laplaceexpo}
\end{equation}
and $\overline{\rho}(x)$ the tail L\'evy intensity
\begin{equation}
\overline{\rho}(x)=\int_x^\infty \rho(w)dw.
\label{eq:taillevy}
\end{equation}

In Section~\ref{sec:specialcases}, we consider special cases including the (compound) Poisson process and generalized gamma process~\cite{Brix1999,Lijoi2007}.

\section{Statistical network models}
\label{sec:model}

Our primary focus is on undirected network models, but implicit in our
construction is the definition of a directed integer-weighted, or
\emph{multigraph}, which in some applications might be the direct quantity of
interest. For example, in social networks, interactions are often not only
directed (\textquotedblleft person \emph{i} messages person \emph{j}%
\textquotedblright), but also have an associated count. Additionally,
interactions might be typed (\textquotedblleft message\textquotedblright,
\textquotedblleft SMS\textquotedblright,\textquotedblleft
like\textquotedblright,\textquotedblleft tag\textquotedblright). Our proposed
framework could be directly extended to model such data.

Our undirected graph simply transforms the directed multigraph by forming an
undirected edge if there is any directed edge between two nodes. Due to the
straightforward relationship between the two graphs, much of the intuition
gained from the directed case carries over to the undirected scenario.

\subsection{Directed multigraphs}

\label{sec:directed}

Let $V=(\theta_{1},\theta_{2},...)$ be a countably infinite set of nodes with
$\theta_{i}\in\mathbb{R}_{+}$. We represent the directed multigraph of interest using an atomic measure on
$\mathbb{R}_{+}^{2}$
\begin{equation}
D=\sum_{i=1}^{\infty}\sum_{j=1}^{\infty}n_{ij}\delta_{(\theta_{i},\theta
_{j})},
\end{equation}
where $n_{ij}$ counts the number of directed edges from node $\theta_{i}$ to
node $\theta_{j}$. See Figure~\ref{fig:simplegraph} for an illustration of the restriction of $D$ to $[0,1]^2$
and the corresponding directed graph. \begin{figure}[ptb]
\begin{center}%
\begin{tabular}
[c]{ccc}%
\includegraphics[width=.32\textwidth]{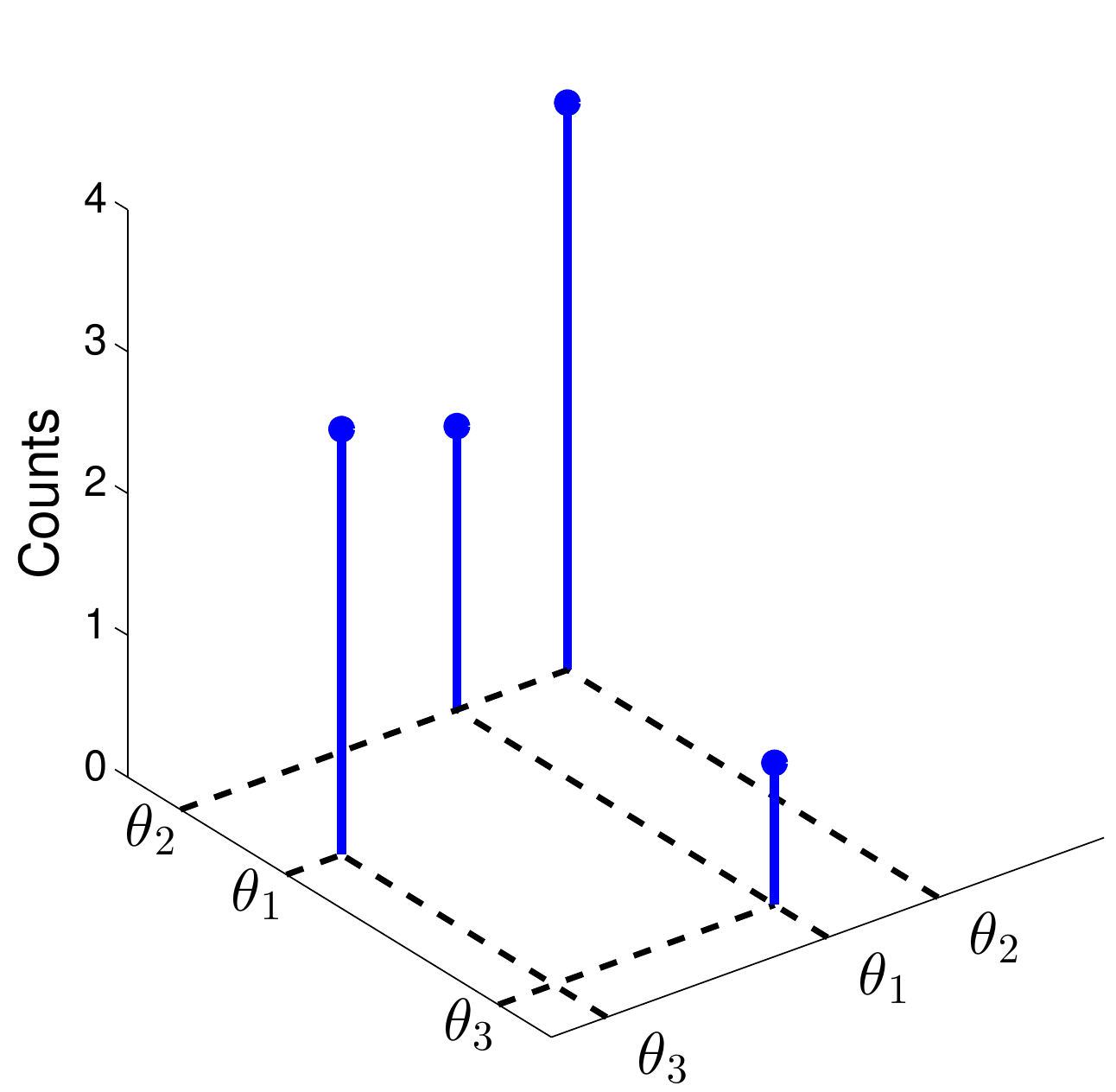} &
\resizebox{0.32\textwidth}{!}{\definecolor {processblue}{cmyk}{0.96,0,0,0}

\begin {tikzpicture}[-latex ,auto ,node distance =1.5 cm and 1.5cm ,on grid ,
semithick ,
state/.style ={ circle ,top color =white , bottom color = processblue!20 ,
draw,processblue , text=blue , minimum width =1 cm}]
\node[state] (C) {$\theta_1$};
\node[state] (A) [above left=of C] {$\theta_2$};
\node[state] (B) [above right =of C] {$\theta_3$};
\path (A) edge [loop left] node[left] {$4$} (A);
\path (C) edge [bend left =25] node[below =0.15 cm] {$2$} (A);
\path (C) edge [bend left =15] node[above =0.15 cm] {$1$} (B);
\path (B) edge [bend right = -25] node[below =0.15 cm] {$3$} (C);
\end{tikzpicture}} &
\resizebox{0.32\textwidth}{!}{\definecolor {processblue}{cmyk}{0.96,0,0,0}

\begin {tikzpicture}[-latex ,auto ,node distance =1.5 cm and 1.5cm ,on grid ,
semithick ,
state/.style ={ circle ,top color =white , bottom color = processblue!20 ,
draw,processblue , text=blue , minimum width =1 cm},every loop/.style={}]
\node[state] (C) {$\theta_1$};
\node[state] (A) [above left=of C] {$\theta_2$};
\node[state] (B) [above right =of C] {$\theta_3$};
\draw[-] (A) edge [loop left] (A);
\draw[-] (C) -- (A);
\draw[-] (C) -- (B);
\end{tikzpicture}}\\
(a) & (b) & (c)
\end{tabular}
\end{center}
\caption{An example of (a) the restriction on $[0,1]^2$ of an atomic measure $D$, (b) the corresponding
directed multigraph, and (c) corresponding undirected graph.}%
\label{fig:simplegraph}%
\end{figure}
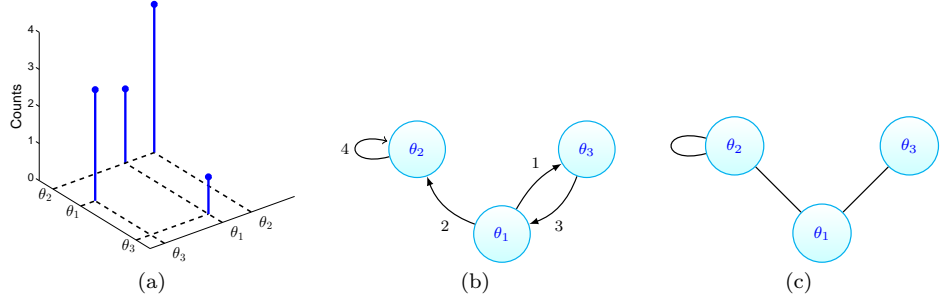

Our generative approach for modeling $D$ associates with each node $\theta
_{i}$ a \emph{sociability} parameter $w_{i}>0$ defined via the atomic random measure
\begin{equation}
W=\sum_{i=1}^{\infty}w_{i}\delta_{\theta_{i}},
\end{equation}
which we take to be distributed according to a homogeneous CRM, $W\sim\mbox{CRM}(\rho
,\lambda)$. Given $W$, $D$ is simply generated from a Poisson process (PP) with intensity
given by the product measure $\widetilde{W}=W\times W$ on $\mathbb{R}_{+}^{2}%
$:
\begin{equation}
D\mid W\sim\mbox{PP}(W\times W).
\end{equation}
That is, informally, the individual counts $n_{ij}$ are generated as
$\mbox{Poisson}(w_{i}w_{j})$. By construction, for any
$A,B\subset\mathbb{R}$, we have $\widetilde{W}(A\times B)=W(A)W(B)$.\ On
any bounded interval $A$ of $\mathbb{R}_{+}$, $W(A)<\infty$ implying $\widetilde{W}(A\times A)$ has finite mass.

\subsection{Undirected graphs}

\label{sec:undirected} We now turn to the primary focus of modeling undirected
graphs. Similarly to the directed case of Section~\ref{sec:directed}, we
represent an undirected graph using an atomic measure
\[
Z=\sum_{i=1}^{\infty}\sum_{j=1}^{\infty}z_{ij}\delta_{(\theta_{i},\theta_{j}%
)},
\]
with the convention $z_{ij}=z_{ji}\in\{0,1\}$. Here, $z_{ij}=z_{ji}=1$
indicates an undirected edge between nodes $\theta_i$ and $\theta_j$. We arise at the
undirected graph via a simple transformation of the directed graph: set
$z_{ij}=z_{ji}=1$ if $n_{ij}+n_{ji}>0$ and $z_{ij}=z_{ji}=0$ otherwise. That
is, place an undirected edge between nodes $\theta_i$ and $\theta_j$ if and only if there is
at least one directed interaction between the nodes. Note that in this
definition of an undirected graph, we allow self-edges. This could represent,
for example, a person posting a message on his or her own profile page. The
resulting hierarchical model is as follows:
\begin{align}
\begin{aligned} \begin{array}{ll} W =\sum_{i=1}^{\infty}w_{i}\delta_{\theta_{i}} & W \sim \mbox{CRM}(\rho,\lambda)\\ D =\sum_{i=1}^{\infty}\sum_{j=1}^{\infty}n_{ij}\delta_{(\theta_{i},\theta_{j})} & D\mid W\sim\text{PP}\left(W \times W\right) \\ Z =\sum_{i=1}^{\infty}\sum_{j=1}^{\infty}\min(n_{ij}+n_{ji},1)\delta_{(\theta _{i},\theta_{j})}. & \end{array} \end{aligned} \label{eq:Zhierarchy}%
\end{align}
This process is depicted graphically in Figure~\ref{fig:undirectedHierarchy}.

\begin{figure}[ptb]
\begin{center}
\subfigure[$\widetilde W = W \times W$]{\includegraphics[width=.3\textwidth]{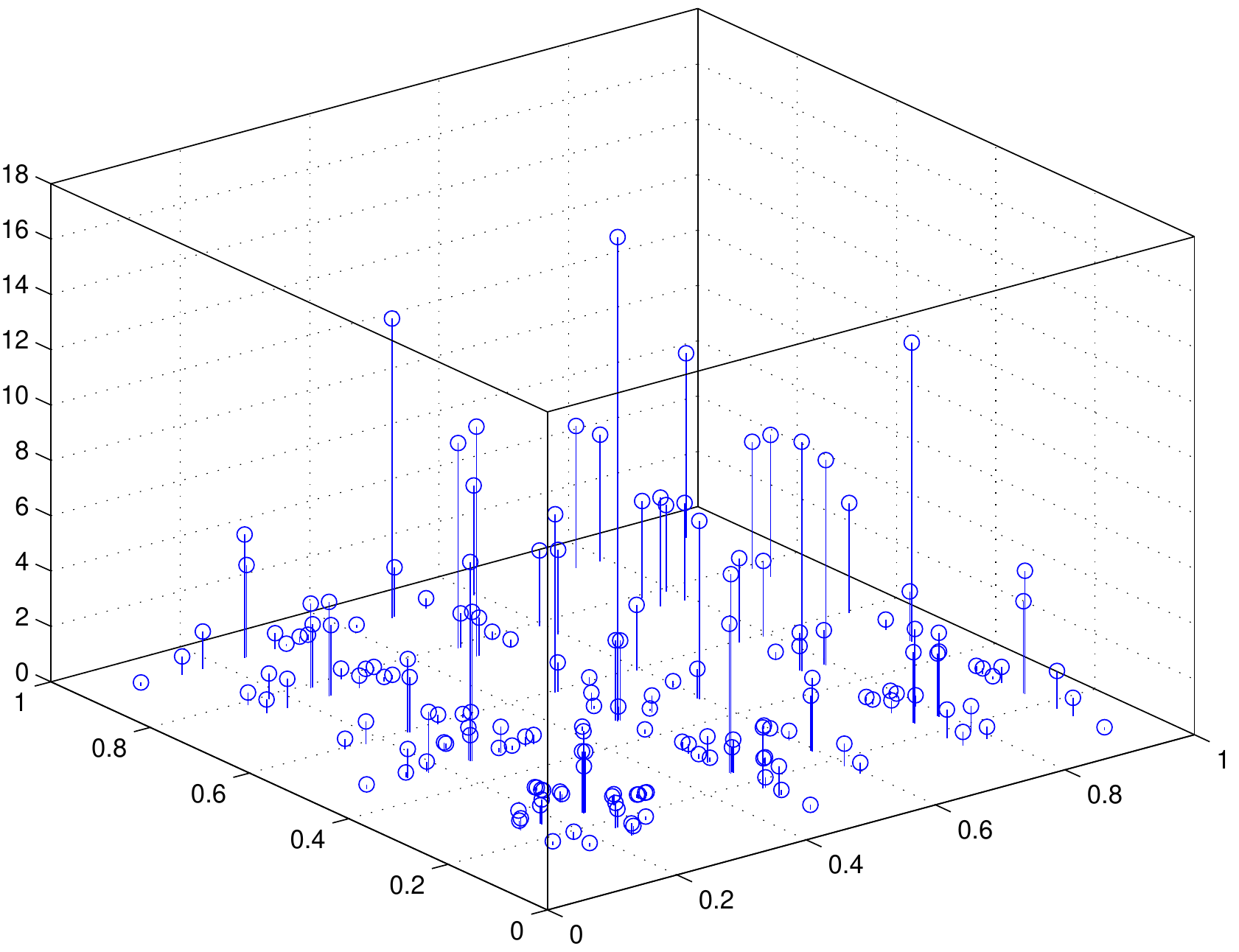}}
\subfigure[Integer point process $D$]{\includegraphics[width=.3\textwidth]{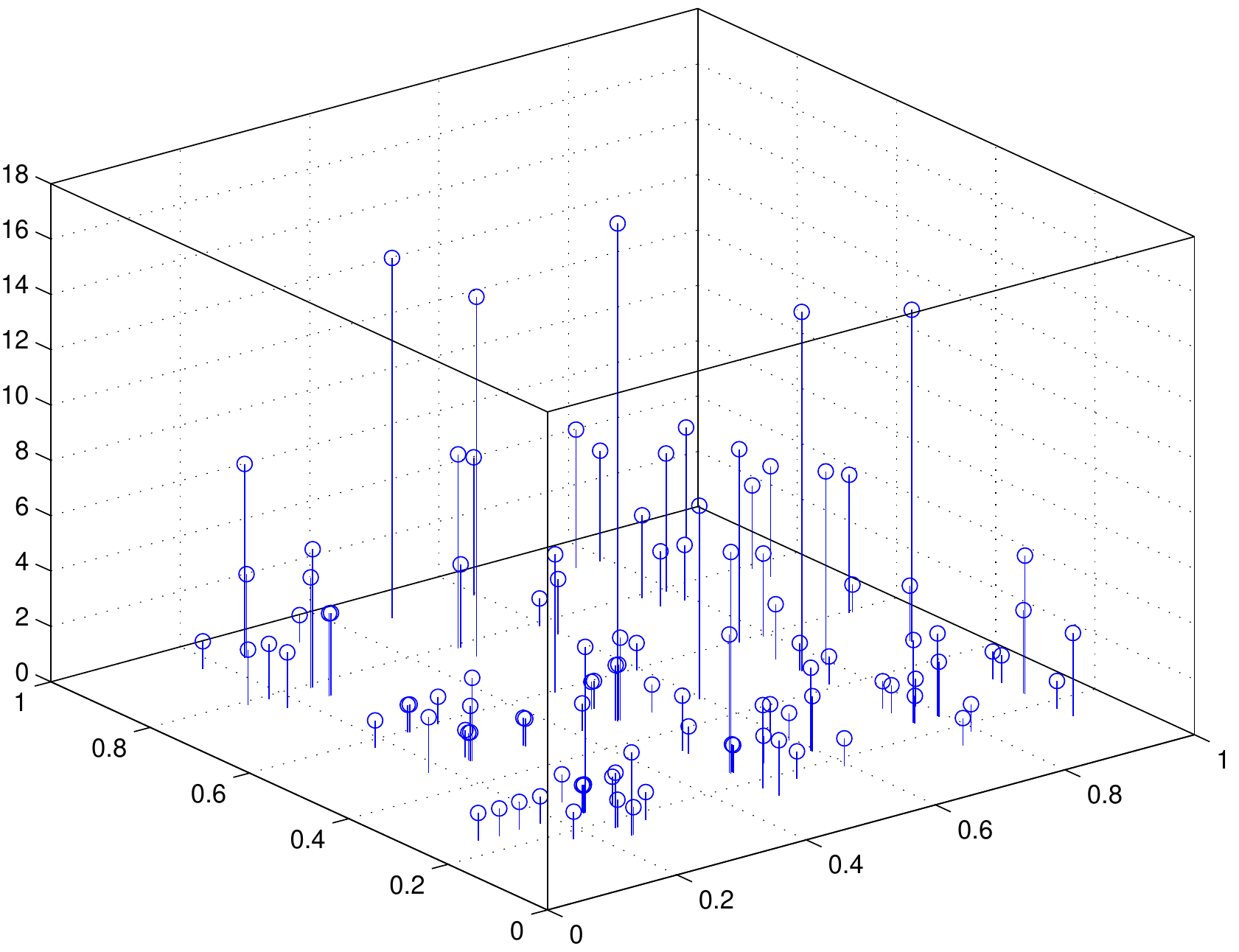}}
\subfigure[Point process $Z$]{\includegraphics[width=.3\textwidth]{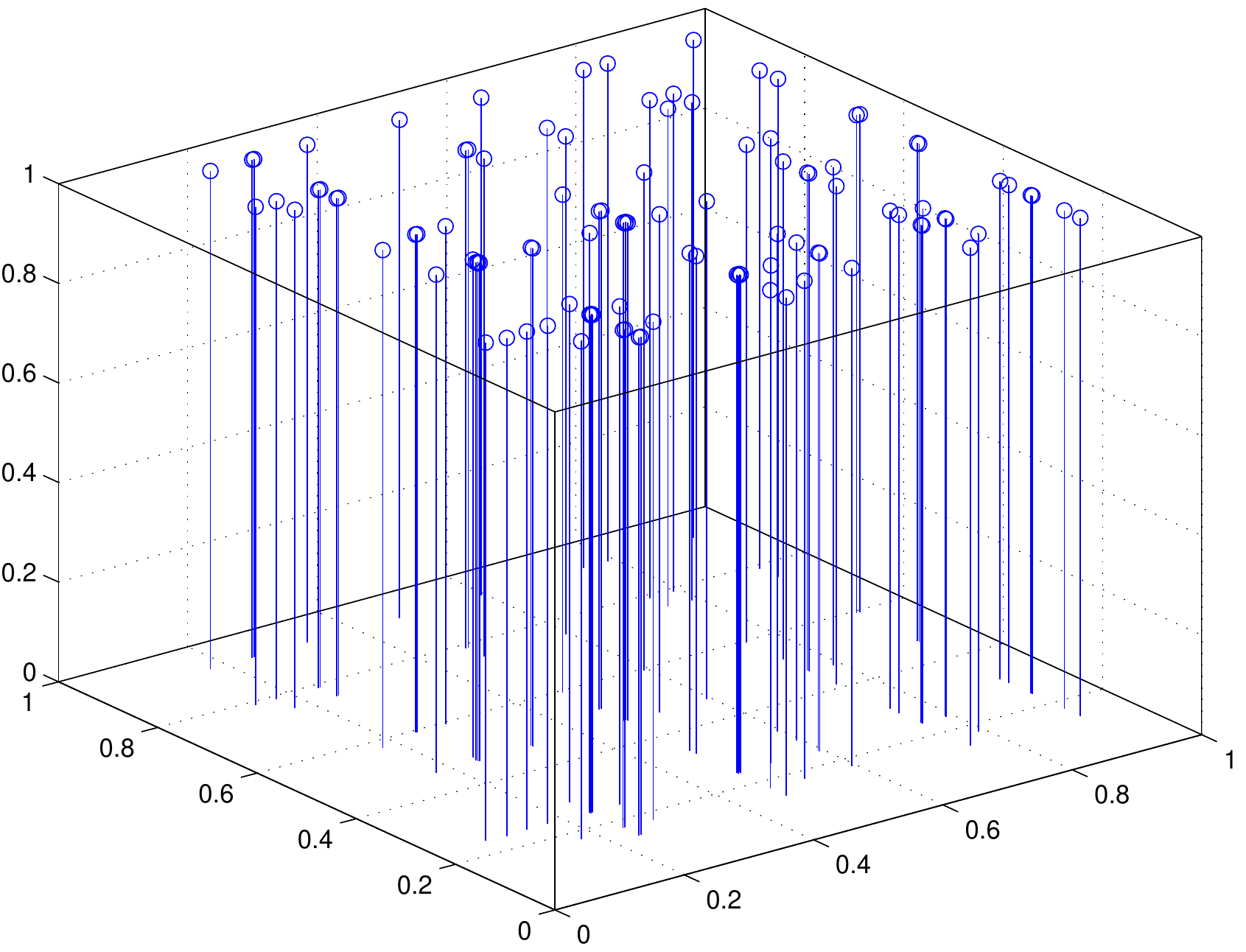}}
\end{center}
\vspace{-0.15in}
\caption{An example of (a) the product measure $\protect\widetilde{W}=W\times
W$ for CRM $W$, (b) a draw of the directed multigraph measure $D\mid W\sim PP(W\times W)$,
(c) corresponding undirected measure $Z=\sum_{i=1}^{\infty}\sum_{j=1}^{\infty
}\min(n_{ij}+n_{ji},1)\delta_{(\theta_{i},\theta_{j})}$.} 
\label{fig:undirectedHierarchy}%
\end{figure}

Equivalently, given the sociability parameters $w = \{w_{i}\}$, we can
directly specify the undirected graph model as
\begin{align}
\Pr(z_{ij}=1\mid w)=\left\{
\begin{array}
[c]{ll}%
1-\exp(-2w_{i}w_{j}) & i\neq j\\
1-\exp(-w_{i}^{2}) & i = j.
\end{array}
\right.  \label{eq:zij}%
\end{align}
To see the equivalence between this formulation and the one obtained from
manipulating the directed multigraph, note that for $i\neq j$, $\Pr
(z_{ij}=1\mid w)=\Pr(n_{ij}+n_{ji}>0\mid w)$. By properties of the Poisson
process, $n_{ij}$ and $n_{ji}$ are independent random variables conditioned on
$W$. The sum of two Poisson random variables, each with rate $w_{i}w_{j}$, is
again Poisson with rate $2w_{i}w_{j}$. The result \eqref{eq:zij} arises from
the fact that $\Pr(n_{ij}+n_{ji}>0\mid w) =1-\Pr(n_{ij}+n_{ji}=0\mid w)$.
Likewise, the $i=j$ case arises using a similar reasoning for $\Pr
(z_{ii}=1\mid w)=\Pr(n_{ii}>0\mid w)$.  

\paragraph{Graph restrictions} Our general network process is defined on $\mathbb R_+^2$ and, due to the fact that $W(\mathbb{R}_+)=\infty$, yields an infinite number of edges.  In applications, we are typically interested in considering graphs with a finite number of edges, but without a bound on or prespecification of this finite number. We therefore consider restrictions $D_\alpha$ and $Z_\alpha$ of $D$ and $Z$, respectively, to the box $[0,\alpha]^2$ and in Section~\ref{sec:MCMC} examine methods for inferring $\alpha$. We also denote by $W_\alpha$ and $\lambda_\alpha$ the corresponding CRM and Lebesgue measure on $[0,\alpha]$. We write $Z_\alpha^\ast=Z_\alpha([0,\alpha]^2)$, the total mass on $[0,\alpha]^2$, and similarly for $D_\alpha^\ast$ and $W_\alpha^\ast$. By definition, $D_\alpha$ is drawn from a Poisson process with finite mean measure $W_\alpha\times W_\alpha$, so we have the following generative model for directly simulating $D_\alpha$ and $Z_\alpha$:
\begin{align}
W_\alpha &\sim \mbox{CRM}(\rho,\lambda_\alpha)\nonumber\\
D_\alpha^{\ast}|W^{\ast}_\alpha&\sim \mbox{Poisson}(W_\alpha^{\ast\ 2}).\nonumber\\
\intertext{For $k=1,\ldots,D_\alpha^{\ast}$ and $j=1,2$}
U_{kj}|W_\alpha&\overset{iid}{\sim} \frac{W_\alpha}{W_\alpha^\ast}\nonumber\\
D_\alpha&=\sum_{k=1}^{D_\alpha^{\ast}}\delta_{(U_{k1},U_{k2})}.
\label{eq:condpoisson}
\end{align}
Here, the variables $U_{kj}\in\mathbb{R}_+$ correspond to nodes in the graph, and pairs of variables $(U_{k1},U_{k2})$ correspond to a directed edge from node $U_{k1}$ to node $U_{k2}$. The number of directed edges, $D_\alpha^\ast$, depends on the total mass of the CRM, $W_\alpha^\ast$. For each such directed edge, the defining nodes $U_{kj}$ are drawn from a normalized CRM, $\frac{W_\alpha}{W_\alpha^\ast}$; since $\frac{W_\alpha}{W_\alpha^\ast}$ is discrete with probability 1, the $U_{kj}$ take a number $N_\alpha\leq 2D_\alpha^\ast$ of distinct values. That is, $N_\alpha$ corresponds to the number of nodes with degree at least one in the network.
Recall that the undirected network construction simply forms an undirected edge between a set of nodes if there exists at least one directed edge between them. If we consider unordered pairs $\{U_{k1},U_{k2}\}$, the number of such unique pairs takes a number $N_\alpha^{(e)}\leq D_\alpha^\ast$ of distinct values, where $N_\alpha^{(e)}$ corresponds to the number of edges in the undirected network.

The construction \eqref{eq:condpoisson}, enables us to re-express our Cox process model in terms of normalized CRMs~\cite{Regazzini2003}. This is very attractive both practically and theoretically; as we show in Section~\ref{sec:specialcases}, one can use this framework to build on the various results on urn processes and power-law properties of normalized CRMs in order to get exact samplers for our graph models as well as to show its sparsity.
\paragraph{Finite-dimensional generative process}
We now describe the urn formulation that allows us to obtain a finite-dimensional generative process. Recall that in practice, we cannot sample $W_\alpha \sim \mbox{CRM}(\rho,\lambda_\alpha)$ if the CRM is infinite activity.

Let $(U^\prime_1,\ldots,U^\prime_{2D_\alpha^\ast})=(U_{11},U_{12},\ldots U_{D_\alpha^{\ast}1},U_{D_\alpha^{\ast}2})$.
For some classes of L\'evy measure $\rho$, it is possible to integrate out the normalized CRM $\mu_\alpha=\frac{W_\alpha}{W_\alpha^\ast}$ in \eqref{eq:condpoisson} and derive the conditional distribution of $U^\prime_{n+1}$ given $(W_\alpha^\ast,U^\prime_1,\ldots,U^\prime_{n})$.
We first recall some background on random partitions. As $\mu_\alpha$ is discrete with probability 1, variables $U^\prime_1,\ldots,U^\prime_n$ take $k\leq n$ distinct values $\widetilde U^\prime_j$, with multiplicities $1\leq m_j\leq n$. The distribution on the underlying partition is usually defined in terms of an exchangeable partition probability function (EPPF)~\cite{Pitman1995} $\Pi_n^{(k)}(m_1,\ldots,m_k|W_\alpha^\ast)$ which is symmetric in its arguments.
The predictive distribution of $U^\prime_{n+1}$ given $(W_\alpha^\ast,U^\prime_1,\ldots,U^\prime_{n})$ is then given in terms of the EPPF:
\begin{align}
U^\prime_{n+1}|(W_\alpha^\ast,U^\prime_1,\ldots,U^\prime_{n}) &\sim \frac{\Pi_{n+1}^{(k+1)}(m_1,\ldots,m_k,1|W_\alpha^\ast)}{\Pi_n^{(k)}(m_1,\ldots,m_k|W_\alpha^\ast)}\frac{1}{\alpha}\lambda_\alpha\nonumber\\
&+\sum_{j=1}^k \frac{\Pi_{n+1}^{(k)}(m_1,\ldots,m_j+1,\ldots,m_k|W_\alpha^\ast)}{\Pi_n^{(k)}(m_1,\ldots,m_k|W_\alpha^\ast)} \delta_{\widetilde U^\prime_j}.
\label{eq:generalurn}
\end{align}
Using this urn representation, we can rewrite our generative process as 
\begin{align}
W_\alpha^\ast &\sim P_{W_\alpha^\ast}\nonumber\\
D_\alpha^{\ast}|W^{\ast}_\alpha&\sim \mbox{Poisson}(W_\alpha^{\ast\ 2}).\nonumber\\
(U_{kj})_{k=1,\ldots,D_\alpha^{\ast};j=1,2}|W_\alpha^\ast &\sim \text{Urn process \eqref{eq:generalurn}} \nonumber\\
D_\alpha&=\sum_{k=1}^{D_\alpha^{\ast}}\delta_{(U_{k1},U_{k2})},
\label{eq:gengeneralurn}
\end{align}
where $P_{W_\alpha^\ast}$ is the distribution of the CRM total mass, $W_\alpha^\ast$. The representation of \eqref{eq:gengeneralurn} can be used to sample exactly from our graph model, assuming we can sample from $P_{W_\alpha^\ast}$ and evaluate the EPPF. In Section~\ref{sec:specialcases}
we show that this is indeed possible for specific CRMs of interest.  If this is not possible, in Section~\ref{sec:simulation} we present alternative, though potentially more computationally complex, methods for simulation.

\subsection{Bipartite graphs}

\label{sec:bipartite} The above construction can also be extended to bipartite
graphs. Let $V=(\theta_{1},\theta_{2},...)$ and $V^\prime=(\theta^\prime_{1},\theta^\prime_{2},...)$ be two countably infinite set of nodes with
$\theta_{i},\theta^\prime_{i}\in\mathbb{R}_{+}$. We assume that only connections between nodes of different sets are allowed.

We represent the \emph{directed bipartite multigraph} of interest using an atomic measure on
$\mathbb{R}_{+}^{2}$
\begin{equation}
D=\sum_{i=1}^{\infty}\sum_{j=1}^{\infty}n_{ij}\delta_{(\theta_{i},\theta^\prime
_{j})},
\end{equation}
where $n_{ij}$ counts the number of directed edges from node $\theta_{i}$ to
node $\theta^\prime_{j}$. Similarly, the \emph{bipartite graph} is represented by an atomic measure
\[
Z=\sum_{i=1}^{\infty}\sum_{j=1}^{\infty}z_{ij}\delta_{(\theta_{i},\theta^\prime_{j}%
)}.
\]
Our bipartite graph formulation introduces two CRMs, $W\sim \mbox{CRM}(\rho,\lambda)$  and $W^\prime\sim \mbox{CRM}(\rho^\prime,\lambda)$, whose jumps correspond to sociability parameters for nodes in sets $V$ and $V^\prime$, respectively. The generative model for the bipartite graph mimics that of the non-bipartite one:
\begin{align}
\begin{aligned}
\begin{array}{ll}
W^{} =\sum_{i=1}^{\infty}w_{i}\delta_{\theta_{i}} & W \sim \mbox{CRM}(\rho,\lambda)\\
W^\prime =\sum_{j=1}^{\infty}w^\prime_{j}\delta_{\theta^\prime_{j}} & W^\prime \sim \mbox{CRM}(\rho^\prime,\lambda)\\
D =\sum_{i=1}^{\infty}\sum_{j=1}^{\infty}n_{ij}\delta_{(\theta_{i},\theta^\prime_{j})} & D\mid W,W^\prime\sim\text{PP}\left(W \times W^\prime\right) \\ Z =\sum_{i=1}^{\infty}\sum_{j=1}^{\infty}\min(n_{ij},1)\delta_{(\theta _{i},\theta^\prime_{j})}. &
\end{array}
\end{aligned} \label{eq:Zhierarchybipartite}%
\end{align}
The model~\eqref{eq:Zhierarchybipartite} has been proposed by~\citet{Caron2012} in a slightly different formulation.  Here, we recast this model within our general framework making connections with an urn representation and the Kallenberg theory of exchangeability, both of which enable new theoretical and practical insights.

\section{General properties and simulation}
\label{sec:simprop}

We provide here general properties of our network model depending on the properties of the L\'{e}vy intensity $\rho.$
In the next section, we provide more
refined properties, depending on specific choices of $\rho$.

\subsection{Exchangeability under the Kallenberg framework}
\label{sec:kallengbergmapping}

\begin{proposition}
{\normalfont\textbf{(Joint exchangeability of the undirected graph measure).}}
\newline\ For any CRM $W \sim\CRM(\rho,\lambda)$, the point process $Z$ defined by \eqref{eq:Zhierarchy}, or
equivalently by \eqref{eq:zij}, is jointly exchangeable.
\end{proposition}

\begin{proof}
The proof follows from the properties of $W\sim \CRM(\rho,\lambda)$. Let
$A_{i}=[h(i-1),hi]$ for $h>0$ and $i\in\mathbb{N}.$ We have%
\begin{equation}
(W(A_{i}))\overset{d}{=}(W(A_{\pi(i)}))
\end{equation}
for any permutation $\pi$ of $\mathbb{N}$. As $D(A_{i}\times A_{j})\sim
\mbox{Poisson}(W(A_{i})W(A_{j}))$, it follows that
\begin{equation}
(D(A_{i}\times A_{j}))\overset{d}{=}(D(A_{\pi(i)}\times A_{\pi(j)}))
\end{equation}
for any permutation $\pi$ of $\mathbb{N}$. Joint exchangeability of $Z$
follows directly.
\end{proof}\smallskip

We now reformulate our network process in the Kallenberg representation of
\eqref{theorem:kallenbergjoint}. Due to exchangeability, we know that such a
representation exits.  What we show here is that our CRM-based formulation
has an analytic and interpretable representation.  In particular, the
CRM $W$ can be constructed from a
two-dimensional unit-rate Poisson process on $\mathbb{R}_{+}^{2}$ using the
inverse L\'{e}vy method~\cite{Khintchine1937,Ferguson1972}. Let $(\theta
_{i},\vartheta_{i})$ be a unit-rate Poisson process on $\mathbb{R}_{+}^{2}$.
Let $\overline{\rho}(x)$ be the tail L\'evy intensity defined in \eqref{eq:taillevy}.\ Then the CRM $W=\sum w_{i}%
\delta_{\theta_{i}}$ with L\'{e}vy measure $\rho(dw)d\theta$ can be
constructed from the bi-dimensional point process by taking $w_{i}=\overline{\rho}^{-1}%
(\vartheta_{i})$. $\overline{\rho}^{-1}$ is a monotone function, known as the inverse
L\'{e}vy intensity.\ It follows that our undirected graph model can be
formulated under the representation of \eqref{theorem:kallenbergjoint} by
selecting any $\alpha_0$, $\beta_0=\gamma_0=0$, $g=g'=0$, $h=h'=l=l'=0$ and%
\begin{align}
f(\alpha_0,\vartheta_{i},\vartheta_{j},\zeta_{\{i,j\}})=\left\{
\begin{array}
[c]{ll}%
1 & \zeta_{\{i,j\}}\leq M(\vartheta_{i},\vartheta_{j})\\
0 & \text{otherwise}%
\end{array}
\right.
\label{eq:KallenbergGraphon}
\end{align}
where $M:\mathbb{R}_{+}^{2}\rightarrow\lbrack0,1]$ is defined by%
\[
M(\vartheta_{i},\vartheta_{j})=\left\{
\begin{array}
[c]{ll}%
1-\exp(-2\overline{\rho}^{-1}(\vartheta_{i})\overline{\rho}^{-1}(\vartheta_{j})) & \text{if }\vartheta
_{i}\neq\vartheta_{j}\\
1-\exp(-\overline{\rho}^{-1}(\vartheta_{i})^{2}) & \text{if }\vartheta_{i}=\vartheta_{j}.
\end{array}
\right.
\]

In Section~\ref{sec:specialcases}, we provide explicit forms for $\overline{\rho}$ depending on our choice of L\'{e}vy intensity $\rho$.  The expression \eqref{eq:KallenbergGraphon} represents a direct analog to that of \eqref{eq:graphon} arising from the Aldous-Hoover framework.  In particular, $M$ here is akin to the graphon $\omega$, and thus allows us to connect our CRM-based formulation with the extensive literature on graphons.  An illustration of the network construction from the Kallenberg representation, including the function $M$, is provided in Figure~\ref{fig:Kallenberg}. Note that had we started from the Kallenberg representation and selected an $f$ (or $M$) arbitrarily, we would likely not have yielded a network model with the normalized CRM interpretation that enables both interpretability and analysis of network properties, such as those presented in Section~\ref{sec:GGPspecialcase}.

\begin{figure}[t]
\begin{center}
\resizebox{0.45\textwidth}{!}{ 
%
%
\pgfplotsset{label style={font=\Huge}}
\begin{tikzpicture}

\begin{axis}[%
width=2.08333333333333in,
height=2.08333333333333in,
scale only axis,
xmin=0,
xmax=1.2,
xtick={0.0514672033008299,0.207242878138187,1},
xticklabels={$\theta_j$,$\theta_i$,$\alpha$},
ymin=0,
ymax=5,
ytick={1.39243641323988,3.38127450990066},
yticklabels={$\vartheta_i$,$\vartheta_j$},
axis x line*=bottom,
axis y line*=left
]
\addplot [
color=blue,
mark size=2.5pt,
only marks,
mark=*,
mark options={solid,fill=blue},
forget plot
]
table[row sep=crcr]{
0.207242878138187 1.39243641323988\\
0.0514672033008299 3.38127450990066\\
0.440809843650636 2.95431408708175\\
0.456833224394711 2.79427043995441\\
0.649144047614761 1.29626223453733\\
};
\addplot [
color=black,
dashed,
line width=2.0pt,
forget plot
]
table[row sep=crcr]{
1 0\\
1 5\\
};
\addplot [
color=black,
dashed,
forget plot
]
table[row sep=crcr]{
0 1.39243641323988\\
0.207242878138187 1.39243641323988\\
};
\addplot [
color=black,
dashed,
forget plot
]
table[row sep=crcr]{
0 3.38127450990066\\
0.0514672033008299 3.38127450990066\\
};
\addplot [
color=black,
dashed,
forget plot
]
table[row sep=crcr]{
0.207242878138187 0\\
0.207242878138187 1.39243641323988\\
};
\addplot [
color=black,
dashed,
forget plot
]
table[row sep=crcr]{
0.0514672033008299 0\\
0.0514672033008299 3.38127450990066\\
};
\end{axis}
\end{tikzpicture}%
}\resizebox{0.45\textwidth}{!}{ \input{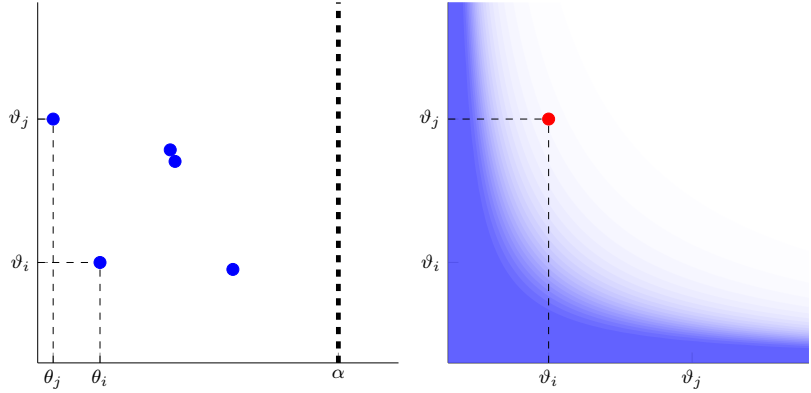}}
\end{center}
\caption{Illustration of the model construction based on the
Kallenberg representation. (left) A unit-rate Poisson process
$(\theta_{i},\vartheta_{i})$, $i\in\mathbb{N}$ on $[0,\alpha]\times
\mathbb{R}_{+}$. (right) For each pair $\{i,j\}\in
\protect\widetilde{\mathbb{N}}^{2}$, set $z_{ij}=z_{ji}=1$ with probability
$M(\vartheta_{i},\vartheta_{j})$. Here, $M$ is indicated by the blue shading
(darker shading indicates higher value) for a stable process
(generalized gamma process with $\tau=0$). In this case there is an analytic expression for $\overline{\rho}^{-1}$
and therefore $M$.}%
\label{fig:Kallenberg}
\end{figure}

For the bipartite graph, an application of Kallenberg's representation theorem for \emph{separate} exchangeability can likewise be made.

\subsection{Sparsity}

In this section we state the sparsity properties of our graph model, which relate to the properties of the L\'evy intensity $\rho$. Of particular interest is the notion of a \textit{regularly varying} L\'evy intensity~\cite{Karlin1967,Gnedin2006,Gnedin2007}, defined as follows.
\begin{definition}{\normalfont\textbf{(Regular variation)}}
\label{def:regularvariation}
Let $W\sim \CRM(\rho,\lambda)$. The CRM is said to be \textit{regularly varying} if the tail L\'evy intensity verifies
\begin{equation}
\overline{\rho}(x)\overset{x\downarrow 0}{\sim} \ell(1/x) x^{-\sigma}
\end{equation}
for $\sigma\in (0,1)$ where $\ell$ is a slowly varying function satisfying $\lim_{t\rightarrow\infty}\ell(at)/\ell(t)=1$ for any $a>0$. For example, constant and logarithmic functions are slowly varying. The equivalence notation $f(x)\overset{x\downarrow 0}{\sim}g(x)$  is used for $\lim_{x\rightarrow 0} \frac{f(x)}{g(x)}=1$ (not to be confused with the notation $\sim$ alone for `distributed from').
\end{definition}

As a trivial (and degenerate) example of obtaining sparse graphs, we note that if $\rho(dw)=0$, then $W([0,\infty))=0$ almost surely and there are no edges, $N_{\alpha}^{(e)}=0$, and thus no nodes of degree at least one, $N_{\alpha}=0$, for all values of $\alpha$. We consider more general L\'evy intensities in Theorem~\ref{th:sparsity}. In this theorem, we follow the notation of~\citet{Janson2011} for probability asymptotics (see Appendix~\ref{sec:asymptnotation} for details).

\begin{theorem}
\label{th:sparsity} Consider the point process $Z$ with $\rho(w)\neq0$. Let $\psi(t)$, defined in \eqref{eq:laplaceexpo}, be the Laplace exponent and $\psi^{\prime}(t)$ its first derivative; here, $\lim_{t\rightarrow0}\psi^{\prime}(t)=\mathbb{E}[W_{1}^{\ast}]$, the expected total mass for $\alpha=1$. Let
$N_{\alpha}^{(e)}$ be the number of edges in the undirected graph restriction $Z_\alpha$, and
$N_{\alpha}$ be the number of nodes. \newline If the CRM $W$ is
finite-activity (i.e., is obtained from a compound Poisson process):
\[
\int_{0}^{\infty}\rho(w)dw<\infty,
\]
then the number of edges scales quadratically with the number of nodes
\begin{equation}
N_{\alpha}^{(e)}=\Theta(N_{\alpha}^{2})
\end{equation}
almost surely as $\alpha$ tends to infinity, and the graph is \textbf{dense}. \newline If the CRM is
infinite-activity, i.e.
\[
\int_{0}^{\infty}\rho(w)dw=\infty
\]
and
\begin{equation}
\lim_{t\rightarrow0}\psi^{\prime}(t)<\infty,
\end{equation}
then the number of edges scales sub-quadratically with the number of nodes
\begin{equation}
N_{\alpha}^{(e)}=o(N_{\alpha}^{2})
\end{equation}
almost surely as $\alpha$ tends to infinity, and the graph is \textbf{sparse}.

The sparsity regime is linked
to the property of regular variation of the L\'evy intensity (Definition~\ref{def:regularvariation}). If the L\'evy intensity $\rho$ is regularly varying, i.e. if  there exists a slowly varying function
$\ell$ such that $\overline{\rho}(x)\overset{x\downarrow
0}{\sim}\ell(1/x)x^{-\sigma}$ with $\sigma\in(0,1)$, and if additionally $\lim_{t\rightarrow\infty}\ell(t)>0$, then%
\begin{equation}
N_{\alpha}^{(e)}=O\left(  N_{\alpha}^{\frac{2}{1+\sigma}}\right)
\end{equation}
almost surely.
\end{theorem}

Theorem~\ref{th:sparsity} is a direct consequence of two theorems that we
state now and prove in Appendix \ref{sec:proofsparsity}. The first theorem states that the number of
edges grows quadratically with $\alpha$, while the second states that the
number of nodes scales superlinearly with $\alpha$ for infinite-activity CRMs,
and linearly otherwise.

\begin{theorem}
\label{th:edgesquad}%
Consider the point process $Z$ with $\rho(w)\neq0$. If $\lim_{t\rightarrow0}\psi^{\prime}(t)=\mathbb{E}%
[W_{1}^{\ast}]<\infty$, then the number of edges in $Z_\alpha$ grows quadratically with
$\alpha$:
\begin{equation}
N_{\alpha}^{(e)}=\Theta(\alpha^{2})
\end{equation}
almost surely. Otherwise, $N_{\alpha}^{(e)}=\Omega(\alpha^{2})$.
\end{theorem}

\begin{theorem}
\label{th:nodesgrowth}%
Consider the point process $Z$ with $\rho(w)\neq0$. Then
\begin{equation}
N_{\alpha}=\left\{
\begin{array}
[c]{ll}%
\Theta(\alpha) & \text{if }W\text{ is a finite-activity CRM}\\
\omega(\alpha) & \text{if }W\text{ is an infinite-activity CRM}%
\end{array}
\right.
\end{equation}
almost surely as $\alpha\rightarrow\infty$. In words, the number of nodes in $Z_\alpha$
scales linearly with $\alpha$ for finite-activity CRMs and superlinearly with
$\alpha$ for infinite-activity CRMs. In particular, for a regularly varying L\'{e}vy intensity with $\lim_{t\rightarrow\infty}\ell(t)>0$, we have
\begin{equation}
N_{\alpha}=\Omega(\alpha^{\sigma+1})
\end{equation}
almost surely as $\alpha\rightarrow\infty$.
\end{theorem}

\subsection{Interactions between groups}

For any disjoint set of nodes $A,B\subset\mathbb{R}_+$, $A\cap B=\emptyset$, the
probability that there is at least one connection between a node in $A$ and a
node in $B$ is given by
\[
\Pr(Z(A\times B)>0|W)=1-\exp(-2W(A)W(B)).
\]
That is, the probability of a between-group edge depends on the sum of the
sociabilities in each group, $W(A)$ and $W(B)$, respectively.

\subsection{Simulation}
\label{sec:simulation}
To simulate an undirected graph, we harness the directed multigraph
representation. That is, we first sample a directed multigraph and then transform
it to an undirected graph as described in Section~\ref{sec:undirected}. One
might imagine simulating a directed network by first sampling $W_{\alpha}$ and
then sampling $D_{\alpha}$ given $W_{\alpha}$. However, recall that
$W_{\alpha}$ may have an infinite number of jumps.
One approximate approach to coping with this issue, which is possible for some L\'{e}vy
intensities $\rho$, is to resort to adaptive thinning
\cite{Lewis1979,Ogata1981,Favaro2013}. A related alternative approximate approach, but
applicable to any L\'{e}vy intensity $\rho$ satisfying
\eqref{eq:infiniteactivity}, is the inverse L\'{e}vy method. This method first
defines a threshold $\varepsilon$ and then samples the weights $\Omega
=\{w_{i}|w_{i}>\varepsilon\}$ using a Poisson measure on $[\varepsilon
,+\infty]$. One then simulates $D_\alpha$ using these truncated weights $\Omega$.

A naive application of this truncated method that considers sampling directed or undirected edges as in \eqref{eq:Zhierarchy} or \eqref{eq:zij}, respectively, can prove computationally problematic since a large number of possible edges must be considered (one Poisson/Bernoulli draw for each $(\theta_i,\theta_j)$ pair for the directed/undirected case).
Instead, we can harness the Cox process representation and resulting sampling procedure of \eqref{eq:condpoisson} to first sample the total number of directed edges and then their specific instantiations.
More specifically, to approximately simulate a point process on $[0,\alpha]^{2}$, we use the
inverse L\'{e}vy method to sample
\begin{align}
\Pi_{\alpha,\varepsilon}=\{(w,\theta)\in\Pi,0<\theta\leq\alpha,w>\varepsilon
\}.
\end{align}
Let $W_{\alpha,\varepsilon}=\sum_{i=1}^{K}w_{i}\delta_{\theta_{i}}$ be the
associated truncated CRM and $W_{\alpha,\varepsilon}^{\ast}=W_{\alpha
,\varepsilon}([0,\alpha])$ its total mass. We then sample $D_{\alpha,\varepsilon}^{\ast}$ and $U_{k,j}$ as in \eqref{eq:condpoisson}
and set $D_{\alpha,\varepsilon}=\sum_{k=1}^{D_{\alpha,\varepsilon}^{\ast
}}\delta_{(U_{k1},U_{k2})}$. The undirected graph measure $Z_{\alpha,\varepsilon}$ is set to the manipulation of
$D_{\alpha,\varepsilon}$ as in \eqref{eq:Zhierarchy}.

In the next section, we show that it is possible to sample a graph \emph{exactly} via an urn scheme
when considering the special case of generalized gamma processes, which includes the standard gamma process.

\section{Special cases}
\label{sec:specialcases}
In this section, we examine the properties of various models and their link to classical random graph models depending on the L\'{e}vy measure $\rho$.
We show that in generalized gamma process case, the resulting graph can be either dense or sparse, with the sparsity tuned by a single hyperparameter. We focus on the undirected graph case, but similar results can be obtained for directed multigraphs and bipartite graphs.

\subsection{Poisson process}
Consider a Poisson process with fixed increments $a$ and%
\[
\rho(dw)=\delta_{w_0}(dw),
\]
where $\delta_{w_0}$ is the dirac delta mass at $w_0>0$. Recalling the definition $\overline{\rho}(x)=\int_x^\infty \rho(dw)$, in this case, we have
\[
\overline{\rho}(x)=\left\{
\begin{array}
[c]{ll}%
1 & \text{if }x<w_0\\
0 & \text{otherwise}.%
\end{array}
\right.
\]
Ignoring self-edges, the graph construction can be described as follows.  To sample $W_\alpha \sim \mbox{PP}(\rho,\lambda_\alpha)$, we generate $n\sim \mbox{Poisson}(\alpha)$ and then sample $\theta_{i}%
\sim\mbox{Uniform}([0,\alpha])$ for $i=1,\ldots n$. We then sample edges according to \eqref{eq:zij}: For $0<i<j<n$, set $z_{ij}=z_{ji}=1$ with probability
$1-$exp$(-2w_0^{2})$ and 0 otherwise.\ The model is therefore equivalent to the
Erd\" os-R\' enyi random graph model $G(n,p)$ with $n\sim \mbox{Poisson}(\alpha)$ and
$p=1-\exp(-2w_0^{2})$. Therefore, this choice of $\rho$ leads to a dense graph
where the number of edges grows quadratically with the number of nodes $n$.

\subsection{Compound Poisson process}
A compound Poisson process is one where%
\[
\rho(dw)=h(w)dw
\]
and $h:\mathbb{R}_{+}\rightarrow\mathbb{R}_{+}$ is such that $\int_{0}%
^{\infty}h(w)dw=1.$ In this case, we have%
\[
\overline{\rho}(x)=1-H(x)
\]
where $H$ is the distribution function associated with $h$. Here, we arrive at a framework similar to the standard graphon. Leveraging the Kallenberg representation of \eqref{eq:KallenbergGraphon}, we first sample $n\sim\mbox{Poisson}(\alpha)$. Then, for $i=1,\ldots n$
we set $z_{ij}=z_{ji}=1$ with probability $M(U_{i},U_{j})$ where $U_{i}$ are uniform
variables and $M$ is defined by%
\[
M(U_{i},U_{j})=1-\exp(-2H^{-1}(U_{i})H^{-1}(U_{j})).
\]
This representation is the same as with the Aldous-Hoover theorem, where
the number of nodes is random and follows a Poisson distribution. As such, the resulting random graph is either trivially empty or dense.

\subsection{Generalized gamma process}
\label{sec:GGPspecialcase}

The generalized gamma process~\cite{Hougaard1986,Aalen1992,Lee1993,Brix1999} (GGP) is a flexible two-parameter CRM, with interpretable parameters and remarkable conjugacy properties~\cite{Lijoi2007,Caron2014}. The process is known as the Hougaard process~\cite{Hougaard1986} when $\lambda$ is the Lebesgue measure, as in this paper, but we will use the term GGP in the rest of this paper. The L\'evy intensity of the GGP is given by
\begin{align}
\rho(dw)=\frac{1}{\Gamma(1-\sigma)}w^{-1-\sigma}\exp(-\tau w)dw,
\label{eq:rhoGGP}
\end{align}
where the two parameters $(\sigma,\tau)$ verify
\begin{align}
(\sigma,\tau)\in(-\infty,0]\times (0,+\infty)\text{ or }(\sigma,\tau)\in(0,1)\times [0,+\infty).
\label{eq:condCRM}
\end{align}
The GGP has different properties if $\sigma\geq 0$ or $\sigma<0$. When $\sigma<0$, the GGP is a finite-activity CRM; more precisely, the number of jumps in $[0,\alpha]$ is finite w.p. 1 and drawn from a Poisson distribution with rate $-\frac{\alpha}{\sigma}\tau^\sigma$ while the jumps $w_i$ are i.i.d. $\Gam(-\sigma,\tau)$.

When $\sigma\geq 0$, the GGP has an infinite number of jumps over any interval $[s,t]$. It includes as special cases the gamma process ($\sigma=0$, $\tau>0$), the stable
process ($\sigma\in(0,1)$, $\tau=0$) and the inverse-Gaussian process ($\sigma=\frac{1}{2}$,$\tau>0$).

The tail L\'evy intensity of the GGP is given by
\[
\overline{\rho}(x)=\int_{x}^{\infty}\frac{1}{\Gamma(1-\sigma)}w^{-1-\sigma}\exp(-\tau
w)dw=\left\{
\begin{array}
[c]{cc}%
\frac{\tau^{\sigma}\Gamma(-\sigma,\tau x)}{\Gamma(1-\sigma)} & \text{if }%
\tau>0\\
\frac{x^{-\sigma}}{\Gamma(1-\sigma)\sigma} & \text{if }\tau=0,
\end{array}
\right.
\]
where $\Gamma(a,x)$ is the incomplete gamma function. Example realizations of the process for various values of $\sigma\geq 0$ are
displayed in Figure~\ref{fig:samples} alongside a realization of an Erd\" os-R\'{e}nyi graph.

\begin{figure}[ptb]
\begin{center}
\subfigure[$G(1000,0.05)$]{\includegraphics[width=0.24\textwidth]{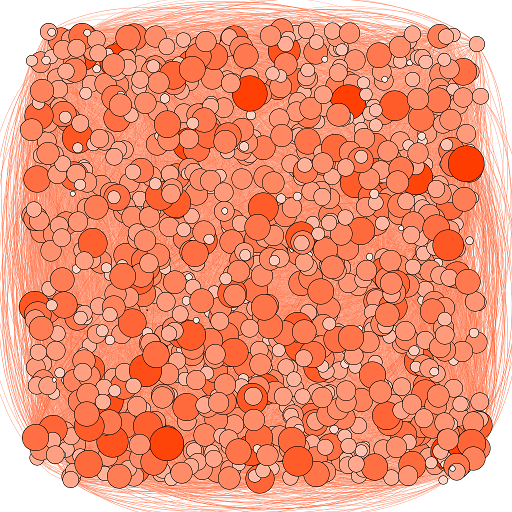}}
\subfigure[$GGP(100,2,0)$]{\includegraphics[width=0.24\textwidth]{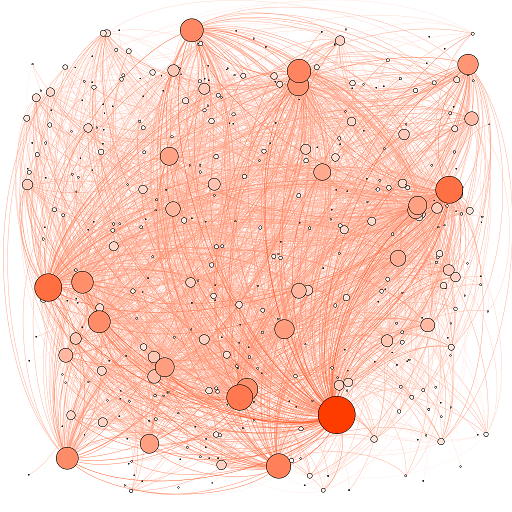}}
\subfigure[$GGP(100,2,0.5)$]{\includegraphics[width=0.24\textwidth]{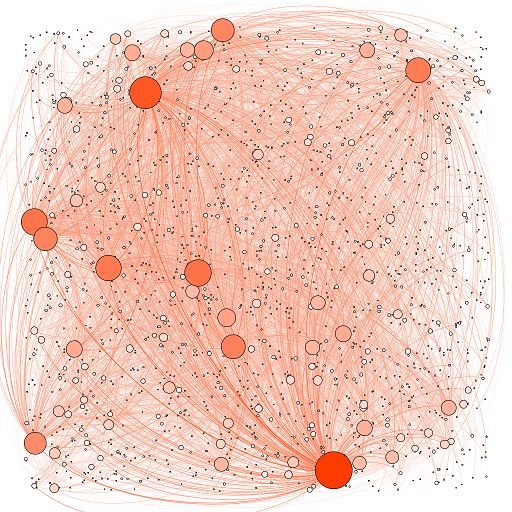}}
\subfigure[$GGP(100,2,0.8)$]{\includegraphics[width=0.24\textwidth]{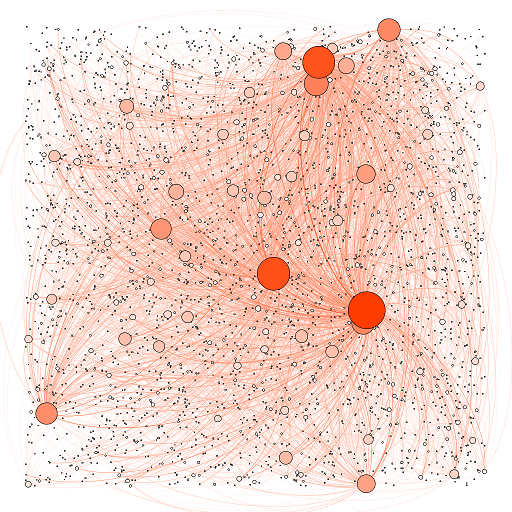}}
\end{center}
\caption{Sample graphs: (a) Erd\" os-R\'{e}nyi graph $G(n,p)$ with $n=1000$ and $p=0.05$ (b-c)
Generalized gamma process graph $GGP(\alpha,\tau,\sigma)$ with $\alpha=100$, $\tau=2$ and (b) $\sigma
=0$, (c) $\sigma=0.5$, (d) $\sigma=0.8$. The size of a node is proportional to
its degree. Graphs have been generated with the software Gephi.}%
\label{fig:samples}
\end{figure}

\paragraph{Exact sampling via an urn approach}
In the case $\sigma\geq 0$, $W_\alpha^\ast$ is an
exponentially tilted stable random variable, for which exact samplers exist~\cite{Devroye2009}. As shown by \citet{Pitman2003} (see also \cite{Lijoi2008}),
the EPPF conditional
on the total mass $W_\alpha^\ast=t$ only depends on the parameter $\sigma$ (and not
$\tau,\alpha$) and is given by
\begin{equation}
\Pi_{k}^{(n)}(m_{1},\ldots,m_{k}|t)=\frac{\sigma^{k}t^{-n}}{\Gamma
(n-k\sigma)g_{\sigma}(t)}\int_{0}^{t}s^{n-k\sigma-1}g_{\sigma}(t-s)ds\left(
\prod_{i=1}^{k}\frac{\Gamma(m_{i}-\sigma)}{\Gamma(1-\sigma)}\right),
\label{eq:EPPF}%
\end{equation}
where $g_{\sigma}$ is the pdf of the positive stable distribution. Plugging the EPPF of \eqref{eq:EPPF} in to \eqref{eq:generalurn} yields the urn process for sampling in the GGP case.  In particular, one can use the generative process \eqref{eq:gengeneralurn} in order to sample exactly from the model.

In the special case of the gamma process ($\sigma =0$), $W_\alpha^\ast$ is a Gamma$(\alpha,\tau)$ random variable and the resulting urn process is given by~\cite{blackwell1973,Pitman1996}:
\begin{align}
U^\prime_{n+1}|(W_\alpha^\ast,U^\prime_1,\ldots,U^\prime_{n}) &\sim \frac{1}{\alpha + n}\lambda_\alpha+\sum_{j=1}^k \frac{m_j}{\alpha +n} \delta_{\widetilde U^\prime_j}.
\label{eq:BlackwellMacQueen}
\end{align}
When $\sigma<0$, the GGP is a compound Poisson process and can thus be sampled exactly.

\paragraph{Expected number of nodes and edges}

In Theorem~\ref{th:inequalitynodes}, we consider bounds on the expected number of nodes in the gamma process case ($\sigma=0,\tau>0$), and the expected number of edges in the multigraph. The proof is in Appendix~\ref{sec:proofGGP}.

\begin{theorem}
\label{th:inequalitynodes}For any $\varepsilon\in(0,1)$,
\begin{multline}
\alpha\log\left(  1+\varepsilon\frac{2(\alpha+1)}{\tau^{2}}\right)  \left(
1-\frac{c_{1}(\alpha)}{1-\varepsilon^{2}}\right)  \leq\mathbb{E}[N_{\alpha
}]\\
\leq\alpha\log\left(  1+\frac{2(\alpha+1)}{\tau^{2}}\right)  +\frac
{2(\alpha+1)}{\tau^{2}+2(\alpha+1)},%
\end{multline}
where $c_{1}(\alpha)=\frac{Var(D_{\alpha}^{\ast})}{\mathbb{E}[D_{\alpha}%
^{\ast}]^{2}}=\frac{\tau^{2}}{\alpha(\alpha+1)}\left(  1+\frac{4\alpha+6}%
{\tau^{2}}\right)  $ is a decreasing function of $\alpha$ with $c(\alpha
)\rightarrow0$ as $\alpha\rightarrow\infty.$ Consequently,
\begin{equation}
\mathbb{E}[N_{\alpha}]=\Theta(\alpha\log\alpha).
\end{equation}
Let $D_\alpha^\ast$ be the number of edges in the directed multigraph. Then
\[
\mathbb{E}[D_\alpha^{\ast}]=\frac{\alpha(\alpha
+1)}{\tau^{2}}
\quad Var[D_\alpha^{\ast}]=\frac{\alpha(\alpha+1)}{\tau^{2}}\left(
1+\frac{4\alpha+6}{\tau^{2}}\right).
\]
\end{theorem}

\paragraph{Power-law properties}
In Theorem~\ref{th:powerlaw}, we show that the GGP directed multigraph has a power-law degree distribution.  A corresponding theorem in the undirected graph case is challenging to show and beyond the scope of this paper, but our empirical results of Figure~\ref{fig:graphproperties} demonstrate that such a power-law property likely holds for the undirected case as well.
\begin{theorem}
Let $N_{\alpha,j}$, $j \geq 1$ be the number of nodes in the directed multigraph $D_\alpha$ with $j$ outgoing or incoming edges (a self edge counts twice for a given node). Then we have the following asymptotic results for the GGP:
\begin{equation}
\frac{N_{\alpha,j}}{N_\alpha}\, \overset{\alpha\uparrow\infty}{\longrightarrow} \,\,\,p_{\sigma,j}
=\frac{\sigma\Gamma(j-\sigma)}{\Gamma(1-\sigma)\Gamma(j+1)},
\end{equation}
almost surely, for fixed $j$. In particular, for large $j$, we have tail behavior
\begin{equation}
p_{\sigma,j}\overset{j\uparrow\infty}\sim\frac{\sigma}{\Gamma(1-\sigma)} j^{-1-\sigma}
\label{eq:tailGGP}
\end{equation}
corresponding to a power-law behavior.
\label{th:powerlaw}
\label{TH:POWERLAW}
\end{theorem}
The proof, which builds on the asymptotic properties of the normalized GGP~\cite{Lijoi2007}, is given in Appendix \ref{sec:proofGGP}.

\paragraph{Sparsity}
The following theorem states that the GGP parameter $\sigma$ tunes the sparsity of the graph. When $\sigma<0$, the graph is dense, whereas it is sparse when $\sigma\geq 0$.
\begin{theorem}
\label{th:sparsityGGP}
Let $N_\alpha$ be the number of nodes and $N_\alpha^{(e)}$ the number of edges in the undirected graph restriction, $Z_\alpha$. Then
\[
N_\alpha^{(e)}=\left \{
\begin{tabular}{ll}
  $\Theta\left (N_\alpha^{2}\right )$ & if $\sigma<0$ \\
  $o\left (N_\alpha^{2}\right )$ & if $\sigma\in [0,1),\tau>0$ \\
  $O\left (N_\alpha^{2/(1+\sigma)}\right )$ & if $\sigma\in (0,1),\tau>0$
\end{tabular}\right .
\]
almost surely as $\alpha\rightarrow\infty$. That is, the underlying graph is sparse if $\sigma\geq 0$ and dense otherwise.
\end{theorem}
\begin{proof}
For $\sigma<0$, the CRM is finite-activity and thus Theorem \ref{th:sparsity} implies that the graph is dense. When $\sigma\geq 0$ the CRM is infinite-activity; moreover, for $\tau>0$, $\mathbb E[W^*_\alpha]<\infty$, and thus Theorem \ref{th:sparsity} implies that the graph is sparse. More precisely, for $\sigma>0$, the tail L\'evy intensity has the asymptotic behavior
\[
\overline{\rho}(x)\overset{x\downarrow 0}{\sim} \frac{\alpha}{\sigma\Gamma(1-\sigma)}x^{-\sigma}
\]
and so Theorem \ref{th:sparsityGGP} follows directly from Theorem \ref{th:sparsity}.
\end{proof}

\begin{remark}
The proof technique requires a finite first moment for the total mass $W^*_\alpha$, and thus excludes the stable process $(\tau=0,\sigma\in(0,1))$, although we conjecture that the graph is also sparse in that case.
\end{remark}

\paragraph{Empirical analysis of graph properties}
\label{sec:empirical}
For the GGP-based formulation, we provide an empirical analysis of our network properties in
Figure~\ref{fig:graphproperties} by simulating undirected graphs using the approach described in Section~\ref{sec:simulation} for
various values of $\sigma,\tau$. We compare to an Erd\" os R\'{e}nyi random graph,
preferential attachment~\cite{Barabasi1999}, and the Bayesian nonparametric network model of~\cite{Lloyd2012}.
The particular features we explore are

\begin{itemize}
\item \textbf{Degree distribution} Figure~\ref{fig:graphproperties}(a)
demonstrates that the model can exhibit power-law behavior providing a
heavy-tailed degree distribution. As shown in Figure~\ref{fig:graphproperties}(b), the model can also handle an exponential cut-off in the tails of the degree distribution, which is an attractive property~\cite{Newman2009}.

\item \textbf{Number of degree 1 nodes} Figure~\ref{fig:graphproperties}(c)
examines the fraction of degree 1 nodes versus number of nodes.

\item \textbf{Sparsity} Figure~\ref{fig:graphproperties}(d) plots the number
of edges versus the number of nodes. The larger $\sigma$, the sparser the
graph. In particular, for the GGP random graph model, we have network growth
at a rate $O(n^{a})$ for $1<a<2$ whereas the Erd\" os R\'{e}nyi (dense) graph
grows as $\Theta(n^{2})$.
\end{itemize}

\begin{figure}[ptb]
\begin{center}%
\begin{tabular}[c]{cc}
\includegraphics[width=.4\textwidth]{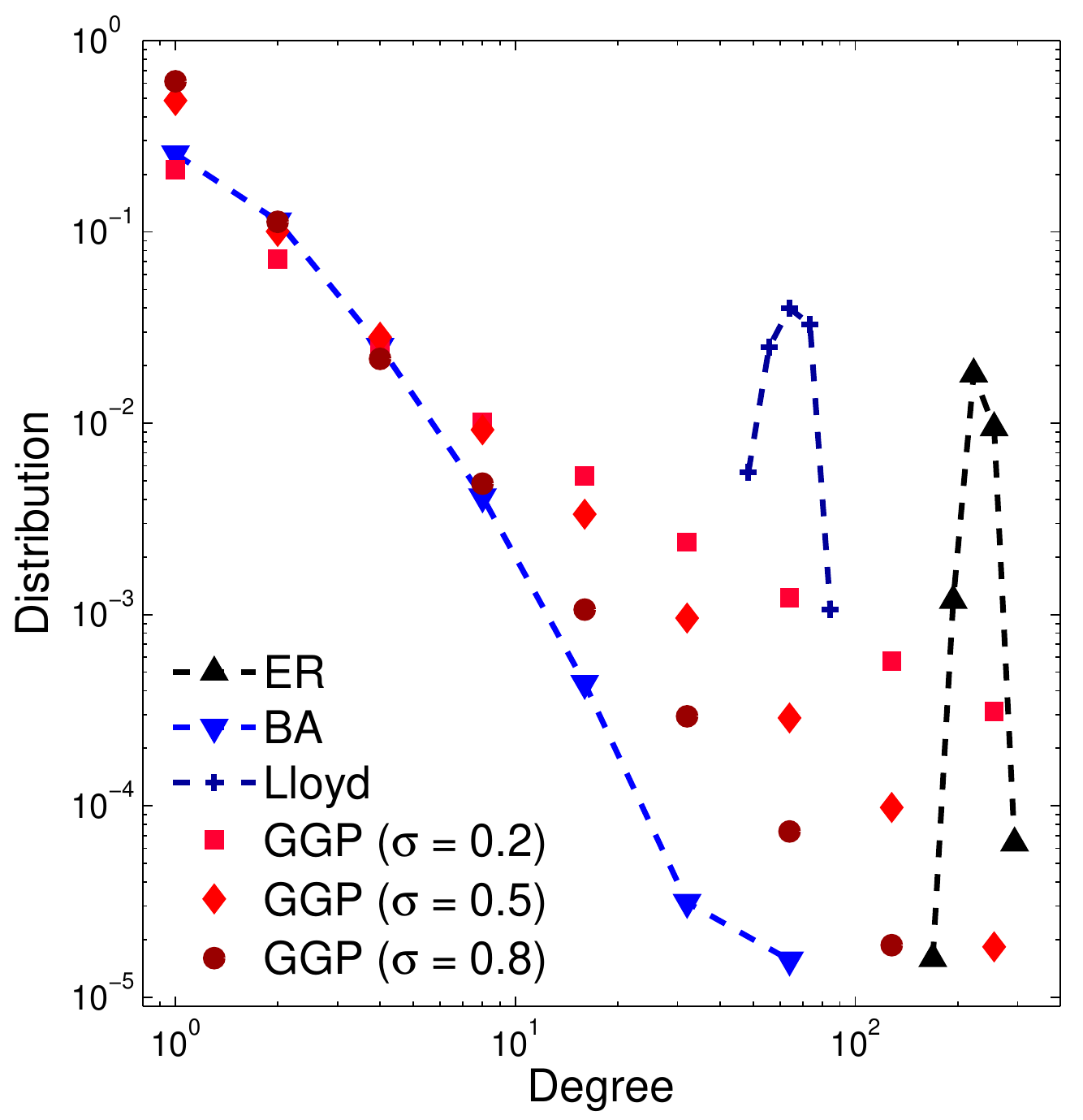} &
\includegraphics[width=.4\textwidth]{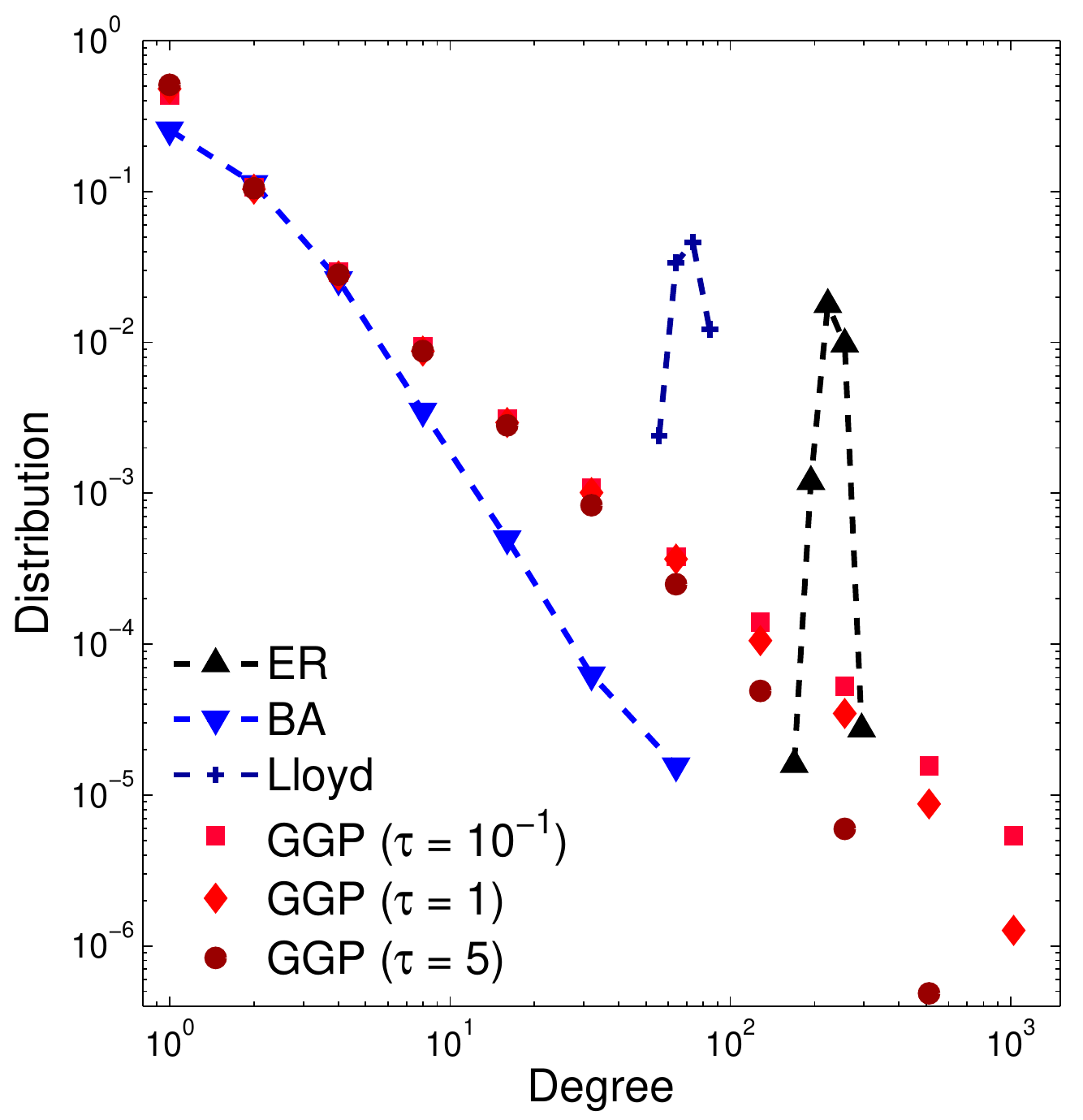}\\
(a) & (b) \\
\includegraphics[width=.4\textwidth]{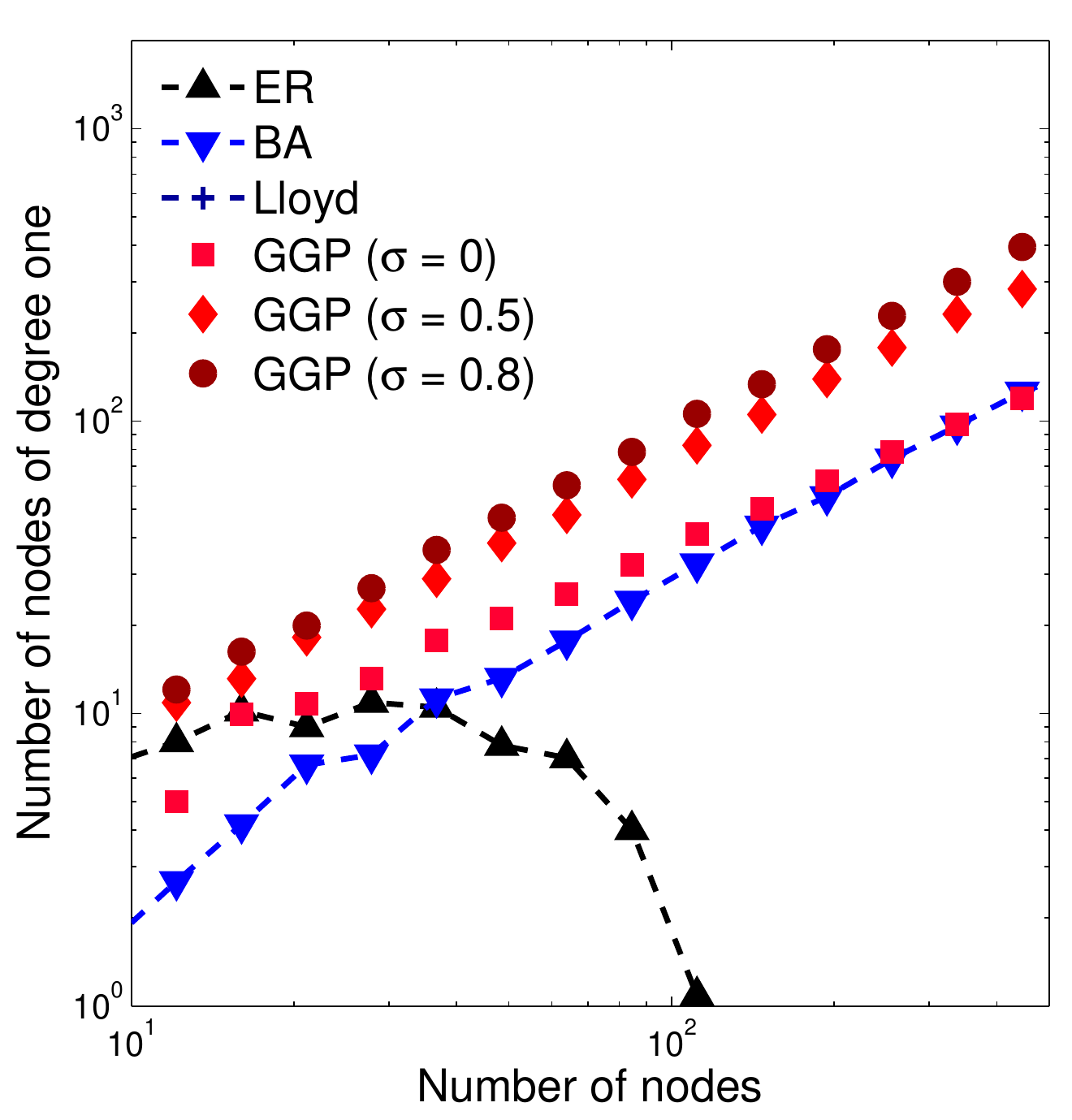} &
\includegraphics[width=.4\textwidth]{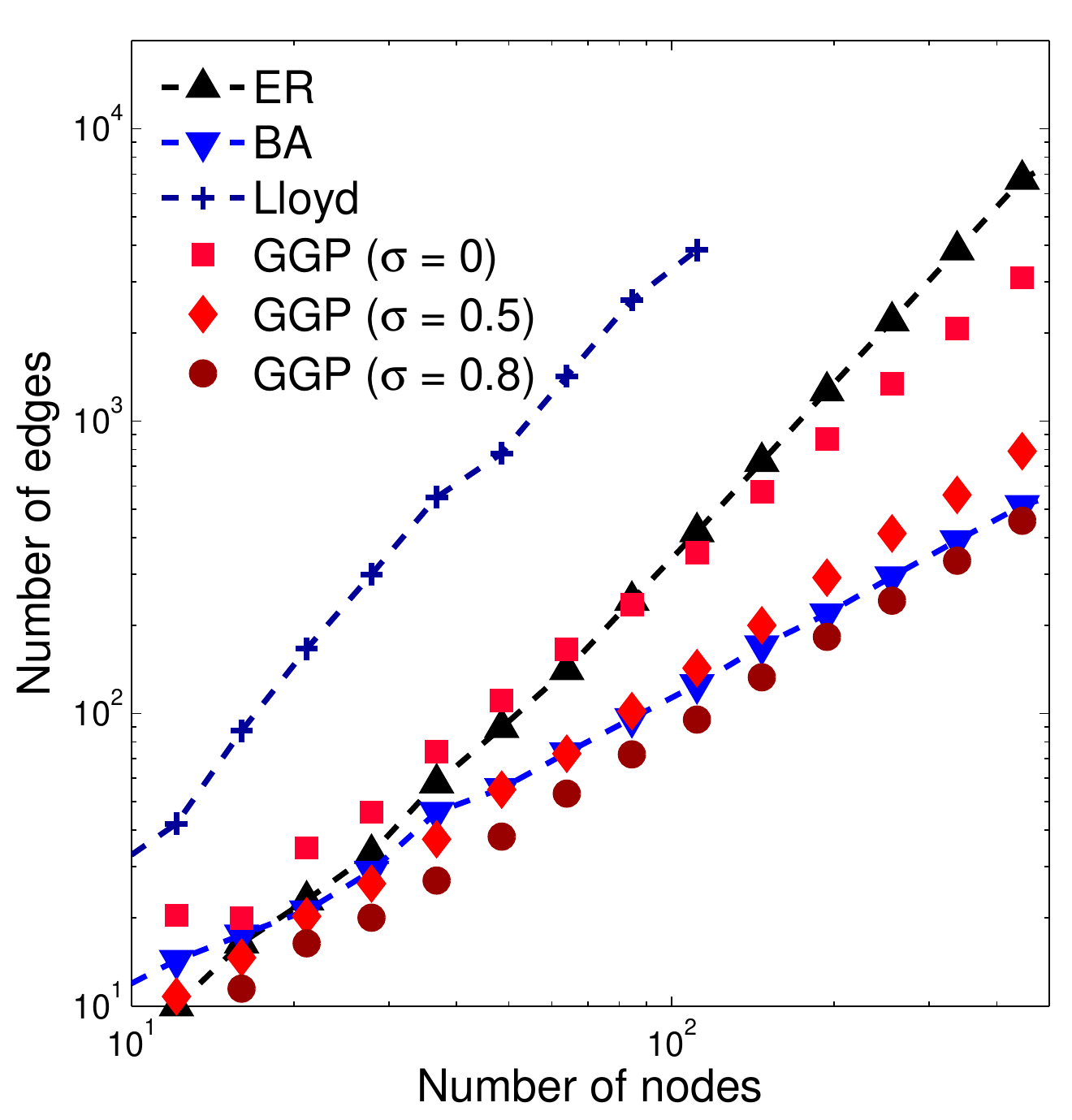}\\
(c) & (d)
\end{tabular}
\end{center}
\caption{\small Examination of the GGP undirected network properties (averaging over graphs with various $\alpha$) in comparison to an Erd\" os R\'{e}nyi $G(n,p)$ model with
$p=0.05$ (ER), the preferential attachment model of~\cite{Barabasi1999} (BA), and the nonparametric formulation of~\cite{Lloyd2012} (Lloyd). (a-b) Degree distribution on a log-log scale for (a) various values of $\sigma$ ($\tau=10^{-2}$) and (b) various values of $\tau$ ($\sigma=0.5$) for the GGP. (c) Number of nodes with
degree one versus the number of nodes on a log-log scale. Note that the Lloyd method leads to dense graphs such that no node has only degree 1. (d) Number of
edges versus the number of nodes. In (d) we note growth at a rate $o(n^{2})$ for our GGP graph models, and $\Theta(n^2)$ for the Erd\" os R\'{e}nyi and Lloyd models (dense graphs).}%
\label{fig:graphproperties}%
\end{figure}

\paragraph{Interpretation of hyperparameters}
Based on the properties derived and explored empirically in this section, we see that our hyperparameters have the following interpretations:
\begin{itemize}
\item $\boldsymbol{\sigma}$--- From Figure~\ref{fig:graphproperties}(a) and (d), $\sigma$ relates to the slope of the degree distribution in its power-law regime and the overall network sparsity.  Increasing $\sigma$ leads to higher power-law exponent and sparser networks.

\item $\boldsymbol{\alpha}$--- From Theorem~\ref{th:inequalitynodes}, $\alpha$ provides an overall scale that affects the number of nodes and directed interactions, with larger $\alpha$ leading to larger networks.

\item $\boldsymbol{\tau}$--- From Figure~\ref{fig:graphproperties}(b), $\tau$ determines the exponential decay of the tails of the power-law degree distribution, with $\tau$ small looking like pure power-law. This is intuitive from the form of $\rho(dw)$ in \eqref{eq:rhoGGP}, where we see that $\tau$ affects large weights more than small ones.
\end{itemize}

\section{Posterior characterization and inference}
\label{sec:MCMC}

In this section we target inferring the posterior distribution of the sociability parameters, $w_i$, restriction value $\alpha$, and CRM hyperparameters.  In the special case of GGPs, our hyperparameters of interest are then the set $(\alpha,\sigma,\tau)$.

\subsection{Directed multigraph and undirected simple graph}
\label{sec:MCMCgraphs}

We first characterize the conditional distribution of the restricted CRM $W_{\alpha}$ given the
directed graph $D_{\alpha}$ (see \eqref{eq:condpoisson} and surrounding text). In what follows, we utilize the fact that the conditional
CRM\ $W_{\alpha}$ given $D_{\alpha}$ can be decomposed as a sum of (i) a measure with fixed locations
$\theta_{i}$ and random weights $w_{i}$, corresponding to nodes for which we
observed at least one connection, and (ii) a measure with random weights and random
atoms, corresponding to the remaining set of nodes. We denote the total mass of this remaining weight as $w_\ast$.

\begin{theorem}
\label{th:posteriorsimple}Let $(\theta_{1},\ldots,\theta_{N_{\alpha}})$, $N_{\alpha}\geq0,$ be
the set of support points of $D_{\alpha}$ such that $D_{\alpha}=\sum_{1\leq
i,j\leq N_{\alpha}}n_{ij}\delta_{(\theta_{i},\theta_{j})}$. Let $m_{i}=\sum_{j=1}%
^{N_{\alpha}}(n_{ij}+n_{ji})>0$ for $i=1,\ldots,N_{\alpha}$. The conditional distribution of
$W_{\alpha}$ given $D_{\alpha}$ is equivalent to the distribution of%
\begin{equation}
w_{\ast}\sum_{i=1}^{\infty}\widetilde{P}_{i}\delta_{\widetilde{\theta}_{i}%
}+\sum_{i=1}^{N_{\alpha}}w_{i}\delta_{\theta_{i}}%
\end{equation}
where $\widetilde{\theta}_{i}\sim\Unif ([0,\alpha])$, and the weights
$(\widetilde{P}_{i})_{i=1,2,\ldots},$ with $\widetilde{P}_{1}>\widetilde{P}%
_{2}>\ldots$ and $\sum_{i=1}^{\infty}\widetilde{P}_{i}=1,$ are distributed
from a Poisson-Kingman distribution~\cite[Definition 3 p.6]{Pitman2003} with L\'evy intensity $\rho$, conditional on $w_{\ast}$
\[
(\widetilde P_{i})|w_{\ast}\sim \PK(\rho|w_{\ast}).
\]
Finally, the weights $(w_{1},\ldots,w_{N_{\alpha}},w_{\ast})$ are jointly dependent
conditional on $D_{\alpha}$, with the following posterior distribution:%
\begin{equation}
p(w_{1},\ldots,w_{N_{\alpha}},w_{\ast}|D_{\alpha})\propto\left[  \prod_{i=1}^{N_{\alpha}}%
w_{i}^{m_{i}}\right]  e^{  -\left(  \sum_{i=1}^{N_{\alpha}}w_{i}+w_{\ast
}\right)  ^{2}}  \left[  \prod_{i=1}^{N_{\alpha}}\rho(w_{i})\right]  \times
g_{\alpha}^{\ast}(w_{\ast})\label{eq: condw}%
\end{equation}
where $g_{\alpha}^{\ast}$ is the probability density function of the random
variable $W^*_\alpha=W_{\alpha}([0,\alpha])$, with Laplace transform
\begin{equation}
\mathbb{E}[e^{-tW^*_{\alpha}}]=e^{  -\alpha\psi(t)}.
\end{equation}
\end{theorem}

\begin{proof}
The proof builds on the Palm formula for Poisson random measures~\cite{Prunster2002,James2002,James2005,James2009} and is described in Appendix \ref{sec:proofposterior}.
\end{proof}\bigskip

Note that the normalized weights $(\widetilde P_{i})_{i=1,2,\ldots}$ and locations $(\widetilde\theta
_{i})_{i=1,2,\ldots}$ are not likelihood identifiable, as the likelihood only brings
information on the weights of the observed nodes, and on the total mass
$w_{\ast}$ of the remaining nodes. Additionally, note that the conditional distribution of $(w_{1},\ldots
,w_{N_{\alpha}},w_{\ast})$ given $D_{\alpha}$ does not depend on the locations
$(\theta_{1},\ldots,\theta_{N_{\alpha}})$ because we considered a homogeneous CRM. This fact is important since the
locations $(\theta_{1},\ldots,\theta_{N_{\alpha}})$ are typically not observed, and our algorithm outlined below will not consider these terms in the inference.

We now specialize to the special case of the GGP, for
which we derive an MCMC\ sampler for posterior inference. Let $\phi
=(\alpha,\sigma,\tau)$ be the set of hyperparameters that we also want to estimate.
We will assume improper priors on those parameters:%
\[
p(\alpha)\propto\frac{1}{\alpha}, \,\,\, p(\sigma)\propto\frac{1}{1-\sigma
}, \,\,\, p(\tau)\propto\frac{1}{\tau}.%
\]
To emphasize the dependence of the L\'{e}vy measure and pdf of the total mass $w_*$ on
the hyperparameters, we write $\rho(w|\sigma,\tau)$ and $g_{\alpha,\sigma
,\tau}^{\ast}(w_{\ast})$. We are interested in approximating the posterior $p(w_{1},\ldots,w_{N_{\alpha}},w_{\ast},\phi|(n_{ij})_{1\leq i,j\leq
N_{\alpha}})$ for a directed multigraph or $p(w_{1},\ldots,w_{N_{\alpha}},w_{\ast},\phi
|(z_{ij})_{1\leq i,j\leq N_{\alpha}})$ for a simple graph.

In the case of a simple graph, we will simply impute the missing directed
edges in the graph. For each $i\leq j$ such that $z_{ij}=1$, we introduce
latent variables $\overline{n}_{ij}=n_{ij}+n_{ji}$ with conditional
distribution%
\begin{equation}
\overline{n}_{ij}|z,w\sim\left\{
\begin{array}
[c]{ll}%
\delta_{0} & \text{if }z_{ij}=0\\
\text{tPoisson}(2w_{i}w_{j}) & \text{if }z_{ij}=1,\text{ }i\neq j\\
\text{tPoisson}(w_{i}^{2}) & \text{if }z_{ii}=1,\text{ }i=j,
\end{array}
\right.  \label{eq:condn}%
\end{equation}
where tPoisson($\lambda$) is the zero-truncated Poisson distribution with pdf%
\[
\frac{k^{\lambda}\exp(-\lambda)}{(1-\exp(-\lambda))k!}\text{, for
}k=1,2,\ldots
\]
By convention, we set $\overline{n}_{ij}=\overline{n}_{ji}$ for $j<i$ and
$m_{i}=\sum_{j=1}^{N_{\alpha}}\overline{n}_{ij}$.

For efficient exploration of the target posterior, we propose using a Hamiltonian Monte Carlo (HMC)
algorithm~\cite{Duane1987,Neal2011} within Gibbs to update the weights
$(w_{1},\ldots,w_{N_{\alpha}})$.  The HMC step requires computing the gradient of the log-posterior, which in our case, letting $\omega_i=\log w_{i}$, is given by
\begin{equation}
\left[  \nabla_{\omega_{1:N_{\alpha}}} \log p(\omega_{1:N_{\alpha}},w_{\ast}|D_{\alpha})\right]_{i}=m_{i}-\sigma-w_{i}\left(  \tau+2\sum_{j=1}^{N_{\alpha}}%
w_{j}+2w_{\ast}\right).  \label{eq:jacobian}%
\end{equation}
For the update of the total mass $w_{\ast}$ and
hyperparameters $\phi$, we use a Metropolis-Hastings step. Note that, except
in some particular cases ($\sigma=0,\frac{1}{2}$), the density $g_{\alpha
,\sigma,\tau}^{\ast}(w_{\ast})$ does not admit any analytical expression. We
therefore use a specific proposal for $w_{\ast}$ based on exponential tilting
of $g_{\alpha,\sigma,\tau}^{\ast}$ that alleviates the need to evaluate this
pdf in the Metropolis-Hasting ratio (see details in Appendix \ref{sec:detailsMCMC}). To summarize,
the MCMC sampler is defined as follows:

\begin{enumerate}
\item Update the weights $(w_{1},\ldots,w_{N_{\alpha}})$ given the rest using an HMC update

\item Update the total mass $w_{\ast}$ and hyperparameters $\phi
=(\alpha,\sigma,\tau)$ given the rest using a Metropolis-Hastings update

\item {}[Undirected graph] Update the latent counts $(\overline{n}_{ij})$
given the rest using the conditional distribution \eqref{eq:condn} or a Metropolis-Hastings update
\end{enumerate}

Note that the computational bottlenecks lie in steps 1 and 3, which roughly
scale linearly in the number of nodes/edges, respectively, although one can parallelize
step 3 over edges. If $L$ is the number of leapfrog steps in the HMC algorithm, $n_{\text{iter}}$
the number of MCMC\ iterations, the overall complexity is in $O(n_{\text{iter}%
}(LN_{\alpha}+N_\alpha^{(e)}))$. We show in Section~\ref{sec:experiments} that the algorithm scales well to large networks with hundreds of thousands of nodes and edges. To efficiently scale HMC to even larger collections of nodes/edges, one can deploy the methods of \citet{ChenFoxGuestrin2014}.

\subsection{Bipartite graph}

For the bipartite graph case, the posterior characterization follows as proposed by \citet{Caron2012}.  However, our proposed data augmentation is different and leads to a simpler form for the sampler.

\begin{theorem}
\label{th:posteriorbipartite}Let $(\theta_{1},\ldots,\theta_{N_{\alpha}})$,
$(\theta_{1}^{\prime},\ldots,\theta_{N_{\alpha}^{\prime}}^{\prime})$ with $N_{\alpha},N_{\alpha}^{\prime
}\geq0,$ be the set of support points of $D_{\alpha}$ and thus $D_{\alpha
}=\sum_{1\leq i,j\leq N_{\alpha}}n_{ij}\delta_{(\theta_{i},\theta_{j}^{\prime})}$. Let
$m_{i}=\sum_{j=1}^{N_{\alpha}^{\prime}}n_{ij}$ and $m_{j}^{\prime}=\sum_{i=1}^{N_{\alpha}}%
n_{ij}$ The conditional distribution of $W_{\alpha}$ given $D_{\alpha
},W_{\alpha}^{\prime}$ is equivalent to the distribution of%
\begin{equation}
\widetilde{W}+\sum_{i=1}^{N_{\alpha}}w_{i}\delta_{\theta_{i}}%
\end{equation}
where $(w_{1},\ldots,w_{N_{\alpha}})$ are independent of $\widetilde{W}$ with
\begin{equation}
p(w_{i}|D_{\alpha},W_{\alpha}^{\prime})\propto w_{i}^{m_{i}}e^{
-w_{i}\left(  \sum_{j=1}^{N_{\alpha}^{\prime}}w_{j}^{\prime}+w_{\ast}^{\prime}\right)
}  \rho(w_{i})
\end{equation}
and $\widetilde{W}\sim \CRM(\widetilde{\rho},\lambda_{\alpha})$ is a CRM with
exponentially tilted L\'{e}vy intensity
\begin{equation}
\widetilde{\rho}(w)=\rho(w)e^{  -w\left(  \sum_{j=1}^{N_{\alpha}^{\prime}}%
w_{j}^{\prime}+w_{\ast}^{\prime}\right)}  .
\end{equation}
In particular, for the generalized gamma process, we have
\begin{equation}
w_{i}|D_{\alpha},W_{\alpha}^{\prime}\sim\Gam\left(  m_{i}-\sigma
,\tau+\sum_{j=1}^{N_{\alpha}^{\prime}}w_{j}^{\prime}+w_{\ast}^{\prime}\right)
\label{eq:wibipartite}%
\end{equation}
and the total mass $w_{\ast}$ of $\widetilde{W}$ is distributed from an
exponentially tilted stable distribution with pdf
\begin{equation}
p(w_{\ast}|\text{rest})=\frac{e^{ - w_{\ast}\left(  \sum_{j=1}%
^{N_{\alpha}^{\prime}}w_{j}^{\prime}+w_{\ast}^{\prime}\right)  }  g_{\alpha
}(w_{\ast})}{e^{  -\psi\left(  \sum_{j=1}^{N_{\alpha}^{\prime}}w_{j}^{\prime
}+w_{\ast}^{\prime}\right)  }  },\label{eq:wstarbipartite}%
\end{equation}
from which one can sample exactly~\cite{Devroye2009,Hofert2011}. Additionally, the
marginal likelihood is expressed as
\begin{equation}
\mathbb{E}_{W_{\alpha}}[p(D_{\alpha}|W_{\alpha}^{\prime})]=e^{
-\alpha\psi\left(  \sum_{j=1}^{N_{\alpha}^{\prime}}w_{j}^{\prime}+w_{\ast}^{\prime
}\right)  }  \alpha^{N_{\alpha}}\prod_{i=1}^{N_{\alpha}}\kappa\left(  m_{i},\sum
_{j=1}^{N_{\alpha}^{\prime}}w_{j}^{\prime}+w_{\ast}^{\prime}\right)d\theta_i,
\label{eq:marginaldist}%
\end{equation}
where $\kappa(n,z)=\int_{0}^{\infty}w^{n}\exp(-zw)\rho(w)dw.$
\end{theorem}
\begin{proof}
The proof is described by \citet{Caron2012} and in Appendix \ref{sec:proofposterior} for completeness.
\end{proof}

Let $\phi=(\alpha,\sigma,\tau)$ and $\phi^{\prime}=(\alpha^{\prime}%
,\sigma^{\prime})$ be, respectively, the parameters of the L\'{e}vy intensity of
$W$ and $W^{\prime}$. To preserve identifiability, we set the
parameter $\tau^{\prime}$ to $1$. The MCMC sampler for approximating
$p(w_{1:N_{\alpha}},w_{\ast},w_{1:N_{\alpha}^{\prime}}^{\prime},w_{\ast}^{\prime},\phi
,\phi^{\prime}|Z_{\alpha})$ iterates as follows:

\begin{enumerate}
\item Update $\alpha,\sigma,\tau$ given $w_{1:N_{\alpha}^{\prime}}^{\prime}$ using a
Metropolis-Hastings step with acceptance ratio calculated with
(\ref{eq:marginaldist})

\item Update $w_{1:N_{\alpha}}$ given $(w_{1:N_{\alpha}^{\prime}}^{\prime},w_{\ast}^{\prime
},\alpha,\sigma,\tau)$ using (\ref{eq:wibipartite})

\item Update $w_{\ast}$ given $(w_{1:N_{\alpha}^{\prime}}^{\prime},w_{\ast}^{\prime
},\alpha,\sigma,\tau)$ using (\ref{eq:wstarbipartite})

\item Update the latent $n_{ij}$ given $w_{1:N_{\alpha}^{\prime}}^{\prime},w_{1:N_{\alpha}}$ as
\[
n_{ij}|z,w,w^{\prime}\sim\left\{
\begin{array}
[c]{ll}%
\delta_{0} & \text{if }z_{ij}=0\\
\text{tPoisson}(w_{i}w_{j}^{\prime}) & \text{if }z_{ij}=1
\end{array}
\right.
\]

\end{enumerate}

The model is symmetric in $(w,w^{\prime})$, so the first three steps can be
repeated for updating $(\alpha^{\prime},\sigma^{\prime},\tau^{\prime
},w_{1:N_{\alpha}^{\prime}}^{\prime},w_{\ast}^{\prime})$. Full algorithmic details are given in Appendix~\ref{sec:detailsMCMC}.

\section{Experiments}
\label{sec:experiments}

\subsection{Simulated data}

We first study the convergence of the MCMC algorithm on simulated data where the graph is simulated from our model. We simulate a GGP undirected graph with parameters $\alpha=300,\sigma=0.5,\tau=1$. Note that we are in the sparse regime. The sampled graph has 13,995 nodes and 76,605 edges. We run 3 MCMC chains each with 40,000 iterations and with different initial values. $L=10$ leapfrog steps are used, and the stepsize of the leapfrog algorithm is adapted during the first 10,000 iterations so as to obtain an acceptance rate of 0.6. Standard deviations of the random walk Metropolis-Hastings for $\log \tau$ and $\log (1-\sigma)$ are set to 0.02. It takes 10 minutes with Matlab on a standard computer (CPU@3.10GHz, 4 cores) to run the 3 chains successively. Trace plots of the parameters $\alpha$, $\sigma$, $\tau$ and $w_*$ are given in Figure~\ref{fig:simupowerlawtrace}. The potential scale factor reduction~\cite{Brooks1998,Gelman2014} is computed for all 13,999 parameters $(w_{1:N_{\alpha}},w_*,\alpha,\sigma,\tau)$ and has a maximum value of 1.01, indicating convergence of the algorithm. This is rather remarkable as the MCMC sampler actually samples from a target distribution of dimension 13,995+76,605+4=90,604. Posterior credible intervals of the sociability parameters $w_i$ of the nodes with highest degrees and log-sociability parameters $\log w_i$ of the nodes with lowest degrees are displayed in Figure~\ref{fig:simupowerlaw_w}(a) and (b), respectively, showing the ability of the method to accurately recover sociability parameters of both low and high degree nodes.

\begin{figure}[ptb]
\begin{center}%
\subfigure[$\alpha$]{\includegraphics[width=.45\textwidth]{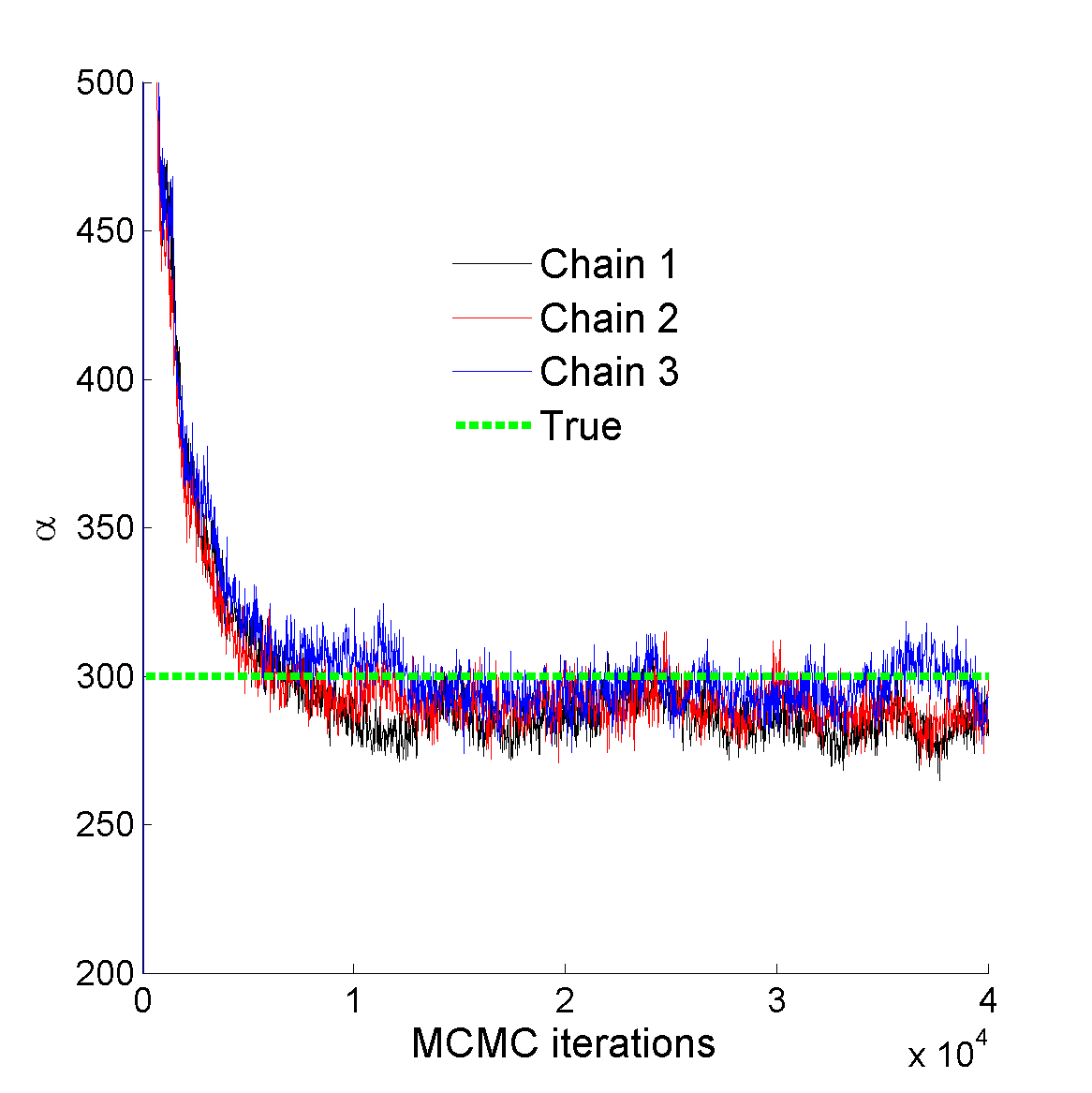}}
\subfigure[$\sigma$]{\includegraphics[width=.45\textwidth]{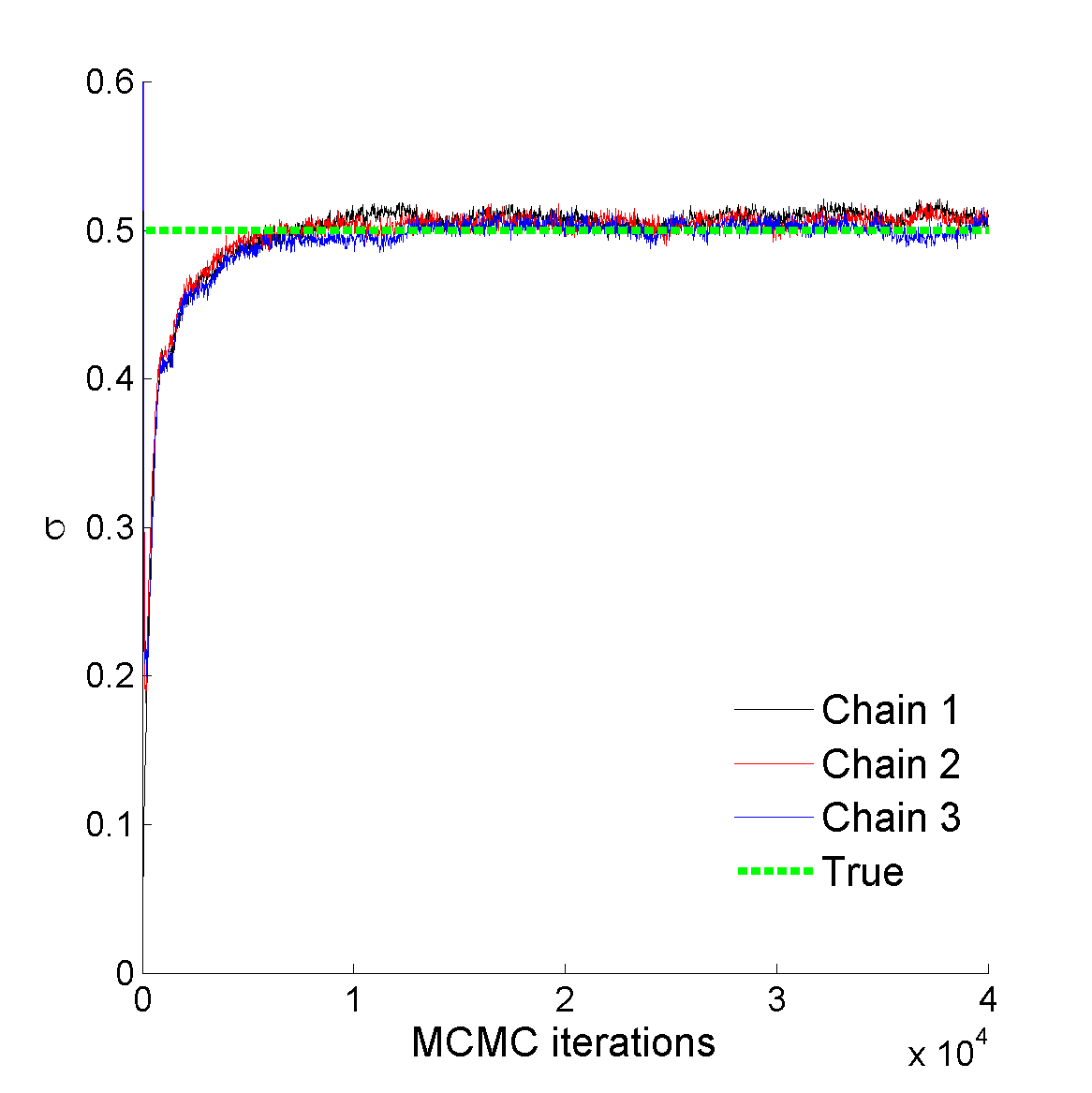}}
\subfigure[$\tau$]{\includegraphics[width=.45\textwidth]{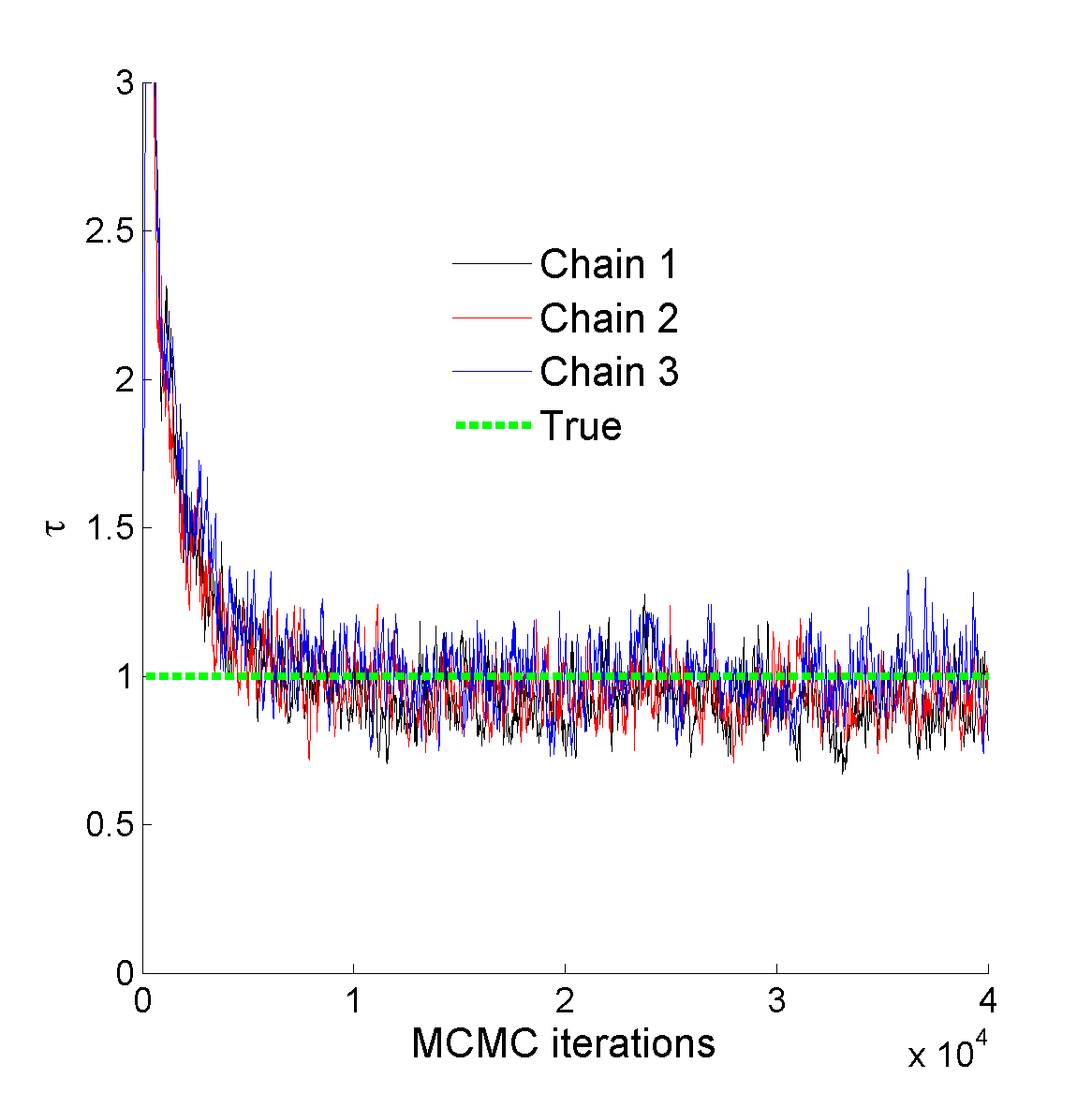}}
\subfigure[$w_*$]{\includegraphics[width=.45\textwidth]{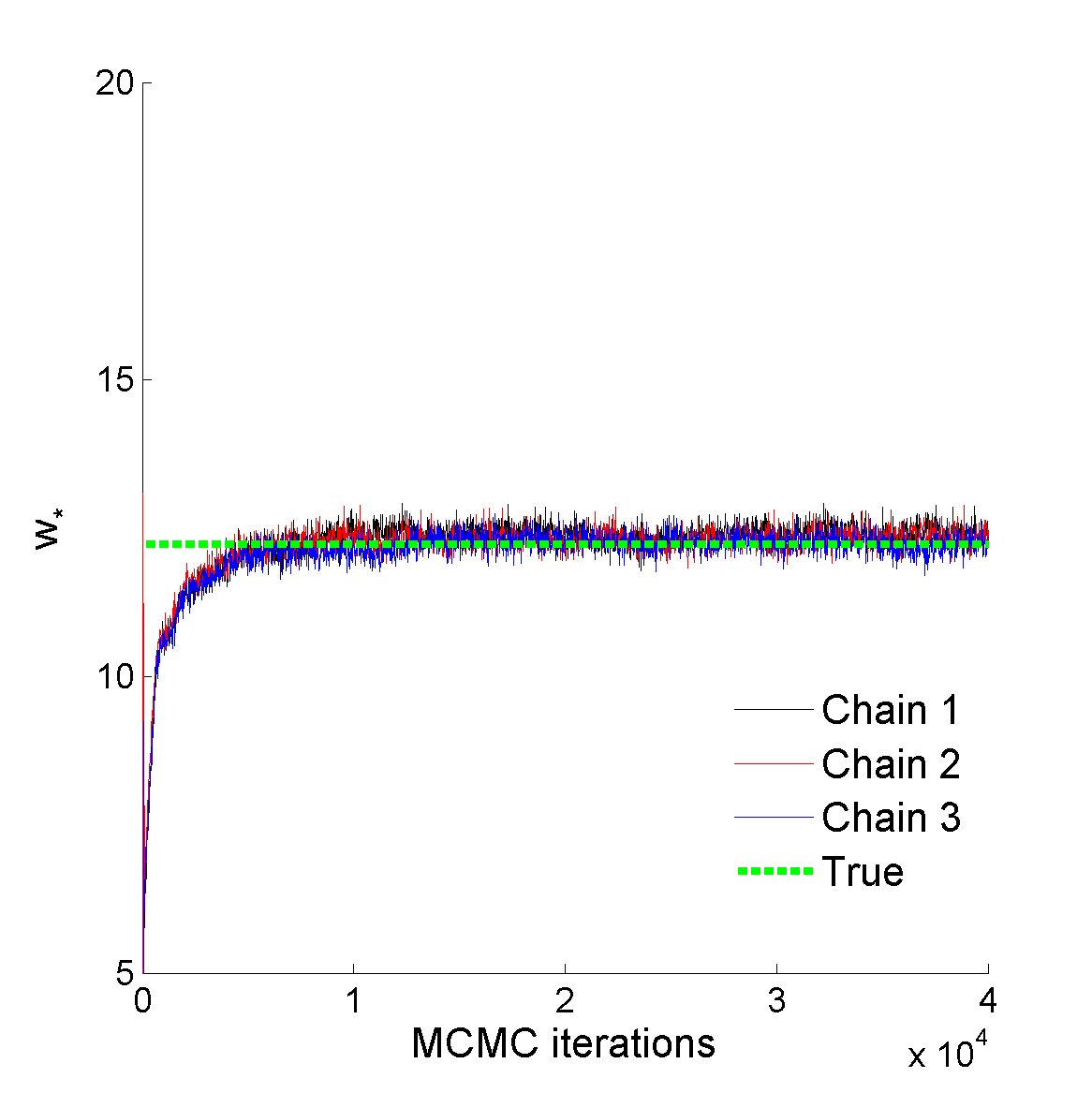}}
\end{center}
\caption{MCMC trace plots of parameters (a) $\alpha$ (b) $\sigma$, (c) $\tau$ and (d) $w_*$ for a graph generated from a GGP model with parameters $\alpha=300,\sigma=0.5,\tau=1$.}%
\label{fig:simupowerlawtrace}%
\end{figure}

\begin{figure}[ptb]
\begin{center}%
\subfigure[50 nodes with highest degree]{\includegraphics[width=.48\textwidth]{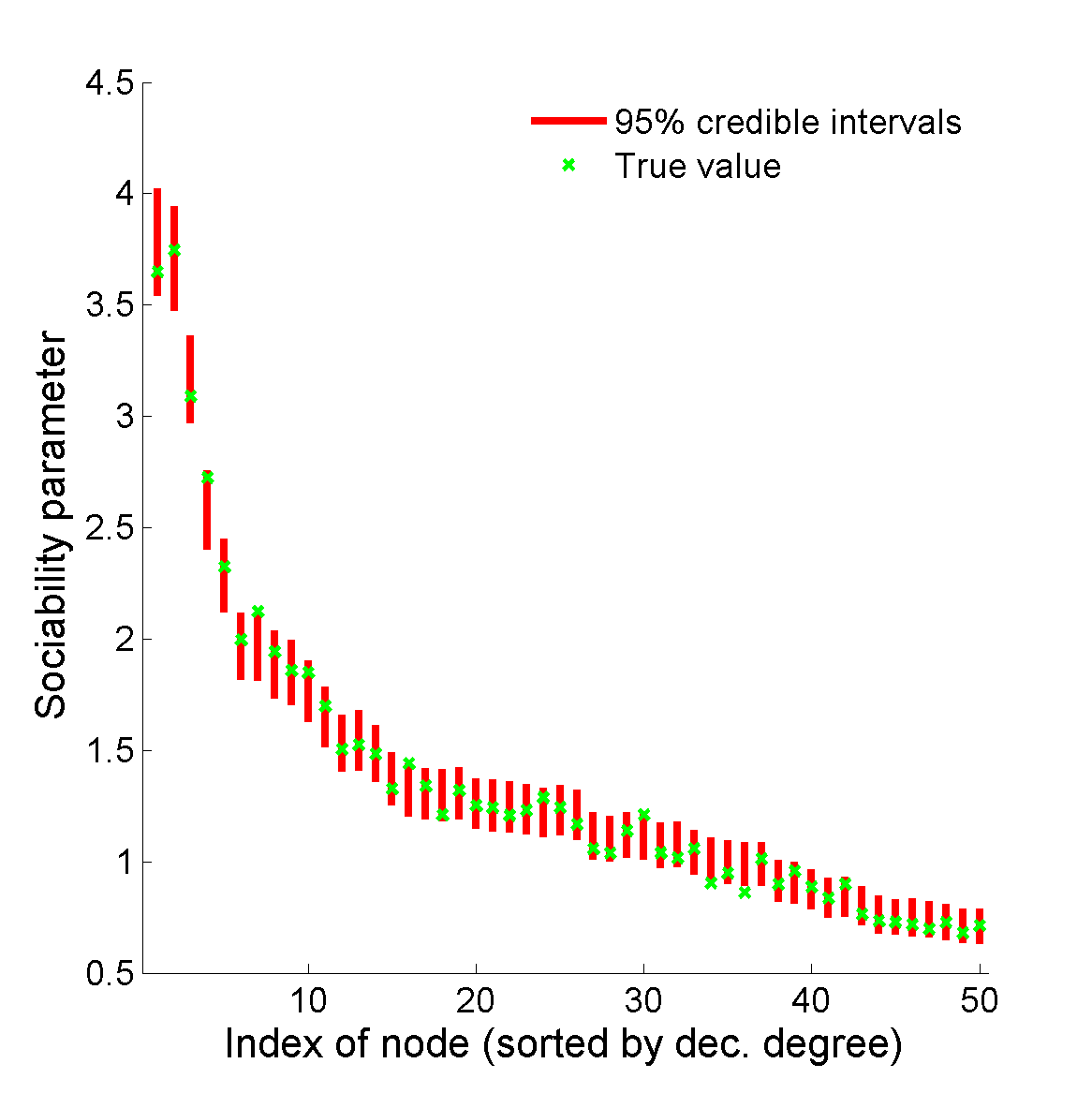}}
\subfigure[50 nodes with lowest degree]{\includegraphics[width=.48\textwidth]{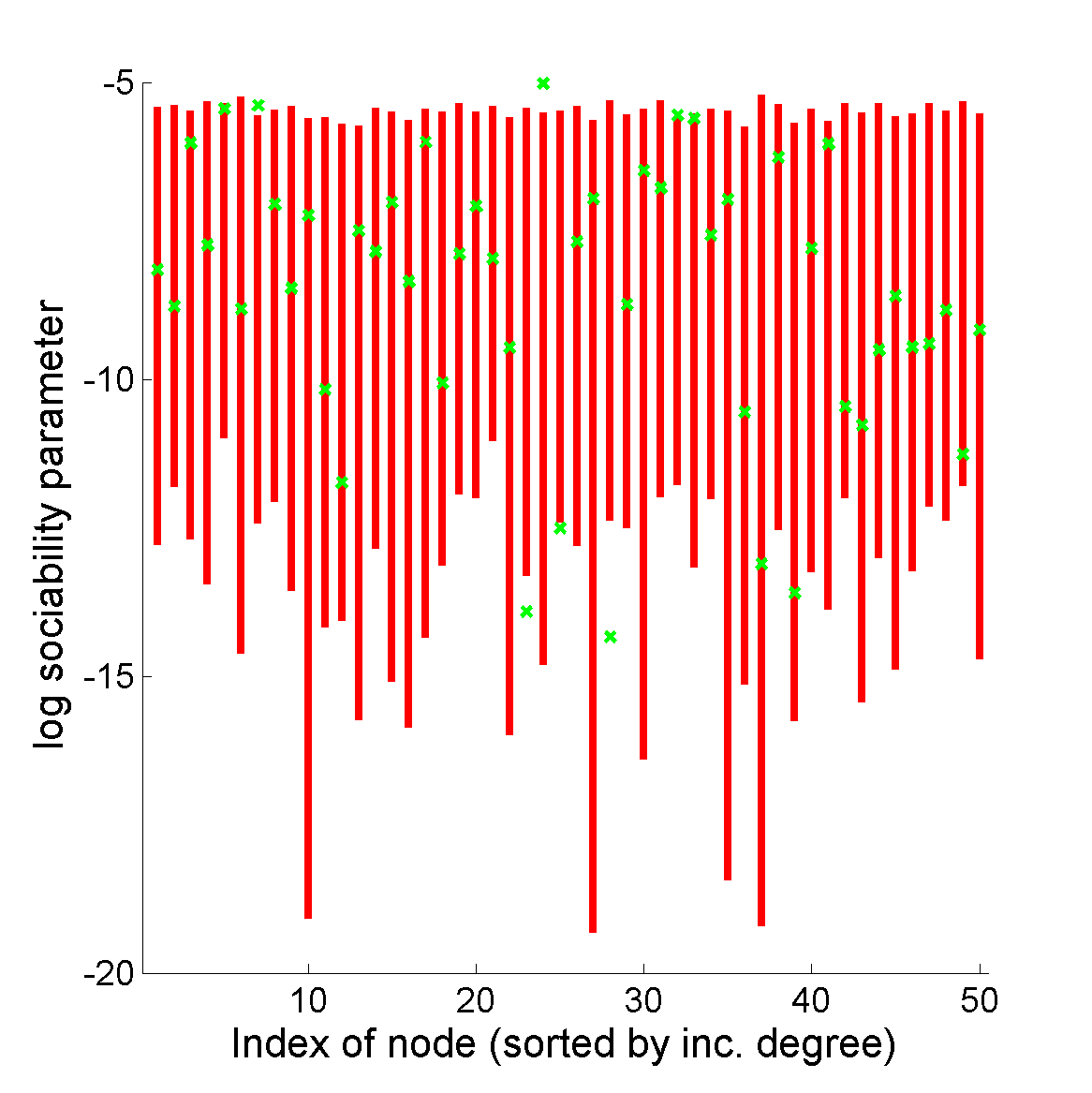}}
\end{center}
\caption{95 \% posterior intervals of (a) the  sociability parameters $w_i$ of the 50 nodes with highest degree and (b) the log-sociability parameter $\log w_i$ of the 50 nodes with lowest degree, for a graph generated from a GGP model with parameters $\alpha=300,\sigma=0.5,\tau=1$. True values are represented by a green star.}%
\label{fig:simupowerlaw_w}%
\end{figure}

To show the versatility of the GGP graph model, we now examine our approach when the observed graph is actually generated from an Erd\"os-R\'enyi model with $n=1,000$ and $p=0.01$. The generated graph has 1,000 nodes and 5,058 edges. We ran 3 MCMC chains with the same specifications as above. In this dense-graph regime, the following transformation of our parameters $\alpha$, $\sigma$ and $\tau$ is more informative: $\varsigma_1 = -\frac{\alpha}{\sigma}\tau^\sigma$, $\varsigma_2 = -\frac{\sigma}{\tau}$ and  $\varsigma_3 = -\frac{\sigma}{\tau^2}$. When $\sigma<0$, $\varsigma_1$ corresponds to the expected number of nodes, $\varsigma_2$ to the mean of the sociability parameters and $\varsigma_3$ to their variance (see Section~\ref{sec:GGPspecialcase}). In contrast, the parameters $\sigma$ and $\tau$ are only weakly identifiable in this case. The potential scale reduction factor is computed on $(w_{1:N_{\alpha}},w_*,\varsigma_1,\varsigma_2,\varsigma_3)$, and its maximum value is 1.01, indicating convergence. Trace plots are shown in Figure~\ref{fig:simuERtrace} for $\varsigma_1$, $\varsigma_2$, $\varsigma_3$ and $w_*$. The value of $\varsigma_1$ converges around the true number of nodes, $\varsigma_2$ to the true sociability parameter $\sqrt{-\frac{1}{2}\log(1-p)}$ (constant across nodes for the Erd\"os-R\'enyi model), while $\varsigma_3$ is close to zero as the variance over the sociability parameters is very small. The total mass is very close to zero, indicating that there are no nodes with degree zero.

\begin{figure}[ptb]
\begin{center}%
\subfigure[$\varsigma_1=-\frac{\alpha}{\sigma} \tau^\sigma$]{\includegraphics[width=.45\textwidth]{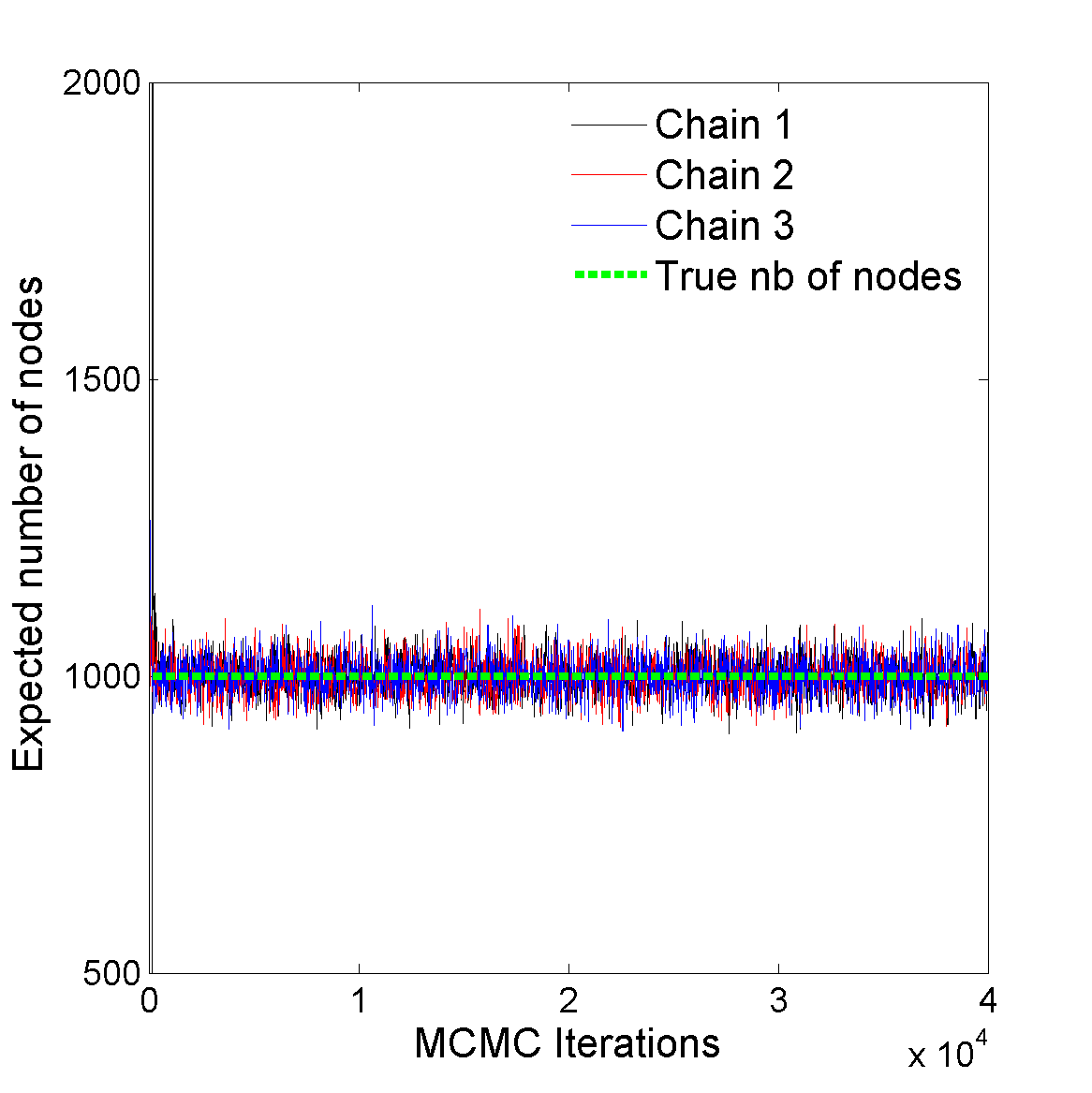}}
\subfigure[$\varsigma_2=-\frac{\sigma}{\tau}$]{\includegraphics[width=.45\textwidth]{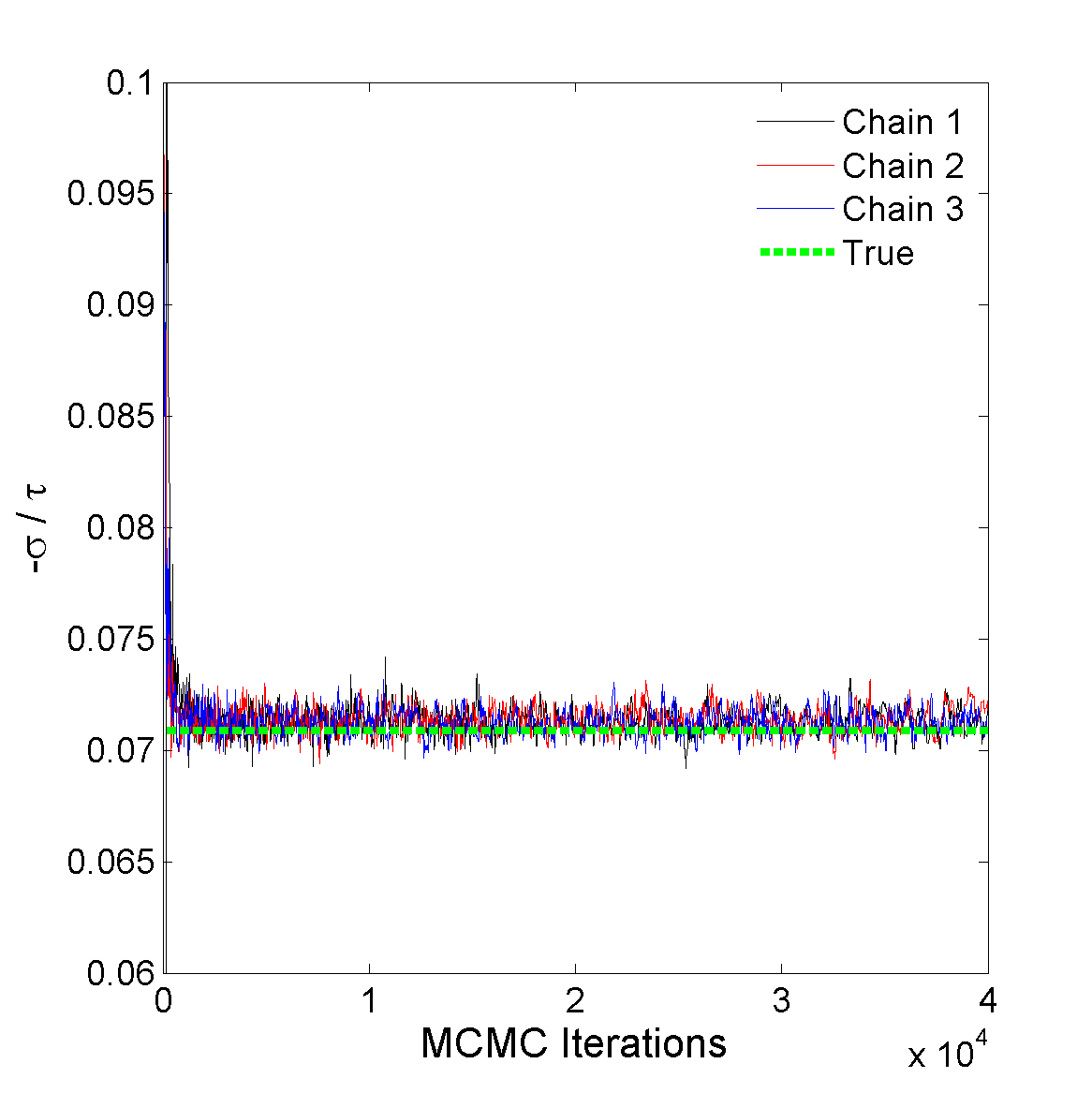}}
\subfigure[$\varsigma_3=-\frac{\sigma}{\tau^2}$]{\includegraphics[width=.45\textwidth]{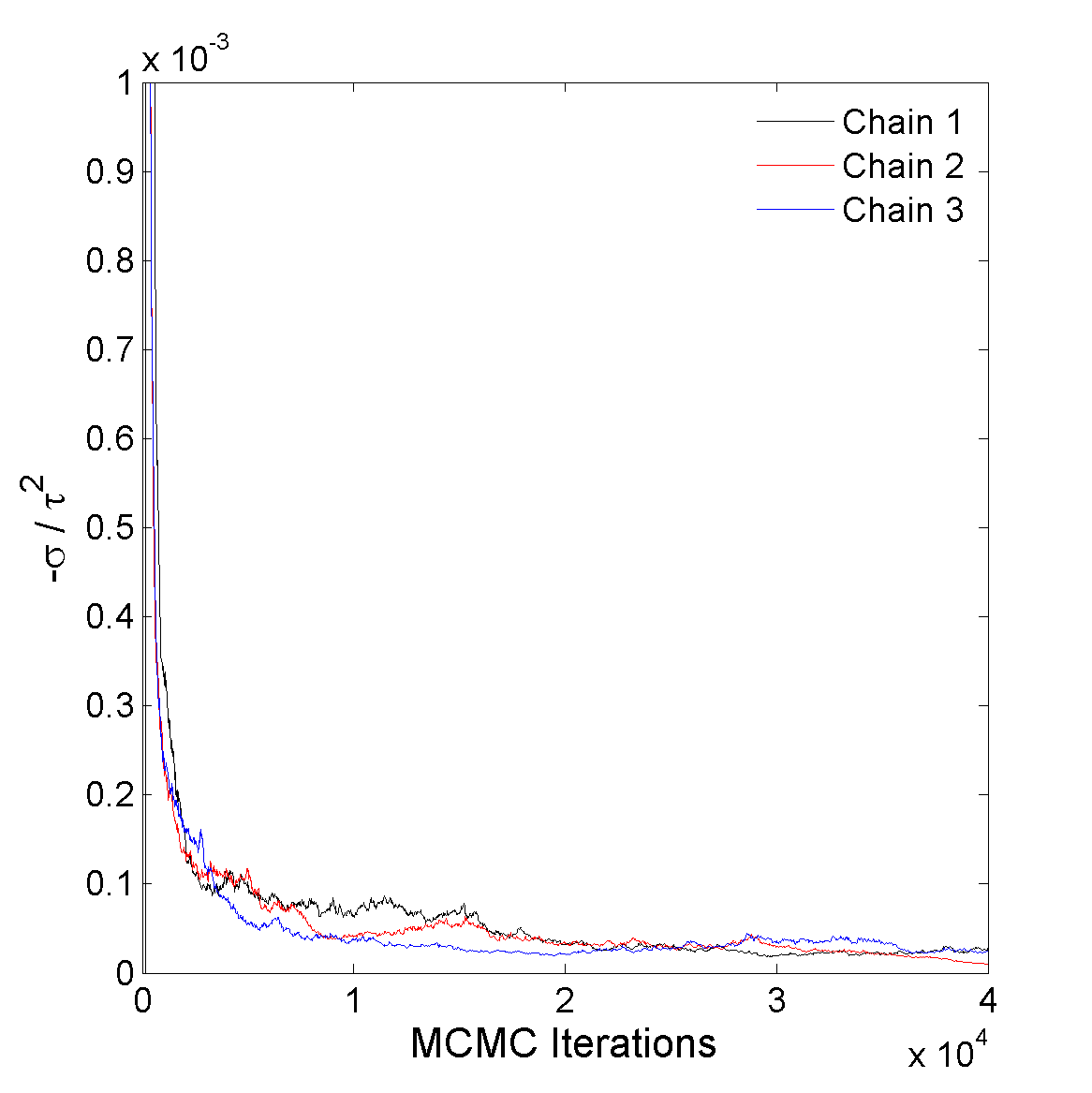}}
\subfigure[$w_*$]{\includegraphics[width=.45\textwidth]{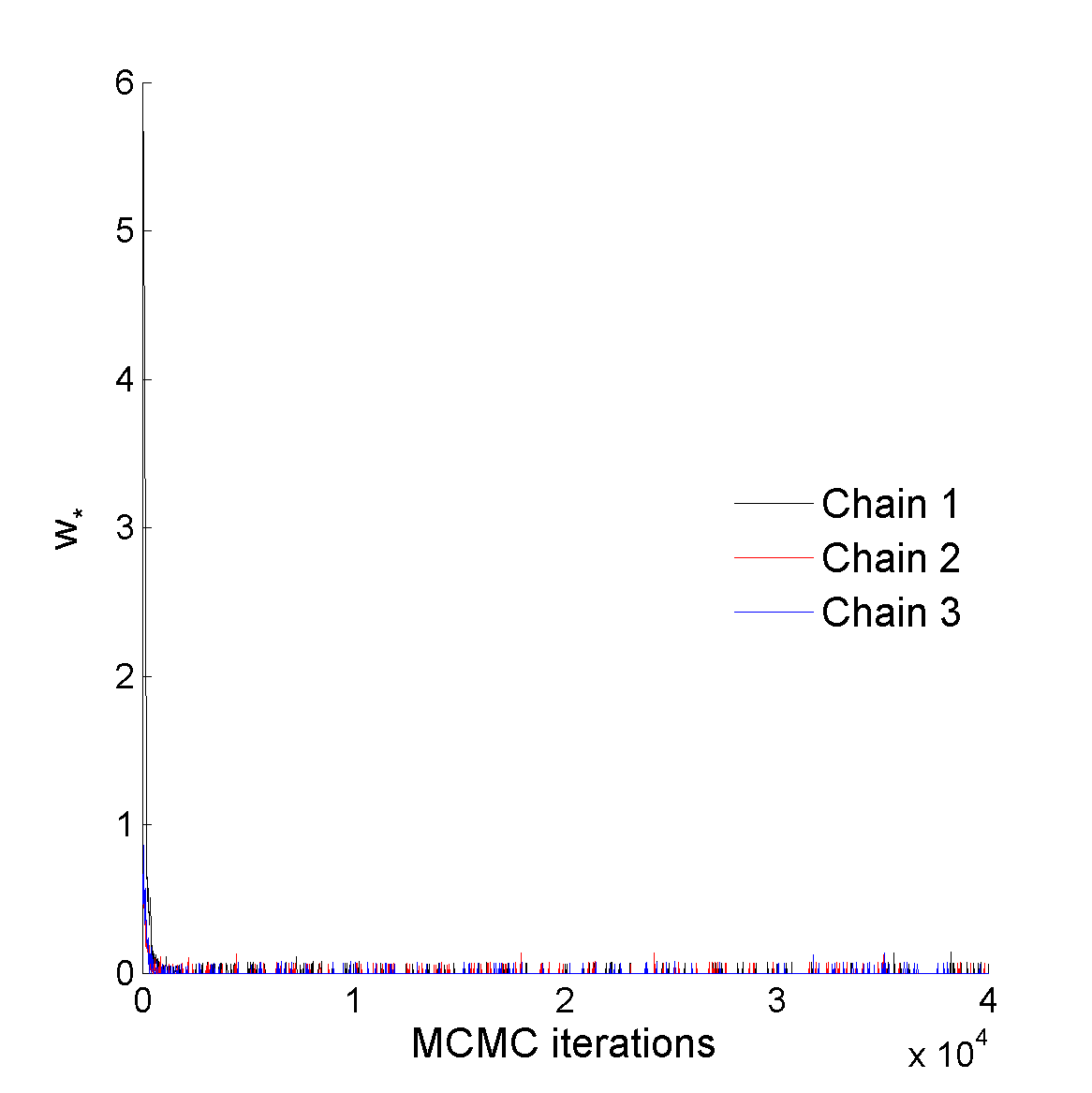}}
\end{center}
\caption{MCMC trace plots of parameters (a) $\varsigma_1$ (b) $\varsigma_2$, (c) $\varsigma_3$, (d) $w_*$ for a graph generated from an Erd\"os-R\'enyi model with parameters $n=1000,p=0.01$. }%
\label{fig:simuERtrace}%
\end{figure}

\begin{figure}[ptb]
\begin{center}%
\subfigure[Nodes with highest degree]{\includegraphics[width=.48\textwidth]{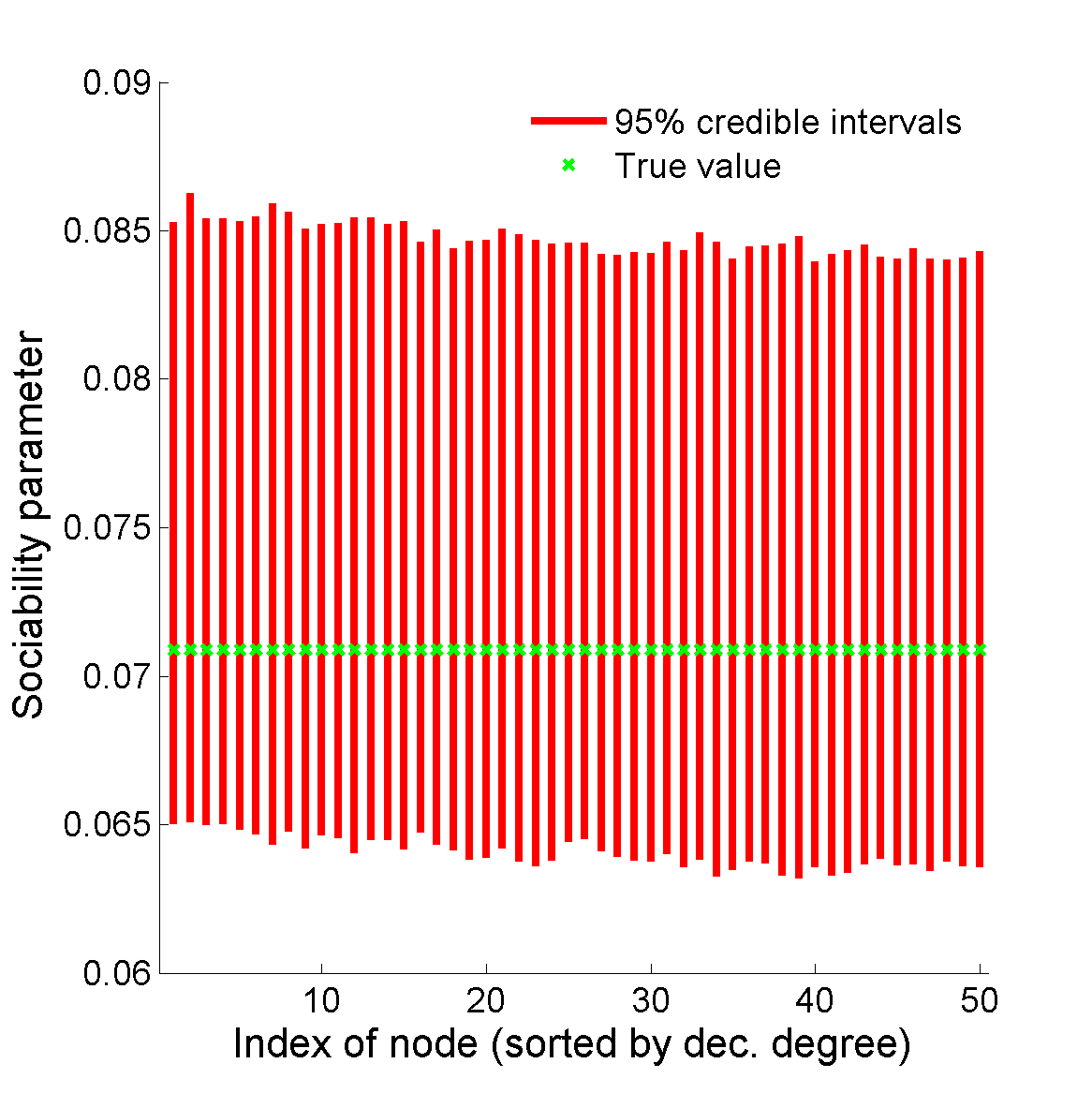}}
\subfigure[Nodes with lowest degree]{\includegraphics[width=.48\textwidth]{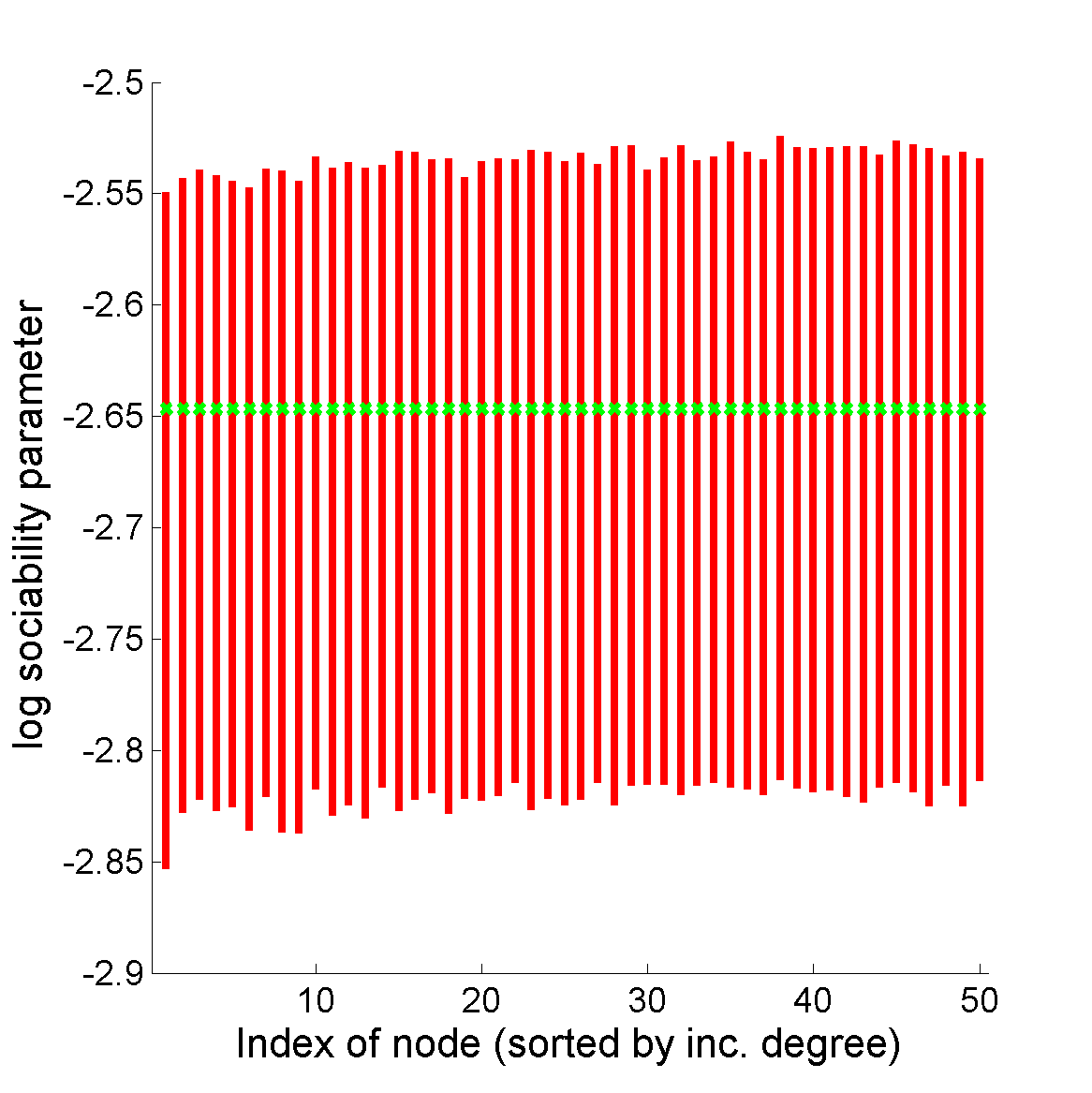}}
\end{center}
\vspace{-0.15in}
\caption{95\% posterior intervals of (a) sociability parameters $w_i$ of the 50 nodes with highest degree and (b) log-sociability parameters $\log w_i$ of the 50 nodes with lowest degree, for a graph generated from an Erd\"os-R\'enyi model with parameters $n=1000,p=0.01$. In this case, all nodes have the same true sociability parameter $\sqrt{-\frac{1}{2}\log(1-p)}$, represented by a green star.} \vspace{-0.15in}
\label{fig:simuER_w}%
\end{figure}

\subsection{Testing for sparsity of real-world graphs}

We now turn to using our methods to test whether a given graph is sparse or not. Such testing based on a single given graph is notoriously challenging as sparsity relates to the asymptotic behavior of the graph. Measures of sparsity from finite graphs exist, but can be costly to implement~\cite{Nesetril2012}. Based on our GGP-based formulation and associated theoretical results described in Section~\ref{sec:specialcases}, we propose the following test:
\[
H_0:\sigma<0 \,\,\, \mbox{ vs } \,\,\, H_1:\sigma\geq 0.
\]
In our experiments, we again consider a GGP-based graph model with improper priors on the unknown parameters $(\alpha,\sigma,\tau)$, as described in Section~\ref{sec:MCMC}. We aim at reporting $\Pr(H_1|z)=\Pr(\sigma>0|z)$ based on a set of observed connections $(z)$, which can be directly approximated from the MCMC output. We consider 12 different datasets:
\begin{itemize}
\item \textbf{facebook107}: Social circles from Facebook\footnote{\url{https://snap.stanford.edu/data/egonets-Facebook.html}}~\cite{Mcauley2012}
\item \textbf{polblogs}: Political blogosphere (Feb. 2005)\footnote{\url{http://www.cise.ufl.edu/research/sparse/matrices/Newman/polblogs}}~\cite{Adamic2005}
\item \textbf{USairport}: US airport connection network in 2010\footnote{\label{foot:tore}\url{http://toreopsahl.com/datasets/}}~\cite{Colizza2007}
\item \textbf{UCirvine}:  Social network of students at University of California, Irvine\footnoteref{foot:tore}~\cite{Opsahl2009}
\item \textbf{yeast}: Yeast protein interaction network\footnote{\url{http://www.cise.ufl.edu/research/sparse/matrices/Pajek/yeast.html}}~\cite{Bu2003}
\item \textbf{USpower}: Network of high-voltage power grid in the Western States of the United States of America\footnoteref{foot:tore}~\cite{Watts1998}
\item \textbf{IMDB}: Actor collaboration network based on acting in the same movie\footnote{\url{http://www.cise.ufl.edu/research/sparse/matrices/Pajek/IMDB.html}}
\item \textbf{cond-mat1}: Co-authorship network\footnoteref{foot:tore}~\cite{Newman2001a}, based on preprints posted to Condensed Matter of Arxiv between 1995 and 1999; obtained from the bipartite preprints/authors network using a one-mode projection
\item \textbf{cond-mat2}: As in cond-mat1, but using Newman's projection method
\item \textbf{Enron}: Enron collaboration network from multigraph email network\footnote{\url{https://snap.stanford.edu/data/email-Enron.html}}
\item  \textbf{internet}: Connectivity of internet routers\footnote{\url{http://www.cise.ufl.edu/research/sparse/matrices/Pajek/internet.html}}
\item  \textbf{www}: Linked www pages in the nd.edu domain\footnote{\url{http://lisgi1.engr.ccny.cuny.edu/~makse/soft_data.html}}
\end{itemize}
The sizes of the different datasets are given in Table~\ref{tab:real} and range from a few hundred nodes/edges to a million. The adjacency matrices for these networks are plotted in Figure~\ref{fig:realG} and empirical degree distributions in Figure~\ref{fig:realdegree} (red).
\begin{table}[h!]
\caption{Size of real-world datasets and posterior probability of sparsity.}
\label{tab:real}
\begin{tabular}{l|l|l|l|l|l}
  Name & Nb nodes& Nb edges& Time& $\Pr(H_1|z)$& 99\% CI $\sigma$ \\
  & & &(min)&\\
  \hline
  facebook107 & 1,034 & 26,749 & 1 & 0.000 & $[-1.057,-0.819]$ \\
  polblogs & 1,224 & 16,715 & 1 & 0.000 & $[-0.348,-0.202]$ \\
  USairport & 1,574 & 17,215 & 1 & 1.000 & $[~0.099,~0.181]$\\
  UCirvine & 1,899 & 13,838 & 1 & 0.000 & $[-0.141,-0.017]$ \\
  yeast & 2,284 & 6,646 & 1 & 0.280 & $[-0.093,0.054]$ \\
  USpower & 4,941 & 6,594 & 1 & 0.000 & $[-4.837,-3.185]$\\
  IMDB & 14,752 & 38,369 & 2 & 0.000 & $[-0.244,-0.173]$ \\
  cond-mat1 & 16,264 & 47,594 & 2 & 0.000 & $[-0.945,-0.837]$\\
  cond-mat2 & 7,883 & 8,586 & 1 & 0.000 & $[-0.176,-0.022]$\\
  Enron & 36,692 & 183,831 & 7 & 1.000 & $[~0.201,~0.221]$ \\
  internet & 124,651 & 193,620 & 15 & 0.000 & $[-0.201,-0.171]$ \\
  www & 325,729& 1,090,108 & 132 & 1.000 & $[0.262,0.298]$ \\
\end{tabular}
\end{table}

We ran 3 MCMC chains for 40,000 iterations with the same specifications as above and report the estimate of  $\Pr(H_1|z)$ and 99\% posterior credible intervals of $\sigma$ in Table~\ref{tab:real}; we additionally provide runtimes. Figure~\ref{fig:realtracesigma} and Figure~\ref{fig:realhistsigma} show MCMC traces and posterior histograms, respectively, for the sparsity parameter $\sigma$ for the different datasets. {Many of the smaller networks fail to provide evidence of sparsity. These graphs may indeed be dense; for example, our \texttt{facebook107} dataset represents a small social circle that is likely highly interconnected and the \texttt{polblogs} dataset represents two tightly connected political parties.} Three of the datasets (\texttt{USairport}, \texttt{Enron}, \texttt{www}) are clearly inferred as sparse; note that two of these datasets are in the top three largest networks considered, where sparsity is more commonplace. In the remaining large, but inferred-dense network, \texttt{internet}, there is not enough evidence under our test that the network is not dense. This may be due to the presence of dense subgraphs or spots (e.g., spatially proximate routers may be highly interconnected, but sparsely connected outside the group) \cite{BorgsChayesCohnZhao2014}. This relates to the idea of \emph{community structure}, though not every node need be associated with a community.  As in many sparse network models that assume no dense spots~\cite{BollobasRiordan2009,WolfeOlhede2013}, our approach does not explicitly model such effects.  Capturing such structure remains a direction of future research likely feasible within our generative framework, though our current method has the benefit of simplicity with three hyperparameters tuning the network properties. Finally, we note in Table~\ref{tab:real} that our analyses finish in a remarkably short time despite the code base being implemented in Matlab on a standard desktop machine, without leveraging possible opportunities for parallelizing and otherwise scaling some components of the sampler (see Section~\ref{sec:MCMC} for a discussion.)

To assess our fit to the empirical degree distributions, we use the methods described in Section~\ref{sec:simulation} to simulate 5000 graphs from the posterior predictive and compare to the observed graph degrees in Figure~\ref{fig:realdegree}.  In all cases, we see a reasonably good fit.  For the largest networks, Figure~\ref{fig:realdegree}(j)-(l), we see a slight underestimate of the tail of the distribution; that is, we do not capture as many high-degree nodes as truly present.  This may be because these graphs exhibit a power-law behavior, but only after a certain cutoff~\cite{Clauset2009}, which is not an effect explicitly modeled by our framework. Likewise, this cutoff might be due to the presence of dense spots. In contrast, we capture power-law behavior with possible exponential cutoff in the tail. We see a similar trend for \texttt{cond-mat1}, but not \texttt{cond-mat2}.  Based on the bipartite articles-authors graph, \texttt{cond-mat1} uses the standard one-mode projection and sets a connection between two authors who have co-authored a paper; this projection clearly creates dense spots in the graph. On the contrary, \texttt{cond-mat2} uses Newman's projection method~\cite{Newman2001}. This method constructs a weighted undirected graph by counting the number of papers co-authored by two scientists, where each count is normalized by the number of authors on the paper. To construct the undirected graph, we set an edge if the weight is equal or greater than 1; \texttt{cond-mat1} and \texttt{cond-mat2} thus have a different number of edges and nodes, as only nodes with at least one connection are considered. It is interesting to note that the projection method used for the \texttt{cond-mat} dataset has a clear impact on the sparsity of the resulting graph, \texttt{cond-mat2} being less dense than \texttt{cond-mat1} (see Figure~\ref{fig:realdegree}(h)-(i)). The degree distribution for \texttt{cond-mat1} is similar to that of \texttt{internet}, thus inheriting the same issues previously discussed. Overall, it appears our model better captures homogeneous power-law behavior with possible exponential cutoff in the tails than it does a graph with perhaps structured dense spots or power-law-after-cutoff behavior.

\begin{figure}[ptb]
\begin{center}%
\subfigure[facebook107]{\includegraphics[width=.32\textwidth]{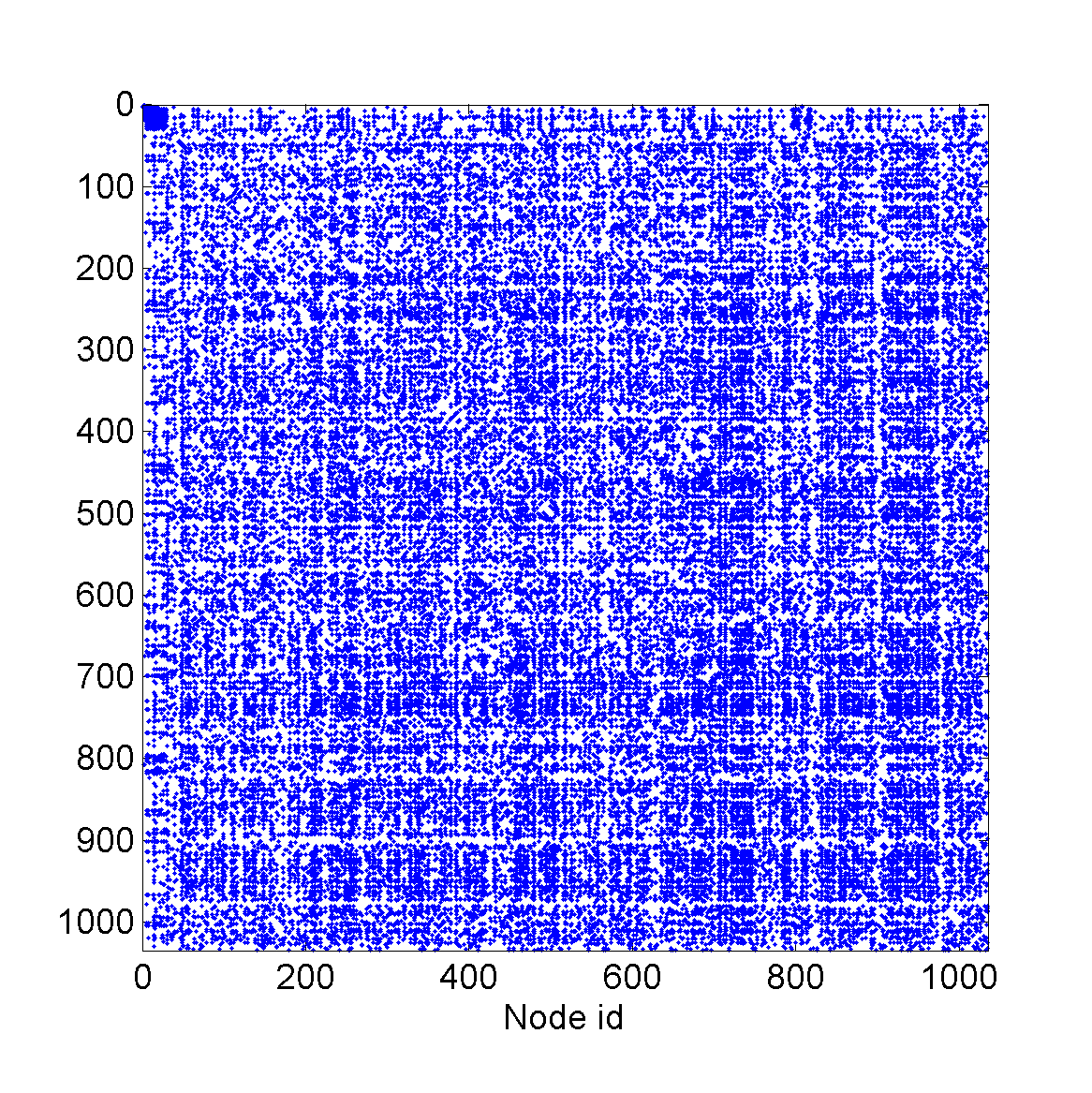}}
\subfigure[polblogs]{\includegraphics[width=.32\textwidth]{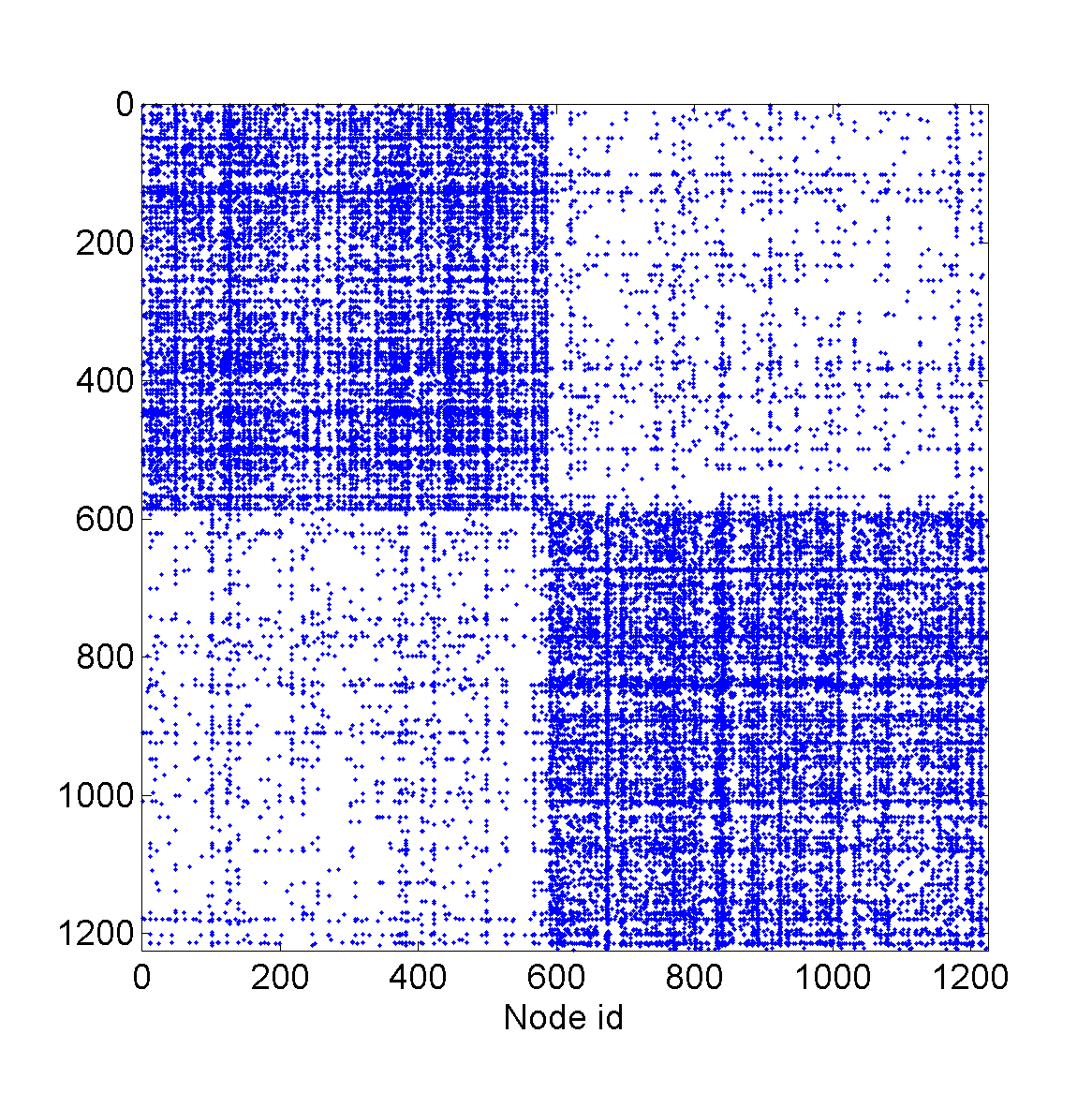}}
\subfigure[USairport]{\includegraphics[width=.32\textwidth]{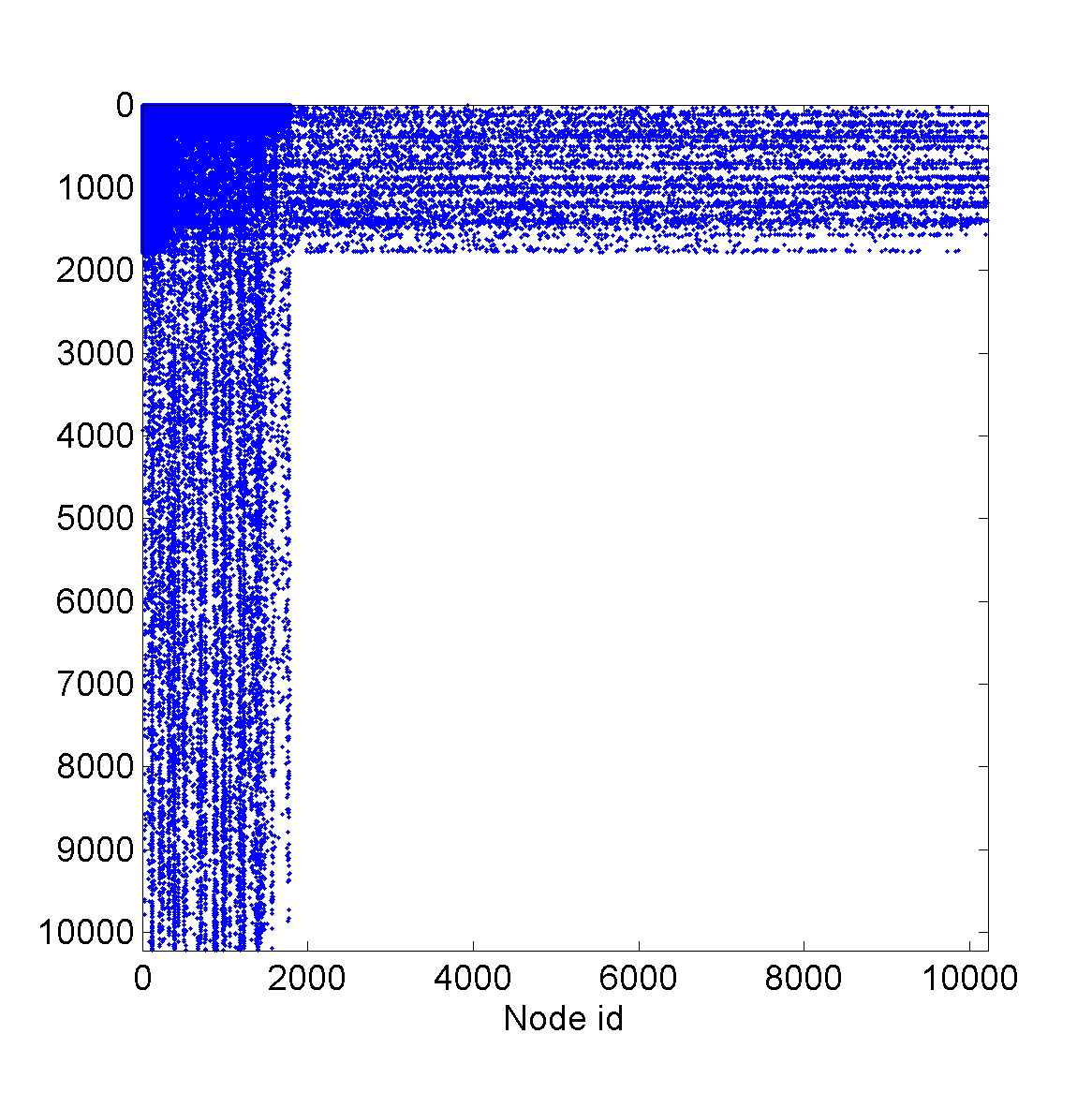}}
\subfigure[UCirvine]{\includegraphics[width=.32\textwidth]{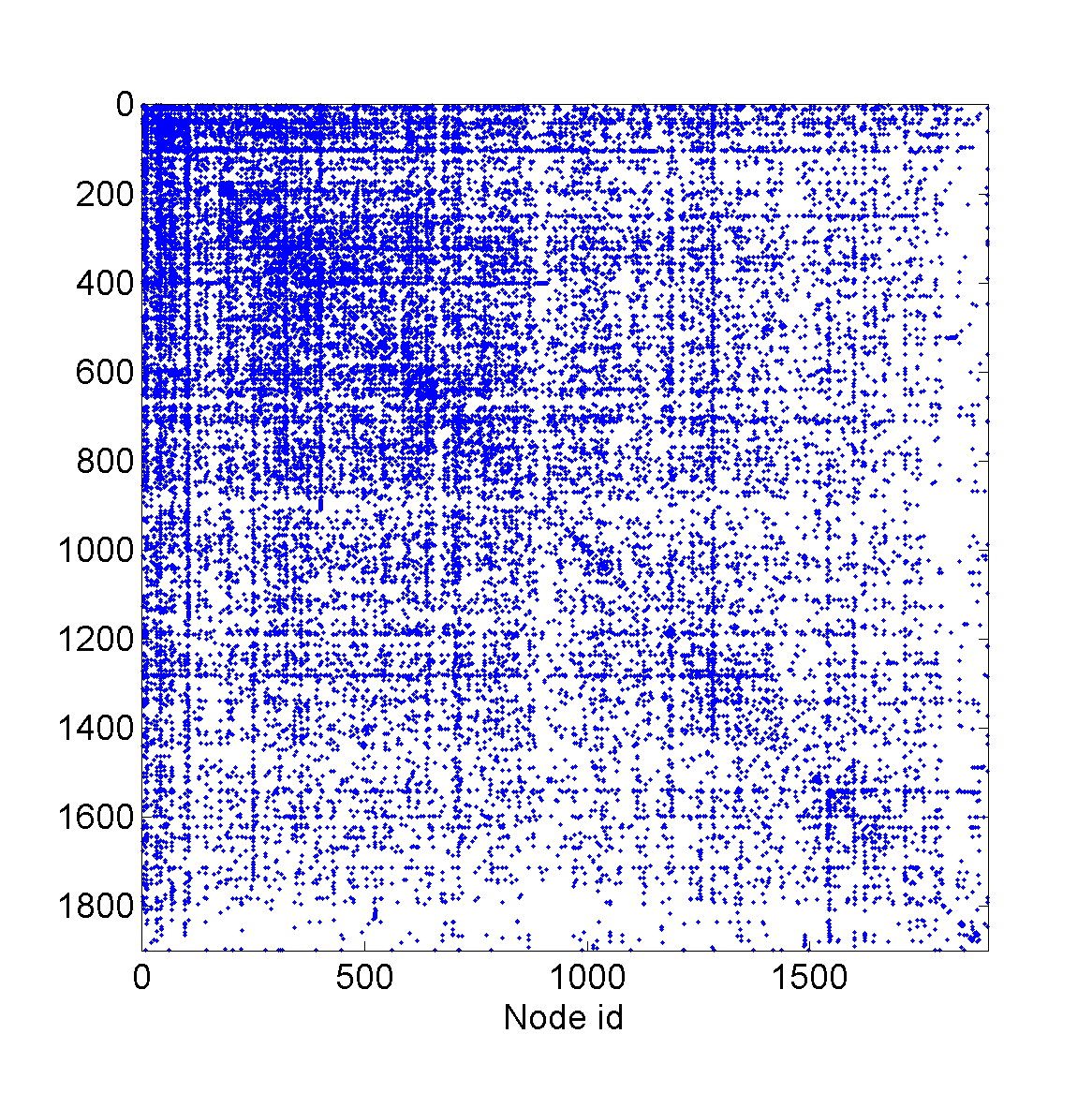}}
\subfigure[yeast]{\includegraphics[width=.32\textwidth]{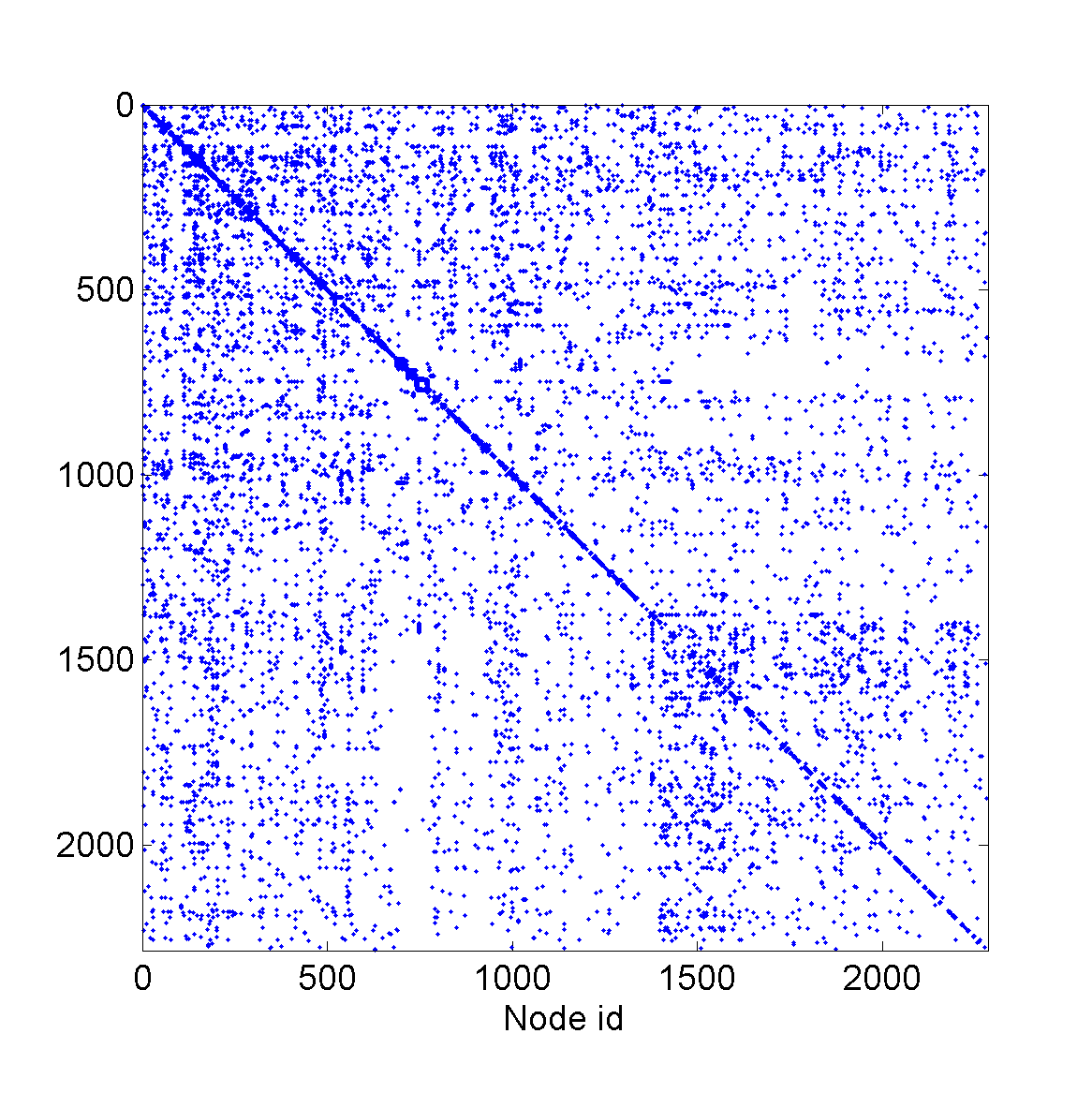}}
\subfigure[USpower]{\includegraphics[width=.32\textwidth]{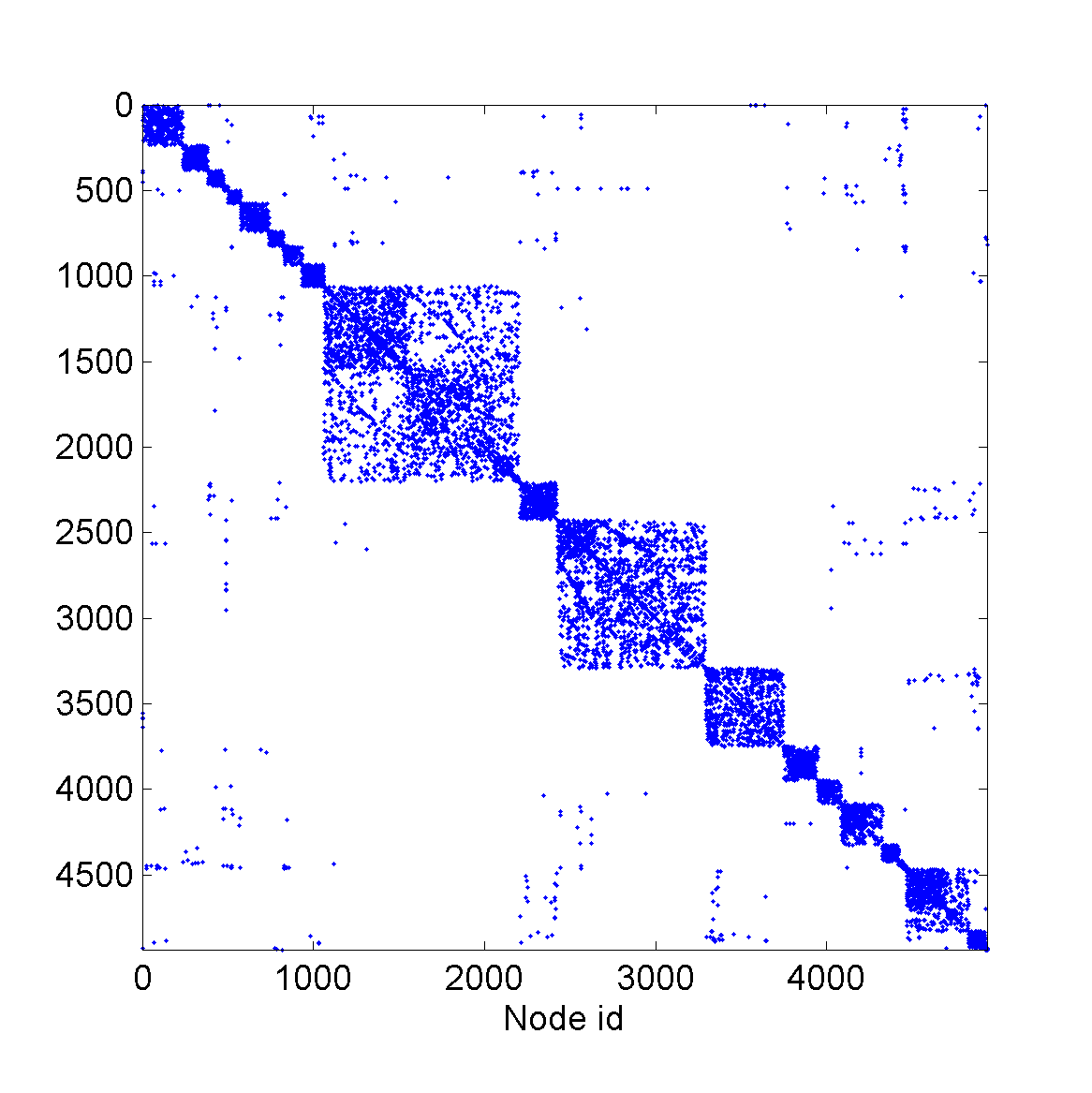}}
\subfigure[IMDB]{\includegraphics[width=.32\textwidth]{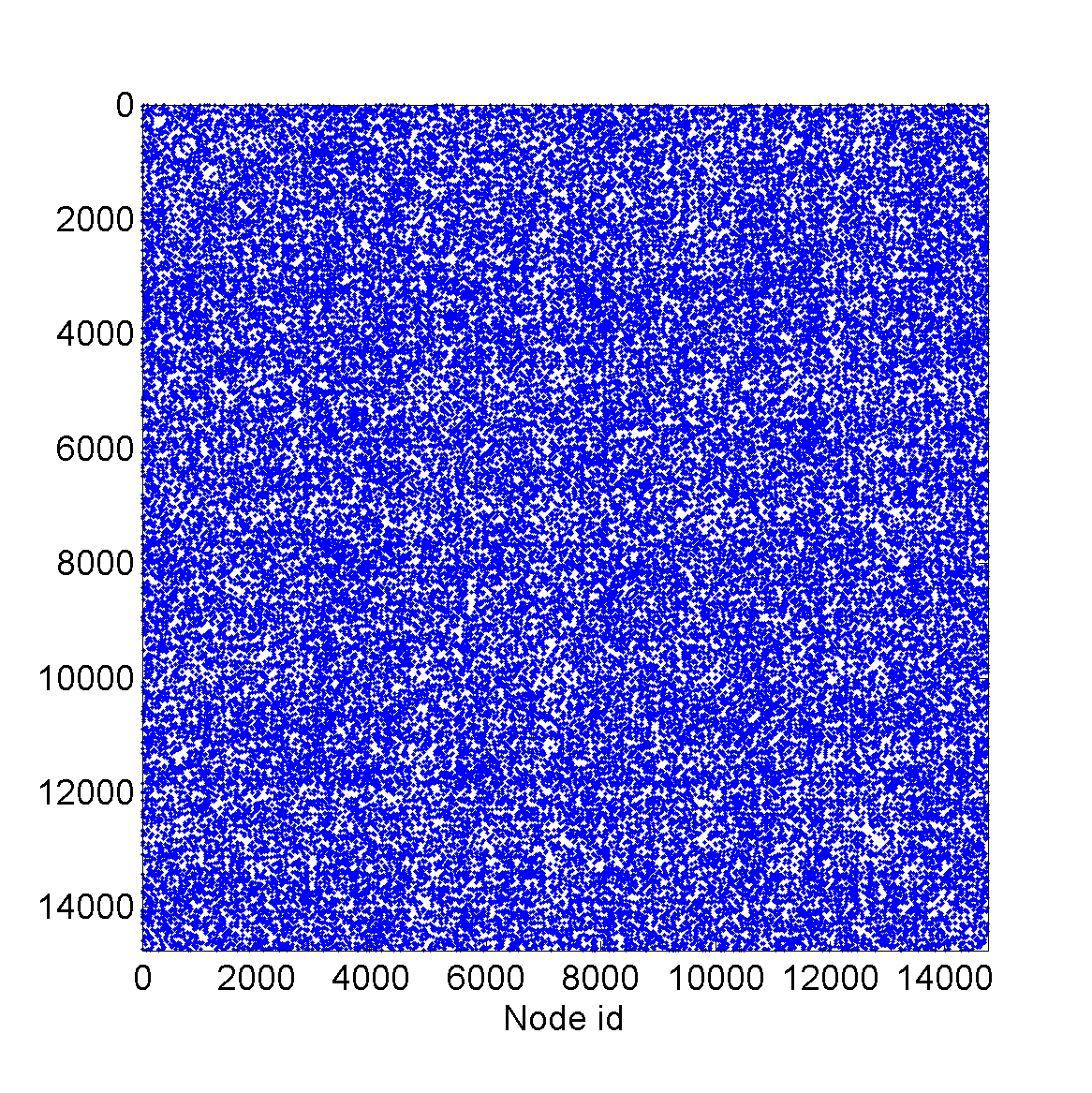}}
\subfigure[cond-mat1]{\includegraphics[width=.32\textwidth]{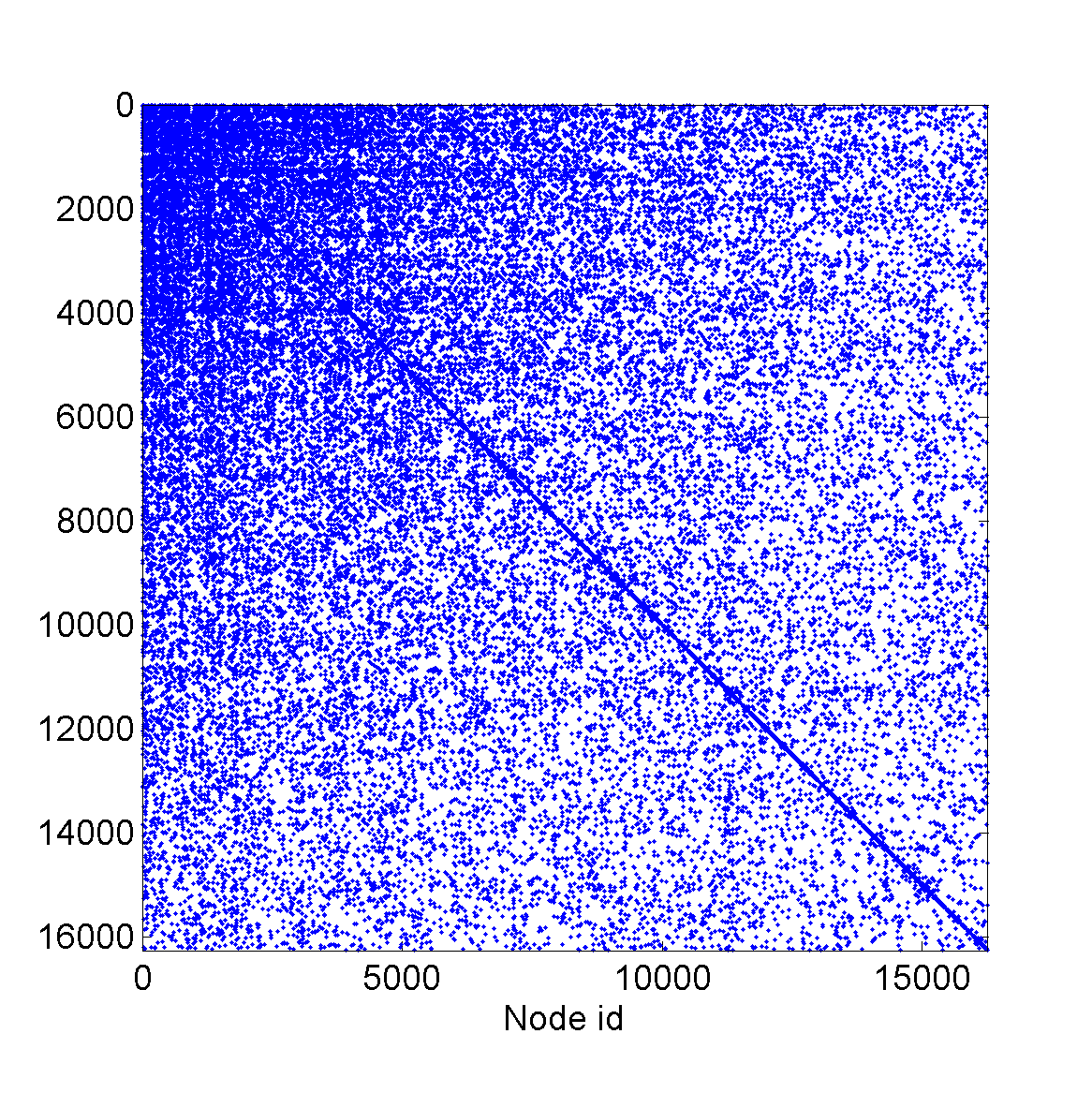}}
\subfigure[cond-mat2]{\includegraphics[width=.32\textwidth]{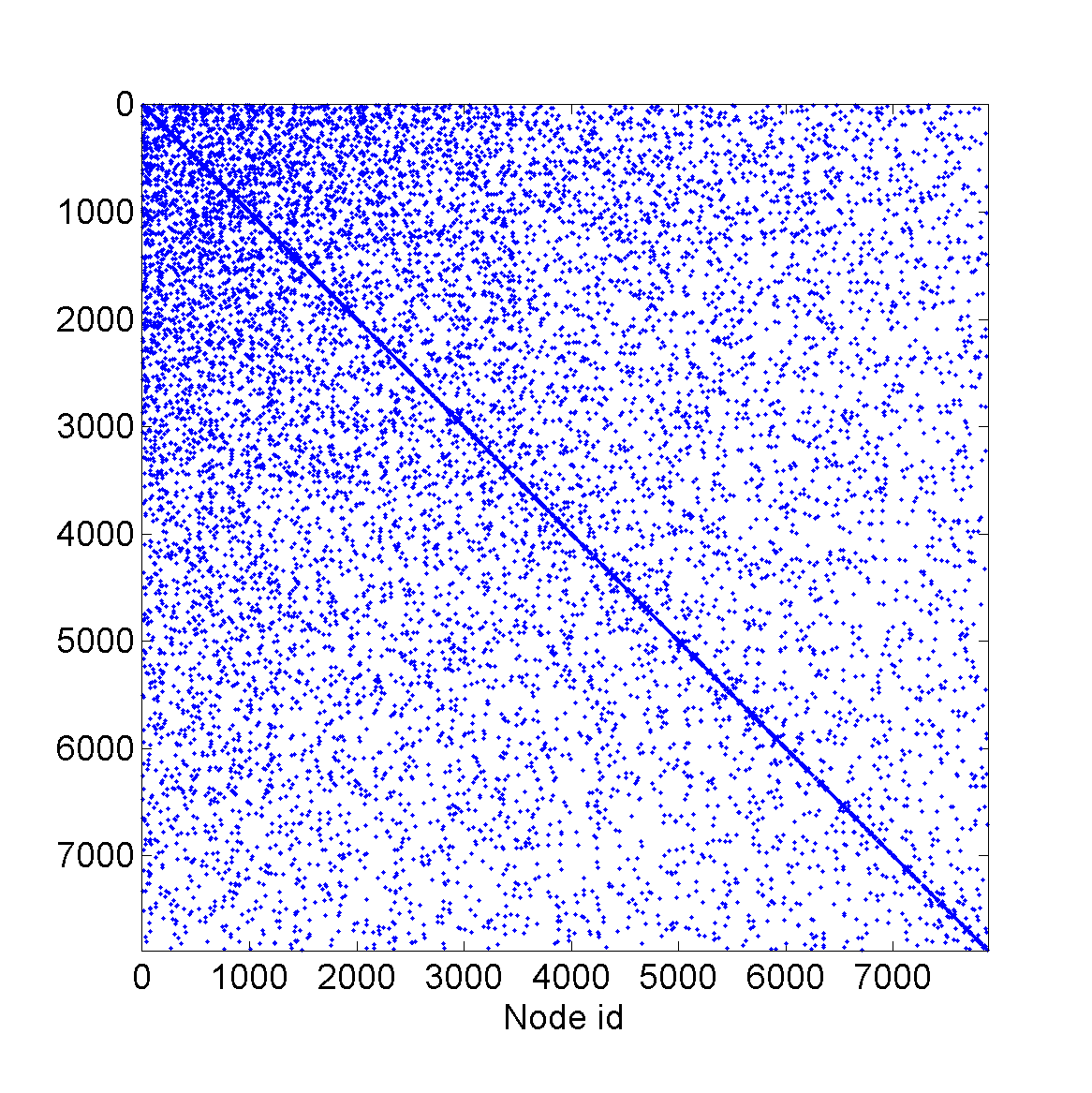}}
\subfigure[enron]{\includegraphics[width=.32\textwidth]{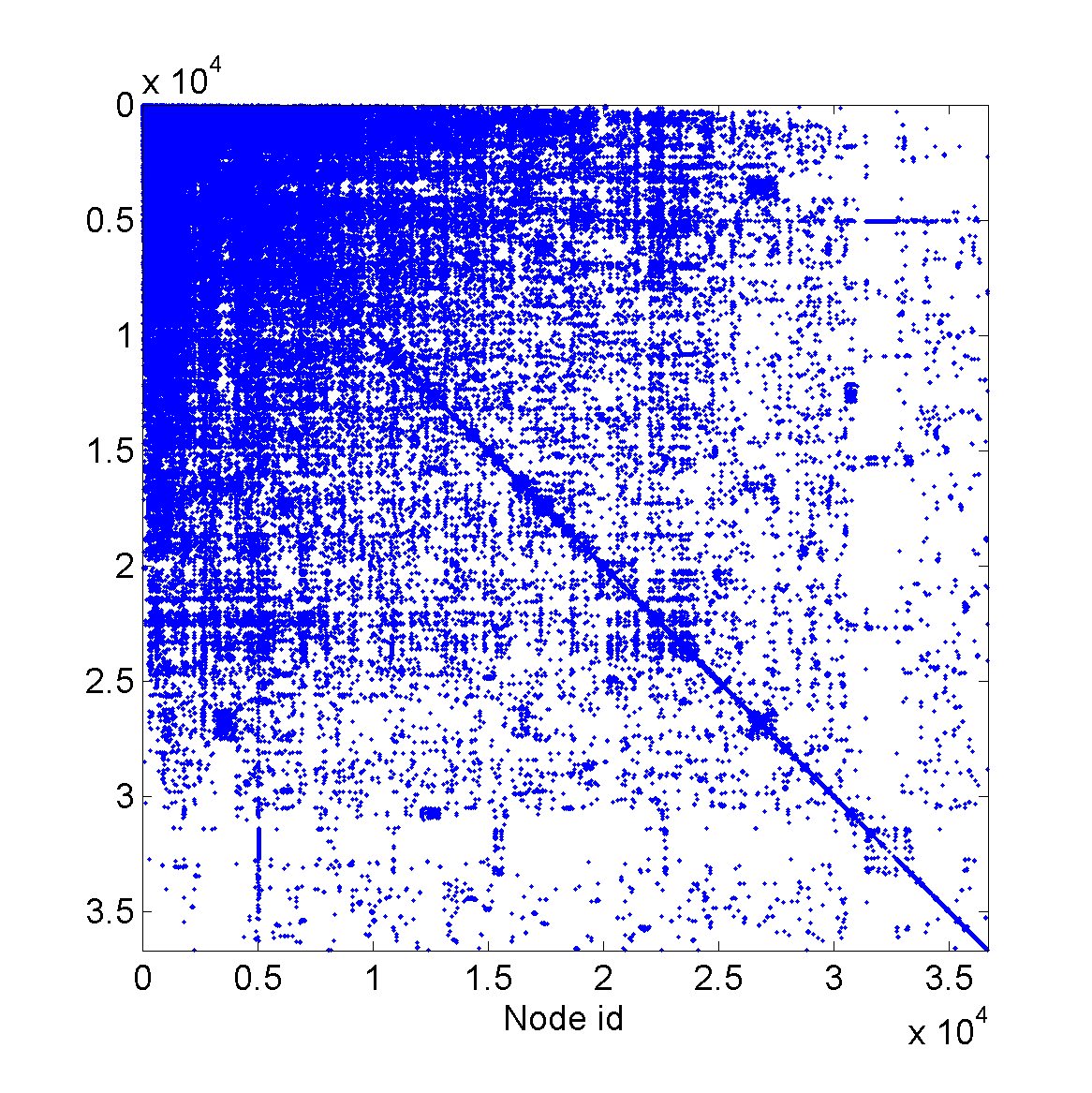}}
\subfigure[internet]{\includegraphics[width=.32\textwidth]{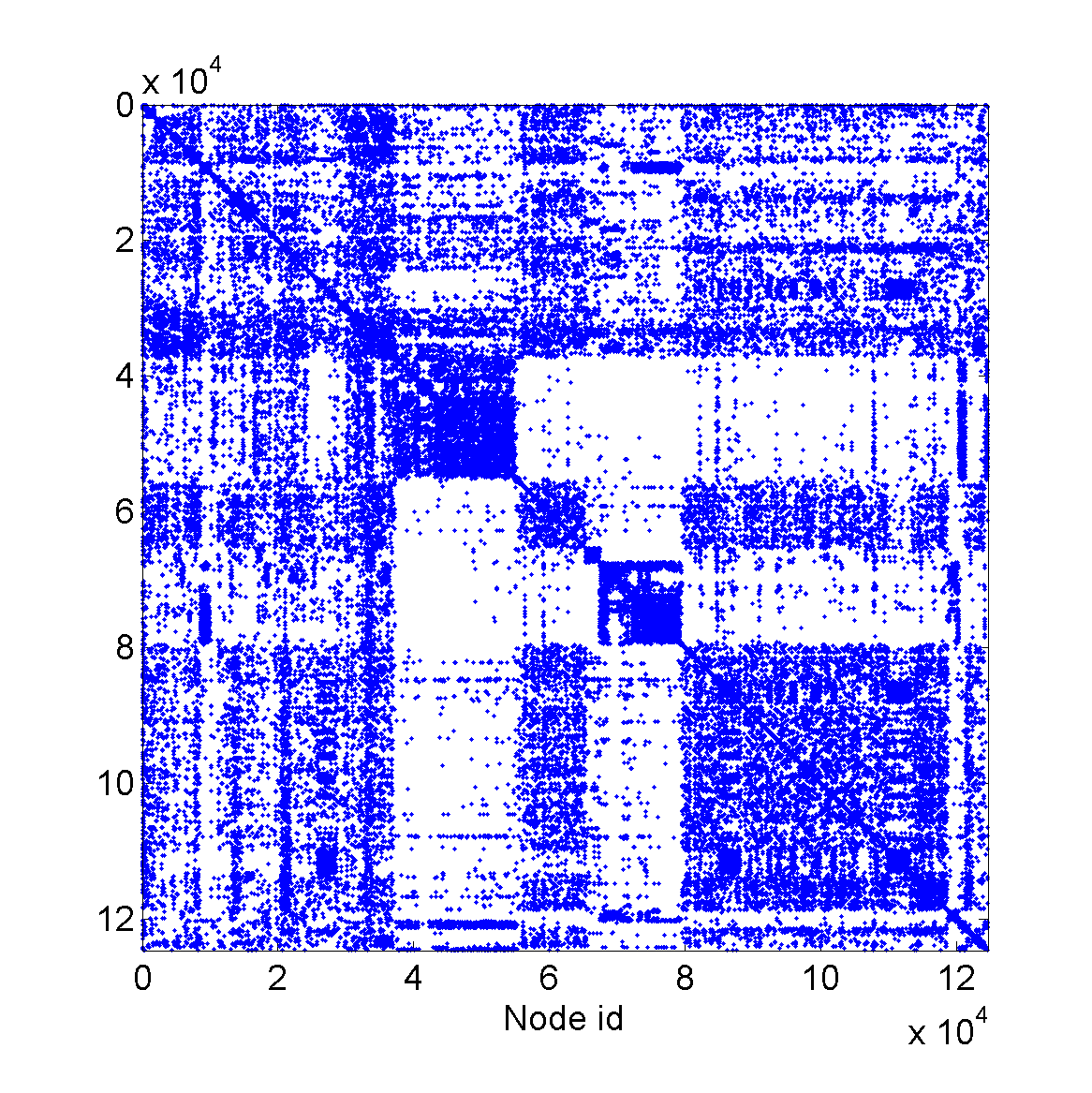}}
\subfigure[www]{\includegraphics[width=.32\textwidth]{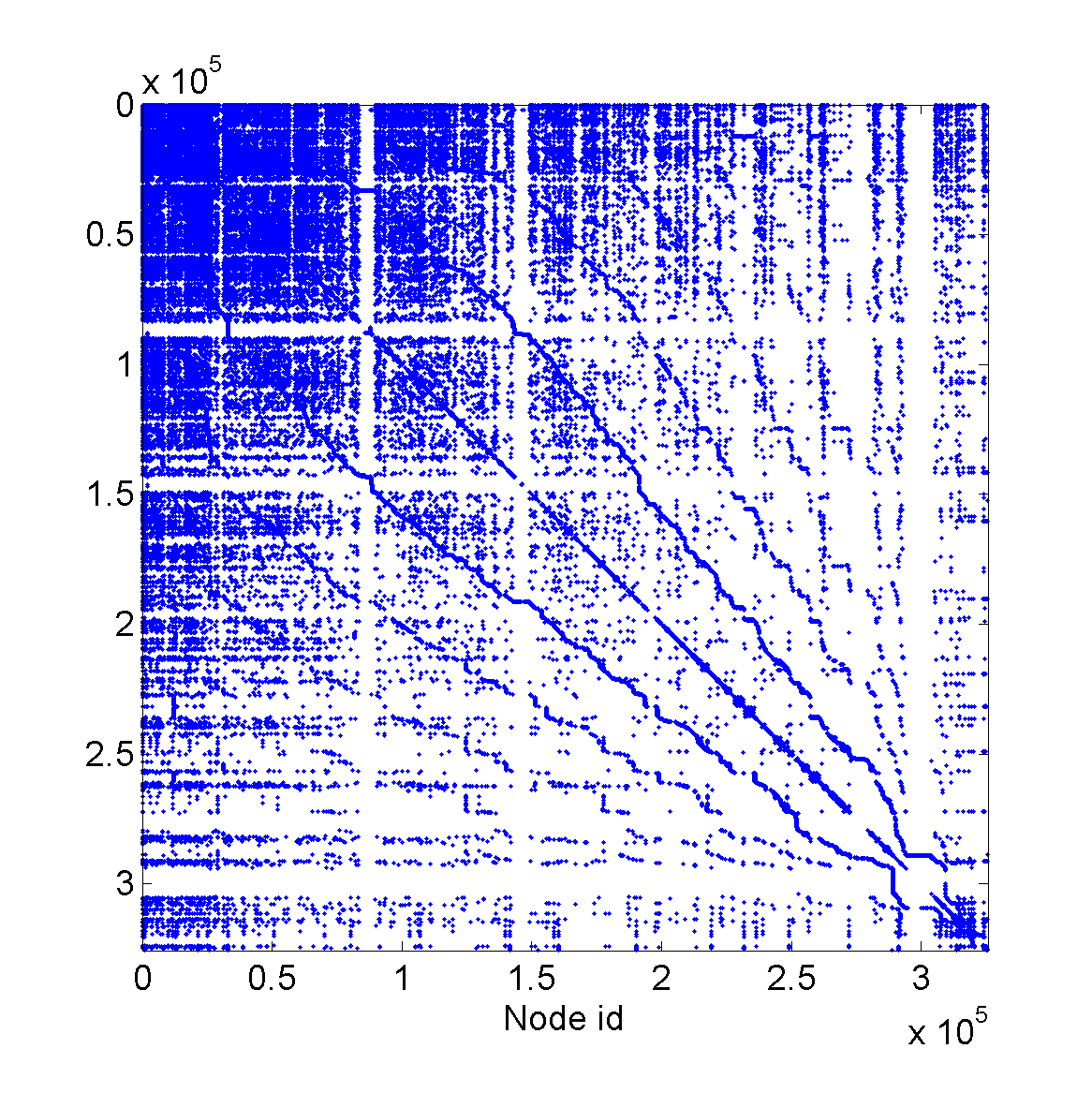}}
\end{center}
\caption{Adjacency matrices for various real-world networks.}%
\label{fig:realG}%
\end{figure}

\begin{figure}[ptb]
\begin{center}%
\subfigure[facebook107]{\includegraphics[width=.32\textwidth]{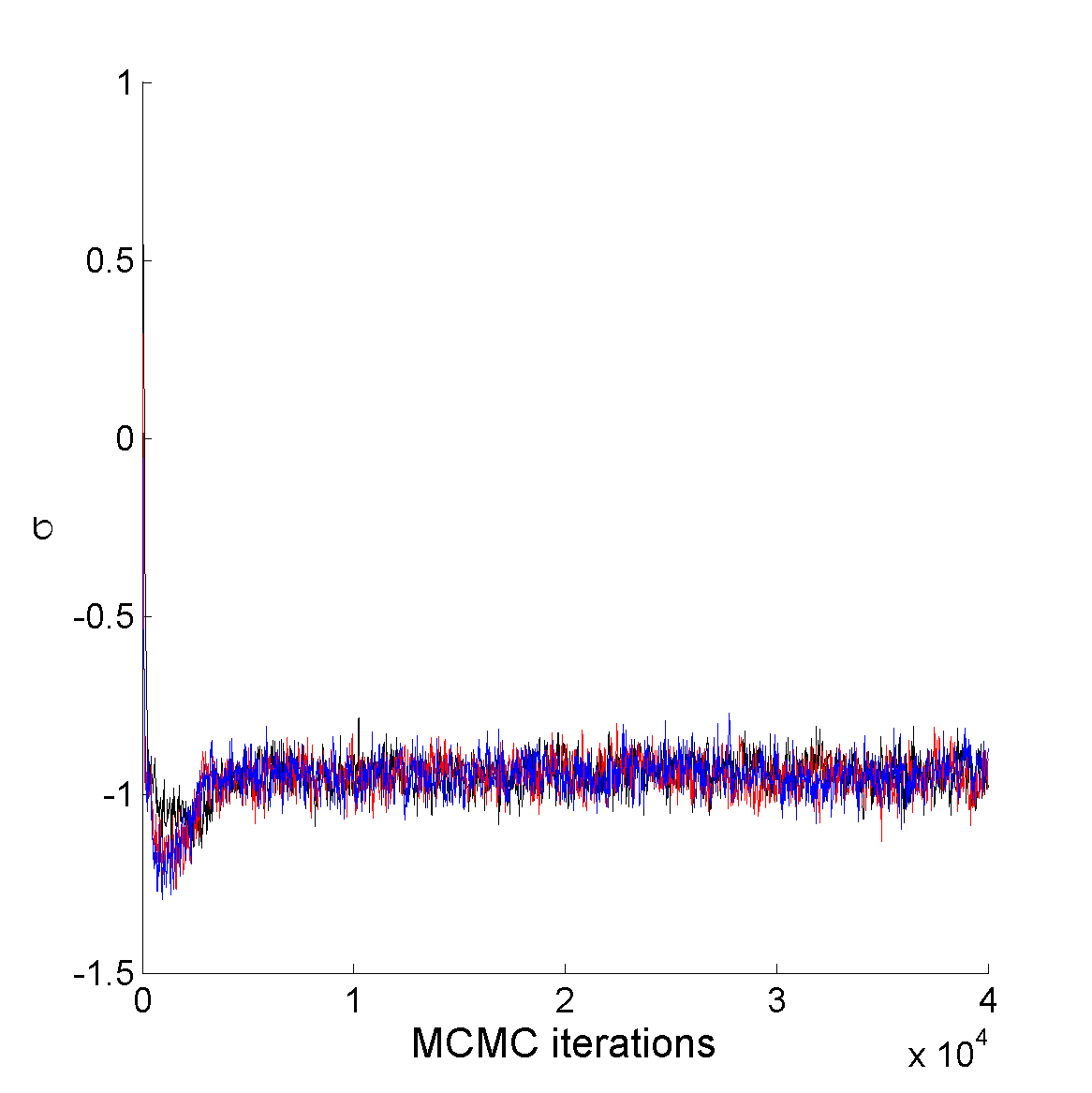}}
\subfigure[polblogs]{\includegraphics[width=.32\textwidth]{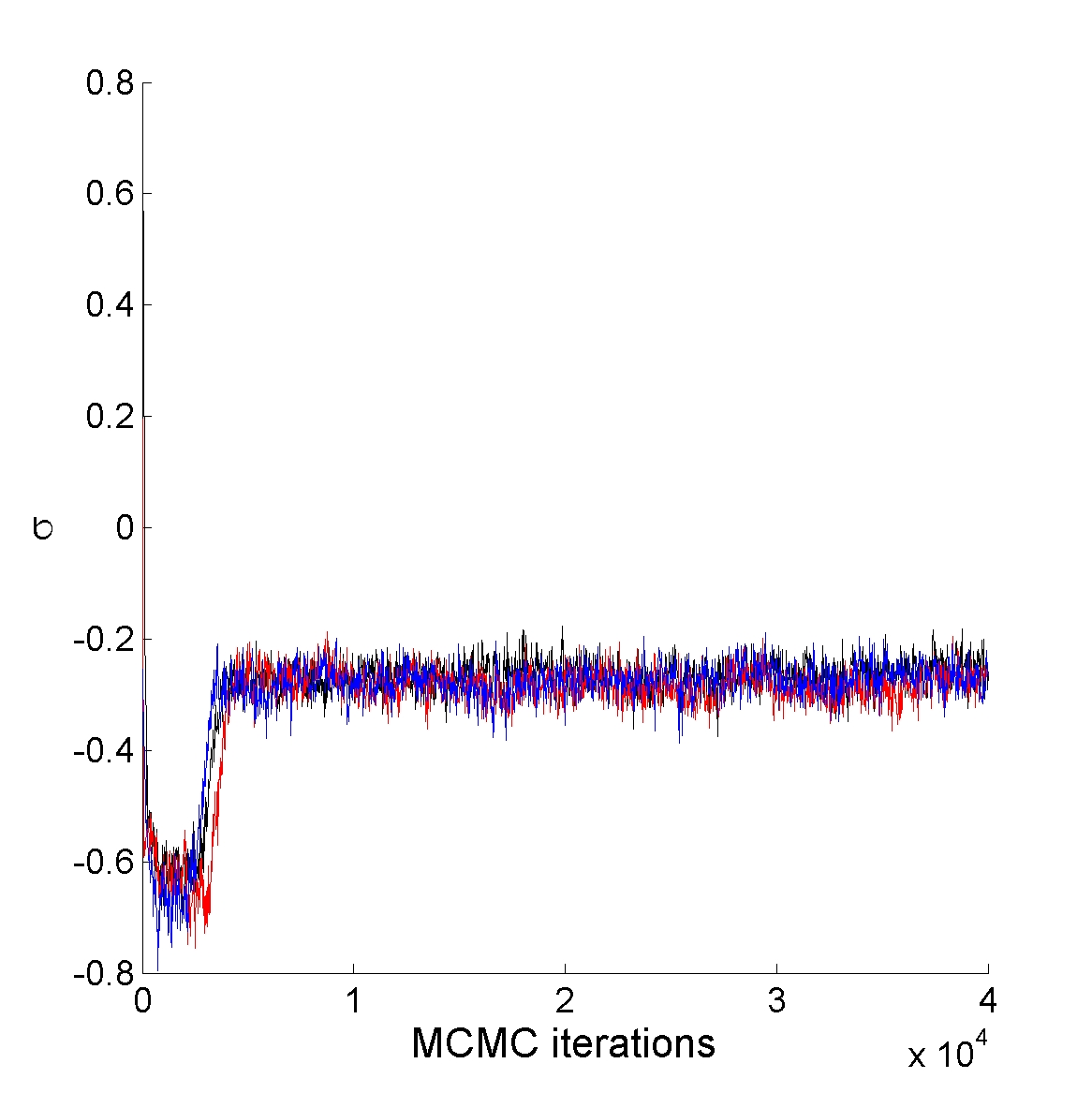}}
\subfigure[USairport]{\includegraphics[width=.32\textwidth]{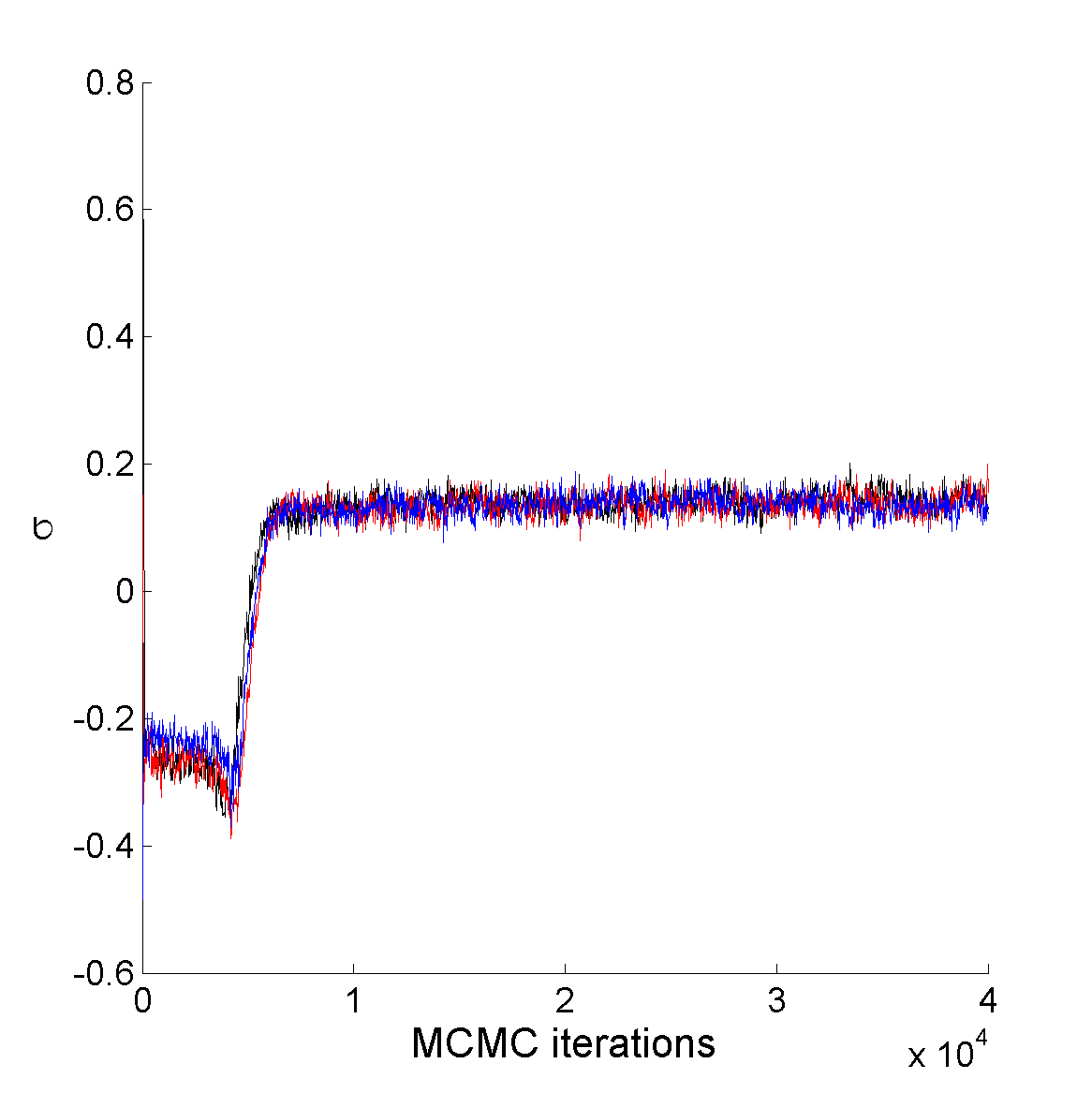}}
\subfigure[UCirvine]{\includegraphics[width=.32\textwidth]{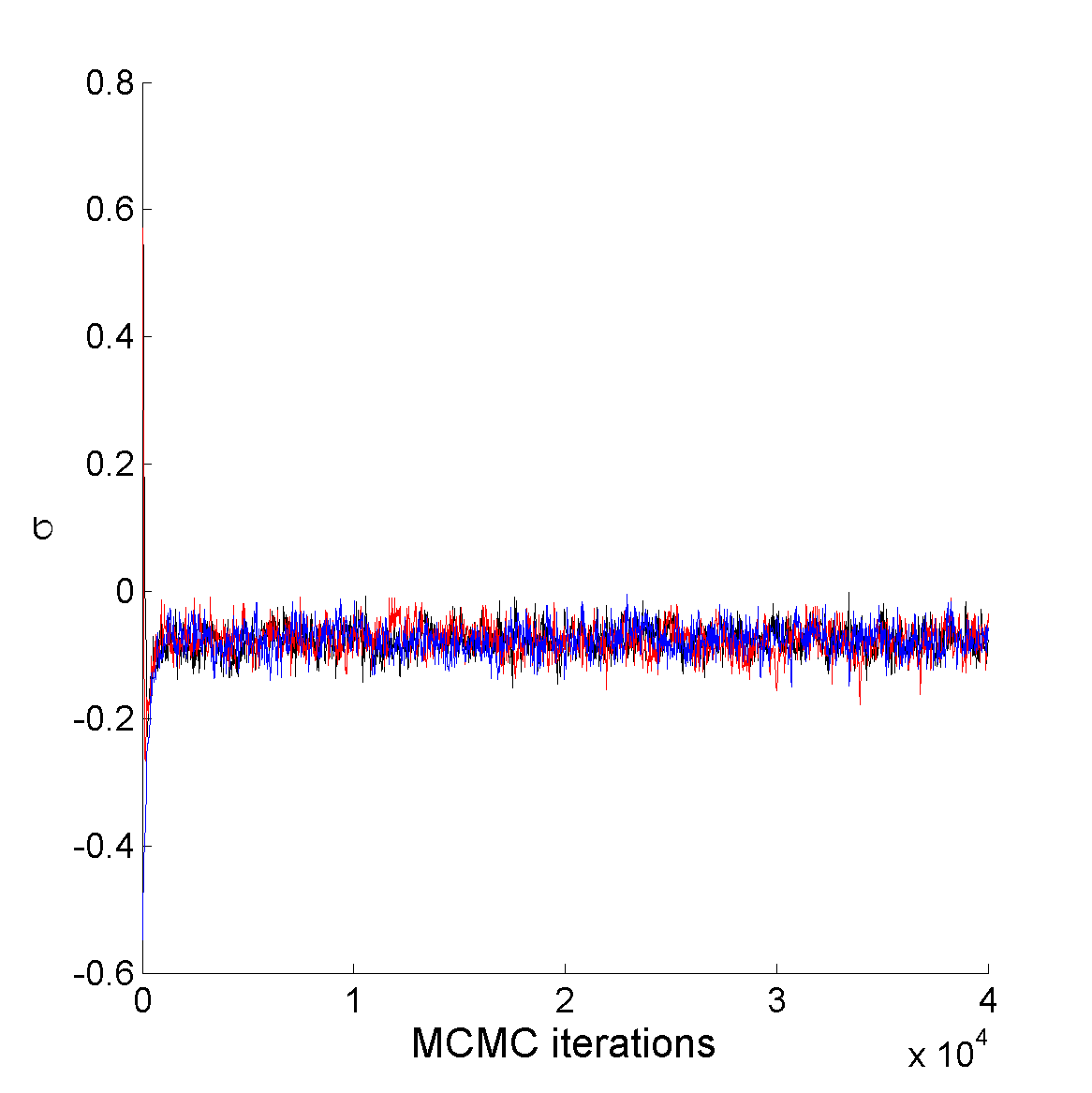}}
\subfigure[yeast]{\includegraphics[width=.32\textwidth]{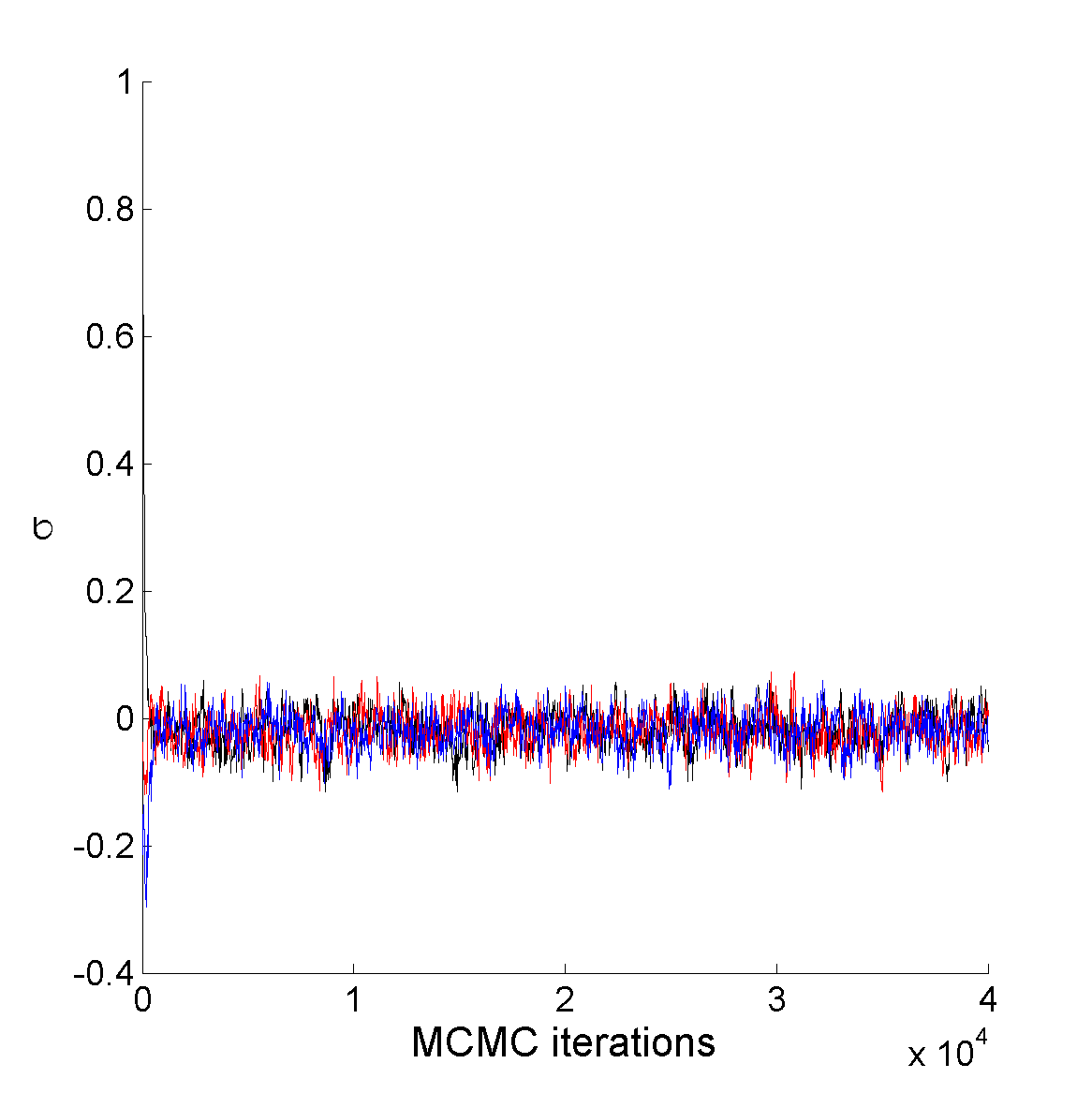}}
\subfigure[USpower]{\includegraphics[width=.32\textwidth]{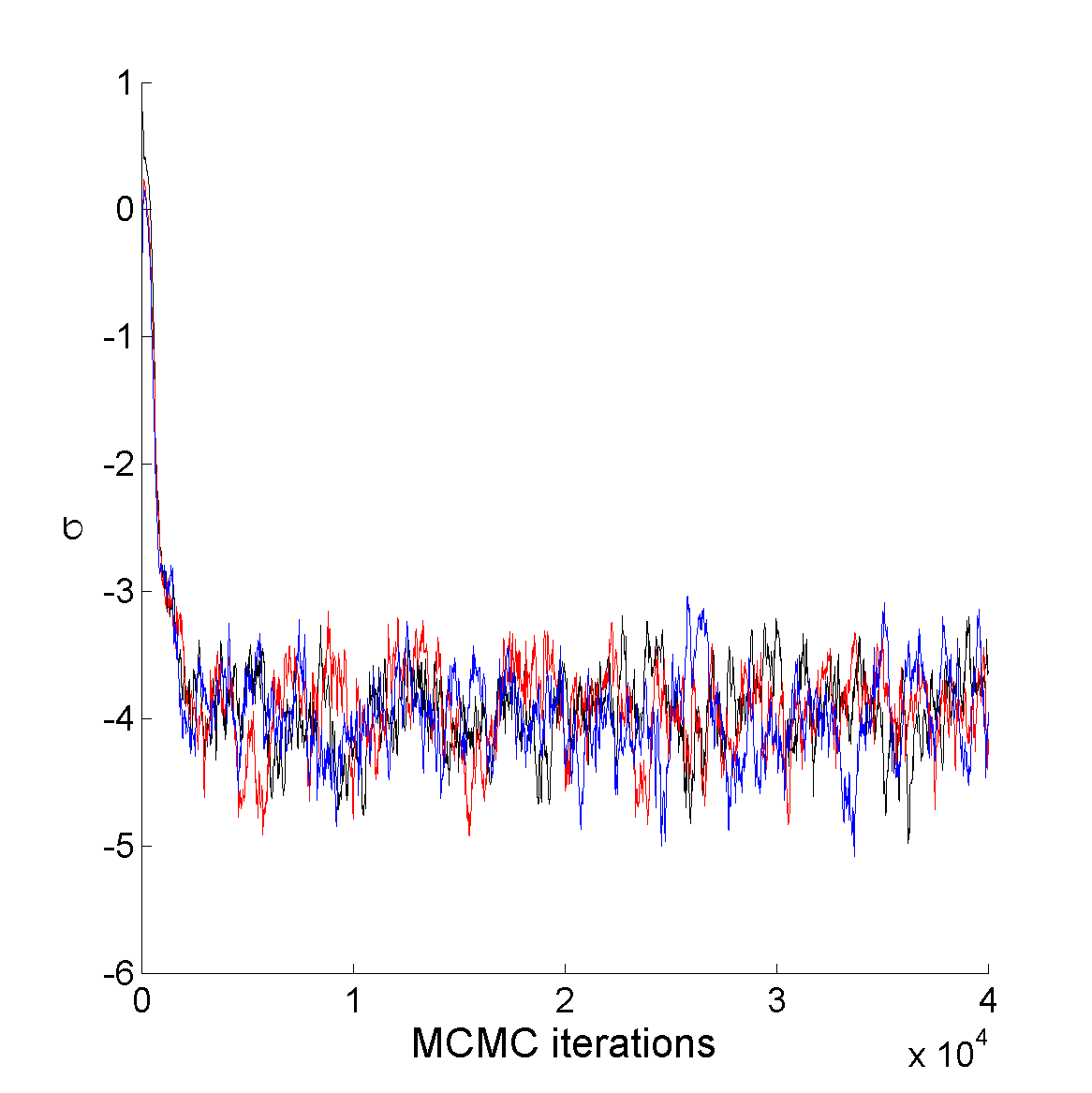}}
\subfigure[IMDB]{\includegraphics[width=.32\textwidth]{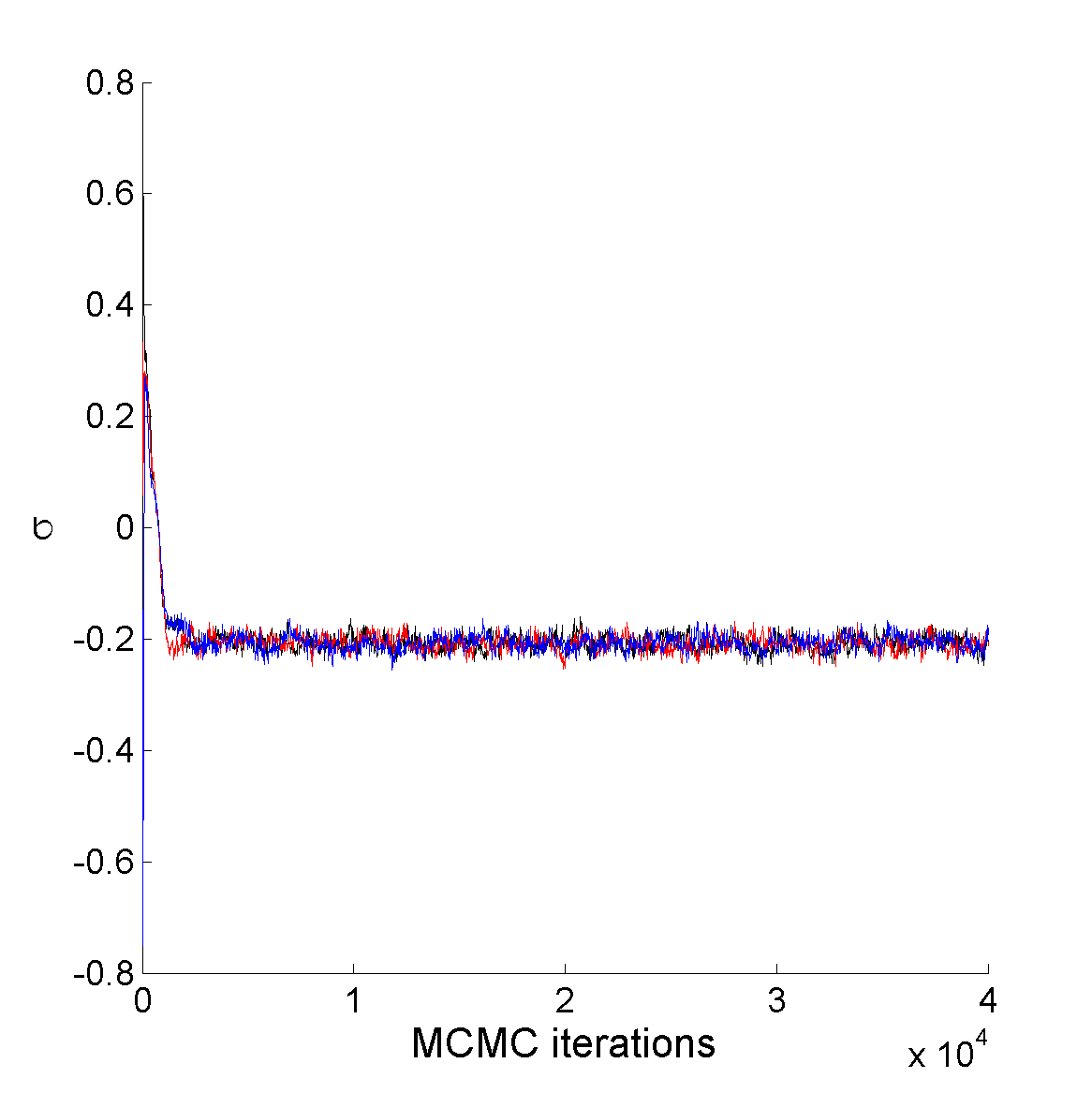}}
\subfigure[cond-mat1]{\includegraphics[width=.32\textwidth]{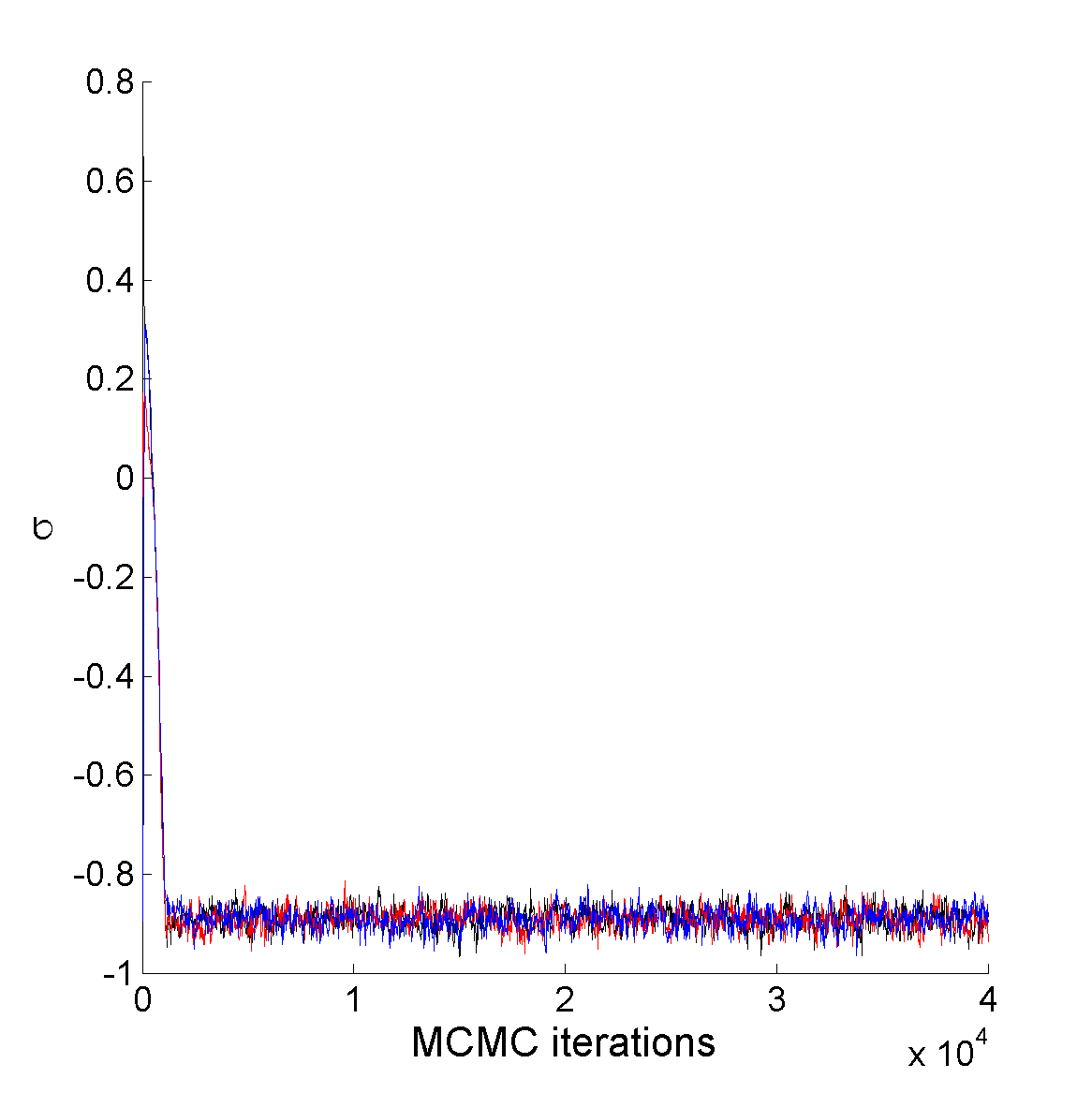}}
\subfigure[cond-mat2]{\includegraphics[width=.32\textwidth]{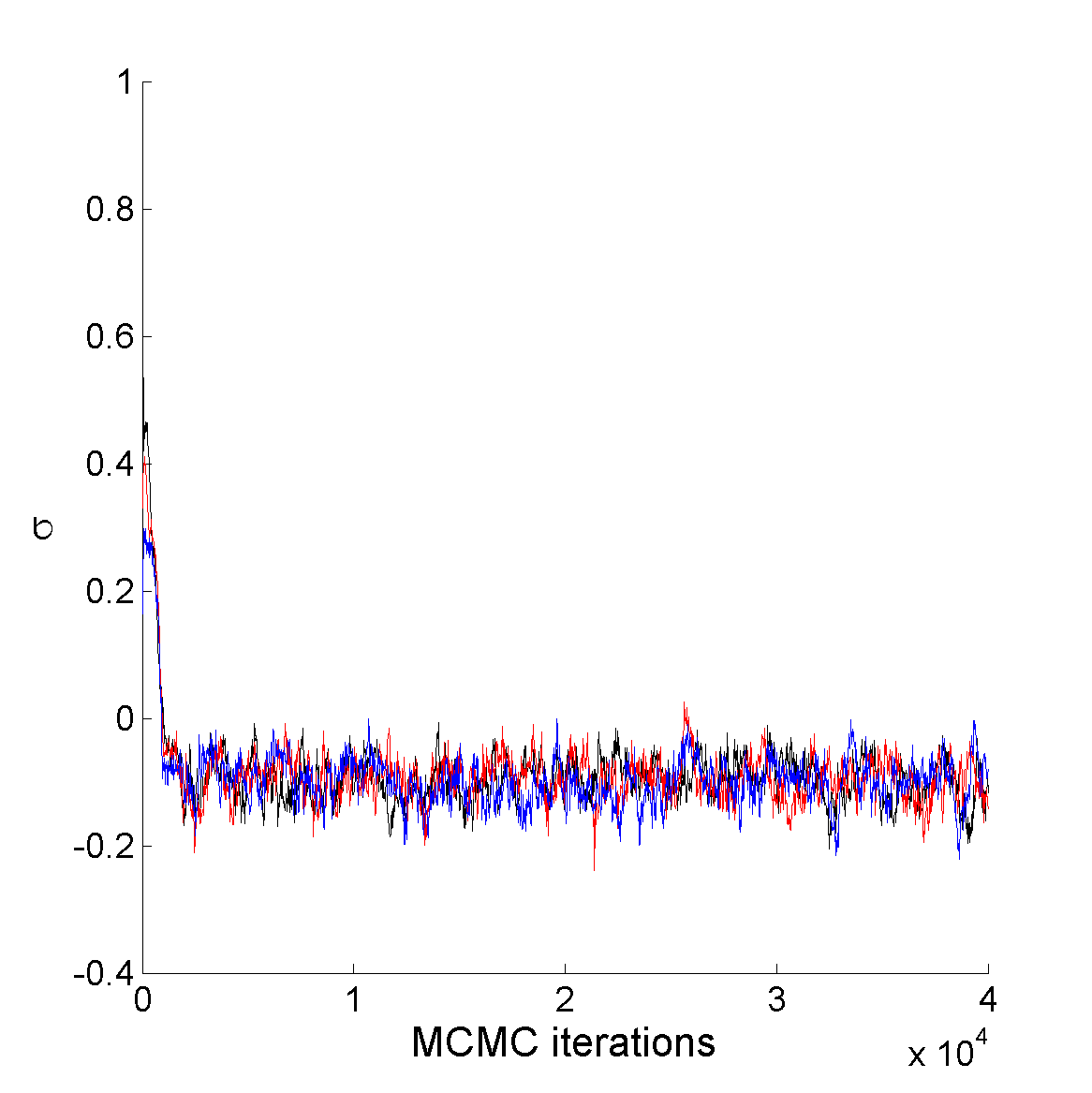}}
\subfigure[enron]{\includegraphics[width=.32\textwidth]{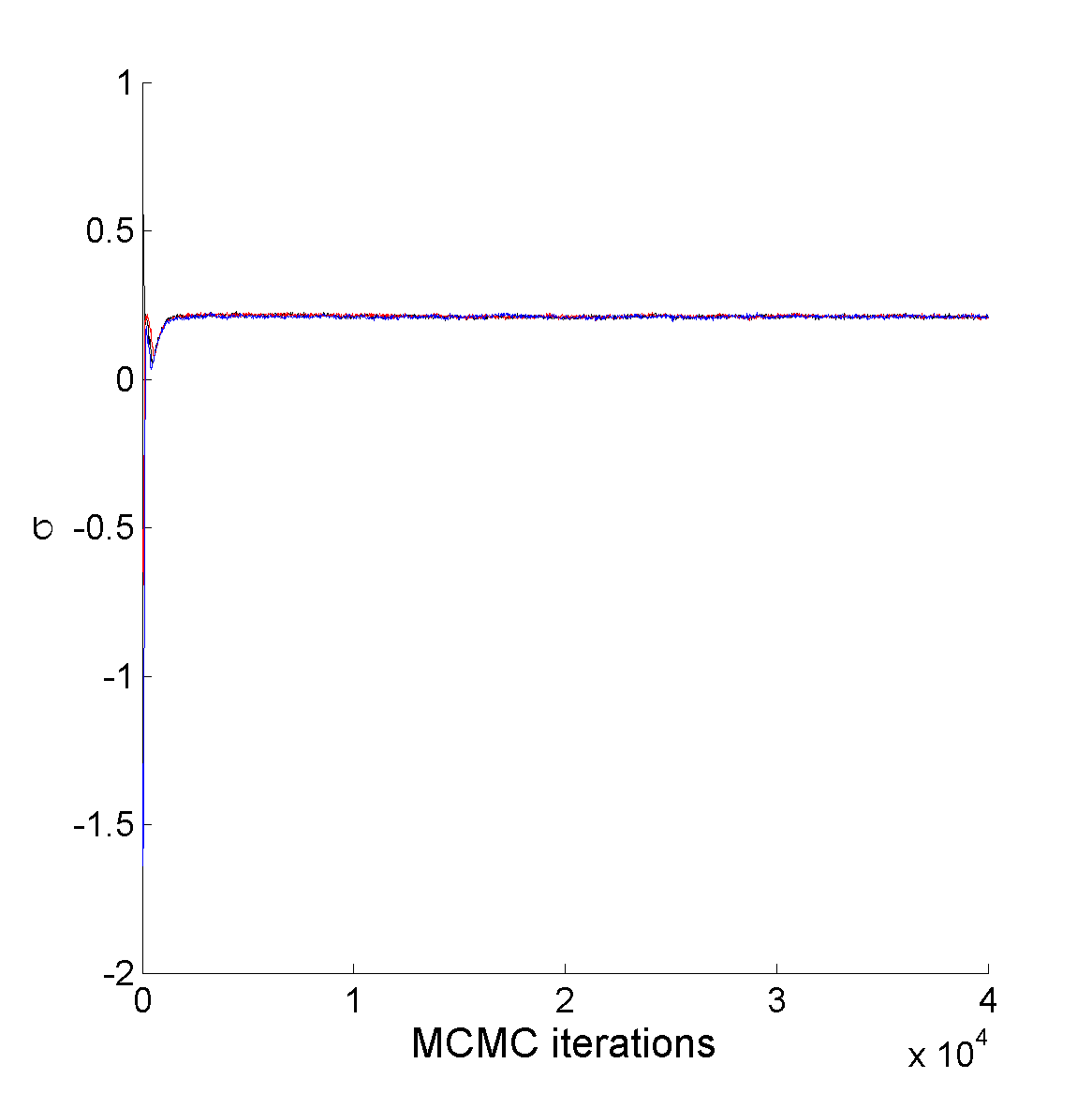}}
\subfigure[internet]{\includegraphics[width=.32\textwidth]{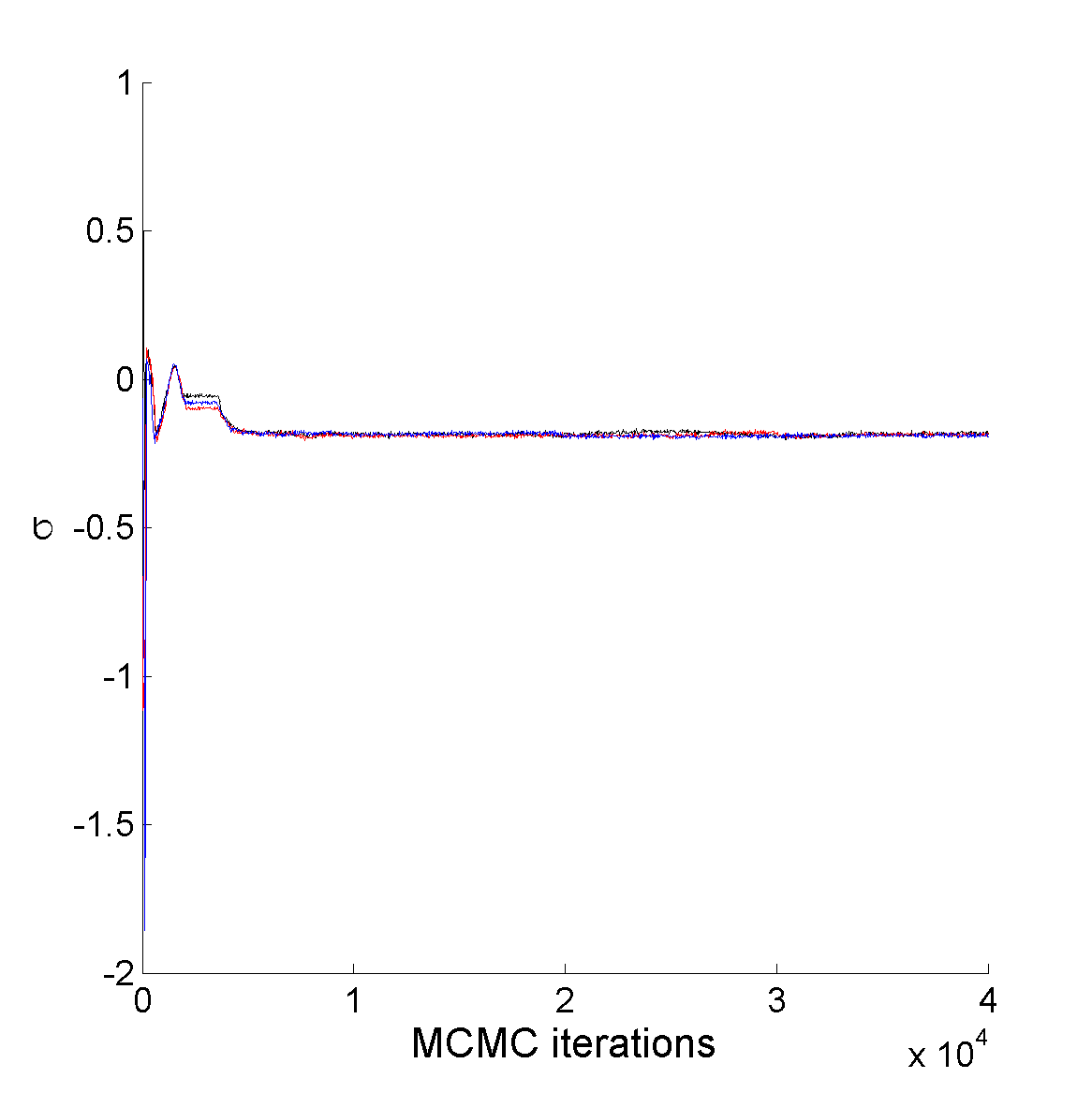}}
\subfigure[www]{\includegraphics[width=.32\textwidth]{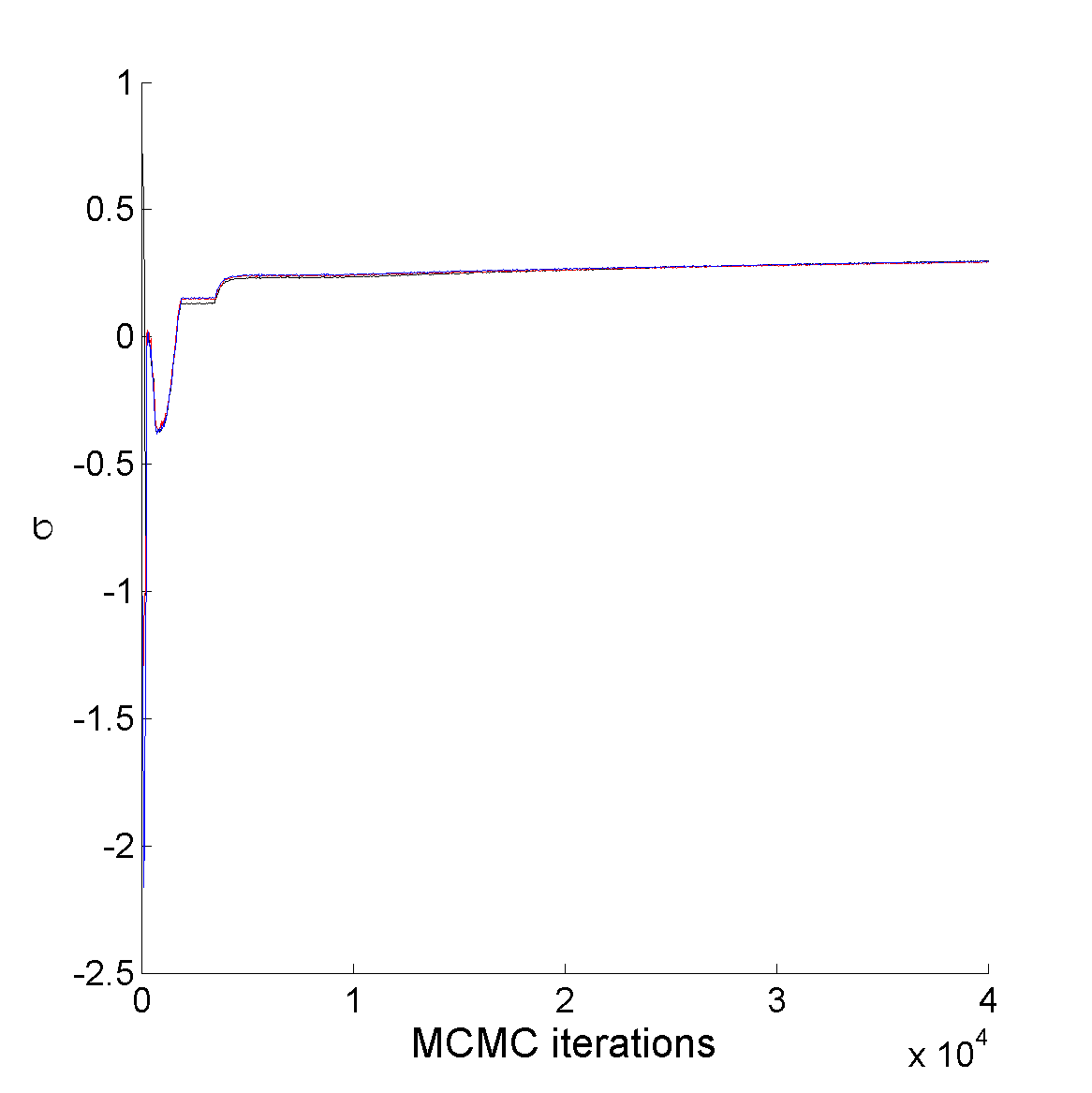}}
\end{center}
\caption{MCMC trace plot for the parameter $\sigma$ for various real-world networks.}%
\label{fig:realtracesigma}%
\end{figure}

\begin{figure}[ptb]
\begin{center}%
\subfigure[facebook107]{\includegraphics[width=.32\textwidth]{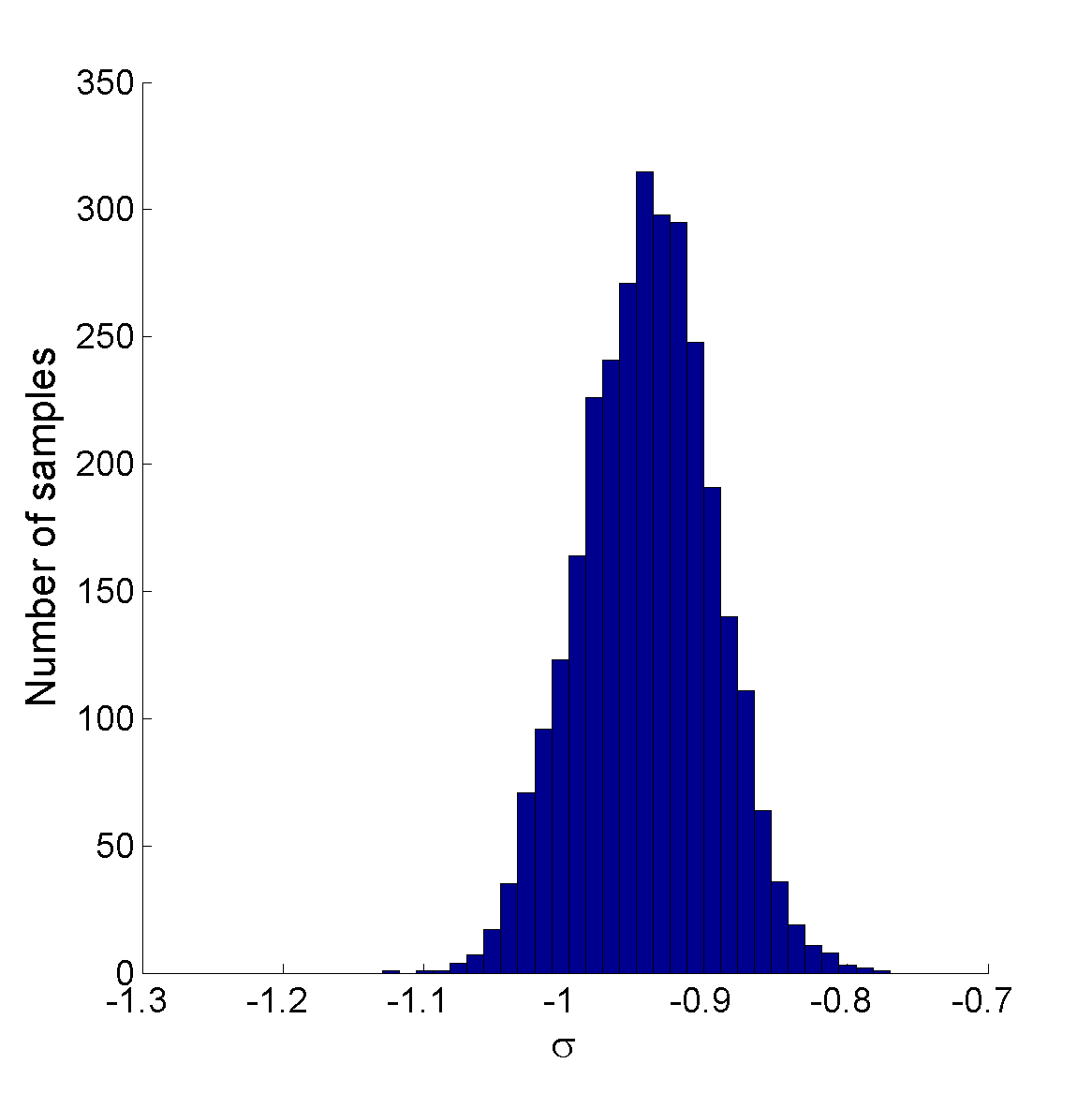}}
\subfigure[polblogs]{\includegraphics[width=.32\textwidth]{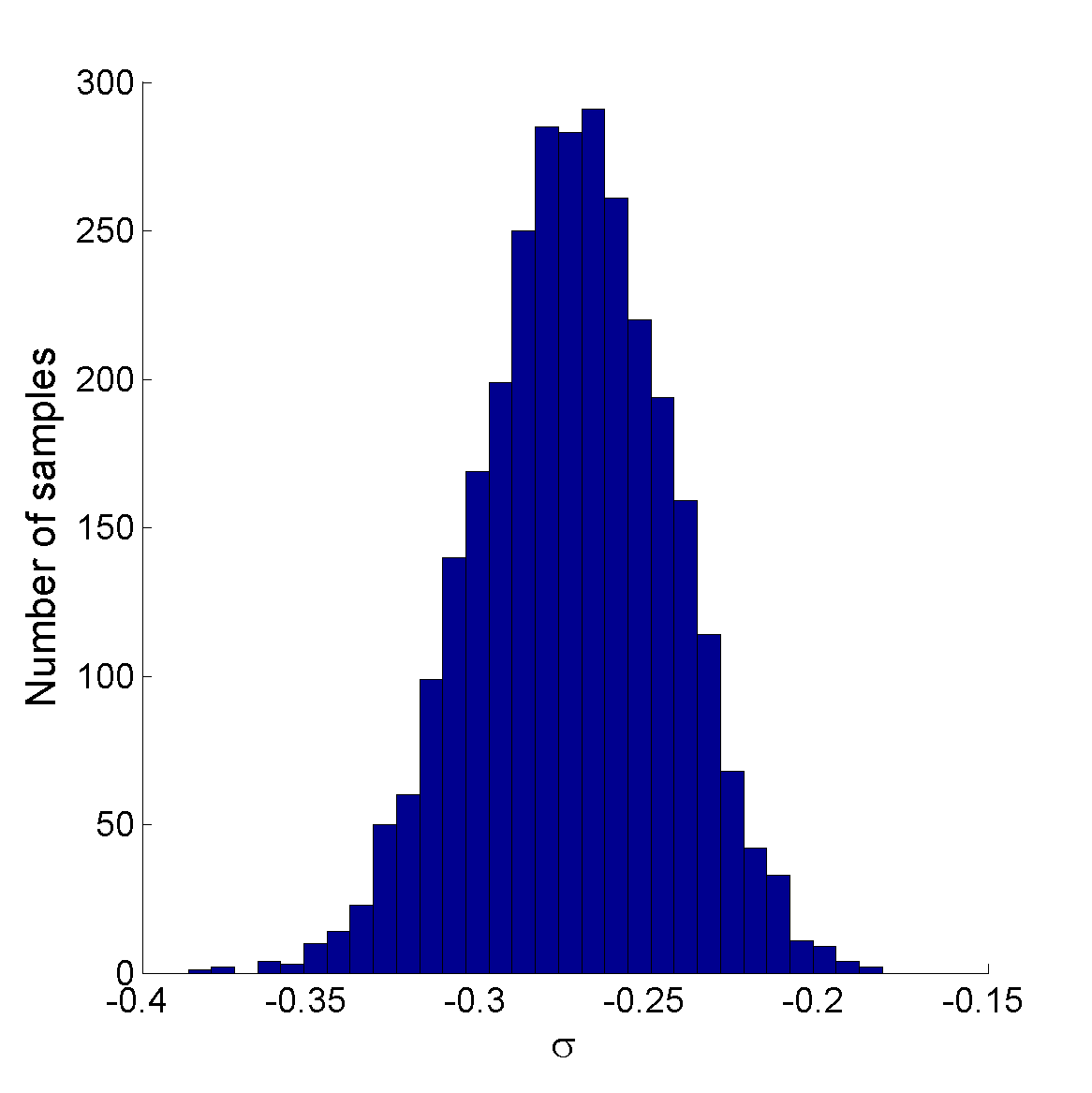}}
\subfigure[USairport]{\includegraphics[width=.32\textwidth]{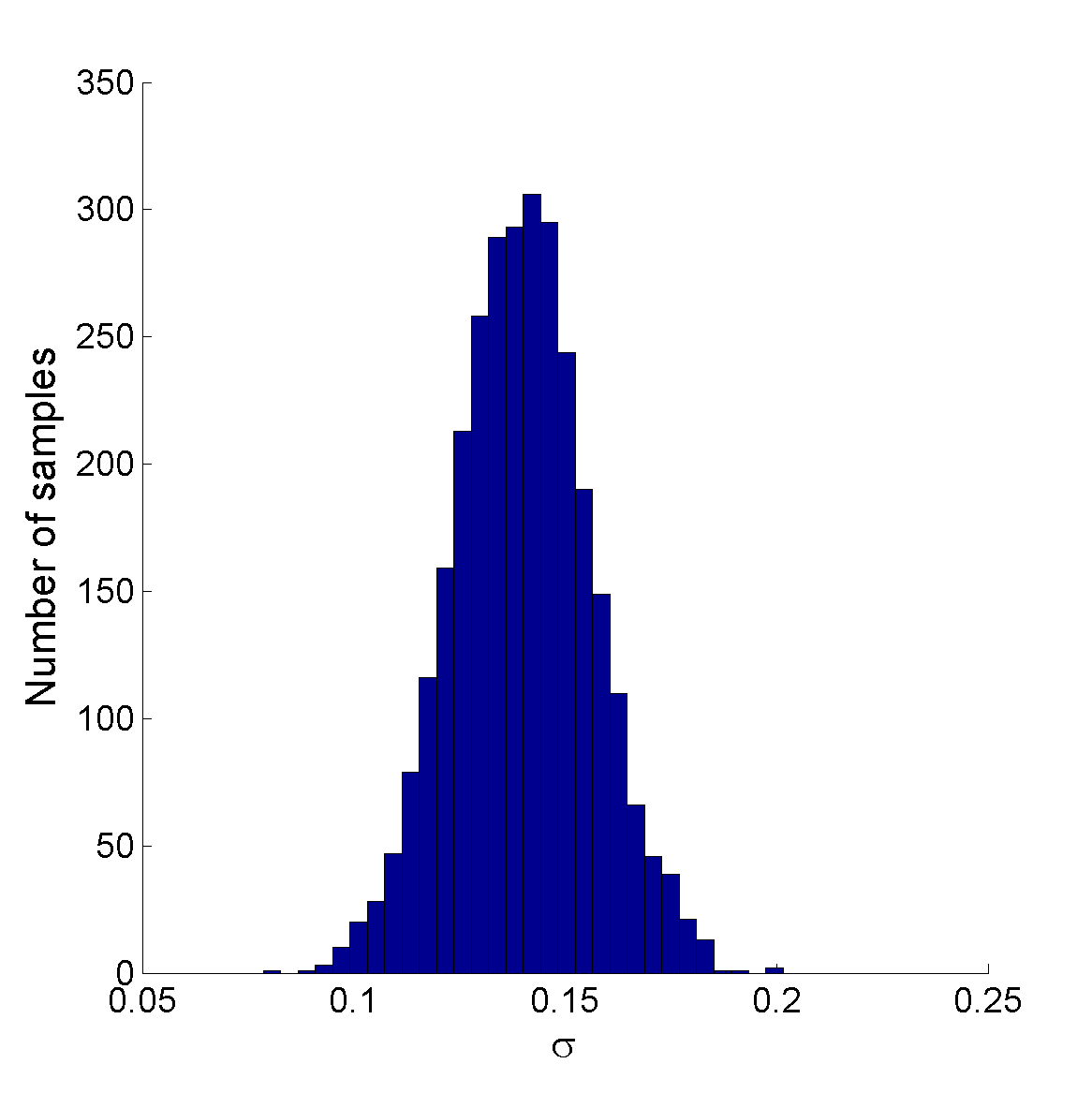}}
\subfigure[UCirvine]{\includegraphics[width=.32\textwidth]{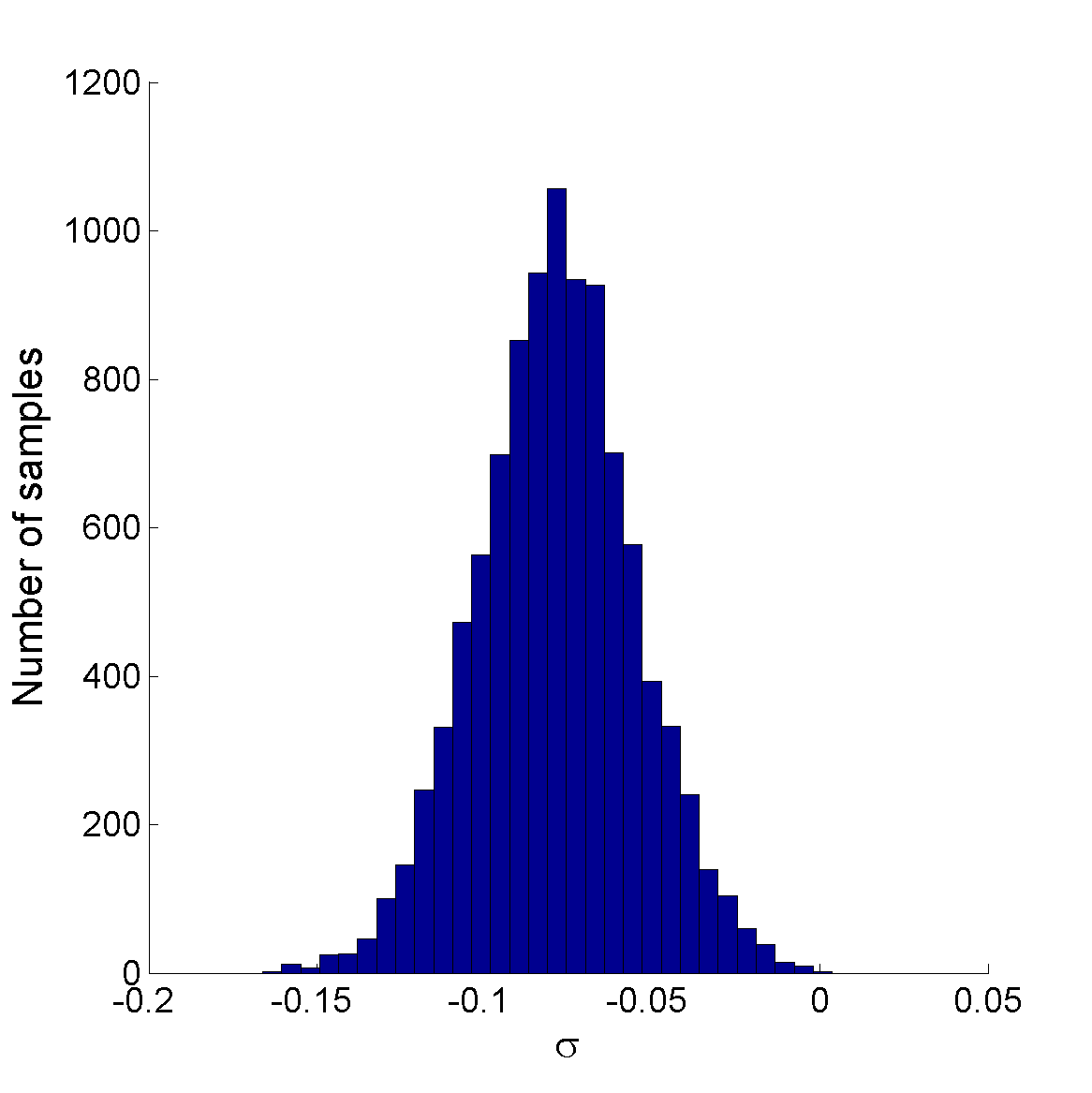}}
\subfigure[yeast]{\includegraphics[width=.32\textwidth]{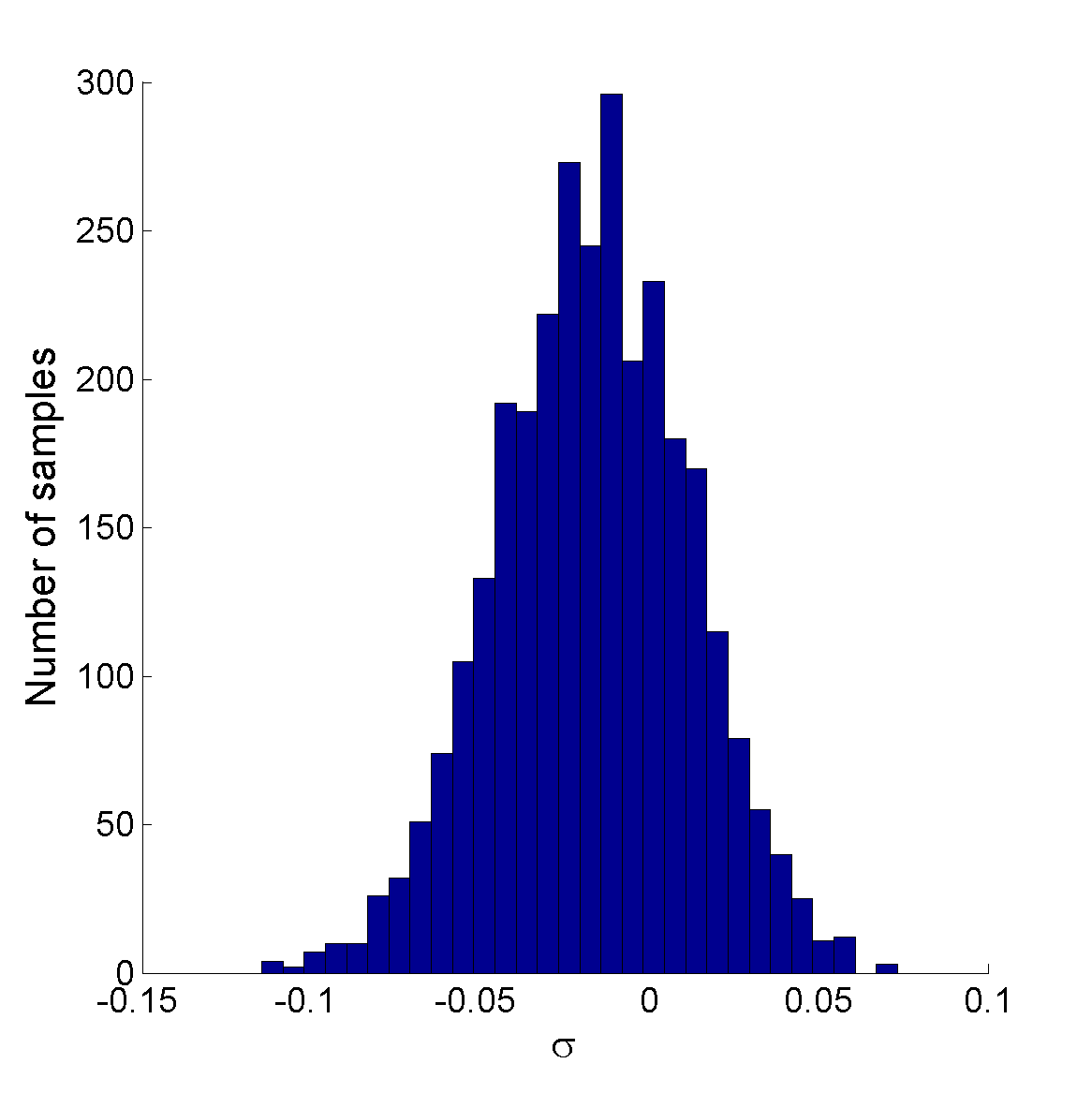}}
\subfigure[USpower]{\includegraphics[width=.32\textwidth]{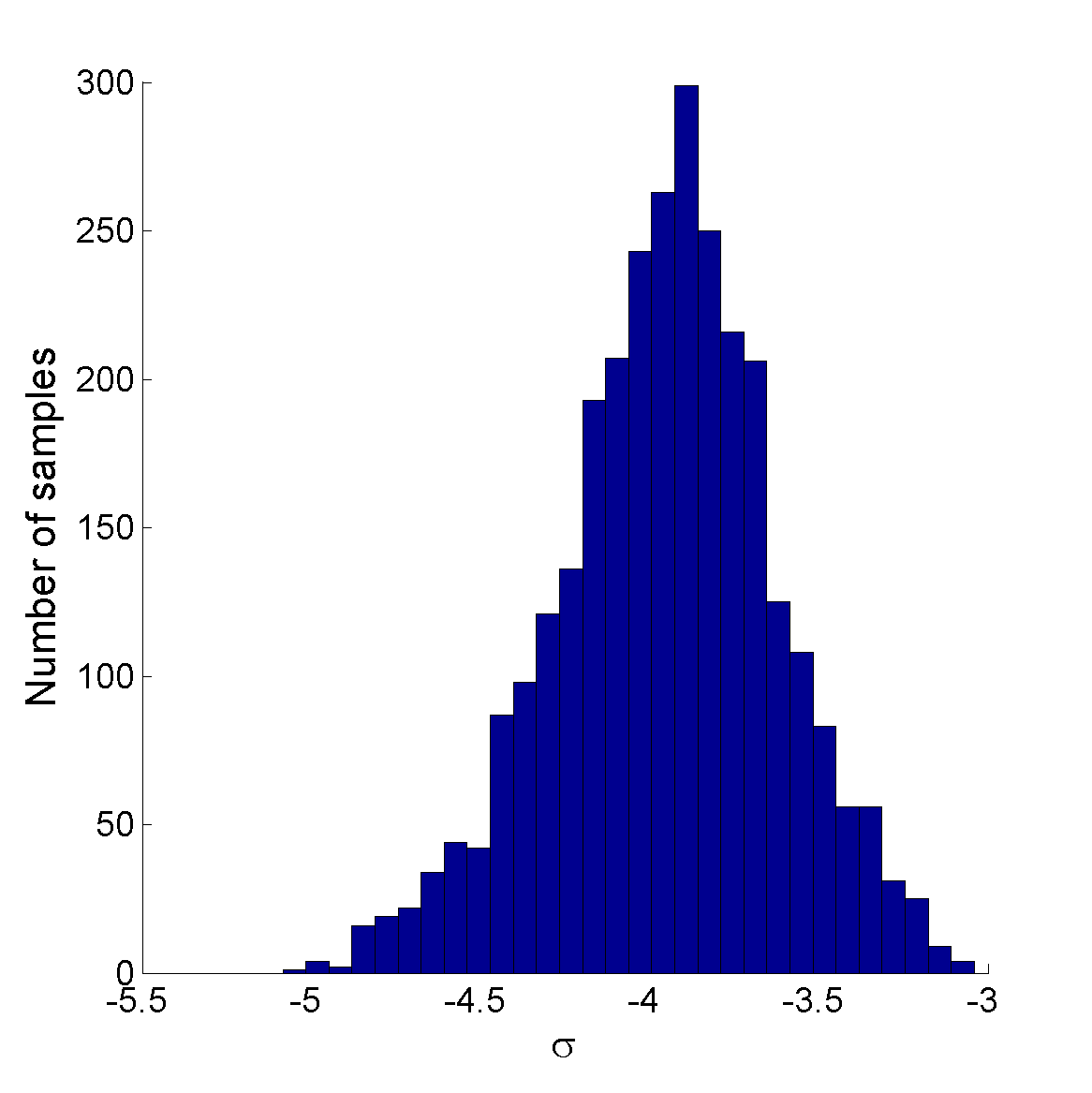}}
\subfigure[IMDB]{\includegraphics[width=.32\textwidth]{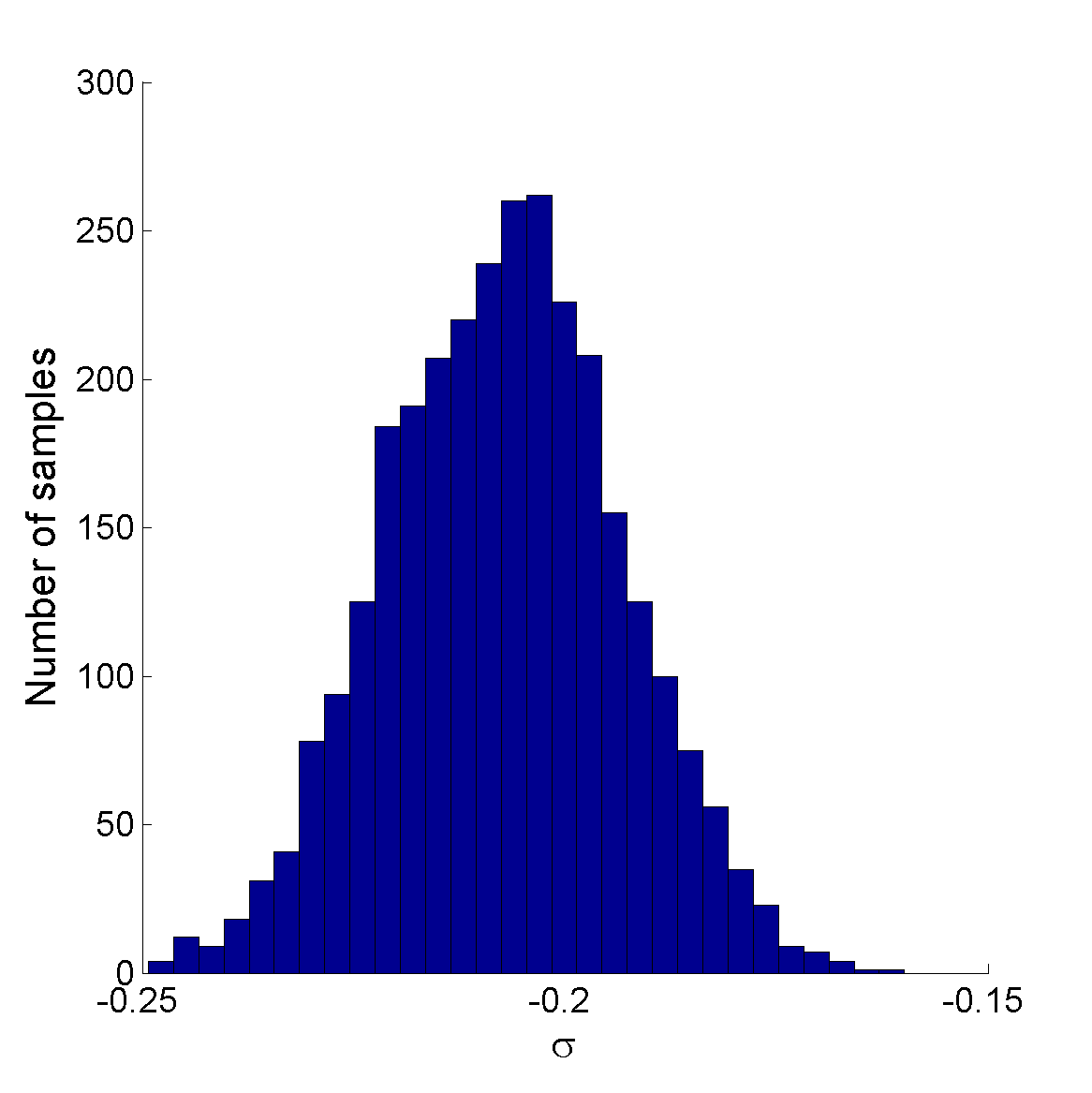}}
\subfigure[cond-mat1]{\includegraphics[width=.32\textwidth]{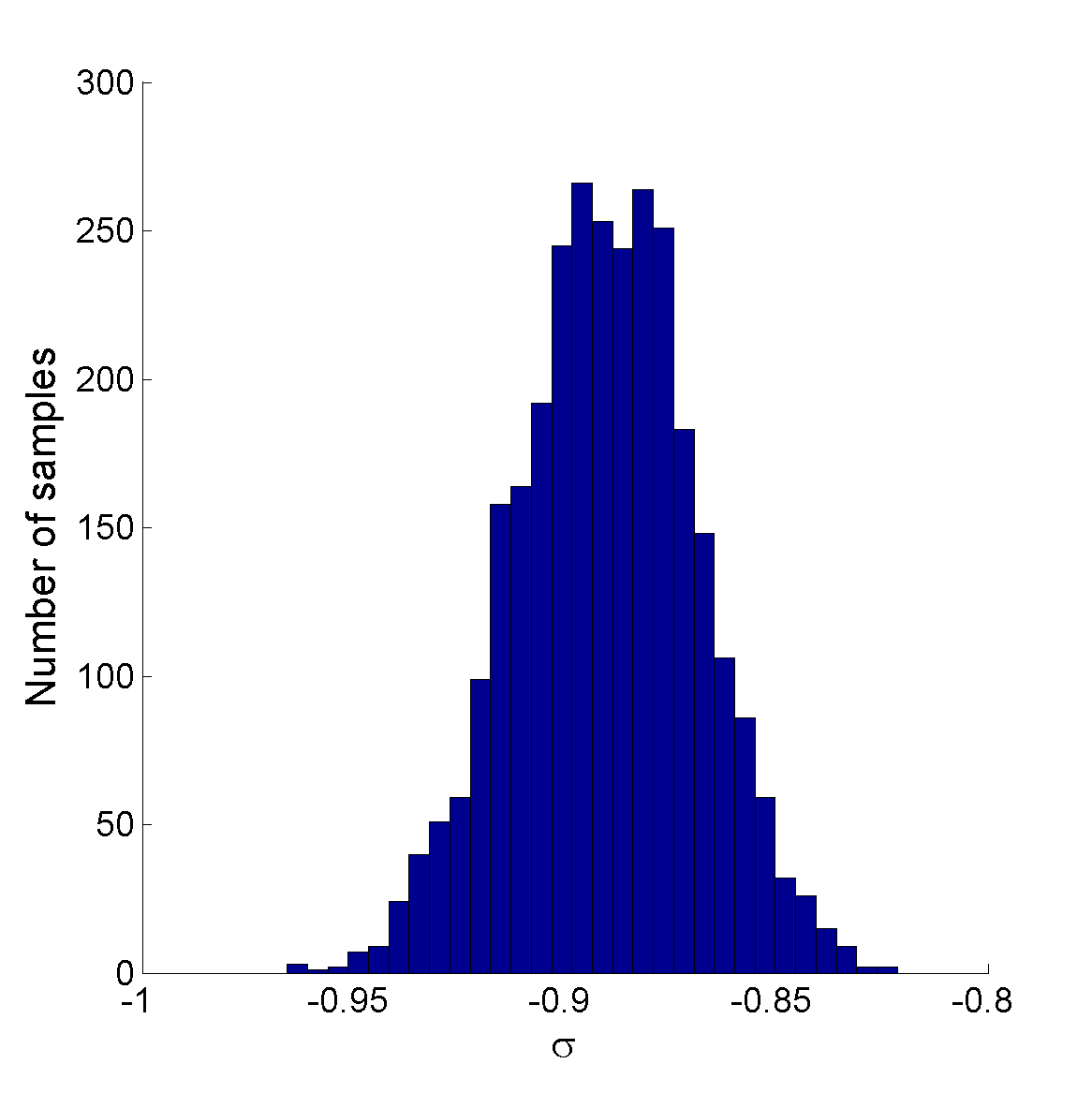}}
\subfigure[cond-mat2]{\includegraphics[width=.32\textwidth]{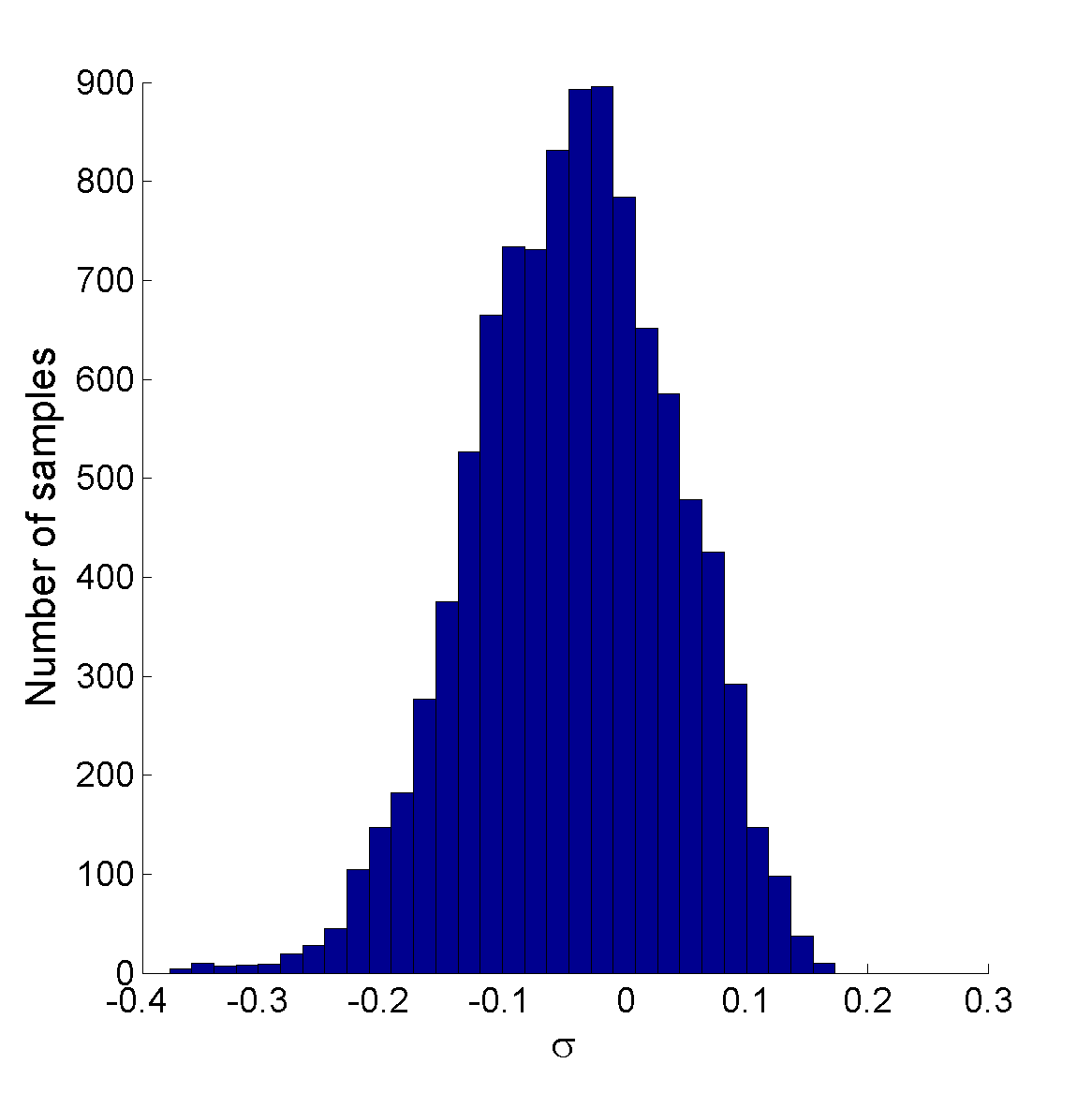}}
\subfigure[enron]{\includegraphics[width=.32\textwidth]{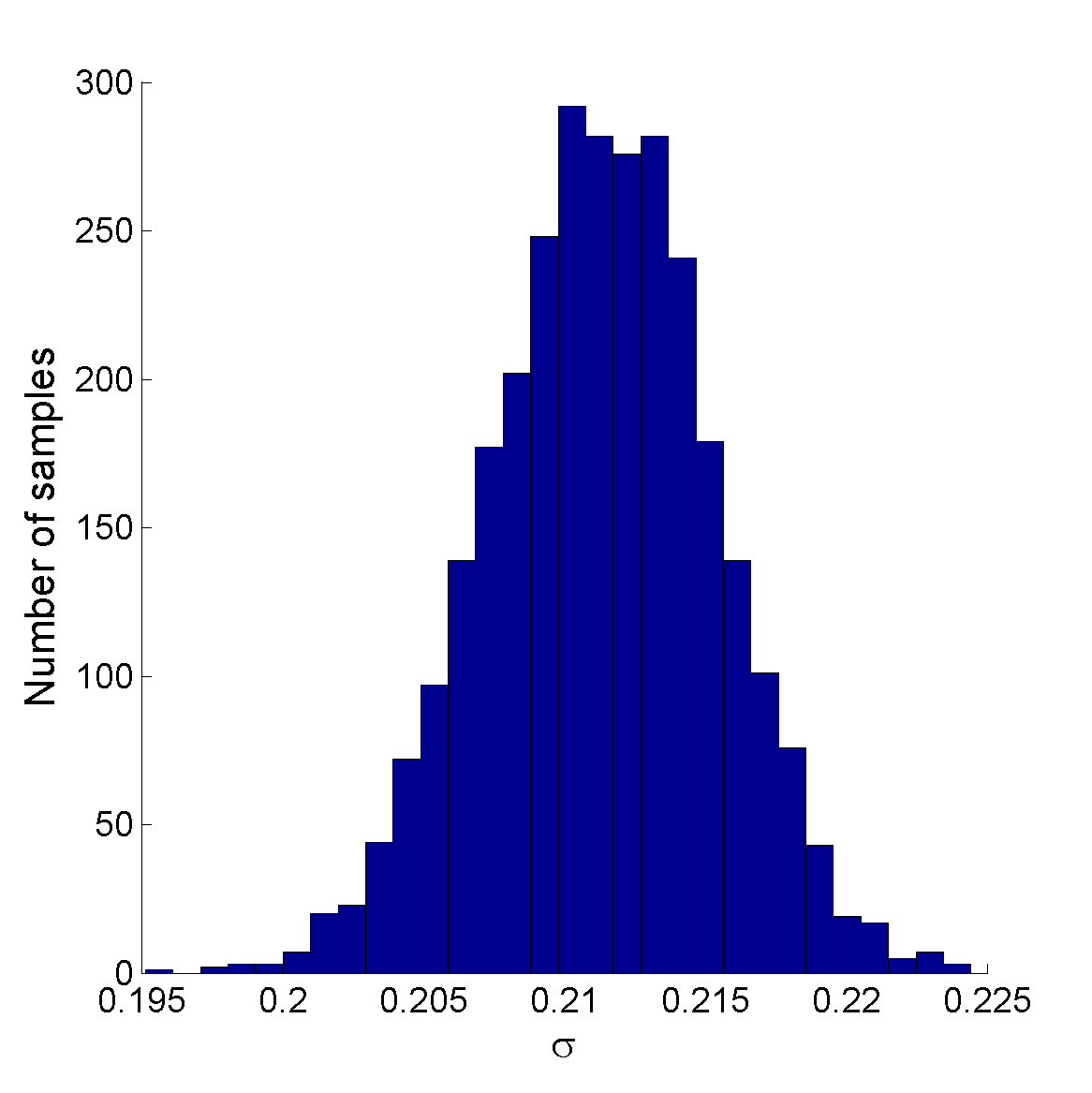}}
\subfigure[internet]{\includegraphics[width=.32\textwidth]{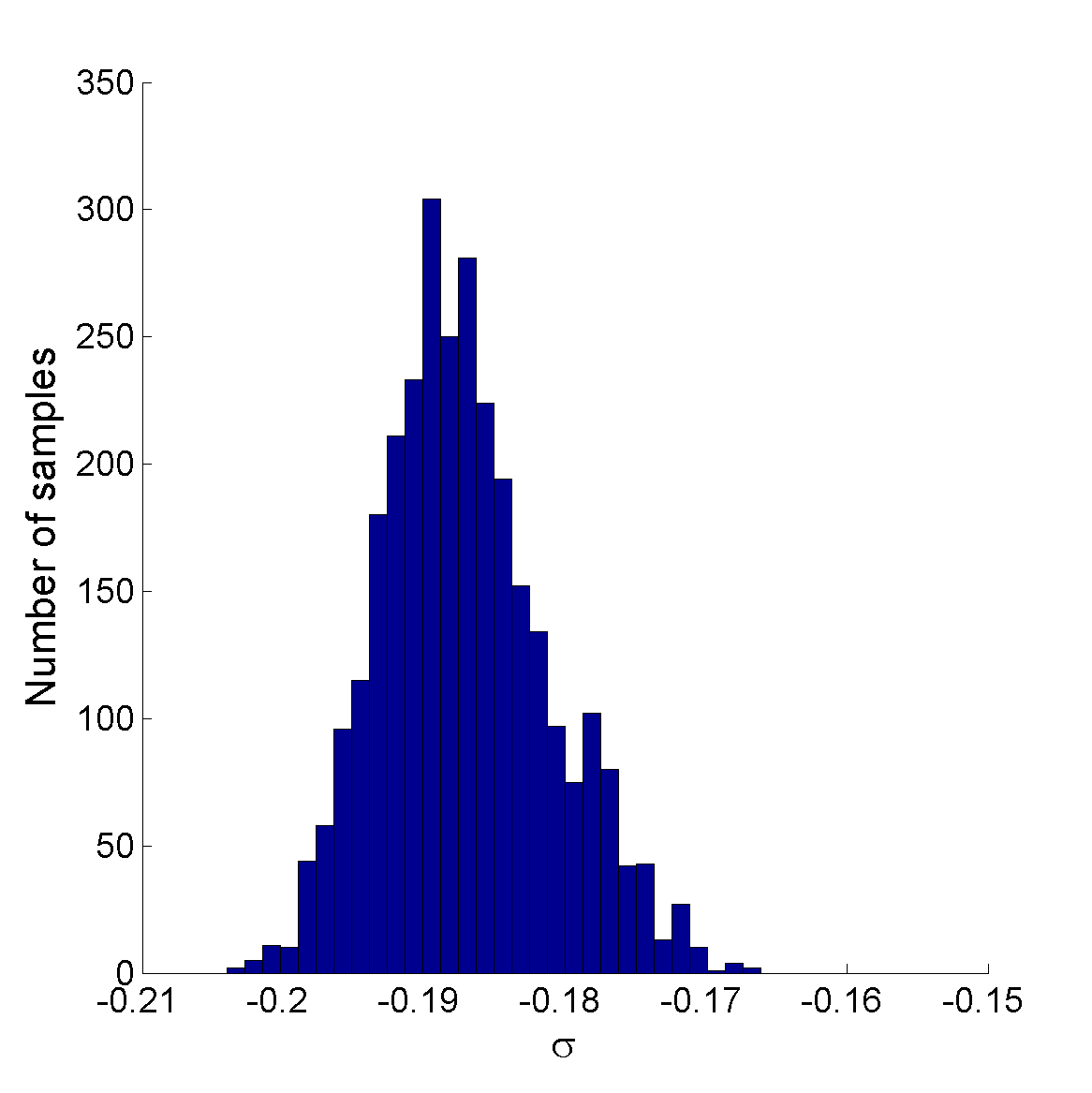}}
\subfigure[www]{\includegraphics[width=.32\textwidth]{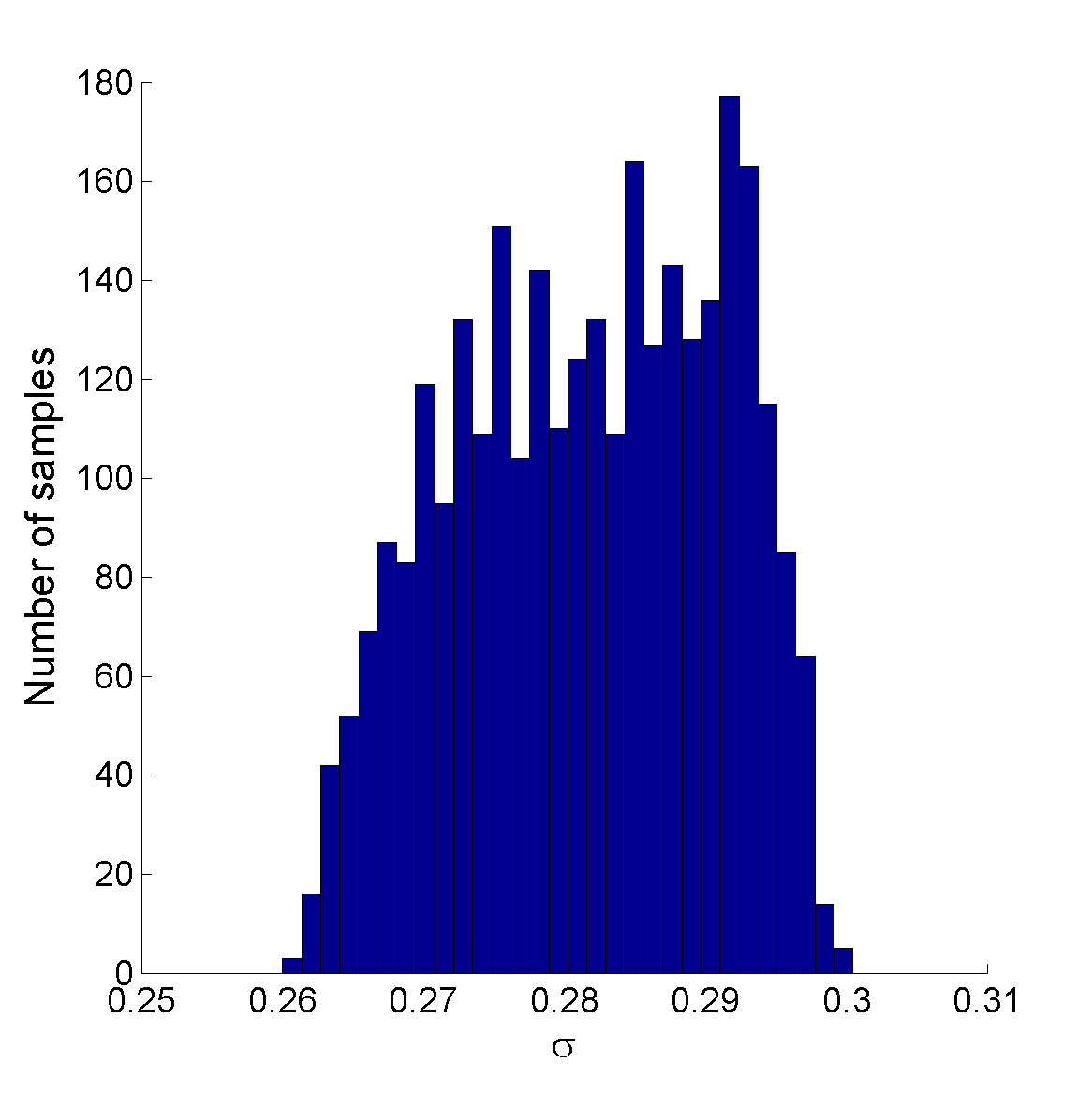}}
\end{center}
\caption{Histograms of MCMC samples of the parameter $\sigma$ for various real-world networks.}%
\label{fig:realhistsigma}%
\end{figure}

\begin{figure}[ptb]
\begin{center}%
\subfigure[facebook107]{\includegraphics[width=.32\textwidth]{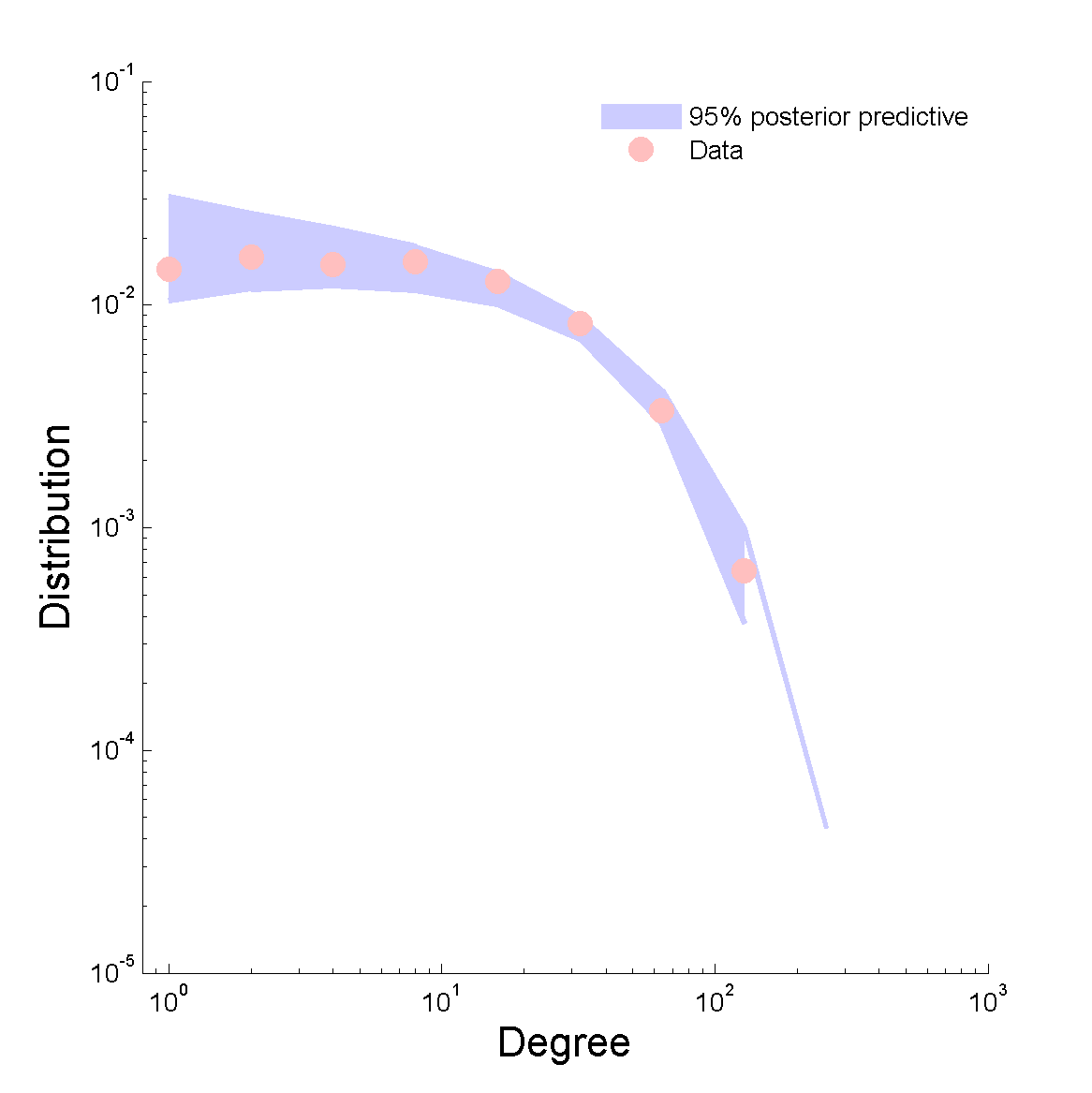}}
\subfigure[polblogs]{\includegraphics[width=.32\textwidth]{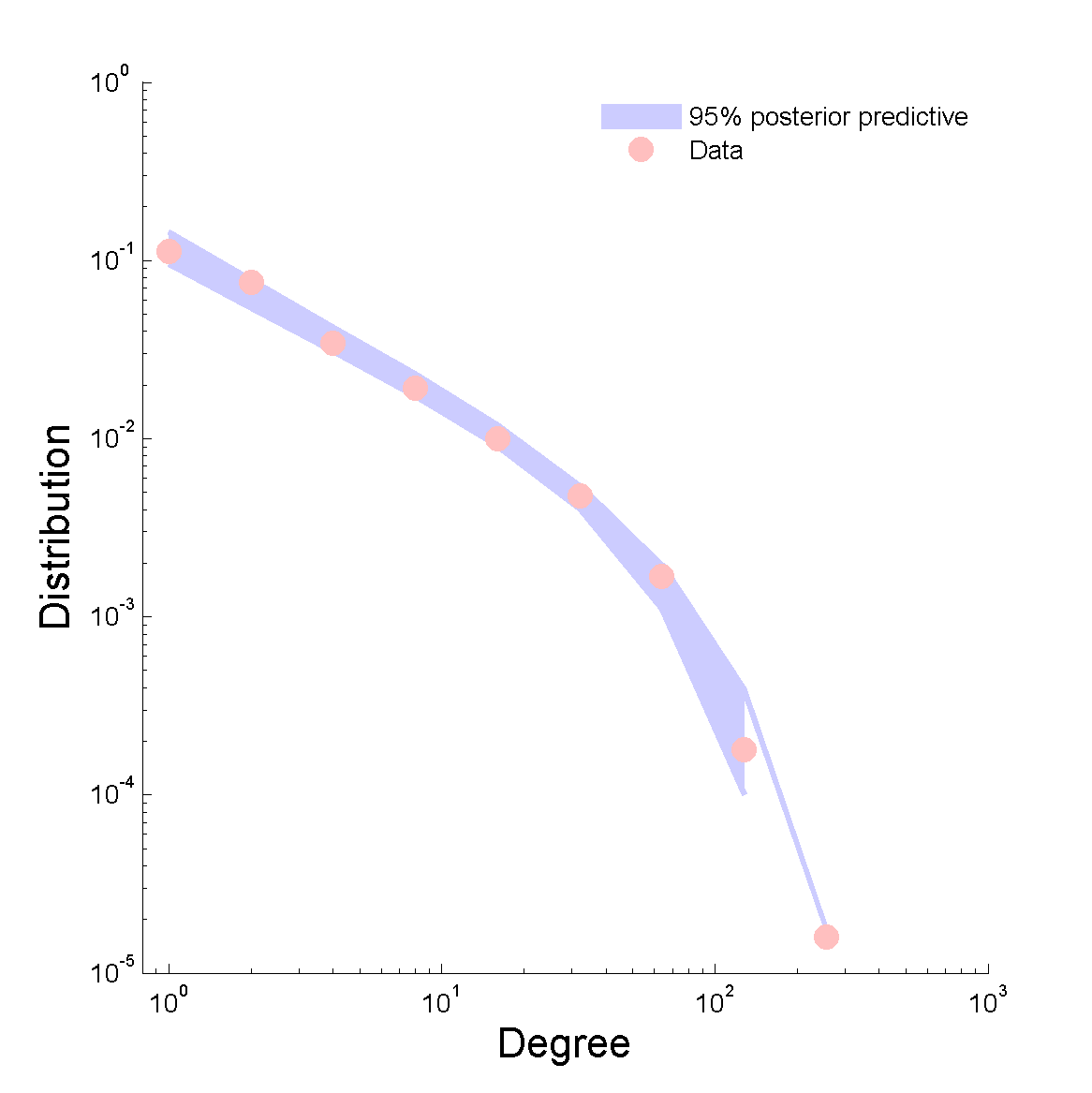}}
\subfigure[USairport]{\includegraphics[width=.32\textwidth]{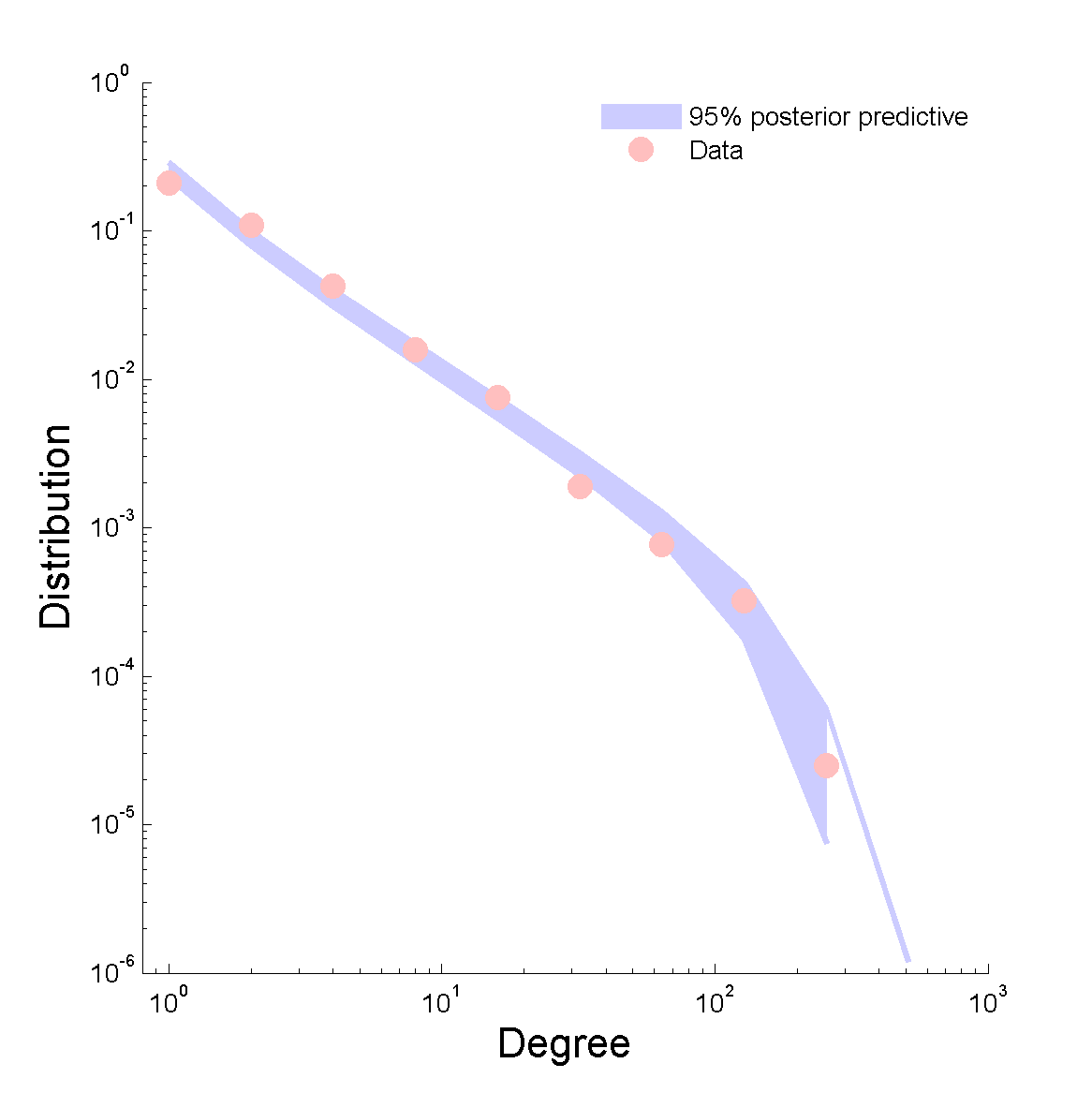}}
\subfigure[UCirvine]{\includegraphics[width=.32\textwidth]{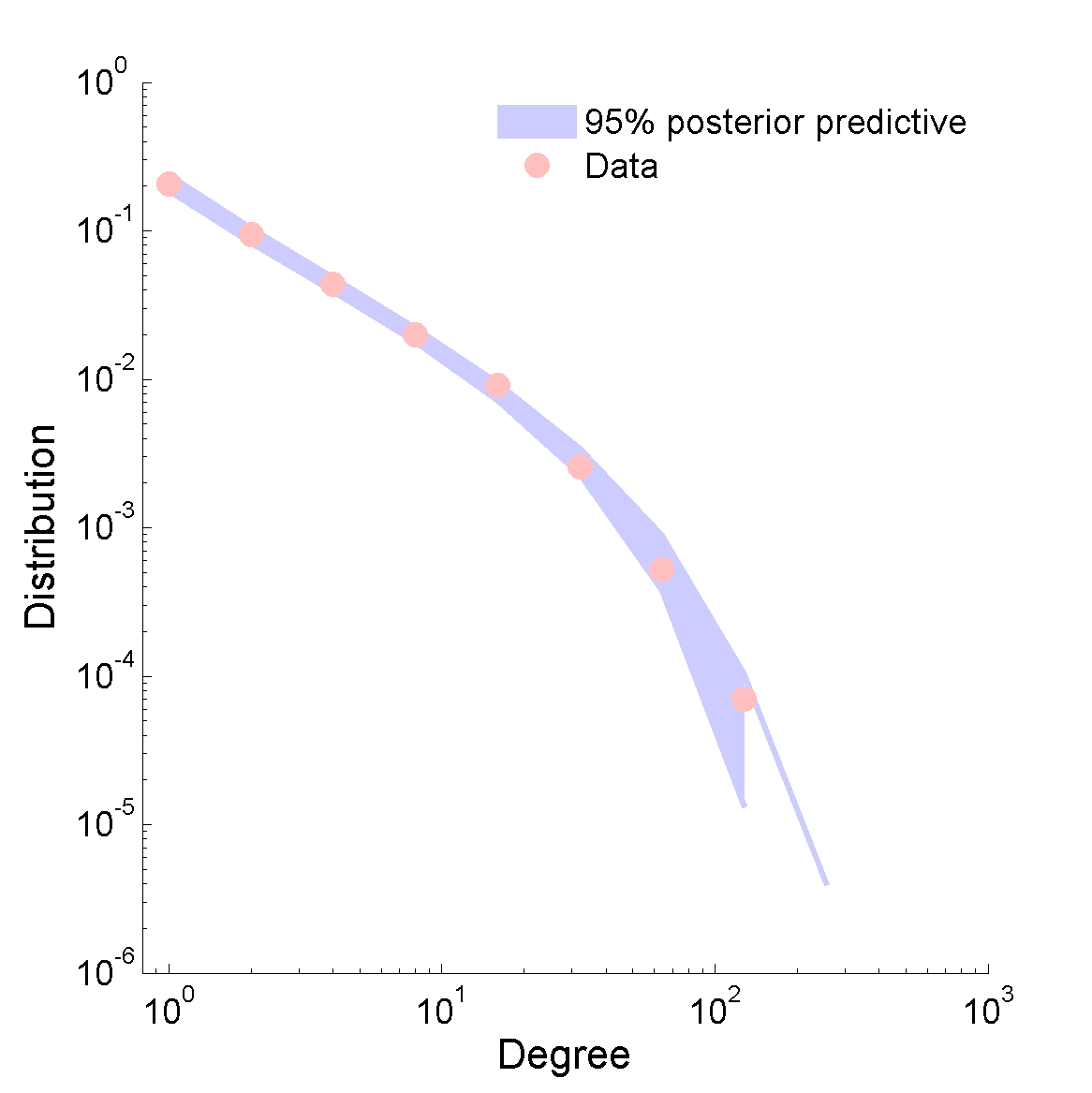}}
\subfigure[yeast]{\includegraphics[width=.32\textwidth]{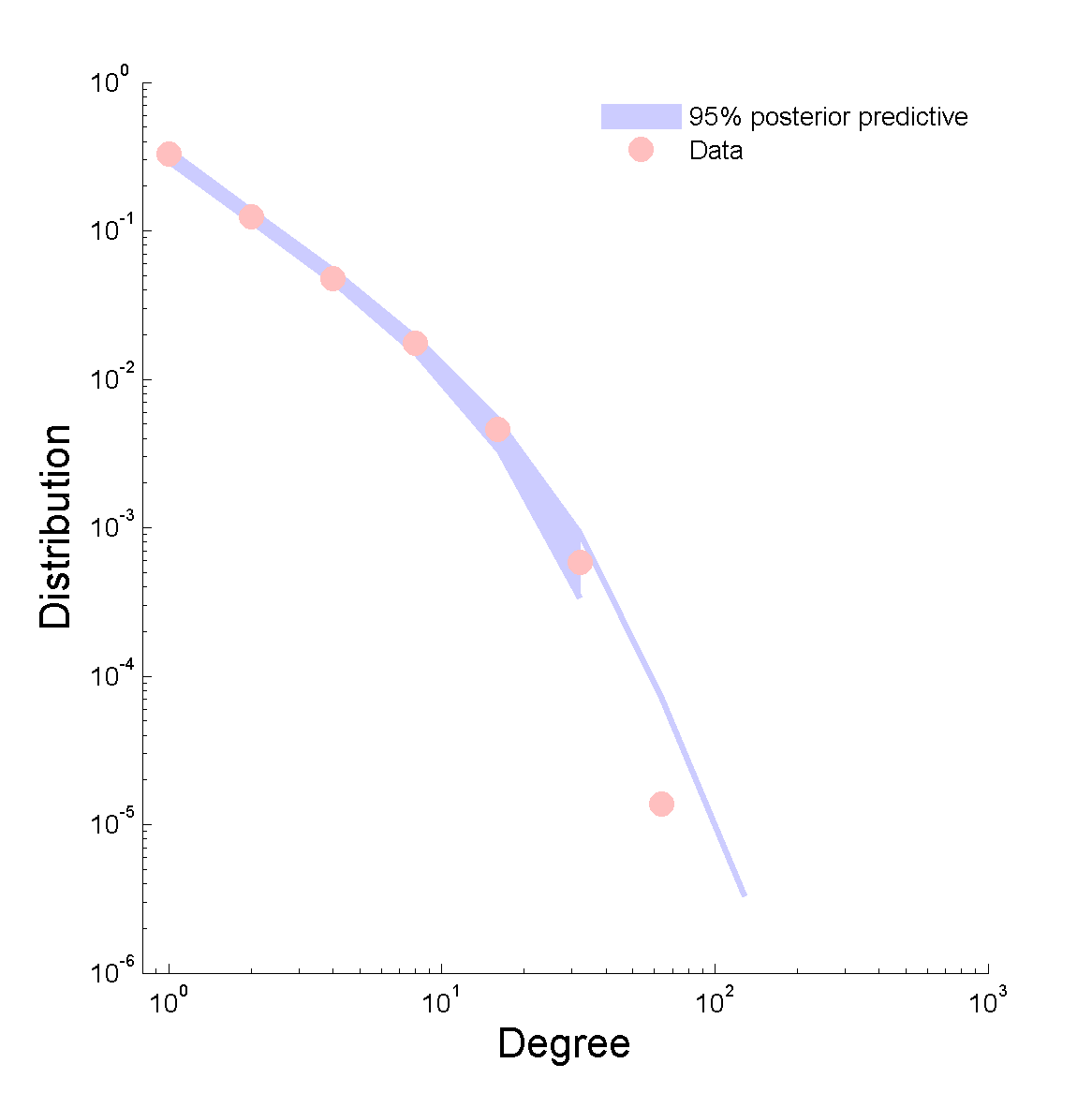}}
\subfigure[USpower]{\includegraphics[width=.32\textwidth]{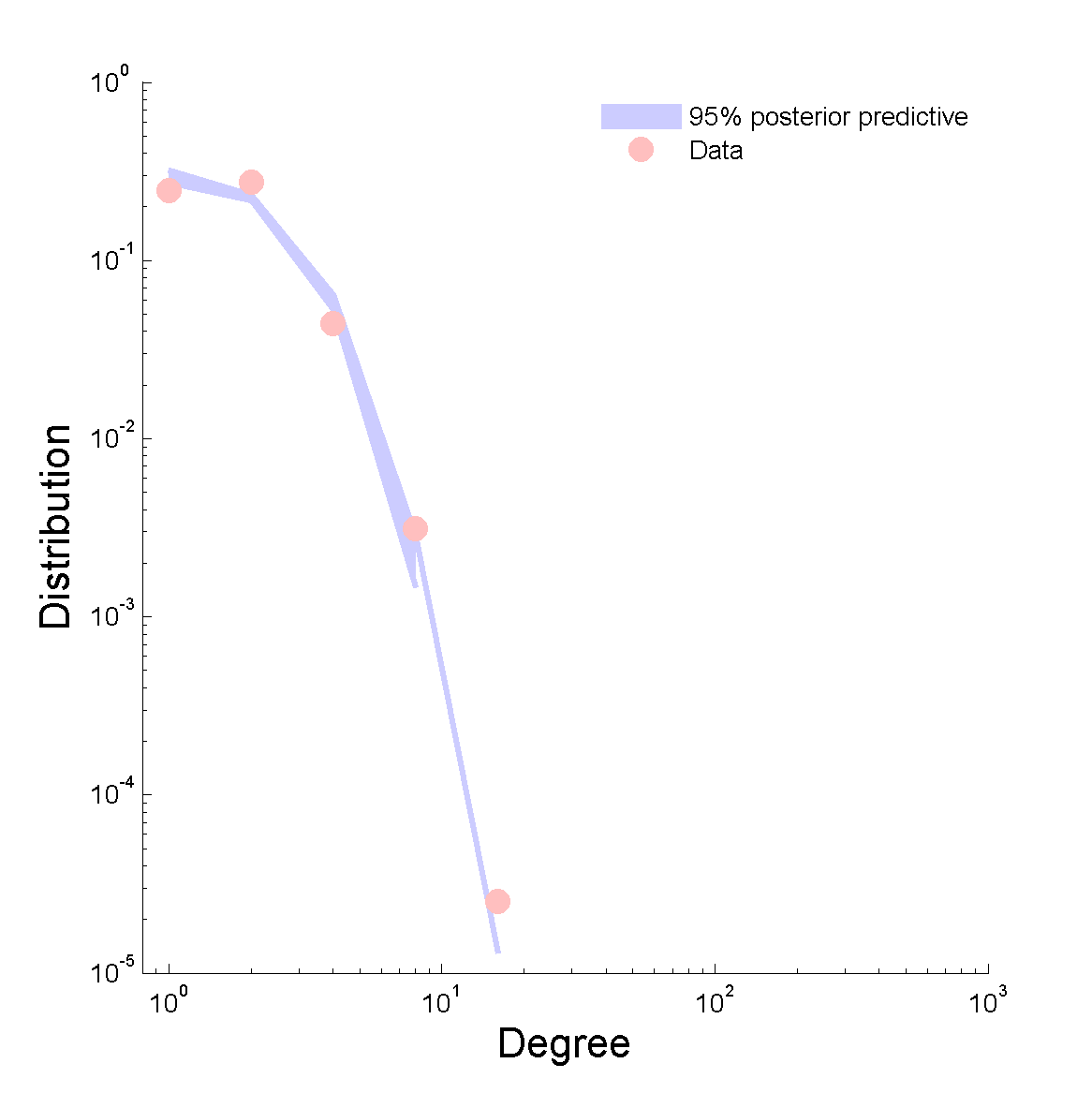}}
\subfigure[IMDB]{\includegraphics[width=.32\textwidth]{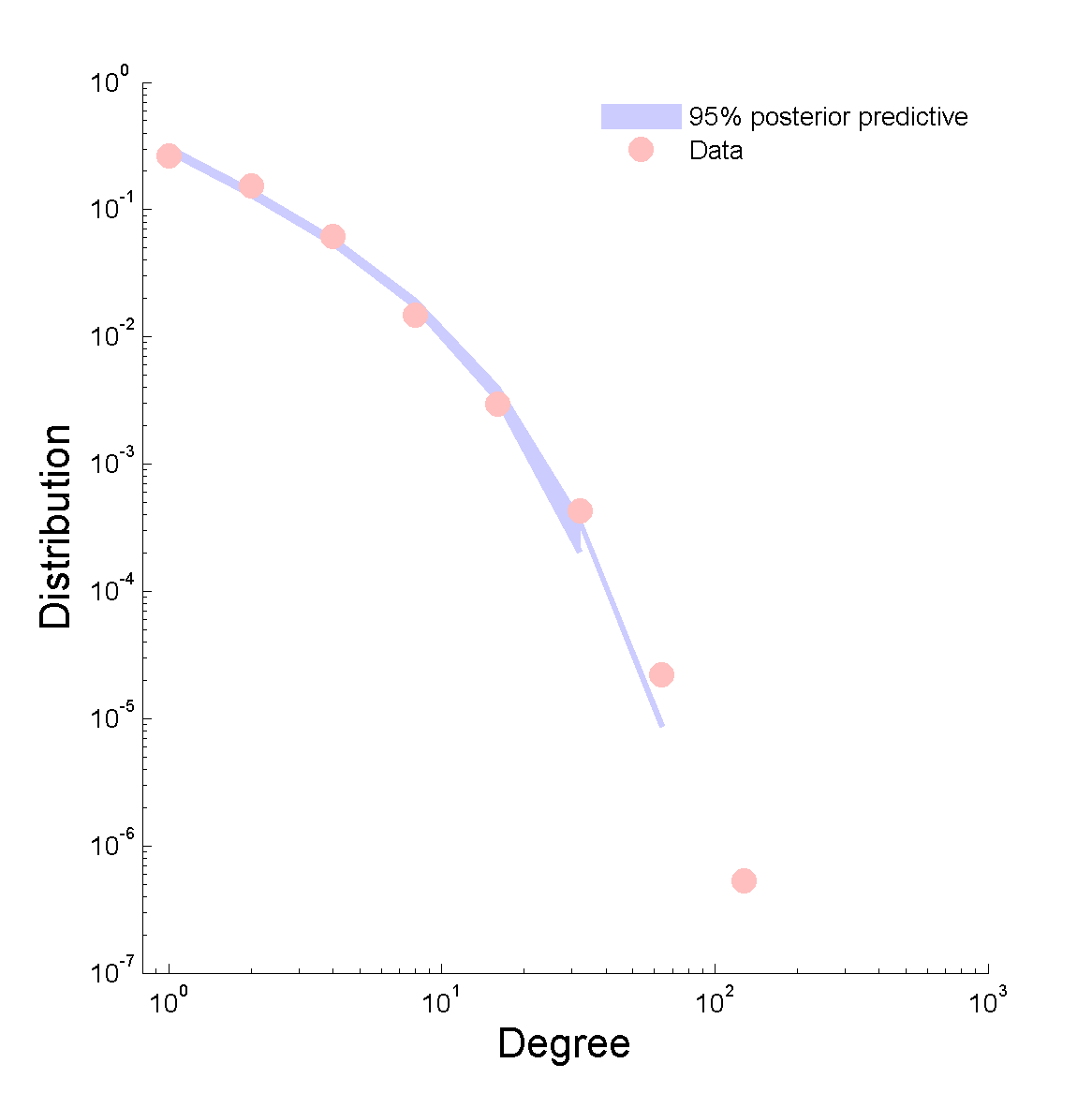}}
\subfigure[cond-mat1]{\includegraphics[width=.32\textwidth]{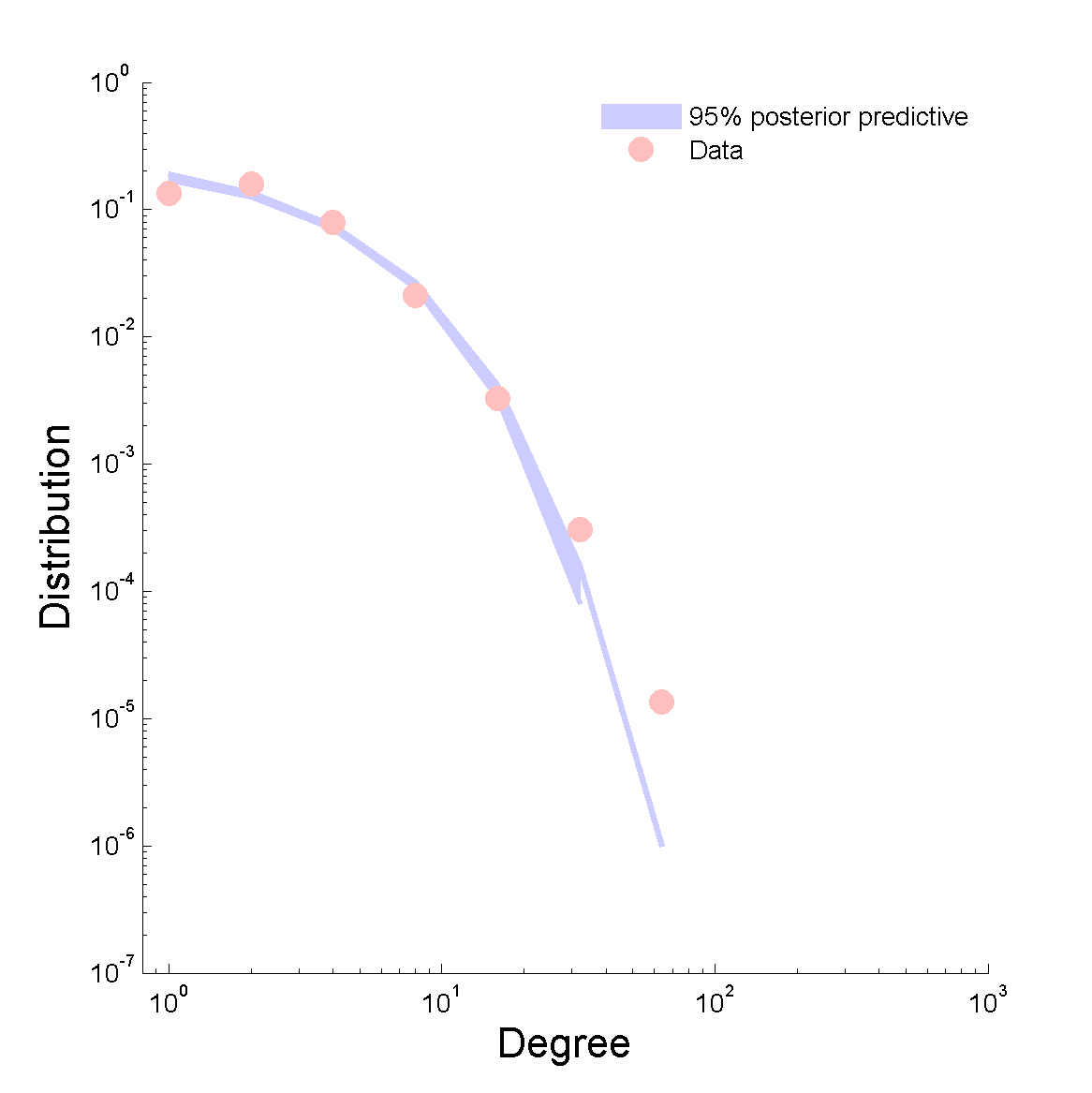}}
\subfigure[cond-mat2]{\includegraphics[width=.32\textwidth]{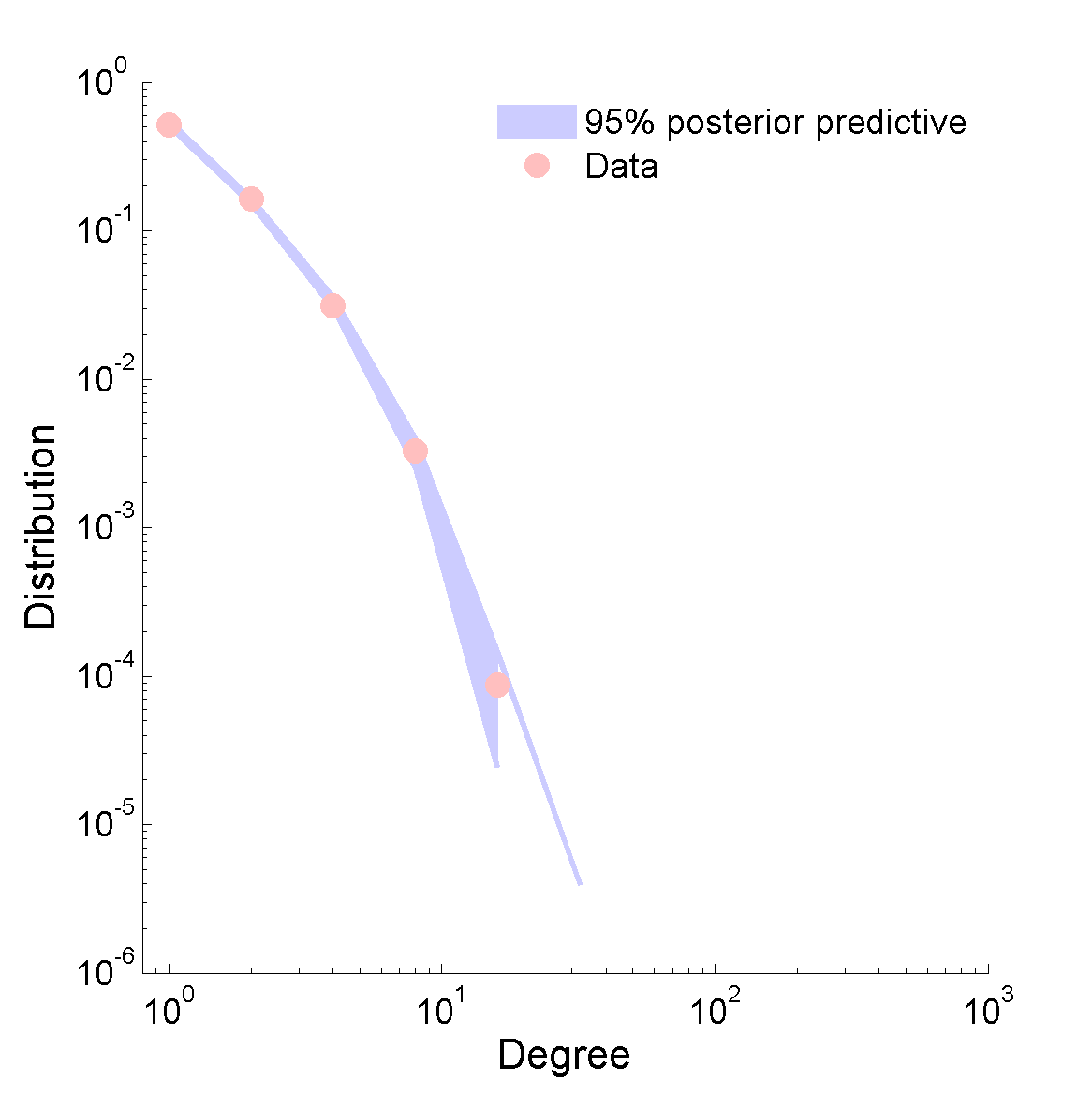}}
\subfigure[enron]{\includegraphics[width=.32\textwidth]{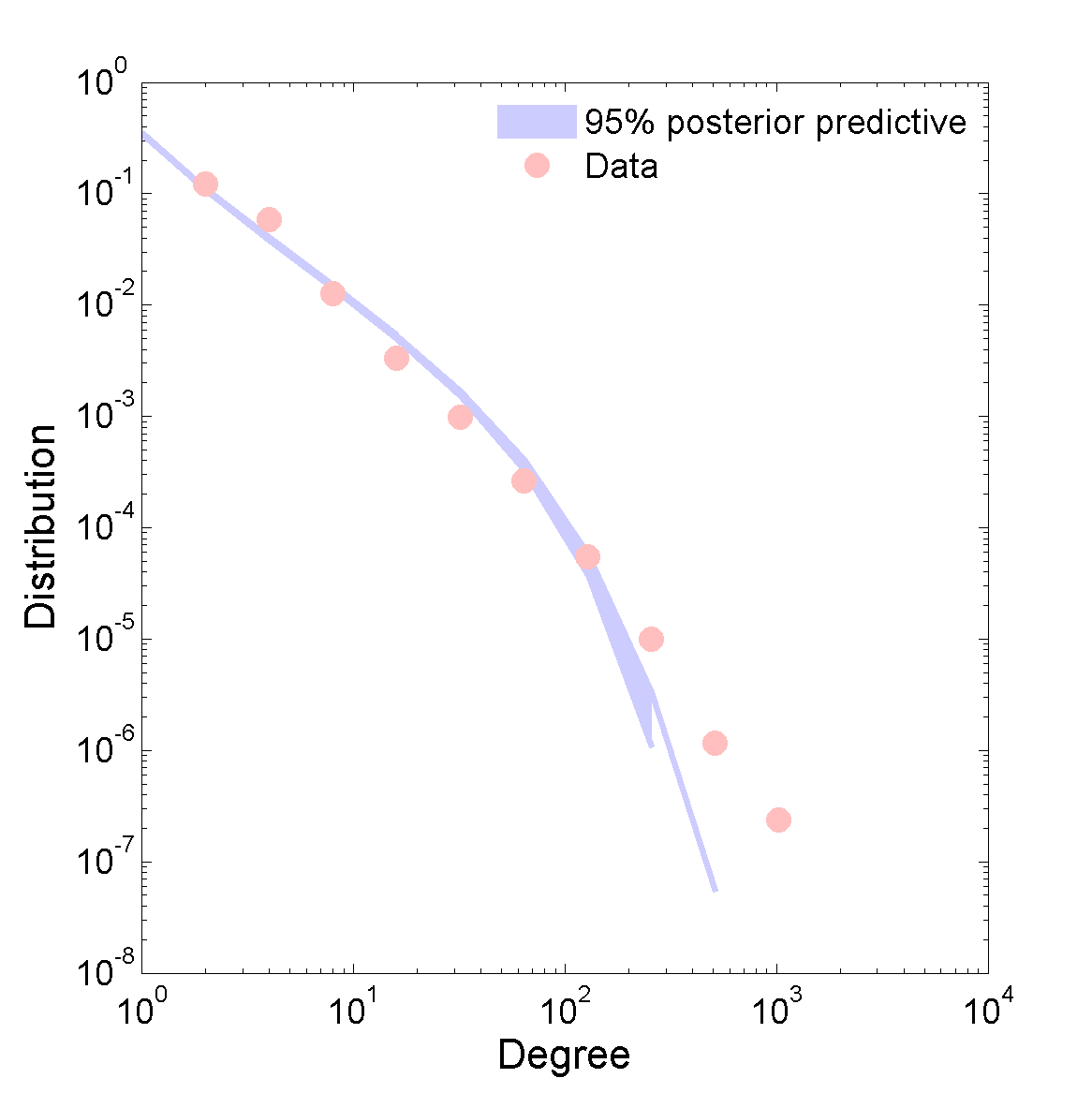}}
\subfigure[internet]{\includegraphics[width=.32\textwidth]{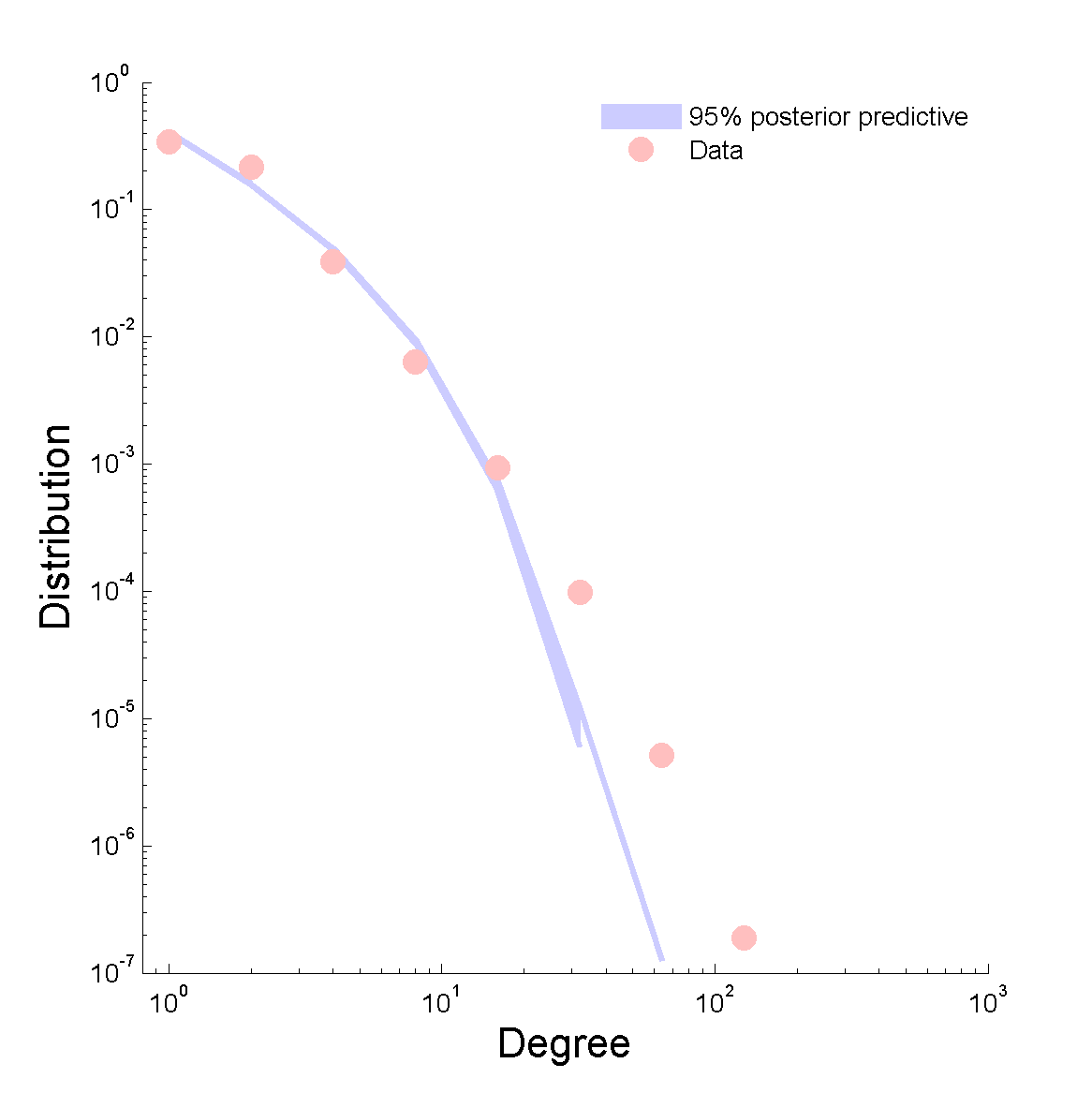}}
\subfigure[www]{\includegraphics[width=.32\textwidth]{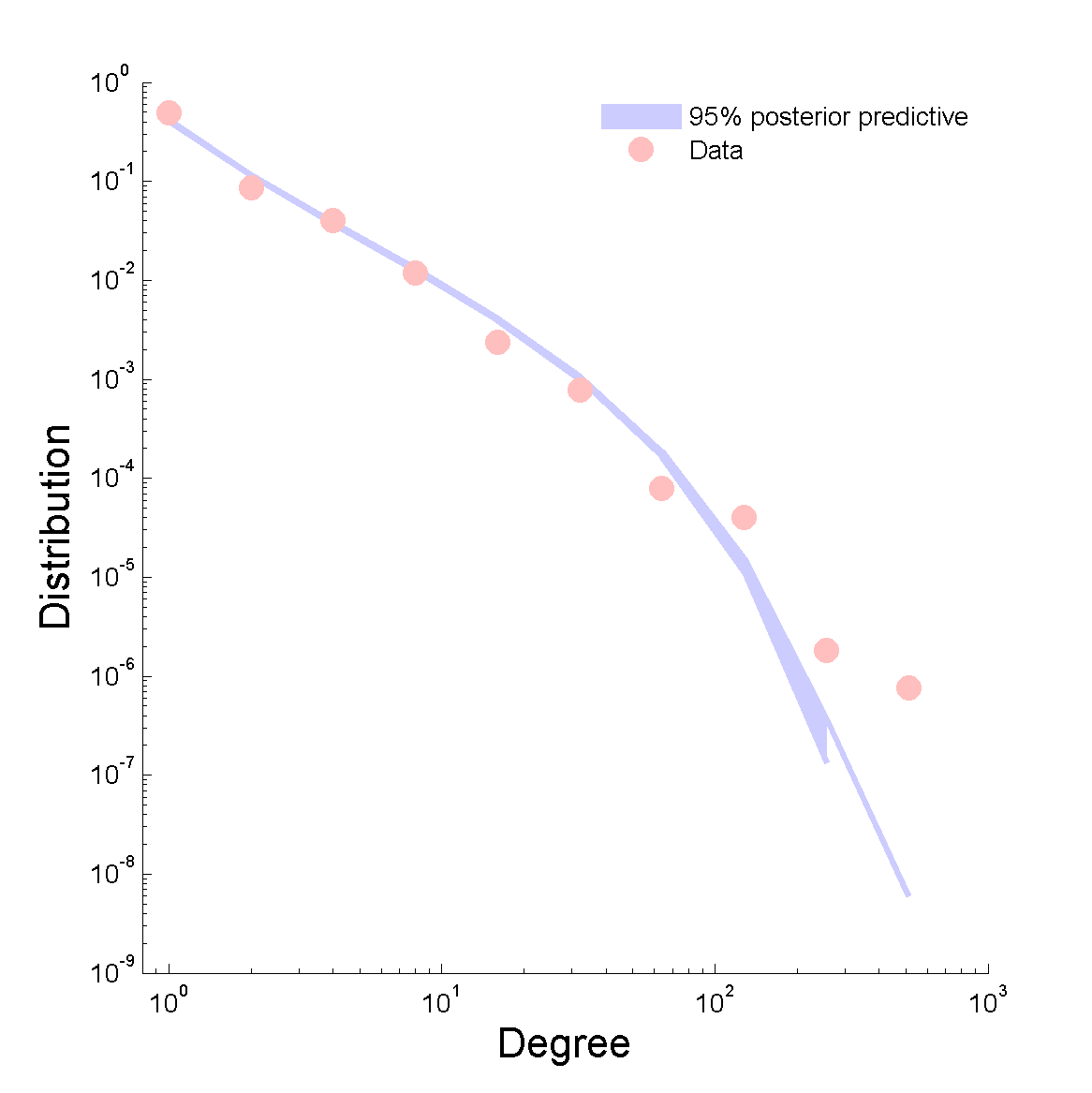}}
\end{center}
\caption{Empirical degree distribution (red) and posterior predictive (blue) for various real-world networks.}%
\label{fig:realdegree}%
\end{figure}

\section{Discussion}

There has been extensive work over the past years on flexible Bayesian nonparametric models for networks, allowing complex latent structures of unknown dimension to be uncovered from real-world networks~\cite{Kemp2006,Miller2009,Lloyd2012,Palla2012,Herlau2014}. However, as mentioned in the unifying overview of~\citet{Orbanz2015}, these methods all fit in the Aldous-Hoover framework and as such produce dense graphs. 

\citet{Norros2006} (see also~\cite{vanderHofstad2014} for a review and~\cite{Britton2006} for a similar model) proposed a conditionally Poissonian multigraph process with similarities to be drawn to our multigraph process. They consider that each node has a given sociability parameter, and the number of edges between two nodes $i$ and $j$ is drawn from a Poisson distribution with rate the product of the sociability parameters, normalized by the sum of the sociability parameters of all the nodes. The normalization makes this model similar to models based on rescaling of the graphon and, as such, does not define a projective model, as explained in Section~\ref{sec:intro}.

Another related model is the degree-corrected random graph model~\cite{Karrer2011}, where edges of the multigraph are drawn from a Poisson distribution whose rate is the product of node-specific sociability parameters and a parameter tuning the interaction between the latent communities to which these nodes belong. When the sociability parameters are assumed to be i.i.d. from some distribution, this model yields an exchangeable matrix and thus a dense graph.

Additionally, there are similarities to be drawn with the extensive literature on latent space modeling~\citep[cf.][]{Hoff2002,Penrose2003,Hoff2009}.  In such models, nodes are embedded in a low-dimensional, continuous latent space and the probability of an edge is determined by a distance or similarity metric of the node-specific latent factors.  In our case, the node position, $\theta_i$, is of no importance in forming edge probabilities.  It would, however, be possible to extend our approach to location-dependent connections by considering inhomogenous CRMs.

Finally, the urn construction described in Section~\ref{sec:undirected} highlights a connection with the configuration model~\citep{Bollobas1980,Newman2009}, a popular model for generating simple graphs with a given degree sequence. The configuration model proceeds as follows. First, the degree $k_i$ of each node $i=1,\ldots,n$ is specified such that the sum of $k_i$ is an odd number. Each node $i$ is given a total of $k_i$ stubs, or \emph{demi-edges}. Then, we repeatedly choose pairs of stubs uniformly at random, without replacement, and connect the selected pairs to form an edge. The simple graph is obtained either by discarding the multiple edges and self-loops (an \emph{erased} configuration model), or by repeating the above sampling until obtaining a simple graph.\smallskip

The connections to this past work nicely place our proposed Bayesian nonparametric network model within the context of existing literature.  Importantly, however, to the best of our knowledge this work represents the first fully generative and projective approach to sparse graph modeling, and with a notion of exchangeability essential for devising our scalable statistical estimation procedure.  For this, we devised a sampler that readily scales to large, real-world networks.  The foundational modeling tools and theoretical results presented herein represent an important building block for future developments, including incorporating notions of community structure, node attributes, etc.

\paragraph{Acknowledgements.}
The authors thank Bernard Bercu for help in deriving the proof of Theorem~\ref{theorem:limitratecondiid}, and Arnaud Doucet, Yee Whye Teh, Stefano Favaro and Dan Roy for helpful discussions and feedback on earlier versions of this paper.
\newpage

\appendix
\noindent

\section{Proofs of results on the sparsity}

\label{sec:proofsparsity}

\subsection{Probability asymptotics notation}
\label{sec:asymptnotation}
We first describe the asymptotic notation used in the remaining of this section, which follows the notation of~\citet{Janson2011}. All unspecified limits are as $\alpha\rightarrow\infty$.

Let $(X_\alpha)_{\alpha\geq 0}$ and $(Y_\alpha)_{\alpha\geq 0}$ be two $[0,\infty)$-valued stochastic processes defined on the same probability space and such that $\lim_{\alpha\rightarrow\infty} X_\alpha=\lim_{\alpha\rightarrow\infty} Y_\alpha=\infty$ a.s. We have
\begin{align*}
X_{\alpha}=O(Y_{\alpha})\text{ a.s.}&\Longleftrightarrow \underset{\alpha\rightarrow\infty}{\lim\sup}\ \frac{X_{\alpha}}{Y_\alpha}<\infty\text{ a.s.} & \\
X_{\alpha}=o(Y_{\alpha})\text{ a.s.}&\Longleftrightarrow  \lim_{\alpha\rightarrow\infty}\frac{X_{\alpha}}{Y_\alpha}=0\text{ a.s.}\\
X_{\alpha}=\Omega(Y_{\alpha})\text{ a.s.}&\Longleftrightarrow Y_{\alpha}=O(X_{\alpha})\text{ a.s.}\\
X_{\alpha}=\omega(Y_{\alpha})\text{ a.s.}&\Longleftrightarrow Y_{\alpha}=o(X_{\alpha})\text{ a.s.}\\
X_{\alpha}=\Theta(Y_{\alpha})\text{ a.s.}&\Longleftrightarrow X_{\alpha}=O(Y_{\alpha})\text{ and } X_{\alpha}=\Omega(Y_{\alpha})\text{ a.s.}
\end{align*}
The relations have the following interpretation
\begin{center}
\begin{tabular}{lll}
$X_{\alpha}=O(Y_{\alpha})$ &``$X_{\alpha}$ does not grow at a faster rate than $Y_\alpha$" & $[~\leq~]$ \\
$X_{\alpha}=o(Y_{\alpha})$ & ``$X_{\alpha}$ grows at a (strictly) slower rate than $Y_\alpha$" & $[~<~]$  \\
$X_{\alpha}=\Omega(Y_{\alpha})$ & ``$X_{\alpha}$ does not grow at a slower rate than $Y_\alpha$"  & $[~\geq~]$  \\
$X_{\alpha}=\omega(Y_{\alpha})$ & ``$X_{\alpha}$ grows at a (strictly) faster rate than $Y_\alpha$" & $[~>~]$  \\
$X_{\alpha}=\Theta(Y_{\alpha})$ & ``$X_{\alpha}$ and $Y_\alpha$ grow at the same rate" & $[~=~]$
\end{tabular}
\end{center}

\subsection{Proof of Theorems \ref{th:sparsity}, \ref{th:edgesquad} and
\ref{th:nodesgrowth} in the finite-activity case}

We first consider the case of a finite-activity CRM. Let
$T=\int_{0}^{\infty}\rho(w)dw<\infty$ and $H(t)=\frac{1}{T}\int_{0}^{t}%
\rho(w)dw$. The point process $Z$ can be equivalently defined as follows. Let
$\Pi=\{\theta_{1},\theta_{2},\ldots\}$ be a homogeneous Poisson process of
rate $T$. For each $1\leq i\leq j$, sample%
\begin{equation}
z_{ij}|U_{i},U_{j}\sim\text{Ber}(W(U_{i},U_{j}))
\label{eq:AHfinite}
\end{equation}
where $U_{1},U_{2},\ldots$ are uniform random variables and%
\[
W(u,v)=\left\{
\begin{array}
[c]{ll}%
1-\exp(-2H^{-1}(u)H^{-1}(v)) & u\neq v\\
1-\exp(-H^{-1}(u)^{2}) & u=v
\end{array}
\right.
\]

Let $J_{\alpha}=\Pi\cap\lbrack0,\alpha].$ As (i) $J_{\alpha}\rightarrow\infty$
almost surely as $\alpha\rightarrow\infty$ and (ii) $\int_{0}^{1}\int_{0}^{1}W(u,v)dudv<\infty$
and $\int_{0}^{1}\sqrt{W(u,u)}du<\infty$, the law of large numbers for $V$
statistics yields (cf Theorem \ref{theorem:SLLNUstats})%
\begin{equation}
\frac{2}{J_{\alpha}(J_{\alpha}+1)}\sum_{1\leq i\leq j\leq J_{\alpha}}%
W(U_{i},U_{j})\rightarrow\int_{0}^{1}\int_{0}^{1}W(u,v)dudv
\label{eq:Jalpha}
\end{equation}
almost surely as $\alpha\rightarrow\infty$. Additionally, applying
 Theorem \ref{theorem:limitratecondiid} to Equation~\eqref{eq:AHfinite}, gives
\[
\frac{N_\alpha^{(e)}}{\sum_{1\leq i\leq j\leq J_\alpha}W(U_i,U_j)}\rightarrow 1
\]
a.s. which, combined with Equation~\eqref{eq:Jalpha} yields
$\frac{N_{\alpha}^{(e)}}{J_{\alpha}^{2}}=\Theta(1)$ almost surely. As
$\frac{N_{\alpha}}{J_{\alpha}}=\Theta(1)$ almost surely, we determine that%
\begin{align*}
N_{\alpha}^{(e)}  &  =\Theta(N_{\alpha}^{2})\text{ a.s.}\\
N_{\alpha}  &  =\Theta(\alpha)\text{ a.s.}\\
N_{\alpha}^{(e)}  &  =\Theta(\alpha^{2})\text{ a.s.}%
\end{align*}

\subsection{Proof of Theorem \ref{th:edgesquad} in the infinite-activity case}

Consider now the infinite-activity case. Assume $\psi^{\prime}(0)=\mathbb{E}%
[W_{1}^{\ast}]<\infty.$ Let
\begin{equation}
\widetilde{Z}_{ij}=\left\{
\begin{array}
[c]{ll}%
1 & \text{if }Z([i-1,i],[j-1,j])>0\\
0 & \text{otherwise}%
\end{array}
\right.
\end{equation}
then, for any $k\in\mathbb{N}$,
\begin{equation}
\sum_{1\leq i<j\leq k}\widetilde{Z}_{ij}\leq N_{k}^{(e)}\leq D_{k}^{\ast}%
\end{equation}

As $Z$ is a jointly exchangeable point process, $(\widetilde{Z}_{ij}%
)_{i,j\in\mathbb{N}}$ is a jointly exchangeable binary matrix, and so by
Theorem \ref{theorem:excharraydense},
\begin{equation}
\sum_{1\leq i<j\leq k}\widetilde{Z}_{ij}=\Theta(k^{2})
\end{equation}
Moreover, we have
\begin{equation}
D_{k}^{\ast}|W_{k}\sim\text{Poisson}(W_{k}^{\ast2})
\end{equation}
where $W_{k}^{\ast}=\sum_{1\leq i\leq j\leq k}W([i-1,i])W([j-1,j])$, so
Theorems \ref{theorem:limitratecondiid} and \ref{theorem:SLLNUstats} imply (as
$\mathbb{E}[W_{1}^{\ast}]<\infty$)%
\begin{equation}
D_{k}^{\ast}=\Theta(k^{2})
\end{equation}
We therefore conclude that%
\begin{equation}
N_{k}^{(e)}=\Theta(k^{2})
\end{equation}
Finally, for any $k\leq \alpha \leq k+1$,
\[
\frac{k^2}{(k+1)^2}\frac{N_{k}^{(e)}}{k^2}\leq \frac{N_{\alpha}^{(e)}}{\alpha^2}\leq\frac{(k+1)^2}{k^2}\frac{N_{k+1}^{(e)}}{(k+1)^2}
\]
 and as $\frac{k+1}{k}\rightarrow1$ we conclude
\begin{equation}
N_{\alpha}^{(e)}=\Theta(\alpha^{2})\text{ a.s. as }\alpha\rightarrow\infty
\end{equation}

\subsection{Proof of Theorem \ref{th:nodesgrowth} in the infinite-activity
case}

Consider sets $\mathcal{S}^{1}=\cup_{k\in\mathbb{N}}[\frac{2k+1}{2}%
,\frac{2k+2}{2})$ and $\mathcal{S}^{2}=\cup_{k\in\mathbb{N}}[\frac{2k}%
{2},\frac{2k+1}{2})$. $\mathcal{S}^{1}$ and $\mathcal{S}^{2}$ define a
partition of $\mathbb{R}_{\mathbb{+}}$ as $\mathcal{S}^{1}\cup\mathcal{S}%
^{2}=\mathbb{R}_{\mathbb{+}}$ and $\mathcal{S}^{1}\cap\mathcal{S}%
^{2}=\emptyset$. Let $\mathcal{S}_{\alpha}^{1}=\mathcal{S}^{1}\cap
\lbrack0,\alpha]$ and $\mathcal{S}_{\alpha}^{2}=\mathcal{S}^{2}\cap
\lbrack0,\alpha]$. Note that for $\alpha$ integer, $\lambda(\mathcal{S}%
_{\alpha}^{1})=\lambda(\mathcal{S}_{\alpha}^{2})=\frac{\alpha}{2}$.\bigskip

Let $\widetilde{N}_{\alpha}$ be the number of nodes $\theta_{i}\in
\mathcal{S}_{\alpha}^{1}$ with at least one directed edge to a node
$\theta_{j}\in\mathcal{S}_{\alpha}^{2}$. Hence%
\[
\widetilde{N}_{\alpha}=\#\left\{  \left.  \theta_{i}\in\mathcal{S}_{\alpha
}^{1}\ \right\vert \ D\left(  \{\theta_{i}\}\times\mathcal{S}_{\alpha}%
^{2}\right)  >0\right\}
\]
Clearly, for all $\alpha\geq0$%
\begin{equation}
\widetilde{N}_{\alpha}\leq N_{\alpha} \label{eq:boundnumbernodes}%
\end{equation}

We have, for $\theta_{i}\in\mathcal{S}_{\alpha}^{1}$
\[
\Pr\left(  \left.  D\left(  \{\theta_{i}\}\times\mathcal{S}_{\alpha}%
^{2}\right)  >0\right\vert W\right)  =1-\exp\left[  -W(\{\theta_{i}\})\times
W\left(  \mathcal{S}_{\alpha}^{2}\right)  \right]
\]
Note the key fact that $W(\{\theta_{i}\})$ is independent of $W\left(
\mathcal{S}_{\alpha}^{2}\right)  $ as $\theta_{i}\notin\mathcal{S}_{\alpha
}^{2}$. Applying Campbell's theorem, we have%
\[
\mathbb{E}\left[  \widetilde{N}_{\alpha}|W\left(  \mathcal{S}_{\alpha}%
^{2}\right)  \right]  =\lambda(\mathcal{S}_{\alpha}^{1})\times\psi(W\left(
\mathcal{S}_{\alpha}^{2}\right)  )
\]
where $\psi(t)=\int_{0}^{\infty}(1-\exp(-wt))\rho(w)dw$ is the Laplace
exponent. And so, by complete randomness of the CRM over $S_{n}^{1}$,%
\begin{equation}
\widetilde{N}_{\alpha}|W\left(  \mathcal{S}_{\alpha}^{2}\right)
\sim\text{Poisson}\left(  \lambda(\mathcal{S}_{\alpha}^{1})\times\psi(W\left(
\mathcal{S}_{\alpha}^{2}\right)  )\right)  \label{eq:poissonNalpha}%
\end{equation}

We have $\lambda(\mathcal{S}_{\alpha}^{1})=\Theta(\alpha)$ and $\lambda
(\mathcal{S}_{\alpha}^{2})=\Theta(\alpha)$. Moreover, as we are in the
infinite-activity case $\int_{0}^{\infty}\rho(w)dw=\infty$, Lemma
\ref{lemma:laplaceexplimit} implies that%
\begin{equation}
\lim_{t\rightarrow\infty}\psi(t)=\infty.
\end{equation}
As $W\left(  \mathcal{S}_{\alpha}^{2}\right)  \rightarrow\infty$ almost surely, we
therefore have $\psi(W\left(  \mathcal{S}_{\alpha}^{2}\right)  )\rightarrow
\infty$ almost surely.  Thus,
\begin{equation}
\psi(W\left(  \mathcal{S}_{\alpha}^{2}\right)  )=\omega(1)\text{ a.s.}%
\end{equation}
and
\begin{equation}
\lambda(\mathcal{S}_{\alpha}^{1})\times\psi(W\left(  \mathcal{S}_{\alpha}%
^{2}\right)  )=\omega(\alpha)\text{ a.s.} \label{eq:asymptoticrate}%
\end{equation}
Combining (\ref{eq:asymptoticrate}) with Theorem
\ref{theorem:limitpoissonprocessrate} and (\ref{eq:boundnumbernodes})
yields%
\begin{align}
\widetilde{N}_{\alpha}  &  =\omega(\alpha)\text{ a.s.}\\
N_{\alpha}  &  =\omega(\alpha)\text{ a.s.}%
\end{align}
\bigskip

Consider now the case where $\overline{\rho}(x)\overset{x\downarrow 0}{\sim}\ell(1/x)x^{-\sigma}$ where $\ell(t)$ is a slowly varying function, i.e. a function
verifying $\frac{\ell(ct)}{\ell(t)}\rightarrow1$ for any $c>0$, and such that
$\lim_{t\rightarrow\infty}\ell(t)>0.$ Then Lemma \ref{th:asymptoticlevy}
implies that $\psi(t)=\Omega(t^{\sigma})$ as $t\rightarrow\infty$ and thus%
\[
\lambda(\mathcal{S}_{\alpha}^{1})\times\psi(W\left(  \mathcal{S}_{\alpha}%
^{2}\right)  )=\Omega(\alpha^{\sigma+1})\text{ a.s.}%
\]
which implies that
\[
N_{\alpha}=\Omega(\alpha^{\sigma+1})\text{ a.s.}%
\]

\section{Technical lemmas}

\label{sec:techlemmas}
\begin{theorem}
[Graphs constructed from exchangeable arrays are dense]%
\label{theorem:excharraydense}Let $(X_{ij})_{i,j\in\mathbb{N}},$ be an
infinitely exchangeable binary symmetric array. Let $N_{n}=\sum_{1\leq i<j\leq
n}X_{ij}$. If $\lim_{n\rightarrow\infty}N_{n}>0$ almost surely, then
\begin{equation}
N_{n}=\Theta(n^{2})\text{ almost surely and in }L_{1}%
\end{equation}

\begin{proof}
From the Aldous-Hoover theorem, there is a random
function\ $W:[0,1]\rightarrow\lbrack0,1]$ such that
\begin{equation}
X_{ij}|W,U_{i},U_{j}\sim\Ber(W(U_{i},U_{j})) \label{eq:aldoushoover}%
\end{equation}
where $(U_{i})_{i\in\mathbb{N}}$ are uniform random variables. Given
$W$, the law of large numbers for $U$ statistics (see Theorem
\ref{theorem:SLLNUstats}) yields%
\[
\frac{2}{n(n-1)}\sum_{1\leq i<j\leq n}W(U_{i},U_{j})\rightarrow\overline
{W}\text{ a.s.}%
\]
If $\lim_{n\rightarrow\infty}N_{n}>0$ almost surely, then $\overline{W}=\int_{0}%
^{1}\int_{0}^{1}W(u,v)dudv>0$ almost surely, thus
\begin{equation}
\sum_{1\leq i<j\leq n}W(U_{i},U_{j})=\Theta(n^{2})\text{ a.s.}
\label{eq:ratedense1}%
\end{equation}
Furthermore, note that if $X|Y\sim\Ber(Y)$, $\mathbb{V}(X|Y)=Y(1-Y)\leq
\mathbb{E}[X|Y\}$. Moreover, $\sum_{1\leq i<j\leq n}W(U_{i},U_{j}%
)\rightarrow\infty$ almost surely. Applying Theorem \ref{theorem:limitratecondiid} to
Eq. (\ref{eq:aldoushoover}) implies%
\begin{equation}
\frac{\sum_{1\leq i<j\leq n}X_{ij}}{\sum_{1\leq i<j\leq n}W(U_{i},U_{j}%
)}\rightarrow1\text{ a.s. as }n\rightarrow\infty\label{eq:ratedense2}%
\end{equation}
and thus, combining (\ref{eq:ratedense1}) with (\ref{eq:ratedense2}) yields%
\begin{equation}
N_{n}=\Theta(n^{2})\text{ a.s. as }n\rightarrow\infty
\end{equation}

As $X_{ij}\leq1$, the dominated convergence theorem implies $\mathbb{E}%
[N_{n}]=\Theta(n^{2})$.
\end{proof}

\end{theorem}

\begin{lemma}
[Strong law of large numbers for U and V statistics]\label{theorem:SLLNUstats}%
\cite{Arcones1992,Gine1992}. Let $(X_{i})$ be i.i.d. real-valued random
variables from $\mu$ and let $h(x_{1},x_{2})$ be a symmetric measurable
function. Consider the $U$ and $V$ statistics defined by%
\begin{align*}
U_{n}(h,\mu)  &  =\frac{2}{n(n-1)}\sum_{1\leq i<j\leq n}h(X_{i},X_{j})\\
V_{n}(h,\mu)  &  =\frac{1}{n^{2}}\sum_{i,j=1}^{n}h(X_{i},X_{j})
\end{align*}
If $\mathbb{E}[\left\vert h(X_{1},X_{2})\right\vert ]<\infty$,
then~\cite{Hoeffding1961}
\[
U_{n}(h,\mu)\rightarrow\mathbb{E}[h(X_{1},X_{2})]\text{ a.s. as }%
n\rightarrow\infty
\]

If $\mathbb{E}[\left\vert h(X_{1},X_{2})\right\vert ]<\infty$ and
$\mathbb{E}[\sqrt{\left\vert h(X_{1},X_{1})\right\vert }]<\infty$, then
\[
V_{n}(h,\mu)\rightarrow\mathbb{E}[h(X_{1},X_{2})]\text{ a.s. as }%
n\rightarrow\infty.
\]

\end{lemma}

\begin{theorem}
\label{theorem:limitratecondiid}Let $Y=(Y_{1},Y_{2},\ldots)$ be a sequence of
positive random variables.\ Consider positive variables $(X_{1},X_{2},\ldots)$
such that%
\begin{align}
\mathbb{E}[X_{i}|X_{1:i-1},Y_{1:i}]  &  =\mathbb{E}[X_{i}|Y_{i}]=Y_{i}%
\label{eq:property1}\\
\mathbb{V}[X_{i}|X_{1:i-1},Y_{1:i}]  &  =\mathbb{V}[X_{i}|Y_{i}]\leq Y_{i}
\label{eq:property2}%
\end{align}
If
\[
\sum_{i=1}^{n}Y_{i}\rightarrow\infty\text{ almost surely as }n\rightarrow
\infty
\]
then
\[
\frac{\sum_{i=1}^{n}X_{i}}{\sum_{i=1}^{n}Y_{i}}\rightarrow1\text{ almost
surely as }n\rightarrow\infty
\]

\end{theorem}

\begin{proof}
We will use a martingale approach here. Note that it is also possible to use an alternative proof via Borel-Cantelli. Let $\mathbb{F}%
=(\mathcal{F}_{1},\mathcal{F}_{2},\ldots)$\ be a filtration with
$\mathcal{F}_{n}=\sigma(X_{1},\ldots,X_{n},Y_{1},\ldots,Y_{n+1})$. We have
\[
\mathbb{E}[X_{n+1}|\mathcal{F}_{n}]=Y_{n+1}%
\]
Let $(M_{n})$ be defined as\footnote{The $1$ at the denominator is used to
ensure that $M_{n}^{2}$ is integrable for all $n$.}%
\[
M_{n}=\sum_{i=1}^{n}\frac{X_{i}-Y_{i}}{1+S_{i}}%
\]
where
\[
S_{n}=\sum_{i=1}^{n}Y_{i}%
\]
$(M_{n})$ is a martingale with respect to the filtration $\mathbb{F}$ as
$M_{n}$ is $\mathcal{F}_{n}$-measurable,%
\begin{align*}
\mathbb{E}[M_{n+1}|\mathcal{F}_{n}]  &  =M_{n}+\frac{1}{1+S_{n}}\left(
\mathbb{E}[X_{n+1}|\mathcal{F}_{n}]-Y_{n+1}\right) \\
&  =M_{n}%
\end{align*}
and, using H\"{o}lder inequalities together with properties
(\ref{eq:property1}) and (\ref{eq:property2})%
\begin{align*}
\mathbb{E}[M_{n}^{2}]  &  \leq\sum_{i=1}^{n}\mathbb{E}\left[  \frac{\left(
X_{i}-Y_{i}\right)  ^{2}}{(1+S_{i})^{2}}\right] \\
&  \leq\sum_{i=1}^{n}\mathbb{E}\left[  \frac{\left(  X_{i}-Y_{i}\right)  ^{2}%
}{(1+Y_{i})^{2}}\right] \\
&  =\sum_{i=1}^{n}\mathbb{E}\left[  \mathbb{E}\left[  \left.  \frac{\left(
X_{i}-Y_{i}\right)  ^{2}}{(1+Y_{i})^{2}}\right\vert Y_{i}\right]  \right] \\
&  =\sum_{i=1}^{n}\mathbb{E}\left[  \frac{\mathbb{V}(X_{i}|Y_{i})}%
{(1+Y_{i})^{2}}\right] \\
&  \leq\sum_{i=1}^{n}\mathbb{E}\left[  \frac{Y_{i}}{(1+Y_{i})^{2}}\right] \\
&  \leq n<\infty
\end{align*}
Let
\begin{align*}
\mathbb{E}[(M_{n+1}-M_{n})^{2}|\mathcal{F}_{n}]  &  =\frac{1}{(1+S_{n+1})^{2}%
}\mathbb{E}[(X_{n+1}-Y_{n+1})^{2}|\mathcal{F}_{n}]\\
&  =\frac{1}{(1+S_{n+1})^{2}}\mathbb{V}[X_{n+1}|\mathcal{F}_{n}]\leq
\frac{Y_{n+1}}{(1+S_{n+1})^{2}}%
\end{align*}
It follows that
\begin{align*}
\left\langle M\right\rangle _{n}  &  =\sum_{i=1}^{n}\mathbb{E}[(M_{i}%
-M_{i-1})^{2}|\mathcal{F}_{i-1}]\\
&  \leq\sum_{i=1}^{n}\frac{Y_{i}}{(1+S_{i})^{2}}%
\end{align*}
By Rieman integration,
\[
\sum_{i=1}^{n}\frac{Y_{i}}{(1+\sum_{j=1}^{n}Y_{j})^{2}}\leq\int_{0}^{n}%
\frac{1}{(1+x)^{2}}dx<\infty
\]
Therefore%
\[
\left\langle M\right\rangle _{\infty}=\lim_{n\rightarrow\infty}\left\langle
M\right\rangle _{n}<\infty\text{ almost surely}%
\]
The law of large numbers for square-integrable martingales (see e..g. Theorem
5.4.9 page 217 of \cite{Durrett2010}) thus implies that%
\[
\lim_{n\rightarrow\infty}M_{n}\text{ exists and is a.s. finite}%
\]
therefore%
\[
\sum_{i=1}^{n}\frac{X_{i}-Y_{i}}{1+S_{i}}<\infty\text{ a.s.}%
\]
As $1+S_{i}\rightarrow\infty$ almost surely, the lemma of Kronecker (Theorem 2.5.5. in
\cite{Durrett2010}) implies
\[
\frac{1}{1+S_{n}}\sum_{i=1}^{n}\left(  X_{i}-Y_{i}\right)  \rightarrow0\text{
a.s.}%
\]
and thus
\[
\frac{\sum_{i=1}^{n}X_{i}}{\sum_{j=1}^{n}Y_{j}}\rightarrow1\text{ a.s.}%
\]

\end{proof}

\begin{theorem}
\label{theorem:limitpoissonprocessrate}Let $\mu$ be a random almost surely positive
measure on $\mathbb{R}_{+}$ and define the random integer valued measure $N$
as
\[
N|\mu\sim\Poi(\mu)
\]
Write $\overline{N}_{t}=N([0,t])$ and $\overline{\mu}_{t}=\mu([0,t])$. Then
\[
\overline{N}_{t}|\mu\sim\Poi(\overline{\mu}_{t})
\]
If $\overline{\mu}_{t}\rightarrow\infty$ almost surely and $\frac{\overline{\mu}_{t+1}%
}{\overline{\mu}_{t}}\rightarrow1$ almost surely, then%
\[
\frac{\overline{N}_{t}}{\overline{\mu}_{t}}\rightarrow1\text{ a.s.}%
\]

\end{theorem}

\begin{proof}
Let $k=1,2,\ldots.$, let $X_{k}=\overline{N}_{k}-\overline{N}_{k-1}$ and
$Y_{k}=\overline{\mu}_{k}-\overline{\mu}_{k-1}$. From Theorem
\ref{theorem:limitratecondiid}, we have
\[
\frac{\sum_{k=1}^{n}X_{k}}{\sum_{k=1}^{n}Y_{k}}\rightarrow1\text{ as
}n\rightarrow\infty
\]
and so%
\[
\frac{\overline{N}_{n}}{\overline{\mu}_{n}}\rightarrow1\text{ as }%
n\rightarrow\infty
\]
As, for $n\leq t\leq n+1$
\[
\frac{\overline{\mu}_{n}}{\overline{\mu}_{n+1}}\frac{\overline{N}_{n}%
}{\overline{\mu}_{n}}\leq\frac{\overline{N}_{t}}{\overline{\mu}_{t}}\leq
\frac{\overline{N}_{n+1}}{\overline{\mu}_{n+1}}\frac{\overline{\mu}_{n+1}%
}{\overline{\mu}_{n}}%
\]
as $\frac{\overline{\mu}_{n+1}}{\overline{\mu}_{n}}\rightarrow1$ almost surely, we
conclude that
\[
\frac{\overline{N}_{t}}{\overline{\mu}_{t}}\rightarrow1\text{ a.s.}%
\]

\end{proof}

\begin{lemma}
\label{lemma:laplaceexplimit}Let $Z_{t}$ be an almost surely increasing L\'{e}vy
process (or subordinator) without deterministic component and with L\'{e}vy
intensity $\rho(w)$. Let
\[
\psi(t)=\int_{0}^{\infty}(1-\exp(-wt))\rho(w)dw
\]
be its Laplace exponent. If
\[
\int_{0}^{\infty}\rho(w)dw=\infty
\]
then
\[
\lim_{t\rightarrow\infty}\psi(t)=\infty.
\]

\end{lemma}

\begin{proof}
Consider sequence of functions $f_{k}(w)=(1-\exp(-wk))\rho(w)$, $k\in
\mathbb{N}.$ We have for%
\[
0\leq f_{k}(w)<f_{k+1}(w)\text{ for all }w>0
\]
and
\[
\lim_{k\rightarrow\infty}f_{k}(w)=\rho(w)
\]
Thus, by Lebesgue's monotone convergence theorem%
\begin{align*}
\lim_{k\rightarrow\infty}\psi(k)  &  =\int_{0}^{\infty}\rho(w)dw\\
&  =\infty
\end{align*}
as $\psi(t)\geq\psi(\left\lfloor t\right\rfloor )$,
\[
\lim_{t\rightarrow\infty}\psi(t)=\infty
\]

\end{proof}

\begin{lemma}
[Relating tail L\'{e}vy intensity and Laplace exponent]\label{th:asymptoticlevy}\cite[Proposition 17]{Gnedin2007}
Let $\rho(w)$ be the L\'{e}vy intensity $\overline{\rho}(x)=\int_{x}^{\infty
}\rho(w)dw$ be the tail L\'{e}vy intensity, and  $\psi(t)=\int_{0}^{\infty
}(1-\exp(-wt))\rho(w)dw$ its Laplace exponent$.$ The following conditions are
equivalent:%
\begin{align}
\overline{\rho}(x)  &\overset{x\downarrow 0}{ \sim}\ell(1/x)x^{-\sigma}\\
\psi(t)  & \overset{t\uparrow\infty}{\sim}\Gamma(1-\sigma)t^{\sigma}\ell(t)
\end{align}
where $0<\sigma<1$ and $\ell$ is a function slowly varying at $\infty$ i.e.
satisfying $\ell(cy)/\ell(y)\rightarrow1$ as $y\rightarrow\infty$, for every
$c>0$.
\end{lemma}

\begin{proof}
Applying integration by part, we have%
\begin{align*}
\psi(t)  & =\int_{0}^{\infty}(1-\exp(-wt))\rho(w)dw\\
& =t\int_{0}^{\infty}\exp(-wt)\overline{\rho}(w)dw
\end{align*}
As $\overline{\rho}(x)$ is positive monotonic, application of Proposition
\ref{th:tauberian} yields the following equivalence%
\begin{align*}
\overline{\rho}(x)  & \overset{x\downarrow 0}{\sim}\ell(1/x)x^{-\sigma}\\
\frac{\psi(t)}{t}  & \overset{t\uparrow\infty}{\sim}\Gamma(1-\sigma)t^{\sigma-1}\ell(t)
\end{align*}

\end{proof}

\begin{proposition}
[Tauberian theorem]\label{th:tauberian}\cite[Chapter XIII,
Section 5,\ Theorem 4 p. 446]{Feller1971} Let $U(dw)$ be a measure on
$(0,\infty)$ with ultimately monontone density $u$, i.e. monotone in some
interval $(x_{0},\infty)$. Assume that%
\[
\mathcal{L}(t)=\int_{0}^{\infty}e^{-tw}u(w)dw
\]
exists for $t>0$. If $\ell$ is slowly varying at infinity and $0\leq a<\infty
$, then the two relations are equivalent%
\begin{align}
\mathcal{L}(\tau)  & \overset{\tau\downarrow0}{\sim}\tau^{-a}\ell(1/\tau)\\
u(x)  & \overset{x\uparrow\infty}{\sim}\frac{1}{\Gamma(a)}x^{a-1}\ell(x)
\end{align}
\newline Additionally\ \cite[Chapter XIII, Section 5, Theorem 3]{Feller1971},
the result remains valid if we interchange the role of infinity and 0, hence
$x\rightarrow\infty$ and $\tau\rightarrow0$%
\begin{align}
\mathcal{L}(\tau)  & \overset{\tau\uparrow\infty}{\sim}\tau^{-a}\ell(\tau)\\
u(x)  & \overset{x\downarrow 0}{\sim}\frac{1}{\Gamma(a)}x^{a-1}\ell(1/x)
\end{align}

\end{proposition}

\begin{proposition}[Chebyshev-type inequality]

Let $X$ be a random variable with $\mathbb{E}[X^{2}]<\infty$ and $\theta\in(0,1)$. Then
\begin{equation}
P(X\geq\theta\mathbb{E}[X])\geq1-\frac{Var(X)}{(1-\theta^{2})E[X]^{2}}
\label{eq:chebyshev2}%
\end{equation}
\end{proposition}
\section{Proofs of results on the properties of the GGP graph}
\label{sec:proofGGP}
\subsection{Proof of Theorem \ref{th:inequalitynodes}}

The hierarchical model for the number of nodes $N_{\alpha}$ in the gamma
process case is%
\begin{align}
N_{\alpha}  &  =\sum_{i=1}^{2D_{\alpha}^{\ast}}Y_{i}\label{eq:Nalpha}\\
Y_{i}  &  \overset{ind}{\sim}\Ber\left(  \frac{\alpha}{\alpha+i-1}\right)
\label{eq:Ybernoulli}\\
D_{\alpha}^{\ast}|W_{\alpha}^{\ast}  &  \sim \Poi\left(  W_{\alpha}%
^{\ast\ 2}\right) \\
W_{\alpha}^{\ast}  &  \sim \Gam(\alpha,\tau)
\end{align}
where $D_{\alpha}^{\ast}$ is the total number of directed edges in the
directed graph, and $W_{\alpha}^{\ast}$ is the total mass. We have
\begin{align*}
\mathbb{E}[D_{\alpha}^{\ast}]  &  =\mathbb{E}[W_{\alpha}^{\ast\ 2}%
]=\frac{\alpha(\alpha+1)}{\tau^{2}}\\
Var(D_{\alpha}^{\ast})  &  =\frac{\alpha(\alpha+1)}{\tau^{2}}\left(
1+\frac{4\alpha+6}{\tau^{2}}\right)
\end{align*}

From\ Equations (\ref{eq:Nalpha}) and (\ref{eq:Ybernoulli}), we have
\[
\mathbb{E}[N_{\alpha}|D_{\alpha}^{\ast}]=\alpha\sum_{i=1}^{2D_{\alpha}^{\ast}%
}\frac{1}{\alpha+i-1}%
\]
As the function $f:x\rightarrow\frac{1}{\alpha+x}$ is decreasing on
$[0,n],~n>0$, we have
\begin{align*}
\sum_{i=2}^{n+1}\frac{1}{\alpha+i-1}  &  \leq\int_{0}^{n}f(x)dx\leq\sum
_{i=1}^{n}\frac{1}{\alpha+i-1}\\
\left(  \sum_{i=1}^{n}\frac{1}{\alpha+i-1}\right)  +\frac{1}{\alpha+n}%
-\frac{1}{\alpha}  &  \leq\log\left(  1+\frac{n}{\alpha}\right)  \leq
\sum_{i=1}^{n}\frac{1}{\alpha+i-1}%
\end{align*}
hence%
\[
\alpha\log\left(  1+\frac{2D_{\alpha}^{\ast}}{\alpha}\right)  \leq\alpha
\sum_{i=1}^{2D_{\alpha}^{\ast}}\frac{1}{\alpha+i-1}\leq\alpha\log\left(
1+\frac{2D_{\alpha}^{\ast}}{\alpha}\right)  +1-\frac{\alpha}{\alpha
+2D_{\alpha}^{\ast}}%
\]
and so
\begin{equation}
\alpha\log\left(  1+\frac{2D_{\alpha}^{\ast}}{\alpha}\right)  \leq
\mathbb{E}[N_{\alpha}|D_{\alpha}^{\ast}]\leq\alpha\log\left(  1+\frac
{2D_{\alpha}^{\ast}}{\alpha}\right)  +1-\frac{\alpha}{\alpha+2D_{\alpha}%
^{\ast}} \label{eq:boundNgivenD}%
\end{equation}
Lets work on the upper bound of Eq. (\ref{eq:boundNgivenD}). The function
$x\rightarrow\alpha\log(1+\frac{2x}{\alpha})+1-\frac{\alpha}{\alpha
+2D_{\alpha}^{\ast}}$ is concave, so by Jensen's inequality%
\begin{align*}
\mathbb{E}\left[  \alpha\log\left(  1+\frac{2D_{\alpha}^{\ast}}{\alpha
}\right)  +1-\frac{\alpha}{\alpha+2D_{\alpha}^{\ast}}\right]   &  \leq
\alpha\log\left(  1+2\frac{\mathbb{E}[D_{\alpha}^{\ast}]}{\alpha}\right)
+1-\frac{\alpha}{\alpha+2\mathbb{E}[D_{\alpha}^{\ast}]}\\
&  =\alpha\log\left(  1+\frac{2(\alpha+1)}{\tau^{2}}\right)  +1-\frac{\tau
^{2}}{\tau^{2}+2(\alpha+1)}%
\end{align*}
Now lets work on the lower bound of Eq. (\ref{eq:boundNgivenD}). For
$\theta\geq0$, Markov inequality gives%
\[
\Pr\left[  \log\left(  1+\frac{2D_{\alpha}^{\ast}}{\alpha}\right)  \geq
\theta\right]  \leq\frac{\mathbb{E}\left[  \log\left(  1+\frac{2D_{\alpha
}^{\ast}}{\alpha}\right)  \right]  }{\theta}%
\]
Taking $\theta=\log\left(  1+\varepsilon\frac{2(\alpha+1)}{\tau^{2}}\right)
$, with $\varepsilon\in(0,1)$, we obtain%
\begin{align*}
&\log\left(  1+\varepsilon\frac{2(\alpha+1)}{\tau^{2}}\right)  \Pr\left[
\log\left(  1+\frac{2D_{\alpha}^{\ast}}{\alpha}\right)  \geq\log\left(
1+\varepsilon\frac{2(\alpha+1)}{\tau^{2}}\right)  \right] \\
& \leq\mathbb{E}%
\left[  \log\left(  1+\frac{2D_{\alpha}^{\ast}}{\alpha}\right)  \right]
\end{align*}
hence%
\begin{equation}
\log\left(  1+\varepsilon\frac{2(\alpha+1)}{\tau^{2}}\right)  \Pr\left(
D_{\alpha}^{\ast}\geq\varepsilon\frac{\alpha(\alpha+1)}{\tau^{2}}\right)
\leq\mathbb{E}\left[  \log\left(  1+\frac{2D_{\alpha}^{\ast}}{\alpha}\right)
\right]  \label{eq:bound1}%
\end{equation}
Using the Chebyshev-type inequality (\ref{eq:chebyshev2}) we obtain%
\begin{equation}
\Pr\left(  D_{\alpha}^{\ast}\geq\varepsilon\frac{\alpha(\alpha+1)}{\tau^{2}%
}\right)  \geq1-\frac{Var(D_{\alpha}^{\ast})}{(1-\varepsilon^{2}%
)\mathbb{E}[D_{\alpha}^{\ast}]^{2}} \label{eq:bound2}%
\end{equation}
Let $c_{1}(\alpha)=\frac{Var(D_{\alpha}^{\ast})}{\mathbb{E}[D_{\alpha}^{\ast
}]^{2}}=\frac{\tau^{2}}{\alpha(\alpha+1)}\left(  1+\frac{4\alpha+6}{\tau^{2}%
}\right)  $, which is a decreasing function of $\alpha.$ Combining
Inequalities (\ref{eq:bound1}) and (\ref{eq:bound2}) with
(\ref{eq:boundNgivenD}), we have the following inequalities, for any
$\varepsilon\in(0,1)$
\[
\alpha\log\left(  1+\varepsilon\frac{2(\alpha+1)}{\tau^{2}}\right)  \left(
1-\frac{c_{1}(\alpha)}{1-\varepsilon^{2}}\right)  \leq\mathbb{E}[N_{\alpha
}]\leq\alpha\log\left(  1+\frac{2(\alpha+1)}{\tau^{2}}\right)  +\frac
{2(\alpha+1)}{\tau^{2}+2(\alpha+1)}%
\]
where $c_{1}(\alpha)\rightarrow0$ as $\alpha\rightarrow\infty$, and so,
\[
\mathbb{E}[N_{\alpha}]=\Theta(\alpha\log\alpha),~~~ \alpha\rightarrow\infty
\]

\subsection{Proof of Theorem \ref{th:powerlaw}}

Consider the conditionally Poisson construction%
\begin{align*}
D^{\ast}_{\alpha}|W_{\alpha}^{\ast}  &  \sim \mbox{Poisson}(W_\alpha^{\ast~2})\\
(U^\prime_{1},\ldots,U^\prime_{2D^{\ast}_\alpha})|D^{\ast}_{\alpha},W_{\alpha}  &
\sim\frac{W_{\alpha}}{W_\alpha^{\ast}}.%
\end{align*}

The number of $U^\prime_j$ in any interval $[a,b]$ $a<b\leq\alpha$ is distributed from a Poisson distribution with rate $2W([a,b])W([0,\alpha])$ and therefore goes to infinity as $\alpha$ goes to infinity. We can therefore invoke asymptotic results on i.i.d. sampling from a normalized generalized gamma process, $\frac{W_{\alpha}}{W_\alpha^{\ast}}$.

Let $N_{\alpha,j}$ be
the number of clusters of size $j$ in $(U^\prime_{1},\ldots,U^\prime_{2D^{\ast}_\alpha})$. In the directed graph model, $N_{\alpha,j}$ corresponds to the number of nodes with $j$ incoming/outgoing edges (self-edges count twice for a given node).

As the $U^\prime_j$ are drawn from a normalized generalized gamma process of parameters $(\alpha,\sigma, \tau)$, we have the following asymptotic result~\cite[Corollary 1]{Pitman2006,Lijoi2007}
\[
\frac{N_{\alpha,j}}{N_\alpha}\, \overrightarrow{\alpha\to\infty} \,\,\, p_{\sigma,j}=\frac{\sigma\Gamma(j-\sigma)}{\Gamma
(1-\sigma)\Gamma(j+1)}.%
\]
almost surely, for $j=1,2,\ldots$.

\section{Proofs of results on posterior characterization}
\label{sec:proofposterior}

\subsection{Proof of Theorem \ref{th:posteriorsimple}}

We first state a general Palm formula for Poisson random measures. This result
is used by various authors in similar forms for characterization of conditionals in Bayesian
nonparametric models
\cite{Prunster2002,James2002,James2005,James2009,Caron2012,Caron2014,Zhou2014,James2014}.

\begin{theorem}
\label{th:Palmcalculus}Let $\Pi$ denote a Poisson random measure on a Polish
space $\mathcal{S}$ with non-atomic mean measure $\nu$. Let $\mathcal{M}$ be
the space of boundedly finite measures on $\mathcal{S}$, with sigma-field
$\mathcal{B}(\mathcal{M})$. Let $f_i$, $i=1,\ldots,K$ be functions from $\mathcal{S}$ to $\mathbb R_+$ such that $f_i(s)f_j(s)=0$ for all $i\neq j$. Let $s_{1:K}=(s_1,\ldots,s_K)\in\mathcal S^K$ and $G$ be a measurable function on $\mathcal{S}^{K}%
\times\mathcal{M}$.\ Then we have the following generalized Palm formula
\begin{align}
&\mathbb{E}_{\Pi}\left[  \int_{\mathcal{S}^{K}}G(s_{1:K},\Pi)\prod_{i=1}%
^{K}f_i(s_i)\Pi(ds_{i})\right]  \nonumber\\ &~~~~~~~~~~~~~~~~~~~=\int_{\mathcal{S}^{K}}\mathbb{E}_{\Pi}\left[
G\left (s_{1:K},\Pi+\sum_{i=1}^{K}\delta_{s_{i}}\right )\right]\prod_{i=1}^{K}f_i(s_i)\nu(ds_{i})
\label{eq:Palmrecursive}%
\end{align}

\end{theorem}

\begin{proof}
The proof is obtained by induction from the classical Palm formula
\cite{Bertoin2006,Daley2008}%
\begin{equation}
\mathbb{E}_{\Pi}\left[  \int_{\mathcal{S}}f(s)G(s,\Pi)\Pi(ds)\right]  =\int%
_{\mathcal{S}}\mathbb{E}_{\Pi}\left[  G(s,\Pi+\delta_{s})\right]f(s)\nu(ds).
\label{eq:Palm}%
\end{equation}
Let $G_{1}(s_{1},\Pi)=\int_{\mathcal{S}^{K-1}%
}G(s_{1:K},\Pi)\prod_{i=2}^{K}f_i(s_i)\Pi(ds_{i})$. Then%
\begin{align*}
&\mathbb{E}_{\Pi}\left[  \int_{\mathcal{S}^{K}}G(s_{1:K},\Pi)\prod_{i=1}%
^{K}f_i(s_i)\Pi(ds_{i})\right]     \\
&=\int_{\mathcal{S}}\mathbb{E}_{\Pi}\left[  G_{1}%
(s_{1},\Pi+\delta_{s_{1}})\right] f_1(s_1) \nu(ds_{1})\\
&  =\int_{\mathcal{S}}\mathbb{E}_{\Pi}\left[  \int_{\mathcal{S}^{K-1}}%
G(s_{1:K},\Pi+\delta_{s_{1}})\prod_{i=2}^{K}f_i(s_i)\left[  \Pi(ds_{i})+\delta_{s_{1}%
}(ds_{i})\right]  \right] f_1(s_1) \nu(ds_{1})\\
&  =\int_{\mathcal{S}}\mathbb{E}_{\Pi}\left[  \int_{\mathcal{S}^{K-1}}%
G(s_{1:K},\Pi+\delta_{s_{1}})\prod_{i=2}^{K}f_i(s_i)\Pi(ds_{i})\right] f_1(s_1) \nu(ds_{1})
\end{align*}
as $f_1(s_1)f_i(s_1)=0$ for all $i=2,\ldots,K$. Applying the same strategy recursively
gives (\ref{eq:Palmrecursive}).
\end{proof}

\bigskip

We now prove Theorem \ref{th:posteriorsimple}. The conditional Laplace
functional of $W_{\alpha}$ given $D_{\alpha}$ is $\mathbb{E}\left[
e^{-W_{\alpha}(f)}|D_{\alpha}\right]  $, for any nonnegative measurable
function $f$ such that $W_{\alpha}(f)=\sum_{i}w_{i}f(\vartheta_{i})1_{\vartheta_{i}%
\in\lbrack0,\alpha]}<\infty$. We have $W_{\alpha}(f)=\Pi(\widetilde{f})$ where
$\Pi=\sum_{i=1}^{\infty}\delta_{(w_{i},\vartheta_{i})}$ is a Poisson random measure
on $\mathcal{S}=(0,+\infty)\times\lbrack0,\alpha]$ with mean measure $\nu$ and
$\widetilde{f}(w,\vartheta)=wf(\vartheta)$. The Laplace functional can thus be
expressed in terms of the Poisson random measure $\Pi$
\begin{align}
\mathbb{E}\left[  e^{-W_{\alpha}(f)}|D_{\alpha}\right] & =\mathbb{E}%
_{\Pi}\left[  e^{-\Pi(\widetilde{f})}|D_{\alpha}\right] \nonumber \\
& =\frac{\mathbb{E}_{\Pi}\left[  \int_{\mathcal{S}^{N_{\alpha}}}e^{-\Pi(\widetilde{f}%
)} \exp(-\Pi(h)^{2})\prod_{i=1}%
^{N_{\alpha}}g_i(w_i,\vartheta_i)\Pi(dw_{i},d\vartheta_{i})\right]  }{\mathbb{E}_{\Pi}\left[  \int_{\mathcal{S}%
^{N_{\alpha}}}  \exp(-\Pi(h)^{2})\prod
_{i=1}^{N_{\alpha}}g_i(w_i,\vartheta_i)\Pi(dw_{i},d\vartheta_{i})\right]  }\label{eq:laplacefunctional}%
\end{align}
where $g_i(w,\vartheta)=w^{m_i}1_{d\theta_i}(\vartheta)$, $h(w,\vartheta)=w$, hence $\Pi(h)=\sum_{i=1}^{\infty}w_{i}=W_{\alpha}(1)$.  Applying Theorem
\ref{th:Palmcalculus} to the numerator yields%
\begin{align*}
&  \mathbb{E}_{\Pi}\left[  \int_{\mathcal{S}^{N_{\alpha}}}e^{-\Pi(\widetilde{f})}  e^{-\Pi(h)^{2}}\prod_{i=1}^{N_{\alpha}}%
g_i(w_i,\vartheta_i)\Pi(dw_{i},d\vartheta_{i})\right]  \\
&  =\int_{\mathcal{S}^{N_{\alpha}}}\mathbb{E}_{\Pi}\left[  e^{-\Pi(\widetilde{f}%
)-\sum_{i=1}^{N_{\alpha}}\widetilde{f}(w_{i},\vartheta_{i})}  e^{  -\left(  \Pi(h)+\sum_{i=1}^{N_{\alpha}}w_{i}\right)
^{2}}  \right]\left(\prod_{i=1}^{N_{\alpha}}g_i(w_i,\vartheta_i)\nu(dw_{i},d\vartheta_{i})\right)  \\
&  =\int_{\mathcal{S}^{N_{\alpha}}}\mathbb{E}_{W_{\alpha}}\bigg[  e^{-W_{\alpha}(f)-\sum
_{i=1}^{N_{\alpha}}w_{i}f(\vartheta_{i})}
e^{  -\left(  W_{\alpha}(1)+\sum_{i=1}^{N_{\alpha}}w_{i}\right)  ^{2}}  \bigg] \left(\prod
_{i=1}^{N_{\alpha}}g_i(w_i,\vartheta_i)\nu(dw_{i},d\vartheta_{i})\right)\\
&  =\int_{\mathcal{S}^{N_{\alpha}}}\mathbb{E}_{W_{\alpha}(1)}\bigg[\mathbb{E}_{W_{\alpha}}\left [  e^{-W_{\alpha}(f)} \vert W_{\alpha}(1)\right ]e^{-\sum
_{i=1}^{N_{\alpha}}w_{i}f(\vartheta_{i})}
e^{  -\left(  W_{\alpha}(1)+\sum_{i=1}^{N_{\alpha}}w_{i}\right)  ^{2}}\\
&\hspace{2.75in}\times \left(\prod
_{i=1}^{N_{\alpha}}g_i(w_i,\vartheta_i)\nu(dw_{i},d\vartheta_{i})\right)\bigg].
\end{align*}

The denominator
in (\ref{eq:laplacefunctional}) is obtained by taking $f=0$. Then, after
simplification, we obtain
\begin{align*}
&\mathbb{E}\left[  e^{-W_{\alpha}(f)}|D_{\alpha}\right] =\int_{\mathbb{R}%
_{+}^{N_{\alpha}+1}}\mathbb{E}_{W_{\alpha}}\left [  e^{-W_{\alpha}(f)} \vert W_{\alpha}(1)=w_\ast\right ]\\
&\hspace{1.6in}\times e^{-\sum_{i=1}^{N_{\alpha}}w_if(\theta_{i}%
)}p(w_{1},\ldots,w_{N_{\alpha}},w_\ast|D_{\alpha})dw_{1:N_\alpha}dw_\ast
\end{align*}
where%
\begin{align}
&p(w_{1},\ldots,w_{N_{\alpha}},w_{\ast}|D_{\alpha})\nonumber\\
&\hspace{.5in}=\frac{\left(  \prod_{i=1}^{N_{\alpha}}%
w_{i}^{m_{i}}\rho(w_{i})\right)  e^{  -\left(  w_\ast+\sum_{i=1}^{N_{\alpha}}%
w_{i}\right)  ^{2}}  g_{\alpha}^{\ast}(w_\ast)}{\int_{\mathbb{R}_{+}^{N_{\alpha}+1}}\left[
\prod_{i=1}^{N_{\alpha}}\widetilde{w}_{i}^{m_{i}}\rho(\widetilde{w}_{i})\right]
e^{  -\left(  \widetilde{w}_{\ast}+\sum_{i=1}^{N_{\alpha}}\widetilde{w}%
_{i}\right)  ^{2}}  g_{\alpha}^{\ast}(\widetilde{w}_{\ast})d\widetilde{w}%
_{1:N_{\alpha}}d\widetilde{w}_{\ast}}%
\end{align}

\subsection{Proof of Theorem \ref{th:posteriorbipartite}}

The proof follows the same lines as in \cite{Caron2012} and is included for
completeness. The Laplace functional is expressed as%
\begin{align}
&\mathbb{E}\left[  e^{-W_{\alpha}(f)}|D_{\alpha},W_{\alpha}^{\prime}\right]
  =\mathbb{E}_{\Pi}\left[  e^{-\Pi(\widetilde{f})}|D_{\alpha},W_{\alpha}^{\prime
}\right]  \nonumber\\
&  =\frac{\mathbb{E}_{\Pi}\left[  \int_{\mathcal{S}^{N_{\alpha}}}e^{-\Pi(\widetilde{f}%
)}  e^{-\Pi(h)\left(  \sum
_{j=1}^{N_{\alpha}^{\prime}}w_{j}^{\prime}+w_{\ast}^{\prime}\right)  }\prod_{i=1}%
^{N_{\alpha}}g_i(w_i,\vartheta_i)\Pi(dw_{i},d\vartheta_{i})\right]  }{\mathbb{E}_{\Pi}\left[  \int_{\mathcal{S}%
^{N_{\alpha}}}  e^{-\Pi(h)\left(  \sum
_{j=1}^{N_{\alpha}^{\prime}}w_{j}^{\prime}+w_{\ast}^{\prime}\right)  }\prod_{i=1}%
^{N_{\alpha}}g_i(w_i,\vartheta_i)\Pi(dw_{i},d\vartheta_{i})\right]  }%
\end{align}

where $g_i(w,\vartheta)=w_i^{m_i}1_{d\theta_i}(\vartheta)$, $h(w,\vartheta)=w$, hence $\Pi(h)=\sum_{i=1}^{\infty}w_{i}=W_{\alpha}(1)$. Applying Theorem
\ref{th:Palmcalculus} to the numerator yields%
\begin{align*}
&  \mathbb{E}_{\Pi}\left[  \int_{\mathcal{S}^{N_{\alpha}}}e^{-\Pi(\widetilde{f})}  e^{  -\Pi(h)\left(  \sum
_{j=1}^{N_{\alpha}^{\prime}}w_{j}^{\prime}+w_{\ast}^{\prime}\right) }
\prod_{i=1}^{N_{\alpha}}g_i(w_i,\vartheta_i)\Pi(dw_{i},d\vartheta_{i})\right] \\
&  =\int_{\mathcal{S}^{N_{\alpha}}}\mathbb{E}_{\Pi}\left[  e^{-\Pi(\widetilde{f}%
)-\sum_{i=1}^{N_{\alpha}}\widetilde{f}(w_{i},\vartheta_{i})}  e^{  -\left(  \Pi(h)+\sum_{i=1}^{N_{\alpha}}w_{i}\right)
\left(  \sum_{j=1}^{N_{\alpha}^{\prime}}w_{j}^{\prime}+w_{\ast}^{\prime}\right)
}\right] \\
 &\times\prod_{i=1}^{N_{\alpha}}g_i(w_i,\vartheta_i)\nu(dw_{i},d\vartheta_{i}) \\
&  =\int_{\mathcal{S}^{N_{\alpha}}}\mathbb{E}_{\Pi}\left[  e^{-\Pi(\widetilde{f}%
)-\sum_{i=1}^{N_{\alpha}}\widetilde{f}(w_{i},\vartheta_{i})} e^{  -\Pi(h)\left(
\sum_{j=1}^{N_{\alpha}^{\prime}}w_{j}^{\prime}+w_{\ast}^{\prime}\right)  }
\right]\\
&\times\prod_{i=1}^{N_{\alpha}}g_i(w_i,\vartheta_i)e^{  -w_{i}\left(  \sum_{j=1}^{N_{\alpha}^{\prime}}w_{j}^{\prime
}+w_{\ast}^{\prime}\right)  }\nu(dw_{i},d\vartheta_{i}) \\
&  =\mathbb{E}_{W_{\alpha}}\left[  e^{-W_{\alpha}(f)}e^{  -W_{\alpha}(1)\left(  \sum
_{j=1}^{N_{\alpha}^{\prime}}w_{j}^{\prime}+w_{\ast}^{\prime}\right)  }  \right]\\
&\times\prod_{i=1}^{N_{\alpha}}\int_{\mathcal{S}}\left[  e^{-w_{i}f(\vartheta_{i})}%
w_{i}^{m_{i}}1_{d\theta_i}(\vartheta_i)e^{  -w_{i}\left(  \sum_{j=1}^{N_{\alpha}^{\prime}}w_{j}^{\prime
}+w_{\ast}^{\prime}\right)  }  \nu(dw_{i},d\vartheta_{i})\right]
\end{align*}
The denominator in (\ref{eq:laplacefunctional}) is obtained by taking $f=0$:
\begin{align}
&  \mathbb{E}_{W_{\alpha}}\left[  e^{  -W_{\alpha}(1)\left(  \sum_{j=1}^{N_{\alpha}^{\prime}}%
w_{j}^{\prime}+w_{\ast}^{\prime}\right)  }  \right]  \prod
_{i=1}^{N_{\alpha}}\int_{\mathcal{S}}    w_{i}^{m_{i}}1_{d\theta_i}(\vartheta_i)%
e^{  -w_{i}\left(  \sum_{j=1}^{N_{\alpha}^{\prime}}w_{j}^{\prime}+w_{\ast
}^{\prime}\right)  }  \nu(dw_{i},d\vartheta_{i})
\nonumber\\
&  =e^{  -\alpha\psi\left(  \sum_{j=1}^{N_{\alpha}^{\prime}}w_{j}^{\prime
}+w_{\ast}^{\prime}\right)  }  \alpha^{N_{\alpha}}\prod_{i=1}^{N_{\alpha}}\kappa\left(
m_{i},\sum_{j=1}^{N_{\alpha}^{\prime}}w_{j}^{\prime}+w_{\ast}^{\prime}\right)d\theta_i
\label{eq:marginal}%
\end{align}
where
\[
\kappa(n,z)=\int_{0}^{\infty}w^{n}\exp(-zw)\rho(w)dw
\]

\section{Details on the MCMC\ algorithms}
\label{sec:detailsMCMC}
\subsection{Simple graph}

The undirected graph sampler outlined in Section~\ref{sec:MCMCgraphs} iterates as follows:

\begin{enumerate}
\item Update $w_{1:N_{\alpha}}$ given the rest with Hamiltonian Monte Carlo

\item Update $(\alpha,\sigma,\tau,w_{\ast})$ given the rest using
a Metropolis-Hastings step

\item Update the latent counts $\overline{n}_{ij}$ given the rest using either
the full conditional or a Metropolis-Hastings step
\end{enumerate}

\paragraph{Step 1: Update of $w_{1:N_{\alpha}}$}

We use an Hamiltonian Monte Carlo update for $w_{1:N_{\alpha}}$ via an augmented system with momentum
variables $p$. See \cite{Neal2011} for on overview. Let $L\geq1$ be the number of leapfrog steps and $\varepsilon>0$ the
stepsize. For conciseness, we write $$U^{\prime}(w_{1:N_{\alpha}},w_{\ast},\phi)=\left.
\nabla_{\omega_{1:N_{\alpha}}} \log p(\omega_{1:N_{\alpha}},w_{\ast},\phi|D_{\alpha})\right\vert _{w_{1:N_{\alpha}},w_{\ast},\phi}$$ the gradient of the log-posterior in (\ref{eq:jacobian}). The algorithm proceeds by first sampling momentum variables as
\begin{equation}
p\sim\mathcal{N}(0,I_{N_{\alpha}}).
\end{equation}
The Hamiltonian proposal $q(\widetilde{w}_{1:N_{\alpha}},\widetilde{p}|w_{1:N_{\alpha}},p)$ is
obtained by the following leap\-frog algorithm (for simplicity of exposure, we omit indices $1:N_{\alpha}$). Simulate $L$ steps of the discretized Hamiltonian via
\begin{align*}
\widetilde{p}^{(0)}  &  =p+\frac{\varepsilon}{2} U^{\prime}(w,w_{\ast},\phi)\\
\widetilde{w}^{(0)}  &  =w
\end{align*}
and for $\ell=1,\ldots,L-1$,%
\begin{align*}
\log\widetilde{w}^{(\ell)}  &  =\log\widetilde{w}^{(\ell-1)}+\varepsilon
\widetilde{p}^{(\ell-1)}\\
\widetilde{p}^{(\ell)}  &  =\widetilde{p}^{(\ell-1)}+\varepsilon U^{\prime
}(\widetilde{w}^{(\ell)},w_{\ast},\phi)
\end{align*}
and finally set%
\begin{align*}
\log\widetilde{w}  &  =\log\widetilde{w}^{(L-1)}+\varepsilon
\widetilde{p}^{(L-1)}\\
\widetilde{p}  &  =-\left[  \widetilde{p}^{(L-1)}+\frac{\varepsilon}%
{2}U^{\prime}(\widetilde{w},w_{\ast},\phi)\right] \\
\widetilde{w}  &  =\widetilde{w}^{(L)}.%
\end{align*}

Accept the proposal $(\widetilde{w},\widetilde{p})$ with probability
$\min(1,r)$ with
\begin{align*}
r= &  \frac{\left[  \prod_{i=1}^{N_{\alpha}}\widetilde{w}_{i}^{m_{i}}\right]
\exp\left(  -\left(  \sum_{i=1}^{N_{\alpha}}\widetilde{w}_{i}+w_{\ast}\right)
^{2}\right)  \prod_{i=1}^{N_{\alpha}}\widetilde{w}_{i}\rho(\widetilde{w}_{i})}{\left[  \prod_{i=1}%
^{N_{\alpha}}w_{i}^{m_{i}}\right]  \exp\left(  -\left(  \sum_{i=1}^{N_{\alpha}}w_{i}+w_{\ast
}\right)  ^{2}\right)  \prod_{i=1}^{N_{\alpha}}{w_{i}}{\rho(w_{i})}}e^{  -\frac{1}{2}%
\sum_{i=1}^{N_{\alpha}}\left(  \widetilde{p}_{i}^{2}-p_{i}^{2}\right) }    \\
= &  \left [\prod_{i=1}^{N_{\alpha}}\left(  \frac{\widetilde{w}_{i}}{w_{i}}\right)
^{m_{i}-\sigma}\right ]e^{  -\left(  \sum_{i=1}^{N_{\alpha}}\widetilde{w}_{i}+w_{\ast
}\right)  ^{2}+\left(  \sum_{i=1}^{N_{\alpha}}w_{i}+w_{\ast}\right)  ^{2}-\tau\left(
\sum_{i=1}^{N_{\alpha}}\widetilde{w}_{i}+\sum_{i=1}^{N_{\alpha}}w_{i}\right)  }  \\
&  \times e^{  -\frac{1}{2}\sum_{i=1}^{N_{\alpha}}\left(  \widetilde{p}_{i}%
^{2}-p_{i}^{2}\right)}
\end{align*}

\paragraph{Step 2: Update of $w_{\ast},\alpha,\sigma,\tau$}

For our Metropolis-Hasting step, we propose ($\widetilde{\alpha
},\widetilde{\sigma},\widetilde{\tau}$,$\widetilde{w}_{\ast}$) from
$q(\widetilde{\alpha},\widetilde{\sigma},\widetilde{\tau},\widetilde{w}_{\ast
}|\alpha,\sigma,\tau,w_{\ast})$ and accept with probability $\min(1,r)$ where%
\begin{multline}
r = \frac{e^{  -\left(  \sum_{i=1}^{N_{\alpha}}w_{i}+\widetilde{w}%
_{\ast}\right)  ^{2}}  }{e^{  -\left(  \sum_{i=1}^{N_{\alpha}}%
w_{i}+w_{\ast}\right)  ^{2}}  }\left[  \prod_{i=1}^{N_{\alpha}}\frac{\rho
(w_{i}|\widetilde{\sigma},\widetilde{\tau})}{\rho(w_{i}|\sigma,\tau)}\right]\\
\times\frac{g_{\widetilde{\alpha},\widetilde{\sigma},\widetilde{\tau}}^{\ast
}(\widetilde{w}_{\ast})}{g_{\alpha,\sigma,\tau}^{\ast}(w_{\ast})}\times
\frac{p(\widetilde{\alpha},\widetilde{\sigma},\widetilde{\tau})}%
{p(\alpha,\sigma,\tau)}\times\frac{q(\alpha,\sigma,\tau,w_{\ast}%
|\widetilde{\alpha},\widetilde{\sigma},\widetilde{\tau},\widetilde{w}_{\ast}%
)}{q(\widetilde{\alpha},\widetilde{\sigma},\widetilde{\tau},\widetilde{w}%
_{\ast}|\alpha,\sigma,\tau,w_{\ast})}
\end{multline}

We will use the following proposal%
\[
q(\widetilde{\alpha},\widetilde{\sigma},\widetilde{\tau},\widetilde{w}_{\ast
}|\alpha,\sigma,\tau,w_{\ast})=q(\widetilde{\tau}|\tau)q(\widetilde{\sigma
}|\sigma)q(\widetilde{\alpha}|\widetilde{\sigma},\widetilde{\tau},w_{\ast
})q(\widetilde{w}_{\ast}|\widetilde{\alpha},\widetilde{\sigma},\widetilde{\tau
},\widetilde{w}_{\ast})
\]
where
\begin{align*}
q(\widetilde{\tau}|\tau) &  =\text{lognormal}(\widetilde{\tau};\log
(\tau),\sigma_{\tau}^{2})\\
q(\widetilde{\sigma}|\sigma) &  =\text{lognormal}(1-\widetilde{\sigma}%
;\log(1-\sigma),\sigma_{\tau}^{2})\\
q(\widetilde{\alpha}|\widetilde{\sigma},\widetilde{\tau},w_{\ast}) &
=\Gam\left(  \widetilde{\alpha};N_{\alpha},\frac{(\widetilde{\tau}+2\sum
w_{i}+w_{\ast})^{\widetilde{\sigma}}-\tau^{\widetilde{\sigma}}}%
{\widetilde{\sigma}}\right)  \\
q(\widetilde{w}_{\ast}|\widetilde{\alpha},\widetilde{\sigma},\widetilde{\tau
},\widetilde{w}_{\ast}) &  =g_{\widetilde{\alpha},\widetilde{\sigma
},\widetilde{\tau}+2\sum w_{i}+w_{\ast}}^{\ast}(\widetilde{w}_{\ast})
\end{align*}
The choice of the proposal for $\widetilde{w}_{\ast}$ is motivated by the fact
that it can be written as an exponential tilting of the pdf
$g_{\widetilde{\alpha},\widetilde{\sigma},\widetilde{\tau}}^{\ast
}(\widetilde{w}_{\ast}):$
\[
g_{\widetilde{\alpha},\widetilde{\sigma},\widetilde{\tau}+2\sum w_{i}+w_{\ast
}}^{\ast}(\widetilde{w}_{\ast})=\frac{\exp(-2\sum w_{i}-w_{\ast}%
)g_{\widetilde{\alpha},\widetilde{\sigma},\widetilde{\tau}}^{\ast
}(\widetilde{w}_{\ast})}{\exp(-\psi_{\widetilde{\alpha},\widetilde{\sigma
},\widetilde{\tau}}(w_{\ast}))}%
\]
which will allow the terms involving the intractable pdf $g^{\ast}$ to cancel in the
Metropolis-Hastings ratio. The acceptance probability reduces to having%
\begin{align*}
&r = \frac{e^{  -\left(  \sum_{i=1}^{N_{\alpha}}w_{i}+\widetilde{w}_{\ast}\right)
^{2}}  }{e^{  -\left(  \sum_{i=1}^{N_{\alpha}}w_{i}+w_{\ast}\right)
^{2}}  }  \frac{\widetilde{\alpha}}{\alpha}\frac{\Gamma
(1-\sigma)^{N_{\alpha}}}{\Gamma(1-\widetilde{\sigma})^{N_{\alpha}}}e^{  -(\widetilde{\tau
}-\tau)\sum_{i=1}^{N_{\alpha}}w_{i}}  \left[  \prod_{i=1}^{N_{\alpha}}w_{i}\right]
^{-\widetilde{\sigma}+\sigma}\\
&\quad\quad\times\frac{p(\widetilde{\alpha
},\widetilde{\sigma},\widetilde{\tau})}{p(\alpha,\sigma,\tau)}
\times\frac{\frac{1}{\tau}\frac{1}{1-\sigma}\times\left[  \frac{1}{\sigma
}\left(  (\tau+2\sum w_{i}+\widetilde{w}_{\ast})^{\sigma}-\tau^{\sigma
}\right)  \right]  ^{N_{\alpha}}e^{  -w_{\ast}\left(  2\sum w_{i}+\widetilde{w}%
_{\ast}\right)  }  }{\frac{1}{\widetilde{\tau}}\frac{1}%
{1-\widetilde{\sigma}}\times\left[  \frac{1}{\widetilde{\sigma}}\left(
(\widetilde{\tau}+2\sum w_{i}+w_{\ast})^{\widetilde{\sigma}}-\widetilde{\tau
}^{\widetilde{\sigma}}\right)  \right]  ^{N_{\alpha}}e^{  -\widetilde{w}_{\ast
}\left(  2\sum w_{i}+w_{\ast}\right)  }  }.%
\end{align*}
Finally, if we assume improper priors on $\alpha,\sigma,\tau$%
\[
p(\alpha)\propto\frac{1}{\alpha},p(\sigma)\propto\frac{1}{1-\sigma}%
,p(\tau)\propto\frac{1}{\tau}, %
\]
then
\begin{multline*}
r=
e^{  -\left(  \sum_{i=1}^{N_{\alpha}}w_{i}+\widetilde{w}_{\ast}\right)
^{2}+\left(  \sum_{i=1}^{N_{\alpha}}w_{i}+w_{\ast}\right)  ^{2} }  e^{
-(\widetilde{\tau}-\tau+2w_{\ast}-2\widetilde{w}_{\ast})\sum_{i=1}^{N_{\alpha}}%
w_{i}}  \\
\times\left[  \prod_{i=1}^{N_{\alpha}}w_{i}\right]  ^{-\widetilde{\sigma}+\sigma
}\left[  \frac{\frac{\Gamma(1-\sigma)}{\sigma}\left(  (\tau+2\sum
w_{i}+\widetilde{w}_{\ast})^{\sigma}-\tau^{\sigma}\right)  }{\frac
{\Gamma(1-\widetilde{\sigma})}{\widetilde{\sigma}}\left(  (\widetilde{\tau
}+2\sum w_{i}+w_{\ast})^{\widetilde{\sigma}}-\widetilde{\tau}%
^{\widetilde{\sigma}}\right)  }\right]^{N_{\alpha}}.%
\end{multline*}

\paragraph{Step 3: Update of the latent variables $\overline{n}_{ij}$}

Concerning the latent $\overline{n}_{ij}$, the conditional distribution is a
truncated Poisson distribution (\ref{eq:condn}) from which we can sample
directly. An alternative strategy, which may be more efficient for a large
number of edges, is to use a Metropolis-Hastings proposal:%
\[
q(\widetilde{\overline{n}}_{ij}|\overline{n}_{ij})=\left\{
\begin{array}
[c]{ll}%
\frac{1}{2} & \text{if }\widetilde{n}_{ij}=n_{ij}+1\text{, }n_{ij}>1\\
\frac{1}{2} & \text{if }\widetilde{n}_{ij}=n_{ij}-1\text{, }n_{ij}>1\\
1 & \text{if }\widetilde{n}_{ij}=n_{ij}+1\text{, }n_{ij}=1\\
0 & \text{otherwise}%
\end{array}
\right.
\]
and accept the proposal with probability
\[
\min\left(  1,\frac{\overline{n}_{ij}!}{\widetilde{\overline{n}}_{ij}!}%
((1+\delta_{ij})w_{i}w_{j})^{\widetilde{\overline{n}}_{ij}-\overline{n}_{ij}}\frac
{q(\overline{n}_{ij}|\widetilde{\overline{n}}_{ij})}{q(\widetilde{\overline
{n}}_{ij}|\overline{n}_{ij})}\right).
\]

\subsection{Bipartite graph}

In the bipartite graph case, the sampler iterates as follows:

\begin{enumerate}
\item Propose $(\widetilde{\alpha},\widetilde{\sigma},\widetilde{\tau})\sim
q(\widetilde{\alpha},\widetilde{\sigma},\widetilde{\tau}|\alpha,\sigma,\tau)$
and accept with probability $\min(1,r)$ \ with%
\begin{align*}
r  & =\frac{\exp\left(  -\widetilde{\alpha}\psi_{\widetilde{\sigma
},\widetilde{\tau}}\left(  \sum_{j=1}^{N_{\alpha}^{\prime}}w_{j}^{\prime}+w_{\ast
}^{\prime}\right)  \right)  \widetilde{\alpha}^{N_{\alpha}}\prod_{i=1}^{N_{\alpha}}%
\kappa_{\widetilde{\sigma},\widetilde{\tau}}\left(  m_{i},\sum_{j=1}%
^{N_{\alpha}^{\prime}}w_{j}^{\prime}+w_{\ast}^{\prime}\right)  }{\exp\left(
-\alpha\psi_{\sigma,\tau}\left(  \sum_{j=1}^{N_{\alpha}^{\prime}}w_{j}^{\prime}%
+w_{\ast}^{\prime}\right)  \right)  \alpha^{N_{\alpha}}\prod_{i=1}^{N_{\alpha}}\kappa
_{\sigma,\tau}\left(  m_{i},\sum_{j=1}^{N_{\alpha}^{\prime}}w_{j}^{\prime}+w_{\ast
}^{\prime}\right)  }\\
& \times\frac{p(\widetilde{\alpha})p(\widetilde{\sigma})p(\widetilde{\tau}%
)}{p(\alpha)p(\sigma)p(\tau)}\times\frac{q(\alpha,\sigma,\tau
|\widetilde{\alpha},\widetilde{\sigma},\widetilde{\tau})}{q(\widetilde{\alpha
},\widetilde{\sigma},\widetilde{\tau}|\alpha,\sigma,\tau)}%
\end{align*}

\item For $i=1,\ldots,N_{\alpha}$, sample%
\[
w_{i}|\text{rest}\sim\text{Gamma}\left(  m_{i}-\sigma,\tau+\sum_{j=1}%
^{N_{\alpha}^{\prime}}w_{j}^{\prime}+w_{\ast}^{\prime}\right)
\]

\item Sample%
\[
w_{\ast}|\text{rest}\sim p(w_{\ast}|rest)=\frac{\exp\left(  -w_{\ast}\left(
\sum_{j=1}^{N_{\alpha}^{\prime}}w_{j}^{\prime}+w_{\ast}^{\prime}\right)  \right)
g_{\alpha}(w_{\ast})}{\exp\left[  -\psi\left(  \sum_{j=1}^{N_{\alpha}^{\prime}}%
w_{j}^{\prime}+w_{\ast}^{\prime}\right)  \right]  }%
\]
using the algorithm of \cite{Devroye2009}.

\item Update the latent $n_{ij}$ given $w_{1:N_{\alpha}^{\prime}}^{\prime},w_{1:N_{\alpha}}$
from a truncated Poisson distribution%
\[
n_{ij}|z,w,w^{\prime}\sim\left\{
\begin{array}
[c]{ll}%
\delta_{0} & \text{if }z_{ij}=0\\
\text{tPoisson}(w_{i}w_{j}^{\prime}) & \text{if }z_{ij}=1
\end{array}
\right.
\]

\item Propose $(\widetilde{\alpha}^{\prime},\widetilde{\sigma}^{\prime})\sim
q(\widetilde{\alpha}^{\prime},\widetilde{\sigma}^{\prime}|\alpha^{\prime
},\sigma^{\prime})$ and accept with probability $\min(1,r)$ \ with%
\begin{align*}
r  & =\frac{\exp\left(  -\widetilde{\alpha}^{\prime}\psi_{\widetilde{\sigma
}^{\prime},1}\left(  \sum_{i=1}^{N_{\alpha}}w_{i}+w_{\ast}\right)  \right)
\widetilde{\alpha}^{\prime\ N_{\alpha}}\prod_{j=1}^{N_{\alpha}}\kappa_{\widetilde{\sigma
}^{\prime},1}\left(  m_{j}^{\prime},\sum_{i=1}^{N_{\alpha}}w_{i}+w_{\ast}\right)
}{\exp\left(  -\alpha^{\prime}\psi_{\sigma^{\prime},1}\left(  \sum_{i=1}%
^{N_{\alpha}}w_{i}+w_{\ast}\right)  \right)  \alpha^{\prime\ N_{\alpha}}\prod_{j=1}^{N_{\alpha}}%
\kappa_{\sigma^{\prime},1}\left(  m_{j}^{\prime},\sum_{i=1}^{N_{\alpha}}w_{i}+w_{\ast
}\right)  }\\
& \times\frac{p(\widetilde{\alpha}^{\prime})p(\widetilde{\sigma}^{\prime}%
)}{p(\alpha^{\prime})p(\sigma^{\prime})}\times\frac{q(\alpha^{\prime}%
,\sigma^{\prime}|\widetilde{\alpha}^{\prime},\widetilde{\sigma}^{\prime}%
)}{q(\widetilde{\alpha}^{\prime},\widetilde{\sigma}^{\prime}|\alpha^{\prime
},\sigma^{\prime})}%
\end{align*}

\item For $j=1,\ldots,N_{\alpha}^{\prime}$, sample%
\[
w_{j}^{\prime}|\text{rest}\sim\text{Gamma}\left(  m_{j}^{\prime}-\sigma
,1+\sum_{i=1}^{N_{\alpha}}w_{i}+w_{\ast}\right)
\]

\item Sample%
\[
w_{\ast}^{\prime}|\text{rest}\sim p(w_{\ast}|rest)=\frac{\exp\left(  -w_{\ast
}^{\prime}\left(  \sum_{i=1}^{N_{\alpha}}w_{i}+w_{\ast}\right)  \right)  g_{\alpha
}(w_{\ast}^{\prime})}{\exp\left[  -\psi\left(  \sum_{i=1}^{N_{\alpha}}w_{i}+w_{\ast
}\right)  \right]  }%
\]
using the algorithm of \cite{Devroye2009}.
\end{enumerate}

\bibliographystyle{imsart-nameyear}
\bibliography{bnpnetwork}

\end{document}